\documentclass[a4paper,twocolumn,11pt,accepted=2026-10-07]{quantumarticle}
\pdfoutput=1
\usepackage{hyperref}
\usepackage{graphicx} 
\usepackage{dsfont}
\usepackage{amsthm}
\usepackage{amsmath}
\usepackage{amsfonts}
\usepackage{listings}
\usepackage[dvipsnames]{xcolor}
\usepackage[numbers]{natbib}
\usepackage{tikz,tikz-cd,calc} 

\usetikzlibrary{decorations.pathreplacing,decorations.markings,calc,3d,positioning,quotes}
\usetikzlibrary{shapes.geometric}
\usetikzlibrary{decorations.pathmorphing}
\usetikzlibrary{snakes}

\newcommand\TikzScaling{1.4}
\newcommand\Shift{2.1}

\tikzset{
	tensor/.style={
		inner sep = 0.06cm,
		shape = circle,
		draw,
		fill,
		thick
	},
}

\tikzset{
	tensorempty/.style={
		inner sep = 0.06cm,
		shape = circle,
		draw
	},
}

\tikzset{
	mpo/.style={
		inner sep = 0.1cm,
		shape = rectangle,
		draw,
		fill = violet,
		thick
	},
}

\tikzset{
	circle/.style={
		inner sep = 0.06cm,
		shape = circle,
		draw,
		fill
	},
}

\tikzset{
	square/.style={
		inner sep = 0.1cm,
		shape = rectangle,
		draw,
		fill,
		thick
	},
}

\tikzset{
	rectangle/.style={
		inner sep = 0.18cm,
		shape = rectangle,
		draw,
		fill
	},
}

\tikzset{
	widetensor/.style={
		minimum width= 0.34cm,
		minimum height = 0.22cm,
		inner sep = 0,
		shape = rectangle,
		draw,
		fill
	},
}  
\tikzset{
	simpstensor/.style={
		minimum width= 0.6cm,
		minimum height = 0.25cm,
		inner sep = 0,
		shape = rectangle,
		draw,
		fill = yellow,
		thick
	},
}
\tikzset{
	tritensor/.style={
		inner sep = 0.06cm,
		shape = triangle,
		draw,
		fill,
		thick
	},
}
\tikzset{decoration={snake,amplitude=.4mm,segment length=2mm, post length=0mm,pre length=0mm}}

\tikzset{
	every picture/.style = {
		baseline={([yshift=-.5ex]current bounding box.center)},
	}
}

\tikzset{
	fusion/.style={
		line cap=round,
		line width=1.5mm,
	},
	fusion_dagger/.style={
		line cap=rect,
		line width=1.5mm,
	},
	every picture/.style = {
		baseline={([yshift=-.5ex]current bounding box.center)},
		font=\scriptsize
	},
	irrep/.style={
		anchor=south,
		font = \tiny,
		inner sep=2pt
	},
	tiny/.style={
		font = \tiny,
		inner sep=2pt
	},
	gelL/.style={
		anchor = south east,
		font = \tiny,
		inner sep=2pt
	},
	gelR/.style={
		anchor = south west,
		font = \tiny,
		inner sep=2pt
	}
}

\tikzset{
	pics/fusL/.style args={#1/#2/#3/#4}{
		code = {
			\draw[fusion] (0,0)--(0,#1);
			\node[gelL] at ($0.5*(-0.1,#1)$) {$#2$};
			\node[gelR] at (0.05,#1) {$#3$};
			\node[gelR] at (0.05,0) {$#4$};
	}}
}

\tikzset{
	pics/fusR/.style args={#1/#2/#3/#4}{
		code = {
			\draw[fusion] (0,0)--(0,#1);
			\node[gelR] at ($0.5*(0.1,#1)$) {$#4$};
			\node[gelL] at (-0.05,#1) {$#2$};
			\node[gelL] at (-0.05,0) {$#3$};
	}}
}
\tikzset{
	pics/actL/.style args={#1/#2/#3/#4}{
		code = {
			\draw[fusion, gray] (0,0)--(0,#1);
			\node[gelL] at ($0.5*(-0.1,#1)$) {$#2$};
			\node[gelR] at (0.05,#1) {$#3$};
			\node[gelR] at (0.05,0) {$#4$};
	}}
}

\tikzset{
	pics/actR/.style args={#1/#2/#3/#4}{
		code = {
			\draw[fusion, gray] (0,0)--(0,#1);
			\node[gelR] at ($0.5*(0.1,#1)$) {$#4$};
			\node[gelL] at (-0.05,#1) {$#2$};
			\node[gelL] at (-0.05,0) {$#3$};
	}}
}
\tikzset{
	pics/fusLdag/.style args={#1/#2/#3/#4}{
		code = {
			\draw[fusion_dagger, brightube] (0,0)--(0,#1);
			\node[gelL] at ($0.5*(-0.1,#1)$) {$#2$};
			\node[gelR] at (0.05,#1) {$#3$};
			\node[gelR] at (0.05,0) {$#4$};
	}}
}
\tikzset{
	pics/fusRdag/.style args={#1/#2/#3/#4}{
		code = {
			\draw[fusion_dagger, brightube] (0,0)--(0,#1);
			\node[gelR] at ($0.5*(0.1,#1)$) {$#4$};
			\node[gelL] at (-0.05,#1) {$#2$};
			\node[gelL] at (-0.05,0) {$#3$};
	}}
}

\tikzset{
	ddash/.pic={
		\draw[thick] (-0.075,-0.075) -- (0.075,0);
		\draw[thick] (-0.075,0) -- (0.075,0.075);
	}
}

\tikzset{
	pics/V/.style args={#1/#2/#3}{
		code = {
			\coordinate (-mid) at (-0.4,0);
			\coordinate (-up) at (0.5,0.4);
			\coordinate (-down) at (0.5,-0.4);
			\coordinate (-top) at (0,0.45);
			\coordinate (-bottom) at (0,-0.45);
			\draw[red,thick] (0,0.4)--(-up);
			\draw[red,thick] (0,-0.4)--(-down);
			\draw[fusion] (0,-0.55)--(0,0.55);
			\node[irrep] at (-0.25,0) {$#1$};
			\node[irrep] at (0.25,0.3) {$#2$};
			\node[irrep] at (0.25,-0.3) {$#3$};
		}
	},
	pics/W/.style args={#1/#2/#3}{
		code = {
			\coordinate (-mid) at (0.4,0);
			\coordinate (-up) at (-0.5,0.4);
			\coordinate (-down) at (-0.5,-0.4);
			\coordinate (-top) at (0,0.45);
			\coordinate (-bottom) at (0,-0.45);
			\draw[red,thick] (0,0.4)--(-up);
			\draw[red,thick] (0,-0.4)--(-down);
			\draw[fusion] (0,-0.55)--(0,0.55);
			\node[irrep] at (0.25,0) {$#1$};
			\node[irrep] at (-0.25,0.3) {$#2$};
			\node[irrep] at (-0.25,-0.3) {$#3$};
		}
	}
}

\tikzset{
	pics/V1/.style args={#1/#2/#3}{
		code = {
			\coordinate (-mid) at (-0.4,0);
			\coordinate (-up) at (0.5,0.4);
			\coordinate (-down) at (0.5,-0.6);
			\coordinate (-top) at (0,0.45);
			\coordinate (-bottom) at (0,-0.45);
			\draw[red,thick] (0,0.4)--(-up);
			\draw[thick] (0,-0.6)--(-down);
			\draw[thick] (-0.5,-0.6)--(-down);
			\draw[fusion,gray] (0,-0.55)--(0,0.55);
			\node[irrep] at (-0.25,0) {$#1$};
			\node[irrep] at (0.25,0.45) {$#2$};
			\node[irrep] at (0.25,-0.4) {$#3$};
		}
	},
	pics/W1/.style args={#1/#2/#3}{
		code = {
			\coordinate (-top) at (0,0.45);
			\coordinate (-up) at (-0.5,0.4);   
			\coordinate (-mid) at (0.4,0);
			\coordinate (-down) at (-0.5,-0.6);
			\coordinate (-bottom) at (0,-0.45);
			\draw[red,thick] (0,0.4)--(-up);
			\draw[thick] (0,-0.6)--(-down);
			\draw[thick] (-down) -- (0.5,-0.6);
			\draw[fusion,gray] (0,-0.55)--(0,0.55);
			\node[irrep] at (0.25,0) {$#1$};
			\node[irrep] at (-0.15,0.45) {$#2$};
			\node[irrep] at (-0.25,-0.4) {$#3$};
		}
	},
	pics/AL/.style args={#1}{
		code = {
			\coordinate (-up) at (0.5,0.3);
			\coordinate (-down) at (0.5,-0.3);
			\draw[red,thick] (0,0.3)--(-up);
			\draw[fusion,gray] (0,-0.35)--(0,0.35);
			\node[irrep] at (0.25,0.3) {$#1$};
		}
	},
	pics/AR/.style args={#1}{
		code = {
			\coordinate (-up) at (-0.5,0.3);
			\coordinate (-down) at (-0.5,-0.3);
			\draw[red,thick] (0,0.3)--(-up);
			\draw[fusion,gray] (0,-0.35)--(0,0.35);
			\node[irrep] at (-0.25,0.3) {$#1$};
		}
	},
	pics/ALrot/.style args={#1}{
		code = {
			\coordinate (-up) at (0.5,0.3);
			\coordinate (-down) at (0,-0.3);
			\draw[red,thick] (-0.5,-0.3)--(-down);
			\draw[fusion,gray] (0,-0.35)--(0,0.35);
			\node[irrep] at (-0.2,-0.75) {$#1$};
		}
	},
	pics/FL/.style args={#1/#2/#3}{
		code = {
			\coordinate (-mid) at (-0.4,0);
			\coordinate (-up) at (0.5,0.3);
			\coordinate (-down) at (0.5,-0.3);
			\coordinate (-top) at (0,0.45);
			\coordinate (-bottom) at (0,-0.45);
			\draw[red, thick] (0,0.3)--(-up);
			\draw[red,thick] (0,-0.3)--(-down);
			\draw[red, thick] (0,0)--(-mid);
			\draw[fusion] (0,-0.35)--(0,0.35);
			\node[irrep] at (-0.25,0) {$#1$};
			\node[irrep] at (0.25,0.3) {$#2$};
			\node[irrep] at (0.25,-0.3) {$#3$};
		}
	},
	pics/FR/.style args={#1/#2/#3}{
		code = {
			\coordinate (-mid) at (0.4,0);
			\coordinate (-up) at (-0.5,0.3);
			\coordinate (-down) at (-0.5,-0.3);
			\coordinate (-top) at (0,0.45);
			\coordinate (-bottom) at (0,-0.45);
			\draw[red,thick] (0,0.3)--(-up);
			\draw[red,thick] (0,-0.3)--(-down);
			\draw[red,thick] (0,0)--(-mid);
			\draw[fusion] (0,-0.35)--(0,0.35);
			\node[irrep] at (0.25,0) {$#1$};
			\node[irrep] at (-0.25,0.3) {$#2$};
			\node[irrep] at (-0.25,-0.3) {$#3$};
		}
	},
	pics/FLD/.style args={#1/#2/#3}{
		code = {
			\coordinate (-mid) at (-0.4,0);
			\coordinate (-up) at (0.5,0.3);
			\coordinate (-down) at (0.5,-0.3);
			\coordinate (-top) at (0,0.45);
			\coordinate (-bottom) at (0,-0.45);
			\draw[red, thick] (0,0.3)--(-up);
			\draw[red,thick] (0,-0.3)--(-down);
			\draw[red, thick] (0,0)--(-mid);
			\draw[fusion_dagger, brightube] (0,-0.35)--(0,0.35);
			\node[irrep] at (-0.25,0) {$#1$};
			\node[irrep] at (0.25,0.3) {$#2$};
			\node[irrep] at (0.25,-0.3) {$#3$};
		}
	},
	pics/FRD/.style args={#1/#2/#3}{
		code = {
			\coordinate (-mid) at (0.4,0);
			\coordinate (-up) at (-0.5,0.3);
			\coordinate (-down) at (-0.5,-0.3);
			\coordinate (-top) at (0,0.45);
			\coordinate (-bottom) at (0,-0.45);
			\draw[red,thick] (0,0.3)--(-up);
			\draw[red,thick] (0,-0.3)--(-down);
			\draw[red,thick] (0,0)--(-mid);
			\draw[fusion_dagger, brightube] (0,-0.35)--(0,0.35);
			\node[irrep] at (0.25,0) {$#1$};
			\node[irrep] at (-0.25,0.3) {$#2$};
			\node[irrep] at (-0.25,-0.3) {$#3$};
		}
	},
	pics/ALl/.style args={#1/#2/#3}{
		code = {
			\coordinate (-mid) at (-0.4,0);
			\coordinate (-up) at (0.5,0.3);
			\coordinate (-down) at (0.5,-0.3);
			\draw[red,thick] (0,0.3)--(-up);
			\draw[thick] (0,-0.3)--(-down);
			\draw[fusion,gray] (0,-0.35)--(0,0.35);
			\node[irrep] at (-0.25,-0.3) {$#1$};
			\node[irrep] at (0.25,0.3) {$#2$};
			\node[irrep] at (0.25,-0.3) {$#3$};
		}
	},
	pics/ARl/.style args={#1/#2/#3}{
		code = {
			\coordinate (-mid) at (0.35,0);
			\coordinate (-up) at (-0.5,0.3);
			\coordinate (-down) at (-0.5,-0.3);
			\draw[red,thick] (0,0.3)--(-up);
			\draw[thick] (0,-0.3)--(-down);
			\draw[fusion,gray] (0,-0.35)--(0,0.35);
			\node[irrep] at (0.25,-0.3) {$#1$};
			\node[irrep] at (-0.25,0.3) {$#2$};
			\node[irrep] at (-0.25,-0.3) {$#3$};
		}
	}
}

\tikzset{
	pics/W2/.style args={#1/#2/#3}{
		code = {
			\coordinate (-mid) at (0.4,0);
			\coordinate (-up) at (-0.5,0.4);
			\coordinate (-down) at (-0.5,-1.3);
			\coordinate (-top) at (0,0.45);
			\coordinate (-bottom) at (0,-0.45);
			\draw[red,thick] (0,0.4)--(-up);
			\draw[thick] (-0.5,-1.3)--(-down);
			\draw[thick] (-down) -- (0.5,-1.3);
			\draw[fusion,gray] (0,-1.25)--(0,0.55);
			\node[irrep] at (0.25,0) {$#1$};
			\node[irrep] at (-0.25,0.3) {$#2$};
			\node[irrep] at (-0.25,-0.3) {$#3$};
		}
	}
}

\tikzset{
	pics/ALe/.style args={#1}{
		code = {
			\coordinate (-up) at (0.5,0.3);
			\coordinate (-down) at (0.5,-0.3);
			\draw[fusion,gray] (0,-0.35)--(0,0.35);
			\node[irrep] at (0.25,0.3) {$#1$};
		}
	}
}

\tikzset{
	pics/ALs/.style args={#1}{
		code = {
			\coordinate (-up) at (0.2,0.3);
			\coordinate (-down) at (0.5,-0.3);
			\draw[red,thick] (0,0.3)--(-up);
			\draw[fusion,gray] (0,-0.35)--(0,0.35);
			\node[irrep] at (0.25,0.3) {$#1$};
		}
	}
}

\newcommand{\tr}{\mathrm{tr}}

\newcommand{\ket}[1]{|#1\rangle}
\newcommand{\bra}[1]{\langle#1|}

\newcommand{\ketbra}[2]{|#1\rangle\hspace{-0.5mm}\langle#2|}
\newcommand{\mc}[1]{\mathcal{#1}}
\newcommand{\Z}{\mathbb{Z}}
\newcommand{\R}{\mathbb{R}}
\newcommand{\C}{\mathbb{C}}

\newcommand{\hilb}{\mathcal{H}}
\newcommand{\mcG}{\mathcal{G}}
\newcommand{\ha}{\frac{1}{2}}

\newcommand{\ammaG}{\text{\reflectbox{$\Gamma$}}}
\newcommand{\id}{\mathrm{id}}

\newtheorem{proposition}{Proposition}
\newtheorem{corollary}{Corollary}
\newtheorem{lemma}{Lemma}
\newtheorem{theorem}{Theorem}

\definecolor{wisteria}{rgb}{0.79, 0.63, 0.86}
\definecolor{xanadu}{rgb}{0.45, 0.53, 0.47}
\definecolor{classicrose}{rgb}{0.98, 0.8, 0.91}
\definecolor{celadon}{rgb}{0.67, 0.88, 0.69}
\definecolor{carrotorange}{rgb}{0.93, 0.57, 0.13}
\definecolor{citrine}{rgb}{0.89, 0.82, 0.04}
\definecolor{antiquewhite}{rgb}{0.98, 0.92, 0.84}
\definecolor{brightube}{rgb}{0.82, 0.62, 0.91}

\usepackage{scalerel}

\newcommand\inv[1]{#1^{\text{-}1}}

\newcommand{\I}{\mathds{1}}

\newtheorem{appxprop}{Proposition}[section]

\begin{document}
	
	\title{Symmetry defects and gauging for quantum states with matrix product unitary symmetries}
	\author{Adrián Franco Rubio}
	\email{afrancorubio@ucm.es}
	\affiliation{Max-Planck-Institut f{\"u}r Quantenoptik,Hans-Kopfermann-Str.~1, Garching, 85748, Germany}
	\affiliation{Munich Center for Quantum Science and Technology, Schellingstraße~4, M{\"u}nchen, 80799, Germany}
	\affiliation{University of Vienna, Faculty of Physics, Boltzmanngasse 5, 1090 Wien, Austria}
	\orcid{0000-0002-9034-7125}
	\author{Arkadiusz Bochniak}
	\email{arkadiusz.bochniak@mpq.mpg.de}
	\affiliation{Max-Planck-Institut f{\"u}r Quantenoptik,Hans-Kopfermann-Str.~1, Garching, 85748, Germany}
	\affiliation{Munich Center for Quantum Science and Technology, Schellingstraße~4, M{\"u}nchen, 80799, Germany}
	\orcid{0000-0001-7856-0452}
	\author{J. Ignacio Cirac}
	\email{ignacio.cirac@mpq.mpg.de}
	\affiliation{Max-Planck-Institut f{\"u}r Quantenoptik,Hans-Kopfermann-Str.~1, Garching, 85748, Germany}
	\affiliation{Munich Center for Quantum Science and Technology, Schellingstraße~4, M{\"u}nchen, 80799, Germany}
	\orcid{0000-0003-3359-1743}
	\date{February 28th, 2025}
	\maketitle
	\begin{abstract}
		In this work, we examine the consequences of the existence of a finite group of matrix product unitary (MPU) symmetries for matrix product states (MPS). We generalize the well-understood picture of onsite unitary symmetries, which give rise to virtual symmetry defects given by insertions of operators in the bonds of the MPS. In the MPU case, we can define analogous defect tensors, this time sitting on lattice sites, that can be created, moved, and fused by local unitary operators. We leverage this formalism to study the gauging of MPU symmetries. We introduce a condition, \textit{block independence}, under which we can gauge the symmetries by promoting the symmetry defects to gauge degrees of freedom, yielding an MPS of the same bond dimension that supports a local version of the symmetry given by commuting gauge constraints. Whenever block independence does not hold (which happens, in particular, whenever the symmetry representation is anomalous), a modification of our method which we call \textit{state-level gauging} still gives rise to a locally symmetric MPS by promotion of the symmetry defects, at the expense of producing gauge constraints that do not commute on different sites.
	\end{abstract}
	
	\section{Introduction}
	
	The concept of symmetry has provided crucial guidance for many developments in theoretical physics, particularly for the study of quantum many-body systems. The symmetry properties of a quantum many-body Hamiltonian can have an impact on the structure of its ground states, as manifested by the Lieb-Schultz-Mattis-type theorems \cite{Lieb61, Oshikawa97, Tasaki_book, Fuji16, Po17, Ogata21, Tasaki22}. Symmetries and symmetry breaking are also at the center of Landau's paradigm of phase transitions \cite{Landau}, a formalism that had to be generalized to include more exotic symmetries, like non-invertible \cite{Bhardwaj24, Bhardwaj2024a, Bhardwaj25} or higher-form symmetries \cite{Iqbal22, Liu24}, to account for newly discovered phases such as topological orders \cite{Wen90}.
	
	In lattice quantum systems, symmetry can be implemented in different ways at the microscopic level. While some symmetries act independently on each site (\textit{onsite symmetries}), more general implementations are possible, where the symmetry action is distributed across several sites in a correlated way. Importantly, such symmetries still behave as ordinary global internal symmetries at long distances, but their microscopic realization can be significantly more structured. These appear, for instance, as the symmetries of the boundary theory of nontrivial SPT phases \cite{Chen11, Else:2014vma}.
	
	In recent decades, tensor networks have emerged as a very relevant toolbox for quantum many-body physics, not only as a means of reducing computational costs but also as a powerful language to represent the entanglement structure of many-body ground and excited states \cite{Review, TN_book}. Global symmetries of a tensor network state give rise to symmetry actions on the entanglement degrees of freedom (the contracted legs of the tensor network) represented locally on a single tensor. This fact led, for example, to the classification of one-dimensional (1d) SPT phases, which are, informally, labeled by such virtual symmetry representations \cite{Chen11a, Schuch11, string_order}. An insertion of such a virtual group element can also be interpreted as a symmetry defect or twist, which will be very useful in what follows. From a physical perspective, these virtual symmetry actions and defects provide a unified way of describing both conventional onsite symmetries and more general non-onsite realizations, which may arise, for example, in duality transformations or in effective descriptions of lattice gauge theories.
	
	A key theoretical mechanism associated with symmetries is that of gauging, or the promotion of global to {\it local} symmetries, where they stop being considered physical symmetries and are rather seen as redundancies in the description of the system \cite{Straumann00}. This procedure has found application in the standard model of particle physics \cite{YangMills54, Wu75, Englert64, Higgs64}, which is a gauge theory, as well as in many effective models of condensed matter systems \cite{Anderson63, Kleinert, Fr23}. Lattice gauge theories arise when this procedure is applied to a globally symmetric many-body Hamiltonian \cite{Kogut75}; however, one can also gauge quantum states directly. This was pioneered for onsite symmetries by Haegeman et al. in \cite{PhysRevX.5.011024}, where a gauging map is defined by tensoring the system with additional (gauge) degrees of freedom and projecting onto a gauge-invariant subspace. It can be seen that this projector has a natural tensor network structure. In 1d, it is a matrix product operator (MPO), and thus when applied on a matrix product state (MPS) it produces another MPS, with a potentially higher bond dimension. 
	
	Moving on to non-onsite symmetries, Garre and Kull \cite{GarreKull} proposed a way to gauge nonanomalous group MPO algebras, and certain (trivial) gauging for anomalous ones, which always gives rise to a product state. This procedure does not trivially extend to MPU groups whose tensors satisfy weaker constraints than the MPO algebras, as we will elucidate below.
	
	As opposed to this \textit{projection} gauging picture, where the gauging is effected in a ``top-down'' manner by applying a family of projectors onto the gauge invariant subspace, there exists a ``bottom-up approach'' where the gauging is accomplished by the \textit{promotion} of the symmetry defects associated to the global symmetry to gauge degrees of freedom, in such a way that the resulting state has local symmetry, i.e.~is gauge invariant \cite{Barkeshli19}. The tensor network states arising from such a construction for an onsite symmetry were classified in \cite{KULL2017199} for 1d systems and in \cite{Blanik_2025} for general graphs. This defect promotion picture is equivalent to the projection picture as long as the resulting gauge-invariant subspaces are the same. Recently, in \cite{Seifnashri}, a systematic procedure was proposed to gauge non-onsite symmetries at the Hamiltonian level, based on promoting symmetry defects to gauge degrees of freedom. As we shall see in the following, this approach is remarkably natural at the level of tensor networks, where these defects can be easily represented. 
	
	In this work, we first establish a symmetry defect formalism for translation-invariant 1d matrix product states (MPS) with periodic boundary conditions (PBC) that are invariant under a matrix product unitary (MPU) representation of a finite symmetry group. Then, inspired by \cite{Seifnashri}, we provide a gauging construction that promotes these defects to gauge degrees of freedom, resulting in an MPS of the same bond dimension. This construction works whenever the initial MPU-MPS system satisfies a condition called \textit{block independence}, which is group-cohomological and gives rise to a gauge-invariant subspace defined by commuting Gauss law projectors. Whenever block independence is not satisfied, a related construction, which we call \textit{state-level gauging}, still generates an MPS with the same bond dimension that is locally invariant, however, the associated Gauss law projectors need no longer commute; hence, this construction departs from the traditional notion of gauging.
	
	\section{Setup and summary of results}
	\subsection{MPSs and MPUs}
	We will work with matrix product states (MPS) with periodic boundary conditions. A translation-invariant MPS is defined by a collection of matrices $A = \{A^1, \ldots, A^d\}\subset M_D(\C)$, where $d$ is the dimension of the local Hilbert space (also called the \textit{physical} dimension) and $D$ is the \textit{bond} dimension, the dimension of the auxiliary Hilbert space, and a parameter that controls how much entanglement content an MPS can display:
	\begin{equation}
		\ket{\psi[A]} = \!\!\!\!\sum_{i_1,\ldots,i_N}\!\!\!\!{\tr{\left(A^{i_1}\ldots A^{i_N}\right)}\ket{i_1\ldots i_N}}.
		\label{eq:introMPS}
	\end{equation}
	The set of MPS with bond dimension $D=1$ is exactly that of the product states. Tensor networks have a powerful diagrammatic language that we will use constantly in this work. The state from Eq. \eqref{eq:introMPS} is represented as
	\begin{equation}\ket{\psi[A]} =
		\begin{tikzpicture}[baseline]
			\def\dx{0.5}
			\def\dy{0.5}
			\def\dd{0.2}
			\draw[thick] (-\dd, 0) -- (5*\dx+\dd,0);
			\foreach \x in {0, \dx, 2*\dx, 3*\dx, 4*\dx, 5*\dx}{
				\draw[thick] (\x, 0) -- (\x, \dy);
			}
			\foreach \x in {0, 2*\dx, 4*\dx}{
				\node[tensor] at (\x, 0) {};
			}
			\foreach \x in {\dx, 3*\dx, 5*\dx}{
				\node[tensor] at (\x, 0) {};
			}
			\node[] at (-\dd+6.25*\dx, 0) {$\ldots$}; 
			\node[] at (-2.5*\dd, 0) {$\ldots$}; 
			\node[] at (-\dd+6.75*\dx,0) {$,$};
		\end{tikzpicture}
	\end{equation}
	where the wavefunction arises from the contraction of three-legged tensors. We can allow the tensor to be distinct at different sites, giving rise to translationally noninvariant states. For a review of the basic tensor network language, we refer the reader to \cite{Review}.
	
	A crucial property an MPS tensor can have is \textit{injectivity}, which is defined as the map
	\begin{equation}
		X\longmapsto \sum_{i}{\tr\left(X A^i\right)\ket{i}}
	\end{equation}
	being injective. This amounts to the existence of an inverse tensor,
	\begin{equation}
		\begin{tikzpicture}[baseline=0.25cm]
			\def\dx{0.3}
			\def\dy{0.6}
			\draw[thick] (-\dx, 0) -- (\dx,0);
			\draw[thick] (-\dx, \dy) -- (\dx,\dy);
			\draw[thick] (0, 0) -- (0,\dy);
			\node[tensor] at (0, 0) {};
			\node[tensor, gray] at (0, \dy) {};
		\end{tikzpicture}
		=
		\begin{tikzpicture}[baseline=0.25cm]
			\def\dx{0.3}
			\def\ddx{0.1}
			\def\dy{0.6}
			\draw[thick] (-\dx, 0) -- (-\ddx,0) --  (-\ddx,\dy) -- (-\dx, \dy);
			\draw[thick] (\dx, 0) -- (\ddx,0) --  (\ddx,\dy) -- (\dx, \dy);
		\end{tikzpicture},
	\end{equation}
	which is a useful property in many proofs. An MPS tensor is \textit{normal} if it becomes injective after blocking together a finite number of copies, effectively enlarging the unit cell,
	\begin{equation}
		\begin{tikzpicture}[baseline]
			
			\def\dy{0.5}
			\def\dd{0.05}
			\def\dx{0.3}
			\draw[thick] (-\dx, 0) -- (\dx,0);
			\draw[thick] (\dd, 0) -- (\dd, \dy);
			\draw[thick] (-\dd, 0) -- (-\dd, \dy);
			\node[widetensor] at (0, 0) {};
		\end{tikzpicture}
		\equiv
		\begin{tikzpicture}[baseline]
			\def\dx{0.5}
			\def\dy{0.5}
			\def\dd{0.2}
			\draw[thick] (-\dd, 0) -- (\dx+\dd,0);
			\foreach \x in {0, \dx}{
				\draw[thick] (\x, 0) -- (\x, \dy);
				\node[tensor] at (\x, 0) {};
			}
		\end{tikzpicture}.
	\end{equation}
	It is known \cite{Cirac17} that any MPS tensor can, potentially after blocking, be written in {\it block-injective form}. In such a form, the MPS tensor has a block structure,
	\begin{equation}
		A^i=\bigoplus\limits_{k=1}^r \mu_k A_k^i
		\label{eq:block_inj}
	\end{equation}
	with $\mu_k\in\mathbb{C}$ and $A_k\in M_{D_k}(\C)$ injective MPS tensors \footnote{An MPS tensor that has the form \eqref{eq:block_inj} with $A_k$ normal tensors is said to be in \textit{canonical form}. Any MPS tensor can, potentially after blocking, be put in canonical form. Further blocking, if necessary, will then yield a block-injective tensor. Beware of the terminology coincidence of the \textit{blocking} operation, putting tensors together, and the \textit{blocks} that make up the matrices of a tensor in canonical or in block-injective form.}. If only one block is present ($r=1$), the tensor is called injective. Otherwise, it is noninjective. 
	
	Similarly to MPS, we can define operators with a tensor network structure. In one dimension, these are called matrix product operators (MPO):
	\begin{equation}
	\begin{split}
		\mc O[M]\!=\!\!\!\!\! \sum_{i_1,\ldots,i_N}\!\!
		&\tr\left(M^{i_1, j_1}\ldots M^{i_N, j_N}\right)\\
		&\times\ket{i_1\ldots i_N}\bra{j_1, \ldots, j_N}
	\end{split}    
	\label{eq:introMPO}    
\end{equation}
	and are represented graphically by 
	\begin{equation}
		\begin{tikzpicture}[baseline]
			\def\dx{0.5}
			\def\dy{0.3}
			\def\dd{0.2}
			\draw[thick,red] (-\dd, 0) -- (5*\dx+\dd,0);
			\foreach \x in {0, \dx, 2*\dx, 3*\dx, 4*\dx, 5*\dx}{
				\draw[thick] (\x, -\dy) -- (\x, \dy);
			}
			\foreach \x in {0, 2*\dx, 4*\dx}{
				\node[mpo] at (\x, 0) {};
			}
			\foreach \x in {\dx, 3*\dx, 5*\dx}{
				\node[mpo] at (\x, 0) {};
			}
			\node[] at (-\dd+6.25*\dx, 0) {$\ldots$}; 
			\node[] at (-2.5*\dd, 0) {$\ldots$}; 
			\node[] at (-\dd+6.75*\dx,0) {$\cdot$};
		\end{tikzpicture}
	\end{equation}
	If the tensor $M$ is such that $\mc{O}[M]$ is unitary for any system size $N$, the operator is called a matrix product unitary (MPU). An MPU of bond dimension $D=1$ corresponds to the tensor product of onsite unitary operators, whereas generally, for $D>1$ the action of an MPU is non-onsite. However, it is known that the set of MPUs is exactly that of quantum cellular automata in one dimension, i.e., the set of unitary actions that preserve locality \cite{MPU}, and thus constitutes a physically motivated generalization of onsite unitary symmetries.
	
	\subsection{Defects and gauging: motivating example and summary of results}
	\label{sec:onsite}
	Matrix product states provide a particularly useful formalism for working with symmetry defects and gauging for onsite symmetries. We review this in the simplest case, which we intend to generalize to non-onsite symmetries. Consider a symmetry group $G$ (for simplicity $|G|<\infty$) represented unitarily onsite,
	\begin{equation}
		U_g = \bigotimes_{j}{u^j_g},\qquad g\in G,
	\end{equation}
	and assume we are given an invariant injective MPS. Because of the fundamental theorem of MPS (see, e.g., \cite{Review} and references therein), this implies the existence of a \textit{virtual} unitary representation of $G$, $S_g$, such that
	\begin{equation}
		\begin{tikzpicture}[scale=\TikzScaling]
			\node[tensor] at (0,0) {};
			\draw[thick] (0,0.0) --++ (0,0.6);
			\draw[thick] (-0.7, 0) --++ (1.4,0);
			\node[square, fill=violet] at(0,0.4) {};
			\node[color = violet] at(-0.3,0.4) {$u_g$};
			
			\node[thick] at (1.0,0) {$=$};
			
			\node[tensor] at (0+\Shift,0) {};
			\draw[thick] (0+\Shift,0) --++ (0,0.6);
			\draw[thick] (-0.75+\Shift, 0) --++ (1.5,0);
			\node[square, fill = blue] at (-0.5+\Shift, 0) {};
			\node[square, fill = blue] at (0.5+\Shift, 0) {};
			\node[color = blue] at (-0.5+\Shift, 0.25) {$S^\dagger_g$};
			\node[color = blue] at (0.5+\Shift, 0.25) {$S_g^{\phantom{\dagger}}$}; 
			
			\node[] at (3,0) {$.$};
		\end{tikzpicture}
	\end{equation}
	
	Note that $S_g$ is potentially a projective representation, that is, there exists a 2-cocycle $\Omega: G\times G\to U(1)$ such that $S_g S_h = \Omega(g,h) S_{gh}$. We can identify the presence of a symmetry defect on a particular bond with the insertion of a matrix $S_g$ in the virtual indices (we use the scalar freedom in $S_g$ to ensure that $S_{\inv{g}} = S_g^{-1} = S_g^\dagger$). Therefore, the action of $u_g$ on a single site creates a pair of defects $(g^{-1}, g)$ on the adjacent bonds. More generally, the action of the symmetry restricted to an interval will generate defects at the boundaries of the said interval:
	\begin{equation}
		\begin{tikzpicture}[scale=\TikzScaling]
			\node[tensor] at (0,0) {};
			\node[tensor] at (0.3,0) {};
			\node[tensor] at (0.95,0) {};
			\node[tensor] at (1.25,0) {};
			
			\draw[thick] (0,0.0) --++ (0,0.6);
			\draw[thick] (0.3,0.0) --++ (0,0.6);
			\draw[thick] (0.95,0.0) --++ (0,0.6);
			\draw[thick] (1.25,0.0) --++ (0,0.6);
			
			\draw[thick] (-0.3, 0) --++ (0.8,0);
			\node[] at (0.65,0) {$\ldots$};
			
			\node[square, fill=violet] at(0,0.4) {};
			\node[color = violet] at(-0.3,0.4) {$u_g^i$};
			\node[square, fill=violet] at(0.3,0.4) {};
			\node[square, fill=violet] at(1.25,0.4) {};
			\node[color = violet] at(1.55,0.4) {$u_g^j$};
			\node[square, fill=violet] at(0.95,0.4) {};
			\draw[thick] (0.8, 0) --++ (0.75,0);
			
			\node[thick] at (1.9,0) {$=$};
			\def\Sc{1.35};
			
			\node[tensor] at (0+\Sc*\Shift,0) {};
			\node[tensor] at (0.3+\Sc*\Shift,0) {};
			
			\draw[thick] (0+\Sc*\Shift,0) --++ (0,0.6);
			\draw[thick] (0.3+\Sc*\Shift,0) --++ (0,0.6); 
			
			\draw[thick] (-0.6+\Sc*\Shift, 0) --++ (1.1,0);
			\node[] at (0.65+\Sc*\Shift,0) {$\ldots$};
			
			\node[tensor] at (0.95+\Sc*\Shift,0) {};
			\node[tensor] at (1.25+\Sc*\Shift,0) {};
			
			\draw[thick] (0.95+\Sc*\Shift,0) --++ (0,0.6);
			\draw[thick] (1.25+\Sc*\Shift,0) --++ (0,0.6); 
			\draw[thick] (0.8+\Sc*\Shift, 0) --++ (1.1,0);
			
			\node[square, fill = blue] at (-0.3+\Sc*\Shift, 0) {};
			\node[square, fill = blue] at (1.55+\Sc*\Shift, 0) {};
			
			\node[color = blue] at (-0.3+\Sc*\Shift, 0.25) {$S^\dagger_g$};
			\node[color = blue] at (1.6+\Sc*\Shift, 0.25) {$S_g^{\phantom{\dagger}}$}; 
			
			\node[] at (0+\Sc*\Shift,-0.2) {$i$};
			\node[] at (1.25+\Sc*\Shift,-0.2) {$j$};
			
			\node[] at (4.85,0) {$.$};
		\end{tikzpicture}
		\label{eq:interval}
	\end{equation}
	Note that the action of a single symmetry operator can move the defects:
	\begin{equation}
		\begin{tikzpicture}[scale=\TikzScaling] 
			\node[tensor] at (0,0) {};
			\draw[thick] (0,0) --++ (0,0.6);
			\draw[thick] (-0.3, 0) --++ (1.6,0);
			\node[square, fill = blue] at (0.5, 0) {};
			\node[color = blue] at (0.5, 0.25) {$S_g^{\phantom{\dagger}}$}; 
			\node[tensor] at (1,0) {};
			\draw[thick] (1,0) --++ (0,0.6);
			\node[square, fill=violet] at(1,0.4) {};
			\node[color = violet] at(1.4,0.4) {$u_g^{j+1}$};
			\node[] at (0,-0.2) {$j$};
			\node[] at (1,-0.2) {$j+1$};
			
			\node[thick] at (1.6,0) {$=$};
			\def\Sc{1};
			
			\node[tensor] at (0+\Sc*\Shift,0) {}; 
			\draw[thick] (0+\Sc*\Shift,0) --++ (0,0.6);
			\draw[thick] (-0.3+\Sc*\Shift, 0) --++ (1.4,0);
			\node[tensor] at (0.4+\Sc*\Shift,0) {}; 
			\draw[thick] (0.4+\Sc*\Shift,0) --++ (0,0.6);
			
			\node[square, fill = blue] at (0.75+\Sc*\Shift, 0) {};
			\node[color = blue] at (0.75+\Sc*\Shift, 0.25) {$S_g^{\phantom{\dagger}}$}; 
			\node[] at (0+\Sc*\Shift,-0.2) {$j$};
			\node[] at (0.4+\Sc*\Shift,-0.2) {$j+1$};
			
			\node[] at (3.3,0) {$.$};
		\end{tikzpicture}
	\end{equation}
	
	Also, two defects can be fused by the action of a symmetry operator, potentially resulting in a phase due to the projectiveness of the representation:
	\begin{equation}
		\begin{tikzpicture}[scale=\TikzScaling] 
			\node[tensor] at (0,0) {};
			\draw[thick] (0,0) --++ (0,0.6);
			\draw[thick] (-0.3, 0) --++ (2.1,0);
			\node[square, fill = blue] at (0.5, 0) {};
			\node[color = blue] at (0.5, 0.25) {$S_g^{\phantom{\dagger}}$}; 
			\node[tensor] at (1,0) {};
			\draw[thick] (1,0) --++ (0,0.6);
			\node[square, fill=violet] at(1,0.4) {};
			\node[color = violet] at (0.75,0.4) {$u_g^{j}$};
			\node[] at (0,-0.2) {$j-1$};
			\node[] at (1,-0.2) {$j$};
			\node[square, fill = blue] at (1.5, 0) {};
			\node[color = blue] at (1.5, 0.25) {$S_h^{\phantom{\dagger}}$};

			\node[thick] at (2.1,0) {$=$};
			\def\Sc{1.6};
			
			\node[] at (-0.75+\Sc*\Shift,0) {$\Omega(g,h)$};
			\node[tensor] at (0+\Sc*\Shift,0) {}; 
			\draw[thick] (0+\Sc*\Shift,0) --++ (0,0.6);
			\draw[thick] (-0.3+\Sc*\Shift, 0) --++ (1.4,0);
			\node[tensor] at (0.4+\Sc*\Shift,0) {}; 
			\draw[thick] (0.4+\Sc*\Shift,0) --++ (0,0.6);
			
			\node[square, fill = blue] at (0.75+\Sc*\Shift, 0) {};
			\node[color = blue] at (0.75+\Sc*\Shift, 0.25) {$S_{gh}^{\phantom{\dagger}}$}; 
			\node[] at (0+\Sc*\Shift,-0.2) {$j-1$};
			\node[] at (0.4+\Sc*\Shift,-0.2) {$j$};
			
			\node[] at (1.2+\Sc*\Shift,0) {$.$};
		\end{tikzpicture}
	\end{equation}
	Thus, the action of the symmetry can now be reinterpreted as the result of the following sequence of unitary operations: create a pair of defects (with a creation operator), move one of them around the chain (with movement operators), and fuse them back together to the identity defect (with a fusion operator). We refer to this structure as a \textit{defect system} on the MPS. The nomenclature of movement and fusion operators is inspired by that of \cite{Seifnashri}, where defects are defined as localized modifications of a Hamiltonian.
	
	A gauging procedure can now be stated as promoting the gauge defects to quantum degrees of freedom. Consider a new MPS tensor defined by
	\begin{equation}
		\begin{tikzpicture}[scale=\TikzScaling]
			\draw[thick] (0,0.0) --++ (0,0.5);
			\draw[thick] (-0.5, 0) --++ (1.0,0);
			\node[square, fill=blue] at (0,0) {};
			\node[circle] at (0,0.5) {};
			\node[] at (0.3,0.5) {$\ket{g}$};
			
			\node[thick] at (0.8,0) {$\equiv$};
			\def\Sc{0.75};
			
			\node[color = white] at (0.3+\Sc*\Shift,0.5) {$\ket{g}$};
			\draw[thick] (-0.5+\Sc*\Shift, 0) --++ (1.0,0);
			\node[square, fill=blue] at (0+\Sc*\Shift,0) {};
			\node[color = blue] at (0+\Sc*\Shift, 0.25) {$S_g$};
			\node[] at (2.3,0) {$,$};
		\end{tikzpicture}
		\label{eq:gaugedoftensor}
	\end{equation}
	and contract it alternatively with the original tensor, giving rise to a two-site unit cell MPS, with matter degrees of freedom on sites and gauge degrees of freedom on bonds:
	\begin{equation}
		\begin{tikzpicture}[baseline]
			\def\dx{0.5}
			\def\dy{0.5}
			\def\dd{0.2}
			\draw[thick] (-\dd, 0) -- (5*\dx+\dd,0);
			\foreach \x in {0, \dx, 2*\dx, 3*\dx, 4*\dx, 5*\dx}{
				\draw[thick] (\x, 0) -- (\x, \dy);
			}
			\foreach \x in {0, 2*\dx, 4*\dx}{
				\node[tensor] at (\x, 0) {};
			}
			\foreach \x in {\dx, 3*\dx, 5*\dx}{
				\node[square, fill=blue] at (\x, 0) {};
			}
			\node[] at (-\dd+6.25*\dx, 0) {$\ldots$}; 
			\node[] at (-2.5*\dd, 0) {$\ldots$}; 
			\node[] at (-\dd+6.75*\dx,0) {$\cdot$};
		\end{tikzpicture}
		\label{eq:gauged_MPS}
	\end{equation}
	Define the corresponding local Hilbert spaces for matter $\hilb_j\cong \C^d$, and gauge degrees of freedom $\hilb^G_{(j,j+1)}\cong \C G$. We can then define a local unitary representation of the gauge group,
	\begin{align}
		\mcG^j:G&\longrightarrow   \mc B\bigl(\hilb^G_{(j-1,j)}\otimes\hilb_j \otimes \hilb^G_{(j,j+1)}\bigr)\nonumber\\
		g&\longmapsto R_g^\Omega\otimes u_g \otimes L_g^\Omega,
		\label{eq:localsyms}
	\end{align}
	where we have used the left and right regular representations of $G$ on $\C G$, twisted by the cocycle $\Omega$ that labels the MPS symmetry-protected topological (SPT) class:
	\begin{equation}
		\begin{split}
			&L_g^\Omega\equiv \sum_{h\in G}{\dfrac{1}{\Omega(g, h)}\ketbra{gh}{h}},\\
			&R_g^\Omega\equiv \sum_{h\in G}{\dfrac{1}{\Omega(h, g^{-1})}\ketbra{hg^{-1}}{h}}.
		\end{split}
	\end{equation}
	It can be seen that the newly defined MPS tensors for the gauge degrees of freedom transform nicely under the action of $L^\Omega_g, R^\Omega_g$,
	\begin{equation}
		\begin{tikzpicture}[scale=\TikzScaling]
			\draw[thick] (0,0.0) --++ (0,0.6);
			\draw[thick] (-0.5, 0) --++ (1.,0);
			\node[square, fill = blue] at (0,0) {};
			\node[square, fill=violet] at(0,0.4) {};
			\node[color = violet] at(-0.3,0.4) {$L^\Omega_g$};
			
			\def\Sc{0.85};
			\def\co{1.7};
			\def\coo{2.4};
			\node[] at (0.75,0.25) {$=$};
			
			\draw[thick] (0+\Sc*\Shift,0) --++ (0,0.6);
			\draw[thick] (-0.75+\Sc*\Shift, 0) -- (0.5+\Sc*\Shift,0);
			\node[square, fill = blue] at (0+\Sc*\Shift,0) {};
			\node[square, fill = blue] at (-0.5+\Sc*\Shift, 0) {};
			\node[color = blue] at (-0.5+\Sc*\Shift, 0.25) {$S^\dagger_g$};
			
			\node[] at (0.6+\Sc*\Shift,0) {$,$};
			
			\draw[thick] (0+\co*\Sc*\Shift,0.0) --++ (0,0.6);
			\draw[thick] (-0.5+\co*\Sc*\Shift, 0) --++ (1.,0);
			\node[square, fill = blue] at (0+\co*\Sc*\Shift,0) {};
			\node[square, fill=violet] at(0+\co*\Sc*\Shift,0.4) {};
			\node[color = violet] at(-0.3+\co*\Sc*\Shift,0.4) {$R^\Omega_g$};
			
			\node[] at (0.75+\co*\Sc*\Shift,0.25) {$=$};
			
			\draw[thick] (0+\coo*\Sc*\Shift,0) --++ (0,0.6);
			\draw[thick] (-0.5+\coo*\Sc*\Shift, 0) -- (0.75+\coo*\Sc*\Shift,0);
			\node[square, fill = blue] at (0+\coo*\Sc*\Shift,0) {};
			\node[square, fill = blue] at (0.5+\coo*\Sc*\Shift, 0) {};
			\node[color = blue] at (0.5+\coo*\Sc*\Shift, 0.25) {$S_g^{\phantom{\dagger}}$}; 
			\node[] at (0.85+\coo*\Sc*\Shift, 0.0) {$,$};
		\end{tikzpicture}
	\end{equation}
	and thus combining all transformation properties of the tensors involved we have the local invariance of the gauged MPS \eqref{eq:gauged_MPS} under the local symmetries $\mcG^j(g)$ from \eqref{eq:localsyms}. It can be checked that the local symmetry operators at different sites commute (a consequence of the commutativity of the left- and right-group actions), which means that so do the projection operators onto the gauge-invariant subspace.
	\begin{equation}
		P_j\equiv \dfrac{1}{|G|}\sum_{g\in G}{\mcG_g^j},\qquad [P_j, P_{j'}] = 0,\qquad \forall j, j',
	\end{equation}
	thus, we get a global projector on the gauge-invariant subspace of the form 
	\begin{equation}
		\mc P = \prod_{j}{P_j}.
		\label{eq:projection}
	\end{equation}
	
	The simple formalism we just reviewed contains most of the features of our general framework. In this work, we will show how a similar defect system structure arises for matrix product states invariant under general MPU group representations. Whenever we act with the symmetry on a subsystem of invariant MPS (using a specific spatial truncation that is very natural for MPUs), defect tensors appear at the boundary of the region, mimicking \eqref{eq:interval}. The main difference is that, this time, the defect tensors sit on lattice sites, instead of between them. Additionally, we will give a method that allows for gauging those symmetries by promoting their associated defects to physical degrees of freedom, as in the onsite case we just reviewed, without increasing the bond dimension. As usual, the gauged MPS reduces to the original globally invariant MPS upon projection of the gauge degrees of freedom. This gauging method works whenever the relevant MPU and MPS satisfy a condition that we dub \textit{block independence}: this condition automatically holds if the MPS is injective and automatically does not hold if the MPU group representation is anomalous, while for a noninjective MPS invariant under a nonanomalous MPU symmetry, both situations are possible, even when restricting to onsite symmetries.
	
	Whenever block independence does not hold, there is still a systematic way to find a local representation of the symmetry that leaves the gauged MPS invariant, possibly after blocking so that we approach a renormalization group fixed point. However, the resulting modified gauge constraints do not necessarily give rise to commuting projectors, even for nonanomalous symmetries. It may be the case that the dimension of the associated gauge-invariant Hilbert space, typically thought of as the space of physical states, is smaller or even bounded, in the limit of large system sizes. Thus, to distinguish this procedure from general bona fide gaugings that produce commuting projectors, we call it ``state-level'' gauging. Figure~\ref{fig:summary} summarizes the exposition and the cases just mentioned. 
	
	\begin{figure}[htb!]
		\centering
		\begin{tikzpicture}[scale=1, every node/.style={scale=0.8, align=center}]
			
			% MPS box
			\node[draw, rectangle, rounded corners, fill=purple!10, minimum width=10cm, minimum height=8cm] (MPS) at (0,0) {};
			
			% NA box
			\node[draw, rectangle, rounded corners, fill=yellow!10, opacity=0.7, minimum width=8cm, minimum height=6cm] (NA) at (0,0) {};
			
			% BI box
			\node[draw, rectangle, rounded corners, fill=red!20, opacity=0.7, minimum width=6cm, minimum height=4cm] (BI) at (0,0) {};
			
			% I box
			\node[draw, rectangle, rounded corners, fill=green!20, opacity=0.7, minimum width=4cm, minimum height=2cm] (I) at (0,0) {};
			
			% Labels outside the boxes
			\node[above right=-0.5cm and -0.5cm] at (MPS.north) {{\color{blue}\large\textbf{MPS}}};
			\node[above right=-0.5cm and 0cm] at (NA.north) {{\color{blue}\large\textbf{NA}}};
			\node[above right=-0.5cm and 0.5cm] at (BI.north) {{\color{blue}\large\textbf{BI}}};
			\node[above right=-1cm and -0.25cm] at (I.north) {{\color{blue}\large\textbf{I}}};
		\end{tikzpicture}
		\caption{The Venn diagram representing possible scenarios. The outermost box represents the class of all MPU-symmetric MPSs. NA labels the instances for which the MPU group representation is nonanomalous (see Section \ref{subsec:MPUreps}). It can be seen that block-independence (BI), introduced in Section \ref{subsec:true_gauging_noninj}, is a stronger condition than NA and that it always holds for injective MPS. The NA cases correspond to the situations in which the projective gauging is applicable, whereas the BI sector is the range of applicability of the gauging based on defect promotion. In the block-dependent sector (the complement of BI), we can use a more limited procedure we call ``state-level'' gauging (Section \ref{sec:state_level}).}
		\label{fig:summary}
	\end{figure}
	
	Our promotion gauging is not completely unrelated to the projection gauging introduced in \cite{PhysRevX.5.011024}. One can see our construction as a systematic way to generate the commuting projectors necessary for this construction. We prove that these projectors commute whenever there is no anomaly, even when block independence does not hold. Thus, in the nonanomalous but block-dependent case, we can generate Gauss law projectors to perform projective gauging, paying the price of an increased bond dimension, even if the bond dimension-preserving promotion gauging does not work.
	
	The remainder of the paper is structured as follows. In Section~\ref{sec:MPUs} we have gathered the results about matrix product unitaries and their action as symmetries of MPS, both known and novel, which we will use later in this work. Then in Section~\ref{sec:defs} we show how we can define defect systems on MPU-invariant MPS, in a similar fashion to those in \cite{Seifnashri}. Section~\ref{sec:true_gauging} presents a gauging procedure based on promoting these defects, provided that block independence is satisfied, and Section~\ref{sec:state_level} describes the state-level gauging mechanism to produce a locally invariant gauged MPS whenever block independence is not satisfied. Finally, Section~\ref{sec:projHam} reflects on the connections of our methods with projection-based gauging and gauging implemented at the Hamiltonian level. We conclude with a discussion in Section~\ref{sec:outlook}.
	
	\section{Review of technical background}
	\label{sec:MPUs}
	In this section, we collect important results on MPUs and their action as symmetries of MPSs. These are mostly known from the literature, but we also present some of their extensions that will be crucial in the remaining parts of the paper.
	
	\subsection{Structure of MPUs}
	\label{sec:MPUsstr}
	Most of the results in this section can be found in the seminal paper \cite{MPU}. It is known that any tensor $\mc U$ generating MPU is normal and can be put in canonical form II (CFII), meaning that we have used the gauge freedom to set the left and right fixed points of the transfer matrix to the identity, and a positive diagonal matrix $\rho$, respectively: 
	\begin{equation}
		\begin{tikzpicture}[scale=\TikzScaling]
			
			\draw[thick,red] (0.758, 0.014) --++ (0.2,0);
			\draw[thick,red] (0.758, 0.585) --++ (0.2,0);
			\draw[thick,red] (0.758, 0.014) --++ (-0.2,0);
			\draw[thick,red] (0.758, 0.585) --++ (-0.2,0);
			
			\draw[thick] (0.775, 0.0) --++ (0,0.6);
			\draw[thick] (0.825, 0.0) --++ (0,0.6);
			\node[mpo] at (0.8,0) {};
			\node[mpo] at (0.8,0.6) {};     
			
			\draw[thick, red] (0.565, 0.014) --++ (0,0.571);  
			
			\node[] at (0.95,-0.2) {$\mathcal{U}$};
			\node[] at (0.95,0.8) {$\mc U^\dagger$};
			
			\node[thick] at (1.1,0.3) {$=$};
			
			\def\Sc{0.6};
			
			\draw[thick, red] (0.0+\Sc*\Shift,0.0) --++ (0,0.6);
			\draw[thick, red] (0.0+\Sc*\Shift,0.0) --++ (0.15,0.0);
			\draw[thick, red] (0.0+\Sc*\Shift,0.6) --++ (0.15,0.0);
			\node[thick] at (1.6,0.25) {$,$};
			
			\def\co{0.8}
			
			\draw[thick,red] (0.758+\co*\Shift, 0.014) --++ (0.3,0);
			\draw[thick,red] (0.758+\co*\Shift, 0.585) --++ (0.3,0);
			\draw[thick,red] (0.758+\co*\Shift, 0.014) --++ (-0.15,0);
			\draw[thick,red] (0.758+\co*\Shift, 0.585) --++ (-0.15,0);
			
			\draw[thick] (0.825+\co*\Shift, 0.0) --++ (0,0.6);
			\draw[thick] (0.775+\co*\Shift, 0.0) --++ (0,0.6);
			\node[mpo] at (0.8+\co*\Shift,0) {};
			\node[mpo] at (0.8+\co*\Shift,0.6) {};     
			
			\draw[thick, red] (1.065+\co*\Shift, 0.014) --++ (0,0.571);  
			
			\node[] at (0.95+\co*\Shift,-0.2) {$\mathcal{U}$};
			\node[] at (0.95+\co*\Shift,0.8) {$\mc U^\dagger$};
			
			\node[square, fill=white] at (1.065+\co*\Shift,0.3) {};
			\node[] at (1.25+\co*\Shift,0.3) {$\rho$};
			
			\node[thick] at (1.5+\co*\Shift,0.3) {$=$};
			
			\draw[thick, red] (1.8+\co*\Shift, 0.014) --++ (0,0.571);  
			\node[square, fill=white] at (1.8+\co*\Shift,0.3) {};
			\node[] at (2.0+\co*\Shift,0.3) {$\rho$};
			
			\draw[thick,red] (1.8+\co*\Shift, 0.014) --++ (-0.15,0);
			\draw[thick,red] (1.8+\co*\Shift, 0.585) --++ (-0.15,0);
			\node[] at (2.2+\co*\Shift, 0.3) {$\cdot$};
		\end{tikzpicture}
		\label{eq:cfii}
	\end{equation}
	Note that if $\mc U$ is in CFII, so is $\mc U^\dagger$. We will from now on always assume that our MPU tensors are in CFII.
	
	A tensor generating MPU is called \textit{simple} if it satisfies the following properties:
	\begin{subequations}
		\begin{equation}
			\begin{tikzpicture}[scale=\TikzScaling]
				\draw[thick,red] (-0.3, 0) --++ (0.9,0);
				\draw[thick,red] (-0.3, 0.6) --++ (0.9,0);
				
				\draw[thick] (0,-0.2) --++ (0,1);
				\node[mpo] at (0,0) {};
				\node[mpo] at (0,0.6) {};
				
				\draw[thick] (0.3,-0.2) --++ (0,1);
				\node[mpo] at (0.3,0) {};
				\node[mpo] at (0.3,0.6) {};
				
				\node[] at (0.45,-0.2) {$\mathcal{U}$};
				\node[] at (0.45,0.8) {$\mc U^\dagger$};
				
				\node[] at (-0.15,-0.2) {$\mathcal{U}$};
				\node[] at (-0.15,0.8) {$\mc U^\dagger$};
				
				\node[thick] at (1,0.25) {$=$};
				
				\def\Sc{0.75};
				
				\draw[thick,red] (0+\Sc*\Shift, 0.014) --++ (0.45,0);
				\draw[thick,red] (0+\Sc*\Shift, 0.585) --++ (0.45,0);
				\draw[thick,red] (0+\Sc*\Shift, 0.014) --++ (-0.35,0);
				\draw[thick,red] (0+\Sc*\Shift, 0.585) --++ (-0.35,0);
				
				\draw[thick,red] (1.1+\Sc*\Shift, 0.014) --++ (0.45,0);
				\draw[thick,red] (1.1+\Sc*\Shift, 0.585) --++ (0.45,0);
				\draw[thick,red] (1.1+\Sc*\Shift, 0.014) --++ (-0.35,0);
				\draw[thick,red] (1.1+\Sc*\Shift, 0.585) --++ (-0.35,0);
				
				\draw[thick] (0+\Sc*\Shift,-0.2) --++ (0,1);
				\node[mpo] at (0+\Sc*\Shift,0) {};
				\node[mpo] at (0+\Sc*\Shift,0.6) {};
				
				\draw[thick] (1.2+\Sc*\Shift,-0.2) --++ (0,1);
				\node[mpo] at (1.2+\Sc*\Shift,0) {};
				\node[mpo] at (1.2+\Sc*\Shift,0.6) {};
				
				\draw[thick,red] (0.437+\Sc*\Shift,0) --++ (0,0.595);
				\draw[thick,red] (0.75+\Sc*\Shift,0) --++ (0,0.595);
				
				\node[square, fill=xanadu] at (0.437+\Sc*\Shift,0.2855) {$a$};
				\node[square, fill=xanadu] at (0.75+\Sc*\Shift,0.2855) {$b$};
				
				\node[] at (1.35+\Sc*\Shift,-0.2) {$\mathcal{U}$};
				\node[] at (1.35+\Sc*\Shift,0.8) {$\mc U^\dagger$};
				
				\node[] at (-0.15+\Sc*\Shift,-0.2) {$\mathcal{U}$};
				\node[] at (-0.15+\Sc*\Shift,0.8) {$\mc U^\dagger$}; 
			\end{tikzpicture}
		\end{equation}
		\begin{equation}
			\begin{tikzpicture}[scale=\TikzScaling]
				
				\draw[thick,red] (0.758, 0.014) --++ (0.4,0);
				\draw[thick,red] (0.758, 0.585) --++ (0.4,0);
				\draw[thick,red] (0.758, 0.014) --++ (-0.3,0);
				\draw[thick,red] (0.758, 0.585) --++ (-0.3,0);
				
				\draw[thick] (0.8,-0.2) --++ (0,1);
				\node[mpo] at (0.8,0) {};
				\node[mpo] at (0.8,0.6) {};     
				
				\draw[thick,red] (0.465, 0.014) --++ (0,0.571);  
				\draw[thick,red] (1.155, 0.014) --++ (0,0.571);  
				
				\node[square, fill=xanadu] at (0.465,0.2855) {$a$};
				\node[square, fill=xanadu] at (1.1555,0.2855) {$b$};
				
				\node[] at (0.95,-0.2) {$\mathcal{U}$};
				\node[] at (0.95,0.8) {$\mc U^\dagger$};
				
				\node[thick] at (1.5,0.25) {$=$};
				
				\def\Sc{0.8};
				
				\draw[thick] (0.0+\Sc*\Shift,-0.2) --++ (0,1);
				\node[thick] at (1.8,0.25) {$,$};
				
			\end{tikzpicture}
		\end{equation}
	\end{subequations}
	for some tensors $a,b$, which can be identified as the fixed points of the transfer matrix. Any tensor that generates MPU becomes simple after blocking a finite number of times, therefore we can always assume that we are dealing with a simple MPU.
	
	The main structural result for MPU tensors can now be summarized in the following.
	\begin{theorem}\cite[Section~3]{MPU}
		\label{thm:mpu}
		For any simple MPU in CFII, there exist left-invertible tensors $X_1, X_2$ and right-invertible ones $Y_1, Y_2$ such that
		\begin{itemize}
			\item[(a)] the MPU tensor can be decomposed as 
			\begin{equation}
				\begin{tikzpicture}[scale=\TikzScaling]
					\draw[thick] (0,0.0) --++ (0,0.8);
					\draw[thick,red] (-0.4,0.4) --++ (0.8,0);
					\node[square, fill=violet] at(0,0.4) {};
					
					\def\Sc{0.2};
					\node[thick] at (0.4+\Sc*\Shift,0.4) {$=$};
					\node[tensorempty,thick] at (1+\Sc*\Shift,0.2) {};
					\node[tensorempty,thick] at (1+\Sc*\Shift,0.6) {};
					
					\tikzset{decoration={snake,amplitude=.4mm,segment length=2mm,
							post length=0mm,pre length=0mm}}
					
					\draw[decorate,thick] (1+\Sc*\Shift,0.25) -- (1+\Sc*\Shift,0.55);
					\draw[thick,red] (0.95+\Sc*\Shift,0.2) --++ (-0.25,0.0);
					\draw[thick] (1+\Sc*\Shift,0.15) --++ (0.0,-0.25);
					\draw[thick,red] (1.05+\Sc*\Shift,0.6) --++ (0.25,0.0);
					\draw[thick] (1+\Sc*\Shift,0.65) --++ (0.0,0.25);
					
					\node[] at (1.2+\Sc*\Shift,0.2) {$Y_1$};
					\node[] at (0.8+\Sc*\Shift,0.6) {$X_1$};

					\node[thick] at (0.4+4*\Sc*\Shift,0.4) {$=$};
					\node[tensor,thick] at (1+4*\Sc*\Shift,0.2) {};
					\node[tensor,thick] at (1+4*\Sc*\Shift,0.6) {};
					
					\draw[ultra thick] (1+4*\Sc*\Shift,0.25) -- (1+4*\Sc*\Shift,0.55);
					\draw[thick,red] (1.05+4*\Sc*\Shift,0.2) --++ (0.25,0.0);
					\draw[thick] (1+4*\Sc*\Shift,0.15) --++ (0.0,-0.25);
					\draw[thick,red] (0.95+4*\Sc*\Shift,0.6) --++ (-0.25,0.0);
					\draw[thick] (1+4*\Sc*\Shift,0.65) --++ (0.0,0.25);
					
					\node[] at (0.8+4*\Sc*\Shift,0.2) {$Y_2$};
					\node[] at (1.2+4*\Sc*\Shift,0.6) {$X_2$};
					
					\node[] at (1.4+4*\Sc*\Shift, 0.4) {$,$};
				\end{tikzpicture}
				\label{eq:3-leg}
			\end{equation}
			\item[(b)] the following combinations are unitary:
			\begin{equation}
				\begin{tikzpicture}[scale=\TikzScaling]
					\tikzset{decoration={snake,amplitude=.4mm,segment length=2mm, post length=0mm, pre length=0mm}}
					
					\draw[ultra thick] (-0.1,0.4) --++ (0,0.4);
					\draw[decorate,thick] (0.1,0.4) --++ (0,0.4);
					\draw[thick] (-0.1,0.4) --++ (0,-0.4);
					\draw[thick] (0.1,0.4) --++ (0,-0.4);
					
					\node[rectangle, fill=classicrose] at (0,0.4) {$u$};
					
					\def\Sc{0.1};
					\node[thick] at (0.15+\Sc*\Shift,0.4) {$=$}; 
					
					\draw[thick,red] (0.6+\Sc*\Shift,0.4) --++ (0.55,0.0);
					\draw[ultra thick] (0.6+\Sc*\Shift,0.45) --++ (0.0, 0.4);
					\draw[decorate, thick] (1.2+\Sc*\Shift,0.45) --++ (0.0, 0.4);
					
					\draw[thick] (0.6+\Sc*\Shift,0.35) --++ (0.0, -0.4);
					\draw[thick] (1.2+\Sc*\Shift,0.35) --++ (0.0, -0.4);
					
					\node[tensor, thick] at (0.6+\Sc*\Shift,0.4) {};
					\node[tensorempty, thick] at (1.2+\Sc*\Shift,0.4) {};
					
					\node[] at (0.38+\Sc*\Shift, 0.4) {$Y_2$};
					\node[] at (1.4+\Sc*\Shift, 0.4) {$Y_1$};
					
					\node[] at (1.6+\Sc*\Shift, 0.35) {$,$};
					
					\draw[thick] (-0.1+12*\Sc*\Shift,0.4) --++ (0,0.4);
					\draw[thick] (0.1+12*\Sc*\Shift,0.4) --++ (0,0.4);
					\draw[decorate, thick] (-0.1+12*\Sc*\Shift,0.4) --++ (0,-0.4);
					\draw[ultra thick] (0.1+12*\Sc*\Shift,0.4) --++ (0,-0.4);
					
					\node[rectangle, fill=celadon] at (0+12*\Sc*\Shift,0.4) {$v$};
					
					\node[thick] at (0.15+13*\Sc*\Shift,0.4) {$=$}; 
					
					\draw[thick,red] (0.65+13*\Sc*\Shift,0.4) --++ (0.55,0.0);
					\draw[thick] (0.6+13*\Sc*\Shift,0.45) --++ (0.0, 0.4);
					\draw[thick] (1.2+13*\Sc*\Shift,0.45) --++ (0.0, 0.4);
					
					\draw[decorate, thick] (0.6+13*\Sc*\Shift,0.35) --++ (0.0, -0.4);
					\draw[ultra thick] (1.2+13*\Sc*\Shift,0.35) --++ (0.0, -0.4);
					
					\node[tensorempty, thick] at (0.6+13*\Sc*\Shift,0.4) {};
					\node[tensor, thick] at (1.2+13*\Sc*\Shift,0.4) {};
					
					\node[] at (0.38+13*\Sc*\Shift, 0.4) {$X_1$};
					\node[] at (1.4+13*\Sc*\Shift, 0.4) {$X_2$};
					
					\node[] at (1.6+13*\Sc*\Shift, 0.35) {$,$};
					
				\end{tikzpicture}
				\label{eq:uv}
			\end{equation}
			\item[(c)] the following holds
			\begin{equation}
				\label{eq:thm1c}
				\begin{tikzpicture}
					\tikzset{decoration={snake,amplitude=.4mm,segment length=2mm, post length=0mm, pre length=0mm}}
					
					\draw[thick] (0.0,-0.3) --++ (0.0,0.6);
					
					\draw[ultra thick] (0.0,-0.35) --++ (0.0, -0.3);
					\draw[ultra thick] (0.0,0.35) --++ (0.0, 0.3);
					\draw[red, thick] (0.0,0.3) --++ (-0.4,0.0);
					\draw[red, thick] (0.0,-0.3) --++ (-0.4,0.0);
					\draw[red, thick] (-0.4,-0.3)--++(0.0,0.6);
					
					\node[tensor] at (0.0,0.3) {};
					\node[tensor] at (0.0,-0.3) {};
					
					\node[] at (0.3,-0.3) {$X_2$};
					\node[] at (0.3,0.3) {$X_2^\dagger$};
					\node[] at (0.6,0.0) {$=$};
					\draw[ultra thick] (1.0,-0.4) --++ (0.0,0.8);
					\node[] at (1.2,0.0) {$,$};
					
					\def\Sc{1.2};
					
					\draw[thick] (0.0+\Sc*\Shift,-0.3) --++ (0.0,0.6);
					
					\draw[decorate, thick] (0.0+\Sc*\Shift,-0.35) --++ (0.0, -0.3);
					\draw[decorate, thick] (0.0+\Sc*\Shift,0.35) --++ (0.0, 0.3);
					\draw[red, thick] (0.0+\Sc*\Shift,0.3) --++ (0.4,0.0);
					\draw[red, thick] (0.0+\Sc*\Shift,-0.3) --++ (0.4,0.0);
					\draw[red, thick] (0.4+\Sc*\Shift,-0.3)--++(0.0,0.6);
					
					\node[tensor, fill=white] at (0.0+\Sc*\Shift,0.3) {};
					\node[tensor, fill=white] at (0.0+\Sc*\Shift,-0.3) {};
					\node[square, fill=white] at (0.4+\Sc*\Shift,0.0) {};
					
					\node[] at (-0.3+\Sc*\Shift,-0.3) {$X_1$}; 
					\node[] at (-0.3+\Sc*\Shift,0.3) {$X_1^\dagger$}; 
					\node[] at (0.7+\Sc*\Shift,0.0) {$\rho$};
					
					\node[] at (1.0+\Sc*\Shift,0.0) {$=$};
					\draw[decorate, thick] (1.4+\Sc*\Shift,-0.4) --++ (0.0,0.8);
					\node[] at (1.6+\Sc*\Shift,0.0) {$.$};
				\end{tikzpicture}
			\end{equation}
		\end{itemize}
	\end{theorem}
	For further purposes, we denote by $d_l$ the dimension of space corresponding to the thick leg, by $d_r$ the dimension of the curly one, and by $d$ the physical dimension of the MPU. For any simple MPU, $d_ld_r=d^2$.
	
	As an immediate consequence of the above Theorem, we have the following
	\begin{corollary}
		\label{cor:mpu}
		Suppose an MPU satisfying the assumptions of Theorem~\ref{thm:mpu} is given. Then 
		\begin{itemize}
			\item[(a)] the following combinations of $X_i, Y_i$ ($i=1,2$) operators are unitary: 
			\begin{equation}
				\begin{tikzpicture}[scale=\TikzScaling]
					\tikzset{decoration={snake,amplitude=.4mm,segment length=2mm, post length=0mm, pre length=0mm}}
					
					\draw[ultra thick] (-0.1,0.4) --++ (0,0.4);
					\draw[thick] (0.1,0.4) --++ (0,0.4);
					\draw[thick] (-0.1,0.4) --++ (0,-0.4);
					\draw[ultra thick] (0.1,0.4) --++ (0,-0.4);
					
					\node[rectangle, fill=carrotorange] at (0,0.4) {$w_L$};
					
					\def\Sc{0.1};
					\node[thick] at (0.175+\Sc*\Shift,0.4) {$=$}; 
					
					\draw[thick, red] (0.6+\Sc*\Shift,0.4) --++ (0.55,0.0);
					\draw[ultra thick] (0.6+\Sc*\Shift,0.45) --++ (0.0, 0.4);
					\draw[thick] (1.2+\Sc*\Shift,0.45) --++ (0.0, 0.4);
					
					\draw[thick] (0.6+\Sc*\Shift,0.35) --++ (0.0, -0.4);
					\draw[ultra thick] (1.2+\Sc*\Shift,0.35) --++ (0.0, -0.4);
					
					\node[tensor, thick] at (0.6+\Sc*\Shift,0.4) {};
					\node[tensor, thick] at (1.2+\Sc*\Shift,0.4) {};
					
					\node[] at (0.38+\Sc*\Shift, 0.4) {$Y_2$};
					\node[] at (1.4+\Sc*\Shift, 0.4) {$X_2$};
					
					\node[] at (1.6+\Sc*\Shift, 0.35) {$,$};
					
					\draw[thick] (-0.1+12*\Sc*\Shift,0.4) --++ (0,0.4);
					\draw[decorate, thick] (0.1+12*\Sc*\Shift,0.4) --++ (0,0.4);
					\draw[decorate, thick] (-0.1+12*\Sc*\Shift,0.4) --++ (0,-0.4);
					\draw[thick] (0.1+12*\Sc*\Shift,0.4) --++ (0,-0.4);
					
					\node[rectangle, fill=citrine] at (0+12*\Sc*\Shift,0.4) {$w_R$};
					
					\node[thick] at (0.175+13*\Sc*\Shift,0.4) {$=$}; 
					
					\draw[thick, red] (0.65+13*\Sc*\Shift,0.4) --++ (0.5,0.0);
					\draw[thick] (0.6+13*\Sc*\Shift,0.45) --++ (0.0, 0.4);
					\draw[decorate, thick] (1.2+13*\Sc*\Shift,0.45) --++ (0.0, 0.4);
					
					\draw[decorate, thick] (0.6+13*\Sc*\Shift,0.35) --++ (0.0, -0.4);
					\draw[thick] (1.2+13*\Sc*\Shift,0.35) --++ (0.0, -0.4);
					
					\node[tensorempty, thick] at (0.6+13*\Sc*\Shift,0.4) {};
					\node[tensorempty, thick] at (1.2+13*\Sc*\Shift,0.4) {};
					
					\node[] at (0.38+13*\Sc*\Shift, 0.4) {$X_1$};
					\node[] at (1.4+13*\Sc*\Shift, 0.4) {$Y_1$};    
					\node[] at (1.6+13*\Sc*\Shift,0.4) {$,$};   \end{tikzpicture}
				\label{eq:wLR}
			\end{equation}
			\item[(b)] the following identities are fulfilled:
			\begin{equation}
				\begin{tikzpicture}
					\tikzset{decoration={snake,amplitude=.4mm,segment length=2mm, post length=0mm, pre length=0mm}}
					
					\draw[ultra thick] (0.0,-0.3) --++ (0.0,0.6);
					
					\draw[thick] (0.0,-0.35) --++ (0.0, -0.3);
					\draw[thick] (0.0,0.35) --++ (0.0, 0.3);
					\draw[red, thick] (0.0,0.3) --++ (0.4,0.0);
					\draw[red, thick] (0.0,-0.3) --++ (0.4,0.0);
					\draw[thick,red] (0.4,-0.3)--++(0.0,0.6);
					
					\node[tensor] at (0.0,0.3) {};
					\node[tensor] at (0.0,-0.3) {};
					\node[square, fill=white] at (0.4,0.0) {};
					
					\node[] at (-0.3,-0.3) {$Y_2$}; 
					\node[] at (-0.3,0.3) {$Y_2^\dagger$}; 
					\node[] at (0.7,0.0) {$\rho$};
					
					\node[] at (1.0,0.0) {$=$};
					\draw[thick] (1.4,-0.4) --++ (0.0,0.8);
					\node[] at (1.6,0.0) {$,$};
					
					\def\Sc{1.2};
					
					\draw[decorate, thick] (0.0+\Sc*\Shift,-0.3) --++ (0.0,0.6);
					
					\draw[thick] (0.0+\Sc*\Shift,-0.35) --++ (0.0, -0.3);
					\draw[thick] (0.0+\Sc*\Shift,0.35) --++ (0.0, 0.3);
					\draw[red, thick] (0.0+\Sc*\Shift,0.3) --++ (-0.4,0.0);
					\draw[red, thick] (0.0+\Sc*\Shift,-0.3) --++ (-0.4,0.0);
					\draw[thick,red] (-0.4+\Sc*\Shift,-0.3)--++(0.0,0.6);
					
					\node[tensor, fill=white] at (0.0+\Sc*\Shift,0.3) {};
					\node[tensor, fill=white] at (0.0+\Sc*\Shift,-0.3) {};
					
					\node[] at (0.3+\Sc*\Shift,-0.3) {$Y_1$};
					\node[] at (0.3+\Sc*\Shift,0.3) {$Y_1^\dagger$};
					\node[] at (0.6+\Sc*\Shift,0.0) {$=$};
					\draw[thick] (1.0+\Sc*\Shift,-0.4) --++ (0.0,0.8);
					\node[] at (1.2+\Sc*\Shift,0.0) {$.$};
				\end{tikzpicture}
				\label{eq:MPUnice2}
			\end{equation}
		\end{itemize}
	\end{corollary}
	
	\begin{proof}
		To prove part (a), first notice that
		\begin{equation}
			\begin{tikzpicture}
				\tikzset{decoration={snake,amplitude=.4mm,segment length=2mm, post length=0mm, pre length=0mm}}
				
				\draw[thick] (0.0,-0.6) --++ (0.0,0.6);
				\draw[ultra thick] (0.0,0.0)--++(0.0,0.6);
				\draw[thick] (0.6,-0.6) --++ (0.0,0.6);
				\draw[decorate, thick] (0.6,0.0)--++ (0.0,0.6);
				\draw[thick] (0.6,0.6)--++(0.0,0.6);
				\draw[thick] (1.2,-0.6)--++(0.0,0.6);
				\draw[ultra thick] (1.2,0.0)--++(0.0,0.6);
				\draw[thick] (1.2,0.6)--++(0.0,0.6);
				\draw[thick] (1.8,-0.6)--++(0.0,0.6);
				\draw[decorate, thick] (1.8,0.0)--++(0.0,0.6);
				\draw[thick] (0.0,0.0)--++(0.6,0.0);
				\draw[thick] (1.2,0.0)--++(0.6,0.0);
				\draw[thick] (0.6,0.6)--++(0.6,0.0);
				
				\node[tensor] at (0.0,0.0) {};
				\node[tensor, fill=white] at (0.6,0.0) {};
				\node[tensor] at (1.2,0.0) {};
				\node[tensor, fill=white] at (1.8,0.0) {};
				\node[tensor, fill=white] at (0.6,0.6) {};
				\node[tensor] at (1.2,0.6) {};
				
				\node[] at (2.2,0.0) {$=$};
				
				\draw[thick] (2.6,-0.6)--++(0.0,0.6);
				\draw[ultra thick] (2.6,0.0)--++(0.0,0.6);
				\draw[thick] (2.6,0.0)--++(0.6,0.0);
				\draw[thick] (3.2,-1.2)--++(0.0,0.6);
				\draw[ultra thick] (3.2,-0.6)--++ (0.0,0.6);
				\draw[thick] (3.2,0.0) --++ (0.0,0.6);
				\draw[thick] (3.2,-0.6) --++ (0.6,0.0);
				\draw[thick] (3.8,-1.2) --++ (0.0,0.6);
				\draw[decorate, thick] (3.8,-0.6) --++ (0.0,0.6);
				\draw[thick] (3.8,0.0)--++(0.0,0.6);
				\draw[thick] (3.8,0.0)--++(0.6,0.0);
				\draw[thick] (4.4,-0.6)--++(0.0,0.6);
				\draw[decorate, thick] (4.4,0.0)--++(0.0,0.6);
				
				\node[tensor] at (2.6,0.0) {};
				\node[tensor] at (3.2,0.0) {};
				\node[tensor] at (3.2,-0.6) {};
				\node[tensor,fill=white] at (3.8,-0.6) {};
				\node[tensor, fill=white] at (3.8, 0.0) {};
				\node[tensor, fill=white] at (4.4,0.0) {};
				
				\node[] at (4.6,0.0) {$,$};
			\end{tikzpicture}
		\end{equation}
		so that
		\begin{equation}
			w_L\otimes w_R= (\id\otimes v\otimes \id)(u\otimes u)(\id \otimes u^\dagger \otimes \id).
		\end{equation}
		
		This demonstrates that $w_L\otimes w_R$ is a unitary operator. Therefore, there exists $\delta>0$ such that $w_L^\dagger w_L=\delta$ and $w_R^\dagger w_R =\frac{1}{\delta}$. It remains to show that $\delta=1$. This follows since
		\begin{equation}
			\begin{tikzpicture}
				\tikzset{decoration={snake,amplitude=.4mm,segment length=2mm, post length=0mm, pre length=0mm}}
				\draw[thick] (-0.6,-0.6) --++ (0.0,1.2);
				\draw[ultra thick] (-0.6,0.6) --++ (0.0,0.6);
				\draw[thick] (-0.6,0.6)--++(0.6,0.0);
				\draw[thick] (0.0,0.0) --++ (0.6,0.0);
				\draw[thick] (0.0,-0.6) --++ (0.0,0.6);
				\draw[ultra thick] (0.0,0.0) --++ (0.0,0.6);
				\draw[thick] (0.0,0.6) --++ (0.0,0.6);
				\draw[ultra thick] (0.6,-0.6) --++ (0.0,0.6);
				\draw[thick] (0.6,0.0) --++ (0.0,1.2);
				
				\node[tensor] at (-0.6,0.6) {};
				\node[tensor] at (0.0,0.0) {};
				\node[tensor] at (0.0,0.6) {};
				\node[tensor] at (0.6,0.0) {};
				
				\node[] at (1.0, 0.3) {$=$};
				
				\draw[thick] (1.4,-0.6) --++ (0.0,0.6);
				\draw[ultra thick] (1.4,0.0) --++ (0.0,1.2);
				\draw[thick] (1.4,0.0) --++ (0.6,0.0);
				\draw[thick] (2.0,-0.6) --++ (0.0,0.6);
				\draw[decorate, thick] (2.0,0.0) --++ (0.0,0.6);
				\draw[thick] (2.0,0.6)--++(0.0,0.6);
				\draw[thick] (2.0,0.6)--++(0.6,0.0);
				\draw[ultra thick] (2.6,-0.6)--++(0.0,1.2);
				\draw[thick] (2.6,0.6)--++(0.0,0.6);
				
				\node[tensor] at (1.4,0.0) {};
				\node[tensor, fill=white] at (2.0,0.0) {};
				\node[tensor, fill=white] at (2.0,0.6) {};
				\node[tensor] at (2.6,0.6) {};
			\end{tikzpicture}
		\end{equation}
		is unitary, implying that $\frac{1}{\delta^2}=1$.
		
		For part(b), from the simplicity of the MPU tensor, we know that 
		\begin{equation}
			\begin{tikzpicture}
				\draw[thick] (0.0,-0.4)--++(0.0,1.4);
				\draw[thick,red] (-0.4,0.0)--++(0.8,0.0);
				\draw[thick,red] (-0.4,0.6)--++(0.8,0.0);
				\draw[thick,red] (-0.4,0.0)--++(0.0,0.6);
				\draw[thick,red] (0.4,0.0)--++(0.0,0.6);
				\node[mpo] at (0.0,0.0) {};
				\node[mpo] at (0.0,0.6) {};
				\node[] at (0.2,0.7) {$\dagger$};
				\node[square, fill=white] at (0.4,0.3) {};
				\node[] at (0.6,0.3) {$\rho$};
			\end{tikzpicture}
		\end{equation}
		is the identity tensor. On the other hand, from the decomposition involving the $X_2$ tensors, we see that the above combination can be rewritten as
		\begin{equation}
			\begin{tikzpicture}
				\draw[thick] (0.0,-1.3)--++(0.0,2.6);
				\draw[ultra thick] (0.0,-0.9)--++(0.0,0.6);
				\draw[ultra thick] (0.0,0.3)--++(0.0,0.6);
				\draw[thick, red] (-0.4,-0.3) --++ (0.4,0.0);
				\draw[thick, red] (-0.4,0.3)--++(0.4,0.0);
				\draw[thick, red] (0.0,0.9) --++ (0.4,0.0);
				\draw[thick,red] (0.0,-0.9)--++(0.4,0.0);
				\draw[red, thick] (-0.4,-0.3)--++(0.0,0.6);
				\draw[red, thick] (0.4,-0.9)--++(0.0,1.8);
				\node[tensor] at (0.0,-0.9) {};
				\node[tensor] at (0.0,-0.3) {};
				\node[tensor] at (0.0,0.3) {};
				\node[tensor] at (0.0,0.9) {};
				\node[square, fill=white] at (0.4,0.0) {};
				\node[] at (0.6, 0.0) {$\rho$};
				\node[] at (0.8, 0.0) {$,$};
			\end{tikzpicture}
		\end{equation}
		which after using Eq.~\eqref{eq:thm1c}, leads to the first of the claimed equalities. The second one follows from analogous computation but this time with using the decomposition of the MPU involving the $X_1$ tensors.
	\end{proof}
	Furthermore, we can also characterize different decompositions of a given MPU tensor. We have
	\begin{lemma}
		\label{lem:deco}
		Let $X_i, Y_i$ and $\widetilde{X}_i, \widetilde{Y}_i$, with $i=1,2$, be two decompositions of the MPU tensor into correspondingly left- and right- invertible three-legged tensors, which satisfy conditions (a) and (b) in Theorem~\ref{thm:mpu}. Then there exist $\delta_i\in \C$ and unitaries $W_i$, $i=1,2$, such that \begin{equation}         
			\widetilde{X}_i=\delta_i X_i W_i, \qquad \widetilde{Y}_i=\frac{1}{\delta_i} W_i^\dagger Y_i.
		\end{equation}
		Furthermore, $\delta_1 \delta_{2}=1$.
	\end{lemma}
	\begin{proof}
		For $i=1,2$, we define $\underline{X_i}$ as the left inverse of $X_i$ and $\underline{Y_i}$ as the right inverse of $Y_i$ (and similarly for $\widetilde{X_i}$ and $\widetilde{Y}_i$). For the two decompositions of the MPU tensor, we define 
		\begin{equation}
			K_i\equiv \underline{\widetilde{X}_i}X_i=\widetilde{Y}_i \underline{Y_i}, \quad i=1,2,
		\end{equation}
		so that $X_i=\widetilde{X}_i K_i$ and $\widetilde{Y}_i=K_i Y_i$, $i=1,2$. The operator $\widetilde{u}$, defined as in \eqref{eq:uv} but for the operators $\widetilde{Y}_i$, is unitary, and hence $K_1\otimes K_2$ is so. Therefore, there exists a positive scalar $\delta$ such that the operators $W_1=\frac{1}{\delta}K_1$ and $W_2=\delta K_2$ are unitary. This, together with the definition of the $K_i$ operators, completes the proof.
	\end{proof}
	For completeness, we provide another set of conditions in the next proposition, analogous to those of Theorem~\ref{thm:mpu} and Corollary~\ref{cor:mpu}. 
	\begin{proposition}
		\label{prop:MPUsplus}
		Let $\mathcal{U}$ be an MPU tensor satisfying the assumptions from Theorem~\ref{thm:mpu} and such that $\mathcal{U}^\dagger$ is also simple. Then
		\begin{subequations}
			\begin{equation}
				\begin{tikzpicture}
					\tikzset{decoration={snake,amplitude=.4mm,segment length=2mm, post length=0mm, pre length=0mm}}
					
					\draw[thick] (0.0,-0.3) --++ (0.0,0.6);
					
					\draw[decorate, thick] (0.0,-0.35) --++ (0.0, -0.3);
					\draw[decorate, thick] (0.0,0.35) --++ (0.0, 0.3);
					\draw[red, thick] (0.0,0.3) --++ (-0.4,0.0);
					\draw[red, thick] (0.0,-0.3) --++ (-0.4,0.0);
					\draw[red, thick] (-0.4,-0.3)--++(0.0,0.6);
					
					\node[tensor, fill=white] at (0.0,0.3) {};
					\node[tensor, fill=white] at (0.0,-0.3) {};
					
					\node[] at (0.3,-0.3) {$Y_1^\dagger$};
					\node[] at (0.3,0.3) {$Y_1$};
					\node[] at (0.8,0.0) {$= \frac{d}{d_r}$};
					\draw[decorate, thick] (1.3,-0.4) --++ (0.0,0.8);
					\node[] at (1.5,0.0) {$,$};
					
					\def\Sc{1.2};
					
					\draw[thick] (0.0+\Sc*\Shift,-0.3) --++ (0.0,0.6);
					
					\draw[ultra thick] (0.0+\Sc*\Shift,-0.35) --++ (0.0, -0.3);
					\draw[ultra thick] (0.0+\Sc*\Shift,0.35) --++ (0.0, 0.3);
					\draw[red, thick] (0.0+\Sc*\Shift,0.3) --++ (0.4,0.0);
					\draw[red, thick] (0.0+\Sc*\Shift,-0.3) --++ (0.4,0.0);
					\draw[red, thick] (0.4+\Sc*\Shift,-0.3)--++(0.0,0.6);
					
					\node[tensor] at (0.0+\Sc*\Shift,0.3) {};
					\node[tensor] at (0.0+\Sc*\Shift,-0.3) {};
					\node[square, fill=white] at (0.4+\Sc*\Shift,0.0) {};
					
					\node[] at (-0.3+\Sc*\Shift,-0.3) {$Y_2^\dagger$}; 
					\node[] at (-0.3+\Sc*\Shift,0.3) {$Y_2$}; 
					\node[] at (0.65+\Sc*\Shift,0.0) {$\rho$};
					
					\node[] at (1.2+\Sc*\Shift,0.0) {$= \frac{d}{d_l}$};
					\draw[ultra thick] (1.7+\Sc*\Shift,-0.4) --++ (0.0,0.8);
					\node[] at (1.9+\Sc*\Shift,0.0) {$,$};
				\end{tikzpicture}
				\label{eq:MPUnice3}
			\end{equation}
			\begin{equation}
				\label{eq:MPUnice4}
				\begin{tikzpicture}
					\tikzset{decoration={snake,amplitude=.4mm,segment length=2mm, post length=0mm, pre length=0mm}}
					\draw[decorate, thick] (0.0,-0.3) --++ (0.0,0.6);
					
					\draw[thick] (0.0,-0.35) --++ (0.0, -0.3);
					\draw[thick] (0.0,0.35) --++ (0.0, 0.3);
					\draw[red, thick] (0.0,0.3) --++ (0.4,0.0);
					\draw[red, thick] (0.0,-0.3) --++ (0.4,0.0);
					\draw[red, thick] (0.4,-0.3)--++(0.0,0.6);
					
					\node[tensor, fill=white] at (0.0,0.3) {};
					\node[tensor, fill=white] at (0.0,-0.3) {};
					\node[square, fill=white] at (0.4,0.0) {};
					
					\node[] at (-0.3,-0.3) {$X_1^\dagger$}; 
					\node[] at (-0.3,0.3) {$X_1$}; 
					\node[] at (0.65,0.0) {$\rho$};
					
					\node[] at (1.2,0.0) {$= \frac{d_r}{d}$};
					\draw[thick] (1.6,-0.4) --++ (0.0,0.8);
					\node[] at (1.8,0.0) {$,$};
					
					\def\Sc{1.3};
					
					\draw[ultra thick] (0.0+\Sc*\Shift,-0.3) --++ (0.0,0.6);
					
					\draw[thick] (0.0+\Sc*\Shift,-0.35) --++ (0.0, -0.3);
					\draw[thick] (0.0+\Sc*\Shift,0.35) --++ (0.0, 0.3);
					\draw[red, thick] (0.0+\Sc*\Shift,0.3) --++ (-0.4,0.0);
					\draw[red, thick] (0.0+\Sc*\Shift,-0.3) --++ (-0.4,0.0);
					\draw[red, thick] (-0.4+\Sc*\Shift,-0.3)--++(0.0,0.6);
					
					\node[tensor] at (0.0+\Sc*\Shift,0.3) {};
					\node[tensor] at (0.0+\Sc*\Shift,-0.3) {};
					
					\node[] at (0.3+\Sc*\Shift,-0.3) {$X_2^\dagger$};
					\node[] at (0.3+\Sc*\Shift,0.3) {$X_2$};
					\node[] at (1.0+\Sc*\Shift,0.0) {$=\frac{d_l}{d}$};
					\draw[thick] (1.4+\Sc*\Shift,-0.4) --++ (0.0,0.8);
					\node[] at (1.6+\Sc*\Shift,0.0) {$.$};
					
				\end{tikzpicture}
			\end{equation}
		\end{subequations}
	\end{proposition}
	\begin{proof}
		We consider the following composition of tensors:
		\begin{equation}
			\begin{tikzpicture}
				\draw[thick, red] (-0.4,0.0) --++ (1.4,0.0);
				\draw[thick, red] (-0.4,0.6) --++ (1.4,0.0);
				\draw[thick, red] (-0.4,-0.6) --++ (1.4,0.0);
				\draw[thick] (0.0,-0.8) --++ (0.0, 1.6);
				\draw[thick] (0.6,-0.8) --++ (0.0, 1.6);
				\draw[thick, red] (-0.4,0.0) --++ (0.0, 0.6);
				\draw[thick, red] (1.0,-0.6) --++ (0.0,0.6);
				\node[square, fill=white] at (1.0,-0.3) {};
				\node[mpo]  at (0.0,0.0) {};
				\node[mpo]  at (0.6,0.0) {};
				\node[mpo]  at (0.0,-0.6) {};
				\node[mpo]  at (0.6,-0.6) {};
				\node[mpo]  at (0.0,0.6) {};
				\node[mpo]  at (0.6,0.6) {};
				\node[] at (0.2,0.1) {$\dagger$};
				\node[] at (0.8,0.1) {$\dagger$};
				\node[] at (1.2,-0.3) {$\rho$};
				\node[] at (1.3,0.0) {$.$};
			\end{tikzpicture}
		\end{equation}
		Using the simplicity of $\mc U$, this diagram can be reduced to
		\begin{equation}
			\begin{tikzpicture}
				\draw[thick, red] (-0.4,-0.6) --++ (1.4,0.0);
				\draw[thick, red] (0.3,0.0) --++ (0.7,0.0);
				\draw[thick, red] (0.3,0.6) --++ (0.7,0.0);
				\draw[thick] (0.6,-0.8) --++ (0.0,1.6);
				\draw[thick, red] (0.3,0.0) --++ (0.0,0.6);
				\draw[thick, red] (1.0,-0.6) --++ (0.0,0.6);
				\draw[thick] (0.0,-0.8) --++ (0.0,0.4);
				
				\node[square, fill=white] at (1.0,-0.3) {};
				\node[mpo]  at (0.0,-0.6) {};
				\node[mpo]  at (0.6,-0.6) {};
				\node[mpo] at (0.6,0.0) {};
				\node[mpo] at (0.6,0.6) {};
				\node[] at (1.2,-0.3) {$\rho$};
				\node[] at (0.8,0.1) {$\dagger$};
				
				\node[] at (1.5, 0.0) {$\equiv$};
				
				\def\Sc{1.1};
				\def\dy{0.6};
				
				\draw[thick, red] (-0.4+\Sc*\Shift,-0.6+\dy) --++ (1.4,0.0);
				\draw[thick] (0.0+\Sc*\Shift,-0.8+\dy) --++ (0.0,0.4);
				\draw[thick] (0.6+\Sc*\Shift,-0.8+\dy) --++ (0.0,0.4);
				\node[mpo]  at (0.0+\Sc*\Shift,-0.6+\dy) {};
				\node[draw, diamond, fill=white, inner sep=1.5pt]  at (0.6+\Sc*\Shift,-0.6+\dy) {};
				\node[] at (0.75 +\Sc*\Shift, -0.5+\dy) {$s$};
				\node[] at (1.2+\Sc*\Shift,0){$,$};    
			\end{tikzpicture}
		\end{equation}
		while using the simplicity of $\mc U^\dagger$ we end up with
		\begin{equation}
			\begin{tikzpicture}
				\draw[thick, red] (-0.4,0.6) --++ (1.4,0.0);
				\draw[thick, red] (-0.4,0.0) --++ (0.9,0.0);
				\draw[thick, red] (-0.4,-0.6) --++ (0.9,0.0);
				\draw[thick, red] (-0.4,0.0) --++ (0.0,0.6);
				\draw[thick] (0.0,-0.8) --++ (0.0,1.6);
				\draw[thick] (0.6,0.4) --++ (0.0,0.4);
				\draw[thick, red] (0.5,-0.6) --++ (0.0,0.6);
				
				\node[mpo] at (0.0,0.0) {};
				\node[mpo] at (0.0,0.6) {};
				\node[mpo] at (0.6,0.6) {};
				\node[mpo] at (0.0,-0.6) {};
				\node[square, fill=white] at (0.5,-0.3) {};
				
				\node[] at (0.2,0.1) {$\dagger$};
				\node[] at (0.7,-0.3) {$\rho$};
				
				\node[] at (1.5, 0.0) {$\equiv$};
				
				\def\Sc{1.1};
				\def\dy{0.6};
				
				\draw[thick, red] (-0.4+\Sc*\Shift,-0.6+\dy) --++ (1.4,0.0);
				\draw[thick] (0.0+\Sc*\Shift,-0.8+\dy) --++ (0.0,0.4);
				\draw[thick] (0.6+\Sc*\Shift,-0.8+\dy) --++ (0.0,0.4);
				\node[mpo]  at (0.6+\Sc*\Shift,-0.6+\dy) {};
				\node[draw, diamond, fill=white, inner sep=1.5pt]  at (0.0+\Sc*\Shift,-0.6+\dy) {};
				\node[] at (0.15 +\Sc*\Shift, -0.5+\dy) {$s$};
				\node[] at (1.2+\Sc*\Shift,0){$.$}; 
			\end{tikzpicture}
		\end{equation}
		Thus, the tensor labeled $s$ commutes with the MPU tensor $\mc U$. Because $\mc U$ is normal, this implies that there exists $\delta\in\C$ such that
		\begin{equation}
			\begin{tikzpicture}
				\draw[thick, red] (-0.3,0.0) --++ (0.6,0.0);
				\draw[thick] (0.0,-0.3) --++ (0.0,0.6); 
				\node[draw, diamond, fill=white, inner sep=1.5pt]  at (0.0,0.0) {};
				\node[] at (0.15,0.15) {$s$};
				\node[] at (0.6,0.0) {$= \delta$};
				\draw[thick, red] (0.9,0.0) --++ (0.6,0.0);
				\draw[thick] (1.2,-0.3) --++ (0.0,0.6); 
				\node[mpo] at (1.2,0.0) {};
				\node[] at (1.6,0.0) {$.$};
			\end{tikzpicture}
		\end{equation}
		Indeed, assume that $\mc U$ was injective so that it had an inverse. Then,
		\begin{equation}
			\begin{tikzpicture}
				\def\dx{0.5}
				\def\ddx{0.2}
				\def\dy{0.3}
				\draw[thick, red] (-\dx-\ddx,0.0) -- (\dx+\ddx,0.0);
				\draw[thick] (-\dx,-\dy) --++ (0,2*\dy);
				\draw[thick] (0,-\dy) --++ (0,2*\dy);
				\draw[thick] (\dx,-\dy) --++ (0,2*\dy);
				\node[mpo] at (\dx,0.0) {};
				\node[mpo] at (0,0.0) {};
				\node[draw, diamond, fill=white, inner sep=1.5pt]  at (-\dx,0.0) {};
				\node[] at (-\dx+0.15,0.15) {$s$};
			\end{tikzpicture}
			=
			\begin{tikzpicture}
				\def\dx{0.5}
				\def\ddx{0.2}
				\def\dy{0.3}
				\draw[thick, red] (-\dx-\ddx,0.0) -- (\dx+\ddx,0.0);
				\draw[thick] (-\dx,-\dy) --++ (0,2*\dy);
				\draw[thick] (0,-\dy) --++ (0,2*\dy);
				\draw[thick] (\dx,-\dy) --++ (0,2*\dy);
				\node[mpo] at (-\dx,0.0) {};
				\node[mpo] at (0,0.0) {};
				\node[draw, diamond, fill=white, inner sep=1.5pt]  at (\dx,0.0) {};
				\node[] at (\dx+0.15,0.15) {$s$};
			\end{tikzpicture}
			\implies 
			\begin{tikzpicture}
				\def\dx{0.35}
				\def\ddx{0.25}
				\def\dy{0.3}
				\draw[thick, red] (-\dx-\ddx,0.0) -- (-\dx+\ddx,0.0);
				\draw[thick, red] (\dx-\ddx,0.0) -- (\dx+\ddx,0.0);
				\draw[thick] (-\dx,-\dy) --++ (0,2*\dy);
				\draw[thick] (\dx,-\dy) --++ (0,2*\dy);
				\node[mpo] at (\dx,0.0) {};
				\node[draw, diamond, fill=white, inner sep=1.5pt]  at (-\dx,0.0) {};
				\node[] at (-\dx+0.15,0.15) {$s$};
			\end{tikzpicture}
			=
			\begin{tikzpicture}
				\def\dx{-0.35}
				\def\ddx{0.25}
				\def\dy{0.3}
				\draw[thick, red] (-\dx-\ddx,0.0) -- (-\dx+\ddx,0.0);
				\draw[thick, red] (\dx-\ddx,0.0) -- (\dx+\ddx,0.0);
				\draw[thick] (-\dx,-\dy) --++ (0,2*\dy);
				\draw[thick] (\dx,-\dy) --++ (0,2*\dy);
				\node[mpo] at (\dx,0.0) {};
				\node[draw, diamond, fill=white, inner sep=1.5pt]  at (-\dx,0.0) {};
				\node[] at (-\dx+0.15,0.15) {$s$};
			\end{tikzpicture},
		\end{equation}
		from applying the inverse of the middle tensor. In case $\mc U$ is not injective, we just need to replace the middle tensor with enough copies so that their blocking is injective, which we can do due to normality. Now, in terms of the $X,Y$ tensors, we have
		\begin{equation}
			\begin{tikzpicture}
				\tikzset{decoration={snake,amplitude=.4mm,segment length=2mm, post length=0mm, pre length=0mm}}
				
				\draw[thick] (0.0,-1.9)--++(0.0,0.4);
				\draw[decorate, thick] (0.0,-1.5)--++(0.0,0.6);
				\draw[thick] (0.0,-0.9) --++ (0.0,0.6);
				\draw[decorate,thick] (0.0,-0.3)--++(0.0,0.6);
				\draw[thick] (0.0,0.3)--++(0.0,0.6);
				\draw[decorate,thick] (0.0,0.9)--++(0.0,0.6);
				\draw[thick] (0.0,1.5)--++(0.0,0.4);
				\draw[thick,red] (0.0,1.5) --++ (0.4,0.0);
				\draw[thick,red] (0.0,0.9) --++ (-0.4,0.0);
				\draw[thick,red] (0.0,0.3) --++ (-0.4,0.0); 
				\draw[thick, red] (-0.4,0.3) --++ (0.0,0.6);
				\draw[thick,red] (0.0,-0.3)--++(0.4,0.0);
				\draw[thick,red] (0.0,-0.9)--++(0.4,0.0);
				\draw[thick, red] (0.4,-0.9)--++(0.0,0.6);
				\draw[thick,red] (0.0,-1.5)--++(-0.4,0.0);
				
				\node[tensor, fill=white] at (0.0,1.5) {};
				\node[tensor, fill=white] at (0.0,0.9) {};
				\node[tensor, fill=white] at (0.0,0.3) {};
				\node[tensor, fill=white] at (0.0,-0.3) {};
				\node[tensor, fill=white] at (0.0,-0.9) {};
				\node[tensor, fill=white] at (0.0,-1.5) {};
				\node[square, fill=white] at (0.4,-0.6) {};
				\node[] at (0.6,-0.6) {$\rho$};
				\node[] at (0.3,-1.5) {$Y_1$};
				\node[] at (-0.3,-0.9) {$X_1$};
				\node[] at (-0.3,-0.3) {$X_1^\dagger$};
				\node[] at (0.3,0.3) {$Y_1^\dagger$};
				\node[] at (0.3,0.9) {$Y_1$};
				\node[] at (-0.3,1.5) {$X_1$};
				
				\node[] at (1.0,0.0) {$=\delta$};
				\draw[decorate, thick] (2.0,-0.3)--++(0.0,0.6);
				\draw[thick, red] (2.0,-0.3)--++(-0.4,0.0);
				\draw[thick, red] (2.0,0.3)--++(0.4,0.0);
				\draw[thick] (2.0,0.3)--++(0.0,0.4);
				\draw[thick] (2.0,-0.3)--++(0.0,-0.4);
				
				\node[tensor, fill=white] at (2.0,-0.3) {};
				\node[tensor, fill=white] at (2.0,0.3) {};
				\node[] at (2.3,-0.3) {$Y_1$};
				\node[] at (1.7,0.3) {$X_1$};
				\node[] at (2.8,0.0) {$.$};
			\end{tikzpicture}
		\end{equation}
		Using the (one-sided) invertibility of $X_1, Y_1$, and Theorem~\ref{thm:mpu}, we end up with the conclusion that
		\begin{equation}
			\begin{tikzpicture}
				\tikzset{decoration={snake,amplitude=.4mm,segment length=2mm, post length=0mm, pre length=0mm}}
				
				\draw[thick] (0.0,-0.3) --++ (0.0,0.6);
				
				\draw[decorate, thick] (0.0,-0.35) --++ (0.0, -0.3);
				\draw[decorate, thick] (0.0,0.35) --++ (0.0, 0.3);
				\draw[red, thick] (0.0,0.3) --++ (-0.4,0.0);
				\draw[red, thick] (0.0,-0.3) --++ (-0.4,0.0);
				\draw[thick, red] (-0.4,-0.3)--++(0.0,0.6);
				
				\node[tensor, fill=white] at (0.0,0.3) {};
				\node[tensor, fill=white] at (0.0,-0.3) {};
				
				\node[] at (0.3,-0.3) {$Y_1^\dagger$};
				\node[] at (0.3,0.3) {$Y_1$};
				\node[] at (0.8,0.0) {$= \delta$};
				\draw[decorate, thick] (1.3,-0.4) --++ (0.0,0.8);
				\node[] at (1.5,0.0) {$.$};
			\end{tikzpicture}
		\end{equation}
		Taking the trace of both sides of this equation and using Corollary \ref{cor:mpu}, we infer that $\delta=\frac{d}{d_r}$. The remaining equalities follow from analogous considerations.
	\end{proof}
	
	The decomposition in terms of three-legged tensors allows us to represent any simple MPU as a two-layer quantum circuit made of $u$ and $v$ gates, or as a ``staircase circuit'' of $w$ gates \cite{SeqGen,Chen24}. Also, a topological index $\log(d_r/d_l)$ can be defined for an MPU in terms of the quotient of dimensions of the virtual legs in these decompositions, measuring the disparity of the ``information flow'' from left to right and from right to left. Finally, this decomposition will be the key to our construction, as it gives a natural way to ``truncate'' the MPU to a finite region, preserving its unitarity. This will give rise to defects in the boundary that can be moved by unitary movement operators and re-fused by unitary fusion operators, thus effectuating the action of the symmetry.
	
	\subsection{MPU group representations}
	\label{subsec:MPUreps}
	
	We now consider a group $G$ and a (linear) representation $U_g$ of $G$ made of \textit{simple injective} matrix product unitaries, each given by a four-legged tensor $\mc U_g$. Thanks to the results from \cite{semiinjective}, from the group property $U_gU_h = U_{gh}$ and the injectivity follows the existence of fusion tensors $F^<_{g,h}, F^>_{g,h}$, such that
	\begin{equation}
		\begin{tikzpicture}[baseline=0.2cm]
			\def\dx{0.5}
			\def\ddx{0.3}
			\def\dy{0.5}
			\def\ddy{0.2}
			\def\dddy{0.3}
			\draw[thick,red] (-\ddx, 0) --++ (5*\dx+2*\ddx, 0);
			\draw[thick,red] (-\ddx, \dy) --++ (5*\dx+2*\ddx, 0);
			\foreach \x in {0, \dx, 2*\dx, 3*\dx, 4*\dx, 5*\dx}{
				\draw[thick] (\x, -\ddy) --++ (0, \dy+2*\ddy);
				\foreach \y in {0,\dy}{
					\node[mpo] at (\x,\y) {};};}
			\node[] at (-\ddx-0.2,\dy) {$g$};
			\node[] at (-\ddx-0.2,0) {$h$}; 
			\node[] at (0.5*\dx, \dy + \ddy + \dddy) {$\overbrace{\hspace{.8 cm}}^{m}$};
			\node[] at (2.5*\dx, \dy + \ddy + \dddy) {$\overbrace{\hspace{.8 cm}}^{n+1}$};
			\node[] at (4.5*\dx, \dy + \ddy + \dddy) {$\overbrace{\hspace{.8 cm}}^{m}$};
		\end{tikzpicture}
		=
		\begin{tikzpicture}[baseline=0.2cm]
			\def\dx{0.5}
			\def\dxb{0.4}
			\def\ddx{0.3}
			\def\dy{0.5}
			\def\ddy{0.2}
			\def\dddy{0.3}
			\draw[thick,red] (-\ddx, 0) -- (\dx+\dxb, 0);
			\draw[thick,red] (2*\dx+3*\dxb, 0) -- (3*\dx+4*\dxb+\ddx, 0);
			\draw[thick,red] (\dx+\dxb, \dy/2)-- (2*\dx+3*\dxb, \dy/2);
			\draw[thick,red] (-\ddx, \dy) -- (\dx+\dxb, \dy);
			\draw[thick,red] (2*\dx+3*\dxb, \dy) -- (3*\dx+4*\dxb+\ddx, \dy);
			\foreach \x in {0, \dx, 2*\dx+4*\dxb, 3*\dx+4*\dxb}{
				\draw[thick] (\x, -\ddy) --++ (0, \dy+2*\ddy);
				\foreach \y in {0,\dy}{
					\node[mpo] at (\x,\y) {};};}
			\foreach \x in {\dx+2*\dxb, 2*\dx+2*\dxb}{
				\draw[thick] (\x, \dy/2-\ddy) --++ (0, 2*\ddy);
				\node[mpo] at (\x,\dy/2) {};}
			\pic () at (\dx+\dxb,0.) {fusR=\dy///};
			\pic () at (2*\dx+3*\dxb,0.) {fusL=\dy///};
			\node[] at (-\ddx-0.2,\dy) {$g$};
			\node[] at (-\ddx-0.2,0) {$h$};
			\node[] at (1.5*\dx+2*\dxb,\dy/2-0.35) {$gh$};
			\node[] at (\dx+\dxb,-0.35) {$F^>_{g,h}$};
			\node[] at (2*\dx+3*\dxb,-0.35) {$F^<_{g,h}$};
			\node[] at (0.5*\dx, \dy + \ddy + \dddy) {$\overbrace{\hspace{.8 cm}}^{m}$};
			\node[] at (1.5*\dx+2*\dxb, \dy + \ddy + \dddy) {$\overbrace{\hspace{.8 cm}}^{n+1}$};
			\node[] at (2.5*\dx+4*\dxb, \dy + \ddy + \dddy) {$\overbrace{\hspace{.8 cm}}^{m}$};
		\end{tikzpicture}  
		\label{eq:fusion_1}
	\end{equation}
	and
	\begin{equation}
		\begin{tikzpicture}[scale=\TikzScaling]
			\draw[thick, red] (0.1,0.3) -- (0.4,0.3);
			\draw[thick, red] (1.6,0.3) -- (1.9,0.3);
			\pic (v) at (0.4,0.3) {V=//};
			\pic (w) at (1.6,0.3) {W=//};
			\draw[thick,red] (0.758, 0.014) --++ (0.485,0);
			\draw[thick,red] (0.758, 0.585) --++ (0.485,0);
			\node[anchor=north,inner sep=6pt] at (w-bottom) {$F^>_{g,h}$};
			\node[anchor=north,inner sep=6pt] at (v-bottom) {$F^<_{g,h}$};
			\draw[thick] (0.8,-0.2) --++ (0,1);
			\node[mpo] at (0.8,0) {};
			\node[mpo] at (0.8,0.6) {};     
			\draw[thick] (1.2,-0.2) --++ (0,1);
			\node[mpo] at (1.2,0) {};
			\node[mpo] at (1.2,0.6) {};  
			
			\node[] at (1, 1) {$\overbrace{}^{n}$};
			
			\node[thick] at (2.1,0.28) {$=$};
			
			\def\Sc{1.30};
			
			\node[] at (-1.8+\Sc*\Shift,0.75) {$g$};
			\node[] at (-1.8+\Sc*\Shift,0.15) {$h$}; 
			
			\draw[thick] (0.0+\Sc*\Shift,-0.2) --++ (0,1);
			\draw[thick] (0.4+\Sc*\Shift,-0.2) --++ (0,1);
			
			\draw[thick, red] (-0.3+\Sc*\Shift,0.3) -- (0.7+\Sc*\Shift,0.3);
			\node[square, fill=violet] at (0.0+\Sc*\Shift,0.3) {};
			\node[square, fill=violet] at (0.4+\Sc*\Shift,0.3) {}; 
			
			\node[] at (0.2+\Sc*\Shift, 1) {$\overbrace{}^{n}$};
			\node[] at (0.2+\Sc*\Shift,0.5) {$gh$};
			
		\end{tikzpicture}
		\label{eq:fusion_2}
	\end{equation}
	for all $n\geq 0, m\geq \ell$, for some $\ell$ called the \textit{nilpotency length}. Incidentally, the fact that the latter can be nonzero is one of the reasons why the procedure outlined in \cite{GarreKull} does not straightforwardly work for MPUs, since $\ell\neq 0$ violates an important assumption in their derivation\footnote{We thank José Garre-Rubio for clarifying this point.}. However, $\ell$ can always be reduced to 1 by blocking the MPUs a finite number of times, which we will assume from now on that we did. 
	
	The fusion tensors are defined up to a scalar $\beta_{g,h}$ which constitutes their gauge freedom \footnote{The reader should be cautioned that the word \textit{gauge} is going to appear in multiple senses in this work since we are dealing with the procedure of gauging a physical symmetry as well as with the redundancies in the definition of the tensors that arise naturally in tensor networks.},
	\begin{equation}
		F^{>}_{g,h}\to \beta_{g,h}F^{>}_{g,h}, \qquad F^{<}_{g,h}\to\dfrac{1}{\beta_{g,h}}F^{<}_{g,h}.
		\label{eq:scalar_fus_ten}
	\end{equation}
	
	There is an additional important piece of information that characterizes an MPU representation of a group, and it is cohomological in nature. Given an MPU group and its fusion tensors, we can define a scalar function $\omega:G^{\times 3}\to \C^\times$, such that the following holds \cite{Chen11,semiinjective},
	\begin{equation}
		\begin{tikzpicture}
			\node[irrep] at (-0.35,0.9) {$g$};
			\node[irrep] at (-0.35,0.4) {$h$};
			\node[irrep] at (-0.35,-0.1) {$k$};
			\pic[scale=1.25] at (1.75,0.625) {FR=ghk/g/hk};
			\pic[scale=0.83] at (1.2,0.25 ) {FR=/h/k};
			\draw[red,thick] (1,1) -- (1.2,1);
			\foreach \y in {0,0.5,1}{
				\foreach \x in {0,0.4,0.8}{
					\draw[thick] (\x,\y-0.25) -- (\x,\y+0.25);
					\draw[red,thick] (\x-0.25,\y) -- (\x+0.25,\y);
					\node[square, fill=violet] at (\x,\y) {};   }}
			\def\Sc{2.1};
			\node[] at (3.2,0.5) {$=\omega(g,h,k)$};
			\node[irrep] at (-0.35+\Sc*\Shift,0.9) {$g$};
			\node[irrep] at (-0.35+\Sc*\Shift,0.4) {$h$};
			\node[irrep] at (-0.35+\Sc*\Shift,-0.1) {$k$};
			\pic[scale=1.25] at (1.75+\Sc*\Shift,0.375) {FR=ghk/gh/k};
			\pic[scale=0.83] at (1.2+\Sc*\Shift,0.75) {FR=/g/h};
			\draw[red,thick] (1+\Sc*\Shift,0) -- (1.2+\Sc*\Shift,0);
			\foreach \y in {0,0.5,1}{
				\foreach \x in {0+\Sc*\Shift,0.4+\Sc*\Shift,0.8+\Sc*\Shift}{
					\draw[thick] (\x,\y-0.25) -- (\x,\y+0.25);
					\draw[red,thick] (\x-0.25,\y) -- (\x+0.25,\y);
					\node[square, fill=violet] at (\x,\y) {};   }}
			\node[] at (7, 0.5) {$.$};      
		\end{tikzpicture}
	\end{equation}
	From the equivalence of the different ways of relating two fusion trees, it can be seen that $\omega$ must be a 3-cocycle, that is
	\begin{equation}
	\begin{split}
		\omega(gh, k, l)&\omega(g, h, kl)\\
		&= \omega(g, h, k)\omega(g, hk, l)\omega(h,k,l),
	\end{split}    
	\label{eq:3-cocycle}
\end{equation}
	for any $g,h,k,l\in G$. Adjusting the scalar freedom of the fusion tensors \eqref{eq:scalar_fus_ten} amounts to modifying the 3-cocycle by a 3-coboundary,
	\begin{equation}  \omega(g,h,k)\rightarrow\dfrac{\beta_{g,hk}\beta_{h,k}}{\beta_{g,h}\beta_{gh,k}}\,\omega(g,h,k),
		\label{eq:omegagauge}
	\end{equation}
	thus the physical information is contained in the cohomology class $[\omega]\in H^3(G,\mathbb{C}^\times)$. An MPU symmetry group with nontrivial $\omega$ cannot have an invariant injective MPS and is called an \textit{anomalous} representation. This anomaly is usually regarded as an obstruction to gauging, in the sense that it leads to non-commuting Gauss law projectors \cite{GarreKull, Seifnashri}.
	
	As a remark, note that for an MPU to be part of a unitary representation of a finite group, it must have a finite order, and thus its topological index (which is real-valued, additive upon taking products, and equals zero for the identity MPU) must vanish, leading to $d_r=d_l=d$. We will, however, keep the notation of the thick and squiggly lines of the three-legged tensors from \eqref{eq:3-leg} as an aid to help identify their orientation.
	
	\subsubsection{Compatibility between Hermitian and group structures}
	\label{sec:daggerinverse}
	
	In any MPU group representation, there exists a relation between the dagger operation and the group structure, in particular the group inverse. We assume that we have blocked enough so that all tensors $\mc U_g$ are simple, injective, and in canonical form. Consider a group element $g\in G$ and the associated tensor $\mc U_g$, which generates the unitary $U_g$. The daggered tensor $\mc U_g^\dagger$ generates $U_{\inv{g}} = U_g^\dagger$, thus, because of the injectivity, $\mc U^\dagger_g$ and $\mc U_{\inv{g}}$ are related by a gauge transformation
	\begin{equation}
		\mc U_g^\dagger = T_g^\dagger\,\mc U_{\inv{g}} T_g,
		\label{eq:defT}
	\end{equation}
	where the unitarity of $T_g$ comes from relating two tensors in canonical form. Taking the complex conjugate of \eqref{eq:defT} and using the injectivity of the tensors, it can be shown that
	\begin{equation}
		T_gT^*_{\inv g}=\sigma_g\mathds{1},
		\label{eq:intro_sigma}
	\end{equation}
	where $T_g^\ast$ denotes the complex conjugation of $T_g$, and $\sigma_g\in U(1)$ satisfy $\sigma_g\sigma_{\inv g}=1$.
	As we shall see below, whenever $g\neq\inv{g}$, we can simply absorb the gauge transformation into the definition of the tensor $\mc U_{\inv{g}}$ so that $T_g=\mathds{1}$ and $\sigma_g=1$. However, if $g^{-1}=g$, then by definition $\mc U_{\inv g} = \mc U_g$, and Eq.~\eqref{eq:defT} reads
	\begin{equation}
		\mc U_{g}^\dagger = T_g^{\dagger}\,\mc U_g T_g.
	\end{equation}
	This time, there may not be a gauge choice such that $T_g=\mathds{1}$. Eq.~\eqref{eq:intro_sigma} reads $T^{\phantom{*}}_gT_g^*=\sigma_g \mathds{1}$, with $\sigma_g=\pm 1$, and $\sigma_g$ acts as a topological index for the time reversal SPT phase of $\mc U_g$ \cite[Section~III.A.1]{Review}. Whenever $\sigma_g=1$, there is a gauge transformation that makes the MPU tensor equal to its Hermitian conjugate, and we can assume that we have made that gauge choice, and thus $T_g=\I$. Whenever $\sigma_g = -1$, this is not possible, $T_g$ will be non-trivial and should be carried out in computations. We show in Appendix \ref{app:Z2} that $\sigma_g=-1$ happens exactly whenever the $\Z_2$ representation generated by $U_g$ is anomalous.
	
	This compatibility of the MPU Hermitian conjugation with the group structure gives rise to a relation between the fusion operators, according to the following.
	\begin{proposition}
		There exist scalars $\zeta_{g,h}\in \C^\times$ such that
		\begin{subequations}
			\begin{equation}
				\begin{tikzpicture}
					\draw[red,thick] (-0.6,0.0)--++(0.6,0.0);
					\draw[red,thick] (0.35,0.3)--++(0.65,0.0);
					\draw[red,thick] (0.35,-0.3)--++(0.65,0.0);
					\draw[thick] (0.75,-0.5)--++(0.0,1.0);
					
					\pic[scale=1.0] at (0.0,0.0) {FLD=//}; 
					
					\node[tensor, red] at (-0.35,0.0) {};
					\node[tensor, red] at (0.35,-0.3) {};
					\node[tensor, red] at (0.35,0.3) {};
					\node[mpo] at (0.75,0.3) {};
					\node[mpo] at (0.75,-0.3) {};
					
					\node[red] at (-0.48,0.3) {$T_{\inv{(gh)}}$};
					\node[red] at (0.39,0.70) {$T_{\inv g}^\dagger$};
					\node[red] at (0.39,-0.65) {$T_{\inv h}^\dagger$};
					
					\node[] at (1.5,0.0) {$= \zeta_{g,h}$};
					
					\draw[red, thick] (2.9,0.3)--++(0.3,0.0); 
					\draw[red, thick] (2.9,-0.3)--++(0.3,0.0); 
					\draw[thick] (2.9,-0.5)--++(0.0,1.0);
					\pic[scale=1.0] at (2.4,0.0) {FL=gh/g/h}; 
					\node[mpo] at (2.9,0.3) {};
					\node[mpo] at (2.9,-0.3) {};
					\node[] at (3.4,0.0) {$,$};
					
				\end{tikzpicture}
			\end{equation}
			\begin{equation}
				\begin{tikzpicture}
					\draw[red,thick] (0.6,0.0)--++(-0.6,0.0);
					\draw[red,thick] (-0.35,0.3)--++(-0.65,0.0);
					\draw[red,thick] (-0.35,-0.3)--++(-0.65,0.0);
					\draw[thick] (-0.75,-0.5)--++(0.0,1.0);
					
					\pic[scale=1.0] at (0.0,0.0) {FRD=//}; 
					
					\node[tensor, red] at (0.35,0.0) {};
					\node[tensor, red] at (-0.35,-0.3) {};
					\node[tensor, red] at (-0.35,0.3) {};
					\node[mpo] at (-0.75,0.3) {};
					\node[mpo] at (-0.75,-0.3) {};
					
					\node[red] at (0.55,0.38) {$T_{\inv{(gh)}}^\dagger$};
					\node[red] at (-0.35,0.65) {$T_{\inv{g}}$};
					\node[red] at (-0.35,-0.65) {$T_{\inv{h}}$};
					
					\node[] at (1.4,0.0) {$=\frac{1}{\zeta_{g,h}}$};
					
					\def\Sc{0.4};
					
					\draw[red, thick] (1.5+\Sc,0.3)--++(0.3,0.0); 
					\draw[red, thick] (1.5+\Sc,-0.3)--++(0.3,0.0); 
					\draw[thick] (1.8+\Sc,-0.5)--++(0.0,1.0);
					\pic[scale=1.0] at (2.3+\Sc,0.0) {FR=gh/g/h}; 
					\node[mpo] at (1.8+\Sc,0.3) {};
					\node[mpo] at (1.8+\Sc,-0.3) {};
					\node[] at (2.9+\Sc,0.0) {$,$};
					
				\end{tikzpicture}
				\label{eq:defzeta}
			\end{equation}
		\end{subequations}
		where 
		\begin{subequations}
			\begin{equation}
				\begin{tikzpicture}
					\node[irrep] at (0.35,0.3) {$g$};
					\node[irrep] at (0.35,-0.3) {$h$};
					\node[irrep] at (-0.35,0.0) {$gh$};
					\pic[scale=1.0] at (0.0,0.0) {FLD=//};
					\node[] at (1.75,0.0) {$\equiv \bigl(F^{<}_{h^{-1},g^{-1}}\bigr)^{\dagger}$};
					\node[] at (2.75,0.0) {$,$};
				\end{tikzpicture}
			\end{equation}
			\begin{equation}
				\begin{tikzpicture}
					\node[irrep] at (-0.3,0.3) {$g$};
					\node[irrep] at (-0.3,-0.3) {$h$};
					\node[irrep] at (0.3,0.0) {$gh$};
					\pic[scale=1.0] at (0.0,0.0) {FRD=//};
					\node[] at (1.75,0.0) {$:=\bigl(F^{>}_{h^{-1},g^{-1}}\bigr)^{\dagger}$};
					\node[] at (2.75,0.0) {$,$};
				\end{tikzpicture}
			\end{equation}
		\end{subequations}
		and the dagger operation here requires complex conjugating every entry and swapping the legs on the same side (as if reflecting the tensor along a horizontal axis).
		These scalars satisfy
		\begin{equation}
			\zeta^{\phantom{*}}_{g,h}\zeta^*_{h^{-1},g^{-1}}=\dfrac{\sigma_g\sigma_h}{\sigma_{gh}}.
			\label{eq:zeta}
		\end{equation}
		and 
		\begin{equation}
			\omega(k^{-1},h^{-1}, g^{-1})^*\omega(g,h,k) = \dfrac{\zeta_{g,hk}\zeta_{h,k}}{\zeta_{g,h}\zeta_{gh,k}}.
			\label{eq:inv_omega}
		\end{equation}
		\label{prop:zetas}
	\end{proposition}
	\begin{proof}
		The result is a consequence of Theorem 1 in \cite{GarreSchuch}, which compiles the results of \cite{semiinjective} for the case of nilpotency length 1. From taking the dagger of Eqs.~\eqref{eq:fusion_1}--\eqref{eq:fusion_2}, and using Eq.~\eqref{eq:defT} it follows that $\left(F^{<,>}_{h^{-1},g^{-1}}\right)^\dagger$, up to multiplication by the corresponding $T_g$, is a valid fusion tensor for $\mc U_g,\mc U_h$, hence it is related by a scalar to our original choice $F^{<,>}_{g, h}$. Eq. \eqref{eq:zeta} follows by applying this duality twice, together with Eq.~\eqref{eq:intro_sigma}. Finally, Eq. \eqref{eq:inv_omega} follows from taking the Hermitian conjugate of the following equation
		\begin{equation}
			\begin{tikzpicture}
				\def\sh{0.2};
				\pic[scale=1.0] at (0.0,0.0) {FL=//};
				\pic[scale=0.665] at (0.75,0.3) {FL=//};
				\draw[thick,red] (0.5,-0.3) --++ (0.8,0.0);
				\pic[scale=0.665] at (1.4+\sh,-0.1) {FR=//};
				\pic[scale=1.0] at (2.15+\sh,0.2) {FR=//};
				\draw[thick,red] (0.8+\sh,0.5)--++(0.9,0.0);
				\draw[thick,red] (0.8+\sh,0.1)--++(0.3,0.0);
				\draw[thick] (0.35,-0.5)--++(0.0,1.2);
				\draw[thick] (1.1+\sh,-0.5)--++(0.0,1.2);
				\draw[thick] (1.8+\sh,-0.5)--++(0.0,1.2);
				\node[mpo] at (0.35,-0.3) {};
				\node[mpo] at (1.1+\sh,-0.3) {};
				\node[mpo] at (0.35,0.3) {};
				\node[mpo] at (1.1+\sh,0.1) {};
				\node[mpo] at (1.1+\sh,0.5) {};
				\node[mpo] at (1.8+\sh,-0.1) {};
				\node[mpo] at (1.8+\sh,0.5) {};
				\node[irrep] at (0.85+\sh,0.5) {$k^{\text{-}1}$};
				\node[irrep] at (0.85+\sh,0.1) {$h^{\text{-}1}$};
				\node[irrep] at (0.85+\sh,-0.35) {$g^{\text{-}1}$};
				\node[] at (3.8+\sh,0.0) {$=\omega(k^{-1},h^{-1},g^{-1})$};
				\draw[thick,red] (5.2+\sh,0.0)--++(1.4,0.0);
				\draw[thick] (5.5+\sh,-0.3)--++(0.0,0.6);
				\draw[thick] (5.9+\sh,-0.3)--++(0.0,0.6);
				\draw[thick] (6.3+\sh,-0.3)--++(0.0,0.6);
				\node[mpo] at (5.5+\sh,0.0) {};
				\node[mpo] at (5.9+\sh,0.0) {};
				\node[mpo] at (6.3+\sh,0.0) {};
				\node[irrep] at (6.75+\sh,0.05) {$(ghk)^{\text{-}1}$};
			\end{tikzpicture}
		\end{equation}
		to get
		\begin{equation}
			\begin{tikzpicture}
				\draw[red,thick] (-0.5,0.0)--++(0.2,0.0);
				\draw[red,thick] (0.4,0.3)--++(2.5,0.0);
				\draw[red,thick] (0.4,-0.3)--++(0.75,0.0);
				\draw[red,thick] (1.9,-0.5)--++(2.5,0.0);
				\draw[red,thick] (1.9,-0.1)--++(1.0,0.0);
				\draw[red,thick] (3.3,0.1)--++(1.0,0.0);
				\draw[red,thick] (4.8,-0.2)--++(0.2,0.0);
				\draw[thick] (0.7,-0.8)--++(0.0,1.4);
				\draw[thick] (2.25,-0.8)--++(0.0,1.4);
				\draw[thick] (3.7,-0.8)--++(0.0,1.4);
				
				\pic[scale=1.0] at (0.0,0.0) {FLD=//};
				\pic[scale=0.665] at (1.45,-0.3) {FLD=//};
				\pic[scale=0.665] at (3.0,0.1) {FRD=//};
				\pic[scale=1.0] at (4.5,-0.2) {FRD=//};
				
				\node[tensor, red] at (-0.3,0.0) {};
				\node[tensor, red] at (0.3,-0.3) {};
				\node[tensor, red] at (0.3,0.3) {};
				\node[mpo] at (0.7,-0.3) {};
				\node[mpo] at (0.7,0.3) {};
				\node[tensor, red] at (1.1,-0.3) {};
				\node[tensor, red] at (1.85,-0.5) {};
				\node[tensor, red] at (1.85,-0.1) {};
				\node[mpo] at (2.25,-0.1) {};
				\node[mpo] at (2.25,-0.5) {};
				\node[mpo] at (2.25,0.3) {};
				\node[tensor, red] at (2.65,0.3) {};
				\node[tensor, red] at (2.65,-0.1) {};
				\node[tensor, red] at (3.3,0.1) {};
				\node[mpo] at (3.7,0.1) {};
				\node[mpo] at (3.7,-0.5) {};
				\node[tensor, red] at (4.1,0.1) {};
				\node[tensor, red] at (4.1,-0.5) {};
				\node[tensor, red] at (4.8,-0.2) {};
				
				\node[red] at (0.4,0.5) {$\dagger$};
				\node[red] at (1.95,0.1) {$\dagger$};
				\node[red] at (3.4,0.3) {$\dagger$};
				\node[red] at (4.9,0.0) {$\dagger$};
				\node[red] at (0.4,-0.15) {$\dagger$};
				\node[red] at (1.95,-0.33) {$\dagger$};
				
				\node[irrep] at (1.62,0.3) {$g$};
				\node[irrep] at (1.62,-0.1) {$h$};
				\node[irrep] at (1.62,-0.5) {$k$};
				
				\node[] at (5.6,-0.1) {$= \phi_{g,h,k}$};
				
				\draw[red,thick] (6.15,-0.1)--++(1.3,0.0);
				\draw[thick] (6.4,-0.4)--++(0.0,0.6);
				\node[mpo] at (6.4,-0.1) {};
				\draw[thick] (6.8,-0.4)--++(0.0,0.6);
				\node[mpo] at (6.8,-0.1) {};
				\draw[thick] (7.2,-0.4)--++(0.0,0.6);
				\node[mpo] at (7.2,-0.1) {};
				\node[] at (7.5,-0.1) {$,$};
				
				\node[irrep] at (7.0,0.05) {$g\!h\!k$};
			\end{tikzpicture}
		\end{equation}
		where we have omitted some labels for compactness, and the scalar factor can be computed in two ways that must yield the same result: 
		\begin{equation}
			\phi_{g,h,k}=\dfrac{\zeta_{g,hk}\zeta_{h,k}}{\zeta_{g,h}\zeta_{gh,k}}\cdot \frac{1}{\omega(g,h,k)}=\omega(k^{-1},h^{-1},g^{-1})^\ast.
		\end{equation}
	\end{proof}

	\subsection{MPU action on MPSs}
	\label{sec:MPUonMPS}
	
	We consider the case where we have injective MPS tensors labeled by elements $x$ of a set $\mathsf{X}$ on which $G$ acts transitively ($\forall\, x,y\in \mathsf{X}\ \exists\, g\in G \,:\ y = gx$) by permutation. The assumption of transitivity comes from the fact that otherwise the set $\mathsf{X}$ breaks into disjoint orbits, and we can reduce our study to one of those orbits \footnote{See also \cite{GarreLootensMolnar} and references therein for yet another motivation, based on the analysis of perturbations of parent Hamiltonian, to consider only transitive actions.}. Acting by permutation allows us to build an invariant block-injective MPS by superposing all MPSs in $\mathsf{X}$. In the injective case (cf. Section~\ref{sec:noninj}), the set $\mathsf{X}$ will have only one element, and therefore we will omit the corresponding label.
	
	From the injectivity of the MPSs, once again using the results from \cite{semiinjective}, there exists a collection of action tensors $A_{g,x}^{\Gamma}, A_{g,x}^{\ammaG}$, defined by 
	\begin{equation}
		\begin{tikzpicture}[baseline=0.2cm]
			\def\dx{0.5}
			\def\ddx{0.3}
			\def\dy{0.5}
			\def\ddy{0.2}
			\def\dddy{0.3}
			\draw[thick] (-\ddx, 0) --++ (5*\dx+2*\ddx, 0);
			\draw[thick,red] (-\ddx, \dy) --++ (5*\dx+2*\ddx, 0);
			\foreach \x in {0, \dx, 2*\dx, 3*\dx, 4*\dx, 5*\dx}{
				\draw[thick] (\x, 0) --++ (0, \dy+\ddy);
				\node[tensor] at (\x,0) {};
				\node[mpo] at (\x,\dy) {};}
			\node[] at (0.5*\dx, \dy + \ddy + \dddy) {$\overbrace{\hspace{.8 cm}}^{m}$};
			\node[] at (2.5*\dx, \dy + \ddy + \dddy) {$\overbrace{\hspace{.8 cm}}^{n+1}$};
			\node[] at (4.5*\dx, \dy + \ddy + \dddy) {$\overbrace{\hspace{.8 cm}}^{m}$};
		\end{tikzpicture}
		=
		\begin{tikzpicture}[baseline=0.2cm]
			\def\dx{0.5}
			\def\dxb{0.4}
			\def\ddx{0.3}
			\def\dy{0.5}
			\def\ddy{0.2}
			\def\dddy{0.3}
			\draw[thick] (-\ddx, 0) -- (\dx+\dxb, 0);
			\draw[thick] (2*\dx+3*\dxb, 0) -- (3*\dx+4*\dxb+\ddx, 0);
			\draw[thick] (\dx+\dxb, 0)-- (2*\dx+3*\dxb, 0);
			\draw[thick,red] (-\ddx, \dy) -- (\dx+\dxb, \dy);
			\draw[thick,red] (2*\dx+3*\dxb, \dy) -- (3*\dx+4*\dxb+\ddx, \dy);
			\foreach \x in {0, \dx, 2*\dx+4*\dxb, 3*\dx+4*\dxb}{
				\draw[thick] (\x, 0) --++ (0, \dy+\ddy);
				\node[tensor] at (\x,0) {};
				\node[mpo] at (\x,\dy) {};}
			\foreach \x in {\dx+2*\dxb, 2*\dx+2*\dxb}{
				\draw[thick] (\x, \dy/2-\ddy) --++ (0, 2*\ddy);
				\node[tensor] at (\x,0) {};}
			\pic () at (\dx+\dxb,0.) {actR=\dy///};
			\pic () at (2*\dx+3*\dxb,0.) {actL=\dy///};
			\node[irrep] at (1.5*\dx+2*\dxb,0) {$gx$};
			\node[] at (\dx+\dxb,-0.35) {$A^{\ammaG}_{g,x}$};
			\node[] at (2*\dx+3*\dxb,-0.35) {$A^{\Gamma}_{g,x}$};
			\node[] at (0.5*\dx, \dy + \ddy + \dddy) {$\overbrace{\hspace{.8 cm}}^{m}$};
			\node[] at (1.5*\dx+2*\dxb, \dy + \ddy + \dddy) {$\overbrace{\hspace{.8 cm}}^{n+1}$};
			\node[] at (2.5*\dx+4*\dxb, \dy + \ddy + \dddy) {$\overbrace{\hspace{.8 cm}}^{m}$};
			\node[irrep] at (3*\dx+4*\dxb+\ddx,0.00) {$x$};
			\node[irrep] at (3*\dx+4*\dxb+\ddx,\dy) {$g$};
			\node[irrep] at (-\ddx,0.00) {$x$};
			\node[irrep] at (-\ddx,\dy) {$g$};
		\end{tikzpicture}  
		\label{eq:action_exterior}
	\end{equation}
	and
	\begin{equation}
		\begin{tikzpicture}
			\node at (0.25,0.45) {$\overbrace{\hspace{0.8cm}}^{n}$};
			\draw[thick] (-0.7,-0.6) -- (1.2,-0.6);
			\draw[red,thick] (0,0) -- (1,0);
			\pic at (-0.5,-0.3) {ALl=y/g/x};
			\pic at (1,-0.3) {ARl=y/g/x};
			\draw[thick] (0,-0.6) -- (0,0.3);
			\draw[thick] (0.5,-0.6) -- (0.5,0.3);
			\node[mpo] (t) at (0,0) {};
			\node[tensor] (t) at (0,-0.6) {};
			\node[mpo] (t) at (0.5,0) {};
			\node[tensor] (t) at (0.5,-0.6) {};
			\node[irrep] at (-0.4,-1.2) {$A^\Gamma_{g,x}$};
			\node[irrep] at (1.0,-1.2) {$A^\ammaG_{g,x}$};
			\def\Sc{1.15};
			\def\Sv{0.6};
			\node[] at (1.75,-0.25) {$=$};
			\node at (0.25+\Sc*\Shift,0.65-\Sv) {$\overbrace{\hspace{0.8cm}}^{n}$};
			\draw[thick] (0+\Sc*\Shift,0-\Sv) -- (0+\Sc*\Shift,0.5-\Sv);
			\draw[thick] (0.5+\Sc*\Shift,0-\Sv) -- (0.5+\Sc*\Shift,0.5-\Sv);
			\draw[thick](-0.3+\Sc*\Shift,0-\Sv) -- (0.8+\Sc*\Shift,0-\Sv);
			\node[tensor] (t) at (0+\Sc*\Shift,0-\Sv) {};
			\node[tensor] (t) at (0.5+\Sc*\Shift,0-\Sv) {};
			\node[irrep] at (0.7+\Sc*\Shift,0.05-\Sv) {$gx$};
			\node[] at (1.0+\Sc*\Shift,-0.35) {$,$};
		\end{tikzpicture}
		\label{eq:action_interior}
	\end{equation}
	where $n\geq0$ and $m\geq \ell$ for some nilpotency length $\ell$. We will once again assume that we have blocked until all nilpotency lengths $\ell$ are 1. Note that, as in the case of the fusion tensors, the action tensors enjoy gauge freedom given by a scalar $\gamma_{g,x}$
	\begin{equation}
		A^{\ammaG}_{g,x}\to \gamma_{g,x}A^{\ammaG}_{g,x}, \quad  A^{\Gamma}_{g,x}\to\dfrac{1}{\gamma_{g,x}}A^{\Gamma}_{g,x}.
		\label{eq:scalar_act_ten}
	\end{equation}
	The compatibility of MPU fusion and their action on MPS gives rise to another family of scalars, the $L$-symbols $L^x_{g,h}$, which are defined by \cite{GarreLootensMolnar}
	\begin{equation}
		\begin{tikzpicture}
			\node[irrep] at (-0.35,0.9) {$g$};
			\node[irrep] at (-0.35,0.4) {$h$};
			\node[irrep] at (-0.35,-0.1) {$x$};
			
			\draw[thick](-0.2,0) -- (2.0,0);
			\pic[scale=1.25] at (1.75,0.375) {ARl=ghx/gh/x};
			\pic[scale=0.83] at (1.2,0.75) {FR=/g/h};
			
			\foreach \y in {0.5,1}{
				\foreach \x in {0,0.4,0.8}{
					\draw[thick](\x,\y-0.25) -- (\x,\y+0.25);
					\draw[red,thick] (\x-0.25,\y) -- (\x+0.25,\y);
					\node[square, fill=violet] at (\x,\y) {};
					\node[tensor] at (\x,0) {};  
					\draw[thick](\x,0) -- (\x,+0.25); }}
			\def\Sc{1.8};  
			\node at (2.75,0.5) {$={L}^x_{g,h}$};
			\node[irrep] at (-0.35+\Sc*\Shift,0.9) {$g$};
			\node[irrep] at (-0.35+\Sc*\Shift,0.4) {$h$};
			\node[irrep] at (-0.35+\Sc*\Shift,-0.1) {$x$};
			
			\draw[red,thick] (1+\Sc*\Shift,1) -- (1.6+\Sc*\Shift,1);
			\draw[thick](-0.2+\Sc*\Shift,0)--(1.125+\Sc*\Shift,0);
			\draw[thick](1.67+\Sc*\Shift,0.0)--(2.175+\Sc*\Shift,0.0);
			\pic[scale=1.67] at (1.85+\Sc*\Shift,0.5) {ARl=ghx/g/hx};
			\pic[scale=0.83] at (1.2+\Sc*\Shift,0.25 ) {ARl=/h/x};
			\foreach \y in {0.5,1}{
				\foreach \x in {0+\Sc*\Shift,0.4+\Sc*\Shift,0.8+\Sc*\Shift}{
					\draw[thick](\x,\y-0.25) -- (\x,\y+0.25);
					\draw[red,thick] (\x-0.25,\y) -- (\x+0.25,\y);
					\node[square, fill=violet] at (\x,\y) {};
					\node[tensor] at (\x,0) {};  
					\draw[thick](\x,0) -- (\x,+0.25); }}
			\node[] at (6.45,0.0) {$.$};
		\end{tikzpicture}
		\label{eq:defL}
	\end{equation}
	These $L$-symbols are subject to the scalar gauge freedom deriving from both the fusion and action tensors,
	\begin{equation}
		L^x_{g,h}\rightarrow \dfrac{\gamma_{gh,x}\beta_{g,h}}{\gamma_{g, hx}\gamma_{h,x}}L^x_{g,h},
		\label{eq:Lsymbgauge}
	\end{equation}
	with $\beta_{g,h}$ as in \eqref{eq:scalar_fus_ten} and $\gamma_{g,x}$ as in \eqref{eq:scalar_act_ten}. Given a choice of gauge for the fusion and action tensors, the corresponding anomaly 3-cocycle and the $L$-symbols are related by
	\begin{equation}
		L^x_{g, hk} L^x_{h,k} = \omega(g,h,k) L^{kx}_{g, h} L^{x}_{gh, k}.
		\label{eq:omega_and_Ls}
	\end{equation}
	In particular, whenever $|\mathsf{X}|=1$ (i.e. we have an injective invariant MPS), \eqref{eq:omega_and_Ls} reduces to the 3-coboundary condition for $\omega$, and there is a choice of gauge such that $\omega(g,h,k) = L_{g,h}=1$.
	
	The $L$-symbols were introduced in \cite{GarreLootensMolnar} as a tool to classify the SPT phases with matrix product operator symmetry, with one such phase existing for any equivalence class of solutions of \eqref{eq:omega_and_Ls}, modulo the gauge freedom $\gamma_{g,x}$ derived from that of the action tensors. Whenever the matrix product operator symmetry represents a finite group $G$, these solutions are labeled by pairs $(H, \psi)$ where $H\leq G$ is a subgroup that trivializes the anomaly (i.e., $[\omega|_H]=1$), and $\psi\in H^2(H, \C^\times)$ is a $2$-cocycle representing a $2$-cohomology class of $H$. $H$ is the unbroken symmetry group, so that the ground state degeneracy of a Hamiltonian in the corresponding phase, i.e. the minimal number of injective blocks of an invariant MPS, is 
	\begin{equation}
		|\mathsf{X}|=\left|\dfrac{G}{H}\right|.
	\end{equation}
	For example, a group $G$ represented on-site has a trivial anomaly, so we can choose $H=G$. The resulting phases with fully unbroken symmetry are then labeled by the second cohomology group $H^2(G, \C^\times)$ or, equivalently, by $H^2(G,U(1))$ \footnote{Since $U(1)$ is an abelian topological group (equipped with the standard topology) that is homotopically equivalent to $\C^\times$ (also understood as an abelian topological group), there is an induced isomorphism in sheaf cohomology, $H^2(X,\C^\times)\cong H^2(X,U(1))$ for $X$ being a CW-complex, cf. \cite{Brown}}.
	
	\subsection{Assumptions for the rest of the paper}
	\label{sec:assumptions}
	In the rest of this work, we will make certain choices for the gauges of the MPU tensors $\mc U_g$, fusion tensors $F^{<,>}_{g,h}$ and action tensors $A^{\Gamma, \ammaG}_{g,x}$ that will simplify the formalism. For the convenience of the reader, we present them all here. 
	
	To begin with, we will always assume that the tensor representing the neutral element $e$ of the group is the identity $\mathcal{U}_e=\I$, with bond dimension $1$, and any fusion or action tensor involving $e$ should be trivial. As a consequence thereof, the $3$-cocycle $\omega(g,h,k)$ and the $L$-symbols $L_{g,h}^x$ are normalized, that is
	\begin{align}
		\omega(g, h, e) = \omega(g, h, e) = \omega(g, h, e) &= 1,\qquad \forall g,h,\label{eq:triv_omegas}\\
		L^x_{g,e} = L^x_{e, g} &= 1,\qquad \forall g.\label{eq:triv_Ls}
	\end{align}
	
	Moving on to the compatibility between the dagger and the inverse (cf. Section \ref{sec:daggerinverse}), whenever we have $g\neq g^{-1}$, we will assume that we have made a gauge choice for the tensor such that $T_g=\mathds{1}$, i.e. $\mc U_g^\dagger = \mc U_{\inv{g}}$. Moreover, we will choose the $X$ and $Y$ tensors in such a way that 
	\begin{equation}
		X_{1}^{\inv{g}}=\bigl(Y_{2}^g\bigr)^\dagger,\qquad Y_{1}^{\inv{g}}=\bigl(X_{2}^g\bigr)^\dagger.
		\label{eq:compat}
	\end{equation}
	This, for example, implies that the $(u_{\inv{g}},v_{\inv{g}})$ gates that give the characterization of the two-layer circuit of $U_{\inv g}$ are given by $(v^\dagger_{g},u^\dagger_{g})$.
	
	If $g$ is of order $2$, i.e. $g^{-1} = g$, and $\sigma_g=1$, there is a gauge transformation that makes the MPU tensor equal to its Hermitian conjugate, so we will once again assume that we have made that gauge choice and thus $T_g=\I$. Whenever $\sigma_g = -1$, this is not possible, $T_g$ will be nontrivial and has to be carried along in computations. Note that with our definitions,
	\begin{equation}
		T_g = T_{\inv g} = \sigma_g T^T_{g}.
	\end{equation}
	In Appendix \ref{app:Z2} we show that we can then pick the $X,Y$ tensors so that
	\begin{equation}
		\begin{tikzpicture}[baseline=-1mm]
			\tikzset{decoration={snake,amplitude=.4mm,segment length=2mm, post length=0mm, pre length=0mm}}
			\draw[red, thick] (0,0)--++(0.5,0.);
			\draw[thick] (0,0)--++(0.0,0.4);
			\draw[decorate, thick] (0,0)--++(0.0,-0.4);
			\node[tensor, fill=white] at (0,0) {};
			\node[] at (-0.45, 0) {$X^{\inv g}_1$};
		\end{tikzpicture}
		=
		\begin{tikzpicture}[baseline=-1mm]
			\draw[red, thick] (0,0)--++(0.7,0.);
			\draw[thick] (0,0)--++(0.0,0.4);
			\draw[ultra thick] (0,0)--++(0.0,-0.4);
			\node[tensor] at (0,0) {};
			\node[tensor, red] at (0.4,0) {};
			\node[] at (-0.5, 0) {$(Y^g_2)^\dagger$};
			\node[] at (0.15, 0.15) {$\dagger$};
			\node[red] at (0.45, 0.3) {$T^\dagger_g$};
		\end{tikzpicture}
		,\quad
		\begin{tikzpicture}[baseline=-1mm]
			\tikzset{decoration={snake,amplitude=.4mm,segment length=2mm, post length=0mm, pre length=0mm}}
			\draw[red, thick] (0,0)--++(-0.5,0.);
			\draw[thick] (0,0)--++(0.0,0.4);
			\draw[decorate, thick] (0,0)--++(0.0,-0.4);
			\node[tensor, fill=white] at (0,0) {};
			\node[] at (0.45, 0) {$Y^{\inv g}_1$};
		\end{tikzpicture}
		=
		\begin{tikzpicture}[baseline=-1mm]
			\draw[red, thick] (0,0)--++(-0.6,0.);
			\draw[thick] (0,0)--++(0.0,0.4);
			\draw[ultra thick] (0,0)--++(0.0,-0.4);
			\node[tensor] at (0,0) {};
			\node[] at (0.6, 0) {$(X^g_2)^\dagger$};
			\node[tensor, red] at (-0.35,0) {};
			\node[] at (0.15, 0.15) {$\dagger$};
			\node[red] at (-0.4, 0.27) {$T_g$};
		\end{tikzpicture}
		\label{eq:modif_XY}
		,
	\end{equation}
	which generalizes \eqref{eq:compat} to include $g^{-1} = g$, while keeping all the pleasant properties of the $X$ and $Y$ tensors listed in Section \ref{sec:MPUsstr}.
	
	Once we have $T_g$ as above, we will make the choice
	\begin{equation}
		\begin{tikzpicture}
			\pic[scale=0.83] at (0.0,0.0) {V=//};
			\node[irrep] at (0.25,0.3) {$g$};
			\node[irrep] at (0.3,-0.35) {$\inv{g}$};
			\node[] at (0.7,0.0) {$=$};
			\def\sx{1.0};
			\pic[scale=0.83] at (0.0+\sx,0.0) {V=//};
			\node[irrep] at (0.3+\sx,0.3) {$\inv{g}$};
			\node[irrep] at (0.25+\sx,-0.35) {$g$};
			\node[] at (0.7+\sx,0.0) {$=$};
			\draw[red, thick] (0.0+2*\sx,-0.35)--++(0.0,0.7);
			\draw[red, thick] (0.0+2*\sx,-0.35)--++(0.4,0.0);
			\draw[red,thick] (0.0+2*\sx,0.35)--++(0.4,0.0);
			\node[tensor,red] at (0.2+2*\sx,0.35) {};
			\node[red] at (0.4+2*\sx,0.55) {$T_g^\dagger$};
			\node[] at (0.5+2*\sx,0.0) {$,$};
			
			\pic[scale=0.83] at (0.0+3.5*\sx,0.0) {W=//};
			\node[irrep] at (-0.25+3.5*\sx,0.3) {$g$};
			\node[irrep] at (-0.3+3.5*\sx,-0.35) {$\inv{g}$};
			\node[] at (0.4+3.5*\sx,0.0) {$=$};
			\pic[scale=0.83] at (0.0+4.5*\sx,0.0) {W=//};
			\node[irrep] at (-0.25+4.5*\sx,0.3) {$\inv{g}$};
			\node[irrep] at (-0.3+4.5*\sx,-0.35) {$g$};
			\node[] at (0.4+4.5*\sx,0.0) {$=$};
			\draw[red, thick] (0.0+5.5*\sx,-0.35)--++(0.0,0.7);
			\draw[red, thick] (0.0+5.5*\sx,-0.35)--++(-0.4,0.0);
			\draw[red,thick] (0.0+5.5*\sx,0.35)--++(-0.4,0.0);
			\node[square, fill=white] at (0.0+5.5*\sx,0.0) {};
			\node[] at (0.2+5.5*\sx,0.0) {$\rho$};
			\node[red,tensor] at (-0.2+5.5*\sx,0.35) {};
			\node[red] at (0.0+5.5*\sx,0.5) {$T_g$};
			\node[] at (0.4+5.5*\sx,0.0) {$,$};
		\end{tikzpicture}
		\label{eq:Fgg}
	\end{equation}
	thus fixing the fusion tensors $F^{<,>}_{g, g^{-1}}$. As a consequence, we have,
	\begin{align}
		\omega(g, g^{-1}, g) = \dfrac{L^{x}_{\inv g, g}}{L^{gx}_{g, \inv g}} = \sigma_g,\qquad \forall g,
		\label{eq:wggg}
	\end{align}
	which is proven in Appendix \ref{app:Z2}.
	The choices we have made so far for the fusion tensors also result in 
	\begin{equation}
		\zeta_{e, g} = \zeta_{g, e} = 1, \qquad \zeta_{g,\inv{g}} = \sigma_g.
		\label{eq:zetas_in_gauge}
	\end{equation}
	
	Later in this paper, we will make subsequent gauge choices to ensure convenient properties such as the well-definedness of the defect tensors and the unitarity of the defect fusion operators. For these conditions to be compatible with the ones made in this section, we will from now on always restrict to gauge transformations satisfying
	\begin{align}
		\beta_{g,e} = \beta_{e,g} = \beta_{g,g^{-1}} &= 1, \qquad \forall g,\\
		\gamma_{e,x} &= 1, \qquad \forall x.
	\end{align}

	\section{Defect systems}
	In this Section, we will explain how a defect system, i.e. a collection of defect tensors together with unitary operators that move them and fuse them, arises from an MPU-invariant MPS. To ease the presentation, we will first study the injective MPS case and later move on to the noninjective case, where defects can be interpreted as domain walls separating different ground states of the parent Hamiltonian.
	
	\label{sec:defs}
	\subsection{Injective MPS} \label{sec:inj}
	As in the motivating example of the onsite symmetry in Section \ref{sec:onsite}, we introduce symmetry defects as the modifications caused at the boundary of a region by the action of the symmetry restricted to that region. Thus we need spatially truncated versions of our MPUs. We achieve this by defining:
	\begin{equation}
		\begin{tikzpicture}[scale=\TikzScaling]
			\tikzset{decoration={snake,amplitude=.4mm,segment length=2mm, post length=0mm,pre length=0mm}}
			
			\node[] at (-1,0.6) {$U_g^L$};
			\node[] at (-0.65,0.6) {$=$};
			
			\draw[thick,red] (-0.3, 0.6) --++ (0.75,0);
			\draw[thick,red] (0.65, 0.6) --++ (0.7,0);
			\node[thick, red] at (0.5,0.6) {$\ldots$};
			
			\draw[thick] (0,0.3) --++ (0,0.6);
			\node[square, fill=violet] at (0,0.6) {};
			
			\draw[thick] (0.3,0.3) --++ (0,0.6);
			\node[square, fill=violet] at (0.3,0.6) {};
			
			\draw[thick] (0.8,0.3) --++ (0,0.6);
			\node[square, fill=violet] at (0.8,0.6) {};
			
			\draw[thick] (1.1,0.3) --++ (0,0.6);
			\node[square, fill=violet] at (1.1,0.6) {};
			
			\node[tensor] at (-0.3,0.6) {};
			\node[tensorempty,thick] at (1.4,0.6) {};
			
			\draw[ultra thick] (-0.3,0.6) --++ (0,0.3) ;
			\draw[thick] (-0.3,0.6) --++ (0,-0.3) ;
			
			\draw[decorate, thick] (1.4,0.65) --++ (0,0.3) ;
			\draw[thick] (1.4,0.55) --++ (0,-0.3) ;
			
			\node[] at (0.55, 1.1) {$\overbrace{ \hspace{1.8cm} }^{L}$};
			
		\end{tikzpicture}
		\label{eq:truncsym}
	\end{equation}
	It can be checked that $U^L_g$ is a unitary, and it satisfies the conditions for it to be a truncated version of the symmetry (stated in \cite{Seifnashri} as the fact that it acts on local operators inside the region and far from the boundaries as the global symmetry would do). Note the growth of these operators happens via the movement operators:
	\begin{equation}
		\begin{tikzpicture}[scale=\TikzScaling]
			
			\tikzset{decoration={snake,amplitude=.4mm,segment length=2mm, post length=0mm,pre length=0mm}}
			
			\draw[thick,red] (-0.3, 0.6) --++ (0.75,0);
			\draw[thick,red] (0.65, 0.6) --++ (0.7,0);
			\node[thick, red] at (0.5,0.6) {$\ldots$};
			
			\draw[thick] (0,0.3) --++ (0,0.6);
			\node[square, fill=violet] at (0,0.6) {};
			
			\draw[thick] (0.3,0.3) --++ (0,0.6);
			\node[square, fill=violet] at (0.3,0.6) {};
			
			\draw[thick] (0.8,0.3) --++ (0,0.6);
			\node[square, fill=violet] at (0.8,0.6) {};
			
			\draw[thick] (1.1,0.3) --++ (0,0.6);
			\node[square, fill=violet] at (1.1,0.6) {};
			
			\node[tensor] at (-0.3,0.6) {};
			\node[tensorempty,thick] at (1.4,0.6) {};
			
			\draw[ultra thick] (-0.3,0.6) --++ (0,0.3) ;
			\draw[thick] (-0.3,0.6) --++ (0,-0.3) ;
			
			\draw[decorate, thick] (1.4,0.65) --++ (0,0.3) ;
			\draw[thick] (1.4,0.55) --++ (0,-0.3) ;
			
			\draw[ultra thick] (-0.5,1) --++ (0,0.35) ;
			\draw[thick] (-0.3,1) --++ (0,0.35) ;
			\draw[thick] (-0.5,1) --++ (0,-0.35) ;
			
			\draw[decorate, thick] (1.6,1) --++ (0,0.35) ;
			\draw[thick] (1.4,1.) --++ (0,0.35) ;
			\draw[thick] (1.6,1) --++ (0,-0.35) ;
			\node[] at (0.55, 1.1) {$\overbrace{\hspace{1.8cm}}^{L}$};
			
			\node[rectangle, fill=citrine] at (1.5,1) {$w_R$};
			
			\node[rectangle, fill=carrotorange] at (-0.4,1) {$w_L$};
			
			\def\Sc{1.2};
			\node[] at (1.9, 0.65) {$=$};
			
			\draw[thick,red] (-0.3+\Sc*\Shift, 0.6) --++ (0.75,0);
			\draw[thick,red] (0.65+\Sc*\Shift, 0.6) --++ (0.7,0);
			\node[thick, red] at (0.5+\Sc*\Shift,0.6) {$\ldots$};
			
			\draw[thick] (0+\Sc*\Shift,0.3) --++ (0,0.6);
			\node[square, fill=violet] at (0+\Sc*\Shift,0.6) {};
			
			\draw[thick] (0.3+\Sc*\Shift,0.3) --++ (0,0.6);
			\node[square, fill=violet] at (0.3+\Sc*\Shift,0.6) {};
			
			\draw[thick] (0.8+\Sc*\Shift,0.3) --++ (0,0.6);
			\node[square, fill=violet] at (0.8+\Sc*\Shift,0.6) {};
			
			\draw[thick] (1.1+\Sc*\Shift,0.3) --++ (0,0.6);
			\node[square, fill=violet] at (1.1+\Sc*\Shift,0.6) {};
			
			\node[tensor] at (-0.3+\Sc*\Shift,0.6) {};
			\node[tensorempty,thick] at (1.4+\Sc*\Shift,0.6) {};
			
			\draw[ultra thick] (-0.3+\Sc*\Shift,0.6) --++ (0,0.3) ;
			\draw[thick] (-0.3+\Sc*\Shift,0.6) --++ (0,-0.3) ;
			
			\draw[decorate, thick] (1.4+\Sc*\Shift,0.65) --++ (0,0.3) ;
			\draw[thick] (1.4+\Sc*\Shift,0.55) --++ (0,-0.3) ;
			
			\node[] at (0.55+\Sc*\Shift, 1.1) {$\overbrace{\hspace{1.8cm}}^{L+2}$};
			
		\end{tikzpicture}
		\label{eq:move_trunc_sym}
	\end{equation}
	Thanks to our choice \eqref{eq:compat} we can also ensure that 
	\begin{equation}
		U^L_g U^\infty_{g^{-1}} = U^{\infty-L}_{g^{-1}},
		\label{eq:complementary}
	\end{equation}
	where $U^{\infty}_g$ denotes the MPU group representation on the whole chain with periodic boundary conditions and $U^{\infty-L}_{\inv g}$ denotes the truncated symmetry supported in the complement of $U_g^L$. This will be a useful property later.
	
	Let us now show that these spatially truncated symmetries give rise to well-defined defect tensors, this time sitting on sites instead of bonds.
	\begin{proposition}
		There exists a (non-unique) choice of gauge for the action tensors such that the defect tensors defined by
		\begin{equation}
			\begin{tikzpicture}[baseline]
				\draw[thick] (0,0.0) --++ (0,0.6);
				\draw[thick] (-0.4, 0) --++ (0.8,0);
				\node[square] at (0,0) {};
				\node[] at (0.0,-0.3) {$g$};
			\end{tikzpicture}
			\equiv
			\begin{tikzpicture}[baseline]
				\tikzset{decoration={snake,amplitude=.4mm,segment length=2mm, post length=0mm,pre length=0mm}}
				\def\dx{0.3}
				\def\dxb{0.6}
				\def\dy{0.5}
				\def\dyb{1.0}
				\draw[thick] (-\dxb,0.0) --++ (2*\dxb, 0);
				\draw[thick] (\dx, 0) --++ (0.,\dy);
				\draw[thick, decorate] (\dx,\dy) -- (\dx,\dyb);
				\draw[thick, red] (-\dx, \dy) --++ (2*\dx,0);
				\pic [] at (-\dx,0.) {actL=\dy//g/};
				\node[tensor, thick, fill=white] at (\dx,\dy) {};
				\node[tensor, thick] at (\dx,0) {};
			\end{tikzpicture}
		\end{equation}
		satisfy
		
		\begin{equation}
			\begin{tikzpicture}[baseline]
				\draw[thick] (0,0.0) --++ (0,0.6);
				\draw[thick] (-0.4, 0) --++ (0.8,0);
				\node[square] at (0,0) {};
				\node[] at (0.0,-0.3) {$g$};
			\end{tikzpicture}
			=
			\begin{tikzpicture}[baseline]
				\def\dx{-0.30}
				\def\dxb{0.6}
				\def\dy{0.5}
				\def\dyb{1.0}
				\draw[thick] (-\dxb,0.0) --++ (2*\dxb, 0);
				\draw[thick] (\dx, 0) --++ (0.,\dy);
				\draw[ultra thick] (\dx,\dy) -- (\dx,\dyb);
				\draw[thick, red] (-\dx, \dy) --++ (2*\dx,0);
				\pic [] at (-\dx,0.) {actR=\dy/\quad\inv{g}//};
				\node[tensor, thick] at (\dx,\dy) {};
				\node[tensor, thick] at (\dx,0) {};
			\end{tikzpicture}
		\end{equation}
		and, for any length $L\geq 0$,
		\begin{equation}
			\begin{tikzpicture}[scale=\TikzScaling]
				\def\dx{1.2}
				\tikzset{decoration={snake,amplitude=.4mm,segment length=2mm, post length=0mm,pre length=0mm}}
				
				\draw[thick] (-\dx, 0) -- (\dx,0);
				\node[tensor] at (-0.8,0) {};
				\node[tensor] at (-0.4,0) {};
				\node[tensor] at (0,0) {};
				\node[tensor] at (0.4,0) {};
				\node[tensor] at (0.8,0) {};
				
				\draw[thick, red] (-0.8, 0.4) -- (0.75, 0.4);

				\draw[thick] (-0.8,0) --++ (0,0.35);
				\draw[ultra thick] (-0.8, 0.45) --++ (0,0.25);
				\draw[thick] (-0.4,0) --++ (0,0.7);
				\draw[thick] (0.0,0) --++ (0,0.7);
				\draw[thick] (0.4,0) --++ (0,0.7);
				\draw[thick] (0.8,0) --++ (0,0.35);
				\draw[decorate, thick] (0.8,0.45) --++ (0,0.25);

				\node[tensor] at (-0.8,0.4) {};
				\node[tensorempty,thick] at (0.8,0.4) {};
				
				\node[square, fill=violet] at (-0.4,0.4) {};
				\node[square, fill=violet] at (0,0.4) {};
				\node[square, fill=violet] at (0.4,0.4) {};
				\node at (0.,0.9) {$\overbrace{\hspace{1.5cm}}^{L}$};

				\def\Sc{0.5};
				
				\node[] at (-\dx-0.3, 0-\Sc*\Shift) {$=$};
				
				\draw[thick] (-\dx, 0-\Sc*\Shift) --++ (2*\dx,0);
				\node[square] at (-0.8,0-\Sc*\Shift) {};
				\node[tensor] at (-0.4,0-\Sc*\Shift) {};
				\node[tensor] at (0,0-\Sc*\Shift) {};
				\node[tensor] at (0.4,0-\Sc*\Shift) {};
				\node[square] at (0.8,0-\Sc*\Shift) {};

				\draw[thick] (-0.8,0-\Sc*\Shift) --++ (0,0.5);
				\draw[thick] (-0.4,0-\Sc*\Shift) --++ (0,0.5);
				\draw[thick] (0.0,0-\Sc*\Shift) --++ (0,0.5);
				\draw[thick] (0.4,0-\Sc*\Shift) --++ (0,0.5);
				\draw[thick] (0.8,0-\Sc*\Shift) --++ (0,0.5);

				\node[] at (-0.8,-0.2-\Sc*\Shift) {$\inv{g}$};
				\node[] at (0.8,-0.2-\Sc*\Shift) {$g$};
				
				\node[] at (\dx+0.1, 0-\Sc*\Shift) {$.$}; 
			\end{tikzpicture}
		\end{equation}
		\label{prop:1}
	\end{proposition}
	
	To prove the above proposition, we will need the following.
	\begin{lemma}
		\label{lemma:1}
		Let $\{A_i\}_{i=1}^n$ be injective MPS tensors, and let 
		\begin{equation}
			\begin{tikzpicture}[scale=\TikzScaling]
				
				\tikzset{decoration={snake,amplitude=.4mm,segment length=2mm, post length=0mm,pre length=0mm}}
				
				\draw[thick] (-1.8, 0) --++ (1.6,0);
				\draw[thick] (0.2,0) --++ (0.8,0);
				\node[tensor] at (-1.6,0) {};
				\node[tensor] at (-1.2,0) {};
				\node[tensor] at (-0.8,0) {};
				\node[tensor] at (-0.4,0) {};
				\node[] at (0,0) {$\ldots$};
				\node[tensor] at (0.4,0) {};
				\node[tensor] at (0.8,0) {};

				\draw[thick] (-1.6,0) --++ (0,0.3);
				\draw[thick] (-1.2,0) --++ (0,0.3);
				\draw[thick] (-0.8,0) --++ (0,0.3);
				\draw[thick] (-0.4,0) --++ (0,0.3);
				\draw[thick] (0.4,0) --++ (0,0.3);
				\draw[thick] (0.8,0) --++ (0,0.3);
				
				\node[] at (-1.6,-0.2) {$A_1$};
				\node[] at (-1.2,-0.2) {$B_1$};
				\node[] at (-0.8,-0.2) {$A_2$};
				\node[] at (-0.4,-0.2) {$B_2$};
				\node[] at (0.4,-0.2) {$A_n$};
				\node[] at (0.8,-0.2) {$B_n$};
				
				\def\Sc{0.35};
				
				\node[] at (-2.0, 0-\Sc*\Shift) {$=$};
				
				\draw[thick] (-1.8, 0-\Sc*\Shift) --++ (1.6,0);
				\draw[thick] (0.2,0-\Sc*\Shift) --++ (0.8,0);
				\node[tensor] at (-1.6,0-\Sc*\Shift) {};
				\node[tensor] at (-1.2,0-\Sc*\Shift) {};
				\node[tensor] at (-0.8,0-\Sc*\Shift) {};
				\node[tensor] at (-0.4,0-\Sc*\Shift) {};
				\node[] at (0,0-\Sc*\Shift) {$\ldots$};
				\node[tensor] at (0.4,0-\Sc*\Shift) {};
				\node[tensor] at (0.8,0-\Sc*\Shift) {};

				\draw[thick] (-1.6,0-\Sc*\Shift) --++ (0,0.3);
				\draw[thick] (-1.2,0-\Sc*\Shift) --++ (0,0.3);
				\draw[thick] (-0.8,0-\Sc*\Shift) --++ (0,0.3);
				\draw[thick] (-0.4,0-\Sc*\Shift) --++ (0,0.3);
				\draw[thick] (0.4,0-\Sc*\Shift) --++ (0,0.3);
				\draw[thick] (0.8,0-\Sc*\Shift) --++ (0,0.3);
				
				\node[] at (-1.6,-0.2-\Sc*\Shift) {$A_1$};
				\node[] at (-1.2,-0.2-\Sc*\Shift) {$B_1'$};
				\node[] at (-0.8,-0.2-\Sc*\Shift) {$A_2$};
				\node[] at (-0.4,-0.2-\Sc*\Shift) {$B_2'$};
				\node[] at (0.4,-0.2-\Sc*\Shift) {$A_n$};
				\node[] at (0.8,-0.2-\Sc*\Shift) {$B_n'$};
				
				\node[] at (1.1, 0-\Sc*\Shift) {$.$}; 
			\end{tikzpicture}
			\label{eq:lemma_equality}
		\end{equation}
		Then there exist scalars $\beta_i$, $i=1,\ldots, n$, $\prod_{i=1}^n{\beta_i}=1$, such that $B_i = \beta_i B'_i$.
	\end{lemma}
	\begin{proof}
		Since the $A_i$ are injective, they have inverse tensors such that
		\begin{equation}
			\begin{tikzpicture}[scale=\TikzScaling]
				\def\y{0.6}
				\draw[thick] (-0.4,0) --++ (0.8,0);
				\draw[thick] (-0.4,\y) --++ (0.8,0);
				\draw[thick] (0,0) --++ (0,\y);
				\node[tensor] at (0,0) {};
				\node[tensor] at (0,\y) {};
				\node[] at (0.1, -0.2) {$A_i^{\phantom{-1}}$};
				\node[] at (0.1, \y+0.25) {$A^{-1}_i$};
				\node[] at (0.65,0.3) {$=$};
				\def\Sc{0.6}
				\def\y{0.6}
				\def\dx{0.05}
				\draw[thick] (-0.4+\Sc*\Shift,0) --++ (0.4-\dx,0);
				\draw[thick] (0.4+\Sc*\Shift,0) --++ (-0.4+\dx,0);
				\draw[thick] (-0.4+\Sc*\Shift,\y) --++ (0.35, 0);
				\draw[thick] (0.05+\Sc*\Shift,\y) --++ (0.35, 0);
				\draw[thick] (-0.05+\Sc*\Shift,0) --++ (0,\y);
				\draw[thick] (0.05+\Sc*\Shift,0) --++ (0,\y);
				\node[] at (3*\Sc, 0.5*\y) {$.$};
			\end{tikzpicture}
		\end{equation}
		Applying them to both sides of \eqref{eq:lemma_equality} we have
		\begin{equation}
			\bigotimes_{i=1}{B_i} = \bigotimes_{i=1}{B'_i}
		\end{equation}
		and the conclusion follows.
	\end{proof}
	We now proceed with the proof of the proposition.
	\begin{proof}[Proof of Proposition~\ref{prop:1}] Consider first Eq. \eqref{eq:action_interior}, and decompose the tensors at the boundaries:
		\begin{equation}
			\begin{tikzpicture}[baseline=+3mm]
				\tikzset{decoration={snake,amplitude=.4mm,segment length=2mm, post length=0mm,pre length=0mm}}
				\def\dy{0.5}
				\draw[thick] (0.3,0.) --++ (1.7,0.);
				\draw[thick] (-0.3,0.) --++ (-1.7,0.);
				\node[] at (0.,0.0) {$\ldots$};
				\foreach \x in {-0.9,-0.5,0.5,0.9}{
					\node[tensor] at (\x, 0) {}; }
				\draw[thick] (-0.9,0.) --++ (0.,0.6);
				\draw[thick] (0.9,0.) --++ (0.,0.6);
				\draw[ultra thick] (0.9,0.6) --++ (0.0,\dy);
				\draw[thick, decorate] (-0.9,0.6) --++ (0.,\dy);
				\draw[thick] (-0.9,0.6+\dy) --++ (0.,0.3);
				\draw[thick] (0.9,0.6+\dy) --++ (0.,0.3);
				\draw[thick] (-0.5, 0.) -- (-0.5, 0.6+\dy+0.3);
				\draw[thick] (0.5, 0.) -- (0.5, 0.6+\dy+0.3);
				\draw[thick, red] (0.3,0.6+\dy) --++ (0.6,0.);
				\draw[thick, red] (-0.9,0.6+\dy) --++ (0.6,0.);
				
				\pic at (-1.4,0.3) {AL=g};
				\pic at (1.4,0.3) {AR=g};
				\node[red] at (0.,0.6+\dy) {$\ldots$};
				\node[square, fill=violet] at (0.5, 0.6+\dy) {};
				\node[square, fill=violet] at (-0.5, 0.6+\dy) {};
				\node[tensor, fill=white] at (-0.9, .6) {};
				\node[tensor, fill=white] at (-0.9, 0.6+\dy) {};
				\node[tensor] at (0.9, .6) {};
				\node[tensor] at (0.9, 0.6+\dy) {};
			\end{tikzpicture}
			=
			\begin{tikzpicture}[baseline=3mm]
				\tikzset{decoration={snake,amplitude=.4mm,segment length=2mm, post length=0mm,pre length=0mm}}
				\def\dy{0.5}
				\draw[thick] (0.3,0.) --++ (1.,0.);
				\draw[thick] (-0.3,0.) --++ (-1.,0.);
				\node[] at (0.,0.0) {$\ldots$};
				\foreach \x in {-0.9,-0.5,0.5,0.9}{
					\node[tensor] at (\x, 0) {}; }
				\draw[thick] (-0.9,0.) --++ (0.,0.6);
				\draw[thick] (0.9,0.) --++ (0.,0.6);
				\draw[thick] (-0.5,0.) --++ (0.,0.6);
				\draw[thick] (0.5,0.) --++ (0.,0.6);
				\node[] at (1.4,0.) {$.$};
			\end{tikzpicture}
		\end{equation}
		Thanks to \eqref{eq:compat}, we can identify the tensors on top as the inverse of the truncated symmetry of $g^{-1}$:
		\begin{equation}
			(U^L_{g^{-1}})^\dagger = 
			\begin{tikzpicture}[baseline=+3mm]
				\tikzset{decoration={snake,amplitude=.4mm,segment length=2mm, post length=0mm,pre length=0mm}}
				\def\dy{0.4}
				\draw[ultra thick] (0.9,0.) --++ (0.0,\dy);
				\draw[thick, decorate] (-0.9,0.) --++ (0.,\dy);
				\draw[thick] (-0.9,0.+\dy) --++ (0.,0.3);
				\draw[thick] (0.9,0.+\dy) --++ (0.,0.3);
				\draw[thick] (-0.5, 0.) -- (-0.5, 0.+\dy+0.3);
				\draw[thick] (0.5, 0.) -- (0.5, 0.+\dy+0.3);
				\draw[thick, red] (0.3,0.+\dy) --++ (0.6,0.);
				\draw[thick, red] (-0.9,0.+\dy) --++ (0.6,0.);

				\node[red] at (0.,0.+\dy) {$\ldots$};
				\node[square, fill=violet] at (0.5, 0.+\dy) {};
				\node[square, fill=violet] at (-0.5, 0.+\dy) {};
				\node[tensor, fill=white] at (-0.9, 0.+\dy) {};
				\node[tensor] at (0.9, 0.+\dy) {};
			\end{tikzpicture}
		\end{equation}
		resulting in 
		\begin{equation}
			\begin{tikzpicture}
				\def\dy{0.6}
				\tikzset{decoration={snake,amplitude=.4mm,segment length=2mm, post length=0mm,pre length=0mm}}
				\draw[thick] (0.3,0.) --++ (1.,0.);
				\draw[thick] (-0.3,0.) --++ (-1.,0.);
				\node[] at (0.,0.0) {$\ldots$};
				\foreach \x in {-0.9,-0.5,0.5,0.9}{
					\node[tensor] at (\x, 0) {}; }
				\draw[thick] (-0.9,0.) --++ (0.,\dy);
				\draw[thick] (0.9,0.) --++ (0.,\dy);
				\draw[thick] (-0.5,0.) --++ (0.,\dy+0.3);
				\draw[thick] (0.5,0.) --++ (0.,\dy+0.3);
				\draw[ultra thick] (-0.9,\dy) --++ (0.0,0.3);
				\draw[thick, decorate] (+0.9,\dy) --++ (0.,0.3);
				\draw[thick, red] (0.3,\dy) --++ (0.6,0.);
				\draw[thick, red] (-0.9,\dy) --++ (0.6,0.);
				\node[red] at (0.,\dy) {$\ldots$};
				\node[square,fill=violet] at (0.5, \dy) {};
				\node[square,fill=violet] at (-0.5, \dy) {};
				\node[tensor, fill=white] at (0.9, \dy) {};
				\node[tensor] at (-0.9, \dy) {};
				\node[] at (-1.1, \dy) {$g$};
			\end{tikzpicture}
			=
			\begin{tikzpicture}
				\def\dy{0.6}
				\def\dx{0.9}
				\def\dxb{1.4}
				\tikzset{decoration={snake,amplitude=.4mm,segment length=2mm, post length=0mm,pre length=0mm}}
				\draw[thick] (0.3,0.) --++ (1.3,0.);
				\draw[thick] (-0.3,0.) --++ (-1.3,0.);
				\node[] at (0.,0.0) {$\ldots$};
				\foreach \x in {-\dx,-0.5,0.5,\dx}{
					\node[tensor] at (\x, 0) {}; }
				\draw[thick] (-\dx,0.) --++ (0.,\dy);
				\draw[thick] (\dx,0.) --++ (0.,\dy);
				\draw[thick] (-0.5,0.) --++ (0.,\dy+0.1);
				\draw[thick] (0.5,0.) --++ (0.,\dy+0.1);
				\draw[ultra thick] (\dx,\dy) --++ (0.0,0.3);
				\draw[thick, decorate] (-\dx,\dy) --++ (0.,0.3);
				
				\pic at (\dxb,0.3) {AR=g^{-1}};
				\pic at (-\dxb,0.3) {AL=g^{-1}};
				
				\node[tensor, fill=white] at (-\dx, \dy) {};
				\node[tensor] at (\dx, \dy) {};
			\end{tikzpicture},
		\end{equation}
		where we have exchanged $g$ and $g^{-1}$, for any $L\geq 2$. For $L\geq 3$, we also have, using \eqref{eq:action_exterior}, and the fact that the $X_i$ tensors have (left) inverses,
		\begin{equation}
			\begin{tikzpicture}
				\def\dy{0.6}
				\tikzset{decoration={snake,amplitude=.4mm,segment length=2mm, post length=0mm,pre length=0mm}}
				\draw[thick] (0.3,0.) --++ (1.,0.);
				\draw[thick] (-0.3,0.) --++ (-1.,0.);
				\node[] at (0.,0.0) {$\ldots$};
				\foreach \x in {-0.9,-0.5,0.5,0.9}{
					\node[tensor] at (\x, 0) {}; }
				\draw[thick] (-0.9,0.) --++ (0.,\dy);
				\draw[thick] (0.9,0.) --++ (0.,\dy);
				\draw[thick] (-0.5,0.) --++ (0.,\dy+0.3);
				\draw[thick] (0.5,0.) --++ (0.,\dy+0.3);
				\draw[ultra thick] (-0.9,\dy) --++ (0.0,0.3);
				\draw[thick, decorate] (+0.9,\dy) --++ (0.,0.3);
				
				\draw[thick, red] (0.3,\dy) --++ (0.6,0.);
				\draw[thick, red] (-0.9,\dy) --++ (0.6,0.);
				\node[red] at (0.,\dy) {$\ldots$};
				\node[square,fill=violet] at (0.5, \dy) {};
				\node[square,fill=violet] at (-0.5, \dy) {};
				\node[tensor, fill=white] at (0.9, \dy) {};
				\node[tensor] at (-0.9, \dy) {};
				\node[] at (-1.1, \dy) {$g$};
			\end{tikzpicture}
			=
			\begin{tikzpicture}
				\def\dy{0.6}
				\def\dx{1.4}
				\tikzset{decoration={snake,amplitude=.4mm,segment length=2mm, post length=0mm,pre length=0mm}}
				\draw[thick] (0.3,0.) --++ (1.3,0.);
				\draw[thick] (-0.3,0.) --++ (-1.3,0.);
				\node[] at (0.,0.0) {$\ldots$};
				\foreach \x in {-\dx,-0.5,0.5,\dx}{
					\node[tensor] at (\x, 0) {}; }
				\draw[thick] (-\dx,0.) --++ (0.,\dy);
				\draw[thick] (\dx,0.) --++ (0.,\dy);
				\draw[thick] (-0.5,0.) --++ (0.,\dy+0.1);
				\draw[thick] (0.5,0.) --++ (0.,\dy+0.1);
				\draw[ultra thick] (-\dx,\dy) --++ (0.0,0.3);
				\draw[thick, decorate] (\dx,\dy) --++ (0.,0.3);
				
				\pic at (-0.9,0.3) {AR=g};
				\pic at (0.9,0.3) {AL=g};
				
				\node[tensor, fill=white] at (\dx, \dy) {};
				\node[tensor] at (-\dx, \dy) {};
			\end{tikzpicture}.
		\end{equation}
		Now it is time to use Lemma \ref{lemma:1} to conclude there are scalars $\alpha_g, g\in G$, with $\alpha_g=\alpha_{g^{-1}}$, such that
		\begin{align}
			\begin{tikzpicture}
				\tikzset{decoration={snake,amplitude=.4mm,segment length=2mm, post length=0mm,pre length=0mm}}
				\draw[thick] (-0.5, 0) --++ (1.5, 0);
				\draw[thick] (0.5, 0) --++(0.0, 0.6);
				\draw[thick, decorate] (0.5, 0.6) --++ (0.0, 0.3);
				\pic at (0.0,0.3) {AL=g};
				\node[tensor] at (0.5, 0.0) {};
				\node[tensor, fill=white] at (0.5, 0.6) {};
			\end{tikzpicture}
			= \alpha_g
			\begin{tikzpicture}
				\draw[thick] (-0.5, 0) --++ (1.5, 0);
				\draw[thick] (0., 0) --++(0.0, 0.6);
				\draw[ultra thick] (0., 0.6) --++ (0.0, 0.3);
				\pic at (0.5,0.3) {AR=g^{-1}};
				\node[tensor] at (0., 0.0) {};
				\node[tensor] at (0., 0.6) {};
			\end{tikzpicture},
		\end{align}
		and we can absorb this scalar into the gauge of the action tensors, for instance by doing $A^\ammaG_g\to\sqrt{\alpha_g}A^\ammaG_g$ and $A^\Gamma_g\to A^\Gamma_g/\sqrt{\alpha_g}$. This allows us to define the defect tensors and it finishes the proof. We can recognize $\alpha_g=L_{g, g^{-1}}$, and our partial gauge choice for the action tensors amounts to setting $L_{g,g^{-1}}=1$, together with the rest of assumptions made in Section \ref{sec:assumptions}. We explain this in more detail in the proof of Prop.~\ref{prop:def_tens_anomalous}.
	\end{proof}
	%%%%%%%%%%%%%%%%%%%%%%%%%
	Equation \eqref{eq:move_trunc_sym} suggests that these defects should move by the operators $w_R, w_L$ defined from the MPU tensors
	\begin{equation}
		\begin{tikzpicture}[scale=\TikzScaling]
			\draw[thick] (-0.2,-0.65) --++ (0,0.8);
			\draw[thick] (0.2,-0.65) --++ (0,0.8);
			
			\node[rectangle, fill=carrotorange] at (0,-0.2) {$w_L$};
			\draw[thick] (-0.4,-0.65) --++ (0.8,0);
			\node[tensor] at (-0.2,-0.65) {};
			\node[square] at (0.2,-0.65) {};
			\node at (0.2,-0.85) {$g$};
			
			\node[] at (0.6, -0.65) {$=$};
			
			\def\Sc{0.6};
			
			\draw[thick] (-0.4+\Sc*\Shift,-0.65) --++ (0.8,0);
			\node[square] at (-0.2+\Sc*\Shift,-0.65) {};
			\node[tensor] at (0.2+\Sc*\Shift,-0.65) {};
			\node at (-0.2+\Sc*\Shift,-0.85) {$g$};
			
			\draw[thick] (-0.2+\Sc*\Shift,-0.65) --++ (0,0.4);
			\draw[thick] (0.2+\Sc*\Shift,-0.65) --++ (0,0.4);
			
			\node[] at (0.5+\Sc*\Shift,-0.65) {$,$};
			
		\end{tikzpicture}\qquad
		\begin{tikzpicture}[scale=\TikzScaling]
			\draw[thick] (-0.2,-0.65) --++ (0,0.8);
			\draw[thick] (0.2,-0.65) --++ (0,0.8);
			
			\node[rectangle, fill=citrine] at (0,-0.2) {$w_R$};
			\draw[thick] (-0.4,-0.65) --++ (0.8,0);
			\node[square] at (-0.2,-0.65) {};
			\node[tensor] at (0.2,-0.65) {};
			\node at (-0.2,-0.85) {$g$};
			
			\node[] at (0.6, -0.65) {$=$};
			
			\def\Sc{0.6};
			
			\draw[thick] (-0.4+\Sc*\Shift,-0.65) --++ (0.8,0);
			\node[tensor] at (-0.2+\Sc*\Shift,-0.65) {};
			\node[square] at (0.2+\Sc*\Shift,-0.65) {};
			\node at (0.2+\Sc*\Shift,-0.85) {$g$};
			
			\draw[thick] (-0.2+\Sc*\Shift,-0.65) --++ (0,0.4);
			\draw[thick] (0.2+\Sc*\Shift,-0.65) --++ (0,0.4);
			
			\node[] at (0.5+\Sc*\Shift,-0.65) {$.$};
			
		\end{tikzpicture}
	\end{equation}
	More generally, there exist fusion operators as in the following
	\begin{proposition}
		\label{prop:3}
		Consider defect tensors defined as in Proposition~\ref{prop:1}. Then there exists a choice of the scalar freedom of the fusion and action tensors such that the operators 
		\begin{subequations}
			\begin{equation}
				\begin{tikzpicture}[scale=\TikzScaling]
					\node[] at (-0.45, 0.285) {$\lambda^L_{g,h} \equiv$};
					\draw[thick, red] (0.05,0.0) --++ (0.4,0.0); 
					\draw[thick, red] (0.05,0.0)--++ (0.0,0.5855);
					\pic (w2) at (0.4,0) {V=//};
					\draw[thick, red] (0.05,0.5855) --++ (0.7,0);
					
					\draw[thick,red] (0.75,0.2865) --++ (0.4,0);
					\node[tensor] at (0.8,0.5855) {};
					\node[tensor] at (0.8,-0.285) {};
					\node[tensor] at (1.2,0.2865) {};
					\draw[thick] (0.8,0.5855) -- (0.8,-0.285);
					\draw[ultra thick] (0.8,0.5855) --++ (0,0.4);
					\draw[ultra thick] (0.8,-0.285) --++ (0,-0.4);
					\draw[thick] (1.2,0.2865) --++ (0, 0.4);
					\draw[ultra thick] (1.2,0.2865) --++ (0, -0.4);
					\node[square, fill=violet] at (0.8,0.285) {};
					\node[] at (0.3, -0.5) {$F^<_{h^{\! -\!1}\!,g^{\!-\!1}}$};
					\node[] at (1.1, -0.3) {$X_2^{g^{\!-\!1}}$};
					\node[] at (1.05, 0.4) {$h^{\!-\!1}$};
					\node[] at (1.5,0.4) {$X_2^{h^{\!-\!1}}$};
					\node[] at (1.35, 0.8) {$\bigl(X_2^{h^{\!-\! 1}\! g^{\!-\!1}}\bigr)^{\dagger}$};
					\node[] at (1.6, 0.2) {$,$};
				\end{tikzpicture}          
			\end{equation}
			\begin{equation}
				\begin{tikzpicture}[scale=\TikzScaling]
					\tikzset{decoration={snake,amplitude=.4mm,segment length=2mm, post length=0mm,pre length=0mm}}
					
					\node[] at (-0.15, 0.285) {$\lambda^R_{g,h} \equiv$};
					\def\Sc{0.05};
					
					\draw[thick, red] (1.25+\Sc*\Shift,0.0)--++(0.35,0.0);
					\draw[thick,red] (1.6+\Sc*\Shift,0.0)--++(0.0,0.585);
					\pic (w2) at (1.25+\Sc*\Shift,0) {W=//};
					\draw[thick, red] (0.9+\Sc*\Shift,0.585) --++ (0.7,0);
					\node[tensor, fill=white] at (0.85+\Sc*\Shift,0.585) {};
					\node[tensor, fill=white] at (0.85+\Sc*\Shift,-0.28) {};
					\draw[thick, red] (0.85+\Sc*\Shift, 0.285) -- (0.55+\Sc*\Shift, 0.285);
					\draw[thick] (0.85+\Sc*\Shift,0.535) -- (0.85+\Sc*\Shift,-0.23);
					\node[tensor, fill=white] at (0.5+\Sc*\Shift,0.285) {};
					\draw[decorate, thick] (0.85+\Sc*\Shift,0.635) -- (0.85+\Sc*\Shift, 1.0);
					\draw[decorate, thick] (0.85+\Sc*\Shift,-0.33) -- (0.85+\Sc*\Shift, -0.7);
					\draw[decorate, thick] (0.5+\Sc*\Shift,0.23) -- (0.5+\Sc*\Shift, -0.2);
					\draw[thick] (0.5+\Sc*\Shift,0.34) -- (0.5+\Sc*\Shift,0.75);
					\node[square, fill=violet] at (0.85+\Sc*\Shift,0.285) {};
					\node[] at (0.3+\Sc*\Shift, 0.4) {$X_1^g$};
					\node[] at (1+\Sc*\Shift, 0.4) {$g$};
					\node[] at (0.65+\Sc*\Shift, -0.2) {$X_1^h$};
					\node[] at (1.3+\Sc*\Shift, 0.75) {$(X^{gh}_1)^\dagger$};
					\node[square, fill=white] at (1.6+\Sc*\Shift, 0.285) {};
					\node[] at (1.85+\Sc*\Shift, 0.285) {$\rho_g$};
					\node[] at (1.5+\Sc*\Shift, -0.4) {$F^>_{g,h}$};
				\end{tikzpicture}
				\label{eq:lambda_R}
			\end{equation}
			\label{def:lambdas}
		\end{subequations}
		are unitary and satisfy
		\begin{subequations}
			\begin{equation}
				\begin{tikzpicture}[scale=\TikzScaling]
					
					\node[rectangle, fill=white] at (0,0) {$\lambda^L_{g,h}$};
					
					\draw[thick] (-0.2,0.25) --++ (0,0.2);
					\draw[thick] (0.2,0.25) --++ (0,0.2);
					\draw[thick] (-0.2,-0.25) --++ (0,-0.4);
					\draw[thick] (0.2,-0.25) --++ (0,-0.4);
					
					\draw[thick] (-0.4,-0.65) --++ (0.8,0);
					\node[square] at (-0.2,-0.65) {};
					\node[square] at (0.2,-0.65) {};
					\node at (-0.2,-0.85) {$g$};
					\node at (0.2,-0.85) {$h$};
					
					\node[] at (0.6, -0.65) {$=$};
					
					\def\Sc{0.6};
					
					\draw[thick] (-0.4+\Sc*\Shift,-0.65) --++ (0.8,0);
					\node[square] at (-0.2+\Sc*\Shift,-0.65) {};
					\node[tensor] at (0.2+\Sc*\Shift,-0.65) {};
					\node at (-0.2+\Sc*\Shift,-0.85) {$gh$};
					
					\draw[thick] (-0.2+\Sc*\Shift,-0.65) --++ (0,0.4);
					\draw[thick] (0.2+\Sc*\Shift,-0.65) --++ (0,0.4);
					
					\node[] at (0.5+\Sc*\Shift,-0.65) {$,$};
					
				\end{tikzpicture}
				\label{eq:fus_def_1}
			\end{equation}
			\begin{equation}
				\begin{tikzpicture}[scale=\TikzScaling]
					
					\node[rectangle, fill=white] at (0,0) {$\lambda^R_{g,h}$};
					
					\draw[thick] (-0.2,0.25) --++ (0,0.2);
					\draw[thick] (0.2,0.25) --++ (0,0.2);
					\draw[thick] (-0.2,-0.25) --++ (0,-0.4);
					\draw[thick] (0.2,-0.25) --++ (0,-0.4);
					
					\draw[thick] (-0.4,-0.65) --++ (0.8,0);
					\node[square] at (-0.2,-0.65) {};
					\node[square] at (0.2,-0.65) {};
					\node at (-0.2,-0.85) {$g$};
					\node at (0.2,-0.85) {$h$};
					
					\node[] at (0.6, -0.65) {$=$};
					
					\def\Sc{0.6};
					
					\draw[thick] (-0.4+\Sc*\Shift,-0.65) --++ (0.8,0);
					\node[square] at (0.2+\Sc*\Shift,-0.65) {};
					\node[tensor] at (-0.2+\Sc*\Shift,-0.65) {};
					\node at (0.2+\Sc*\Shift,-0.85) {$gh$};
					
					\draw[thick] (-0.2+\Sc*\Shift,-0.65) --++ (0,0.4);
					\draw[thick] (0.2+\Sc*\Shift,-0.65) --++ (0,0.4);
					\node[] at (0.5+\Sc*\Shift,-0.65) {$.$};
					
				\end{tikzpicture}
				\label{eq:fus_def_2}
			\end{equation}
		\end{subequations}
		In particular, $\lambda^L_{g,e} = \lambda^R_{e,g} = \mathds{1}$, $\lambda^L_{e,g}=w_L^g$, and $\lambda^R_{g,e}=w_R^g$. This choice of gauge also satisfies $\omega(g,h,k)=L_{g,h} = 1$, and is compatible with the one made in Proposition \ref{prop:1} and all our assumptions (Section \ref{sec:assumptions}).
		\label{prop:2}
	\end{proposition}
	
	The proof of this proposition is included in Appendix \ref{app:proofprops}, together with that of Proposition \ref{prop:unitarity_and_gauges}, which tackles the more general noninjective case. We note that the choice of fusion and action tensors in Proposition~\ref{prop:3} is uniquely specified up to a rather constrained choice of phases.
	
	\begin{corollary}
		For any $a,b,c,d\in G$ such that $ab = cd$, there exists a unitary $\lambda_{a,b}^{c,d}$ such that 
		\begin{equation}
			\begin{tikzpicture}[scale=\TikzScaling]
				
				\node[rectangle, fill=white] at (0,0) {$\lambda^{c,d}_{a,b}$};
				
				\draw[thick] (-0.2,0.25) --++ (0,0.2);
				\draw[thick] (0.2,0.25) --++ (0,0.2);
				\draw[thick] (-0.2,-0.25) --++ (0,-0.4);
				\draw[thick] (0.2,-0.25) --++ (0,-0.4);
				
				\draw[thick] (-0.4,-0.65) --++ (0.8,0);
				\node[square] at (-0.2,-0.65) {};
				\node[square] at (0.2,-0.65) {};
				\node at (-0.2,-0.85) {$a$};
				\node at (0.2,-0.85) {$b$};
				
				\node[] at (0.6, -0.65) {$=$};
				
				\def\Sc{0.6};
				
				\draw[thick] (-0.4+\Sc*\Shift,-0.65) --++ (0.8,0);
				\node[square] at (-0.2+\Sc*\Shift,-0.65) {};
				\node[square] at (0.2+\Sc*\Shift,-0.65) {};
				\node at (-0.2+\Sc*\Shift,-0.85) {$c$};
				\node at (0.2+\Sc*\Shift,-0.85) {$d$};
				
				\draw[thick] (-0.2+\Sc*\Shift,-0.65) --++ (0,0.4);
				\draw[thick] (0.2+\Sc*\Shift,-0.65) --++ (0,0.4);
				
				\node[] at (0.5+\Sc*\Shift,-0.65) {$,$};
				
			\end{tikzpicture}
		\end{equation}
		and we have $\lambda^R_{g,h} = \lambda^{e, gh}_{g, h}$, $\lambda^L_{g,h}=\lambda^{gh,e}_{g,h}$. Furthermore, for any $a,b,c,d,e,f\in G$ that satisfies $ab = cd = fg$, we have $\lambda^{f,g}_{a,b} = \lambda_{c,d}^{f,g}\lambda_{a,b}^{c,d}$.
	\end{corollary}
	\begin{proof}
		$\lambda_{a,b}^{c,d}\equiv (\lambda_{c,d}^{R})^\dagger\lambda_{a,b}^{R}$ are a family of unitaries that meet the required conditions.
	\end{proof}
	Thus, we have mimicked the defect formalism from the previous section in the case where the symmetries are no longer onsite, at the cost of having defects lie on sites, instead of being purely virtual.

	\subsection{Noninjective MPS}
	\label{sec:noninj}
	We now turn our attention to the case where the invariant state is not injective but is made up of several injective blocks labeled $x\in \mathsf{X}$, with $|\mathsf{X}|>1$, in the notation of Section \ref{sec:MPUonMPS}. Now, defects can be domain walls that intertwine different ground states of a symmetry-broken state \cite{PhysRevB.85.100408}. Our construction of defects in what follows, in particular, gives a constructive instance of the formalism of \cite{GarreSchuch}.
	
	The definition of spatially truncated symmetries proceeds as in the previous section. We will see that defect tensors arise similarly at the boundaries when acting with a truncated symmetry. However, in the particular case where there is a $\Z_2$ anomaly, $\sigma_g=-1$, (which could not have appeared in the previous section, as an anomalous symmetry cannot have an injective invariant MPS), additional minus signs appear in the definition and properties of the defect tensors, which cannot be gauged away.
	\begin{proposition}
		There exists a (non-unique) choice of gauge for the action tensors such that the defect tensors defined by
		\begin{equation}
			\begin{tikzpicture}[baseline=10]
				\node[square] at (0,0) {};
				\draw[thick] (0,0.0) --++ (0,0.8);
				\draw[thick] (-0.4, 0) --++ (0.8,0);
				\node[] at (0.,-0.3) {$g[x]$};
			\end{tikzpicture}
			\equiv
			\begin{tikzpicture}[baseline=10]
				\tikzset{decoration={snake,amplitude=.4mm,segment length=2mm, post length=0mm,pre length=0mm}}
				\draw[thick] (1.2,0.) --++ (0.7,0);
				\draw[thick] (1.7,0.) --++ (0,0.6);
				\draw[thick, decorate] (1.7,0.6) --++ (0,0.3);
				\draw[thick] (0.8, 0) --++ (0.4, 0);
				\pic  at (1.2,0.3) {ALl=gx/g/x};
				\node[tensor, fill=white] at (1.7,0.6) {};
				\node[tensor] at (1.7,0) {};
			\end{tikzpicture}
		\end{equation}
		satisfy, for $\sigma_g=1$,
		\begin{equation}
			\begin{tikzpicture}[baseline=10]
				\node[square] at (0,0) {};
				\draw[thick] (0,0.0) --++ (0,0.8);
				\draw[thick] (-0.4, 0) --++ (0.8,0);
				\node[] at (0.,-0.3) {$g[x]$};
			\end{tikzpicture}
			=
			\begin{tikzpicture}[baseline=10]
				\draw[thick] (2.5,0.7) --++ (0,-0.7);
				\draw[thick] (2.3,0) --++ (0.7,0);
				\draw[ultra thick] (2.5,0.6) --++ (0,0.4);
				\draw[thick] (3.0, 0) --++ (0.4, 0);
				\pic  at (3.0,0.3){ARl=x/\;\inv{g}/gx};
				\node[tensor] at (2.5,0.6) {};
				\node[tensor] at (2.5,0.0) {};
			\end{tikzpicture}
		\end{equation}
		and for any $L\geq 0$,
		\begin{equation}
			\begin{tikzpicture}[baseline]
				\def\dy{0.6}
				\tikzset{decoration={snake,amplitude=.4mm,segment length=2mm, post length=0mm,pre length=0mm}}
				\draw[thick] (-1.3,0.) -- (1.3,0.);
				\foreach \x in {-0.9,-0.45,0, 0.45,0.9}{
					\node[tensor] at (\x, 0) {}; }
				\draw[thick] (-0.9,0.) --++ (0.,\dy);
				\draw[thick] (0.9,0.) --++ (0.,\dy);
				\draw[thick] (0,0.) --++ (0.,\dy+0.3);
				\draw[thick] (-0.45,0.) --++ (0.,\dy+0.3);
				\draw[thick] (0.45,0.) --++ (0.,\dy+0.3);
				\draw[ultra thick] (-0.9,\dy) --++ (0.0,0.3);
				\draw[thick, decorate] (+0.9,\dy) --++ (0.,0.3);
				
				\draw[thick, red] (-0.9,\dy) -- (0.9,\dy);
				\node[square, fill=violet] at (0.45, \dy) {};
				\node[square, fill=violet] at (0, \dy) {};
				\node[square, fill=violet] at (-0.45, \dy) {};
				\node[tensor, fill=white] at (0.9, \dy) {};
				\node[tensor] at (-0.9, \dy) {};
				\node[] at (-1.1, \dy) {$g$};
				\node[] at (-1.5, 0) {$x$};
				\node at (0.,\dy+0.5) {$\overbrace{\hspace{1.2cm}}^{L}$};
			\end{tikzpicture}
			=
			\begin{tikzpicture}[baseline]
				\def\dy{0.6}
				\draw[thick] (-1.3,0.) -- (1.3,0.);
				\foreach \x in {-0.45,0,0.45}{
					\node[tensor] at (\x, 0) {}; }
				\draw[thick] (-0.9,0.) --++ (0.,\dy);
				\draw[thick] (0.9,0.) --++ (0.,\dy);
				\draw[thick] (-0.45,0.) --++ (0.,\dy);
				\draw[thick] (0,0.) --++ (0.,\dy);
				\draw[thick] (0.45,0.) --++ (0.,\dy);
				\node[square] at (-0.9,0.) {};
				\node[square] at (0.9,0.) {};
				\node[] at (-0.9, -0.3) {$\inv{g}[gx]$};
				\node[] at (0.9, -0.3) {$g[x]$};
				\node[] at (0.0, -0.3) {$gx$};
			\end{tikzpicture};
		\end{equation}
		while for $\sigma_g=-1$,
		\begin{equation}
			\begin{tikzpicture}[baseline=10]
				\node[square] at (0,0) {};
				\draw[thick] (0,0.0) --++ (0,0.8);
				\draw[thick] (-0.4, 0) --++ (0.8,0);
				\node[] at (0.,-0.3) {$g[x]$};
			\end{tikzpicture}
			= \xi_g(x)
			\begin{tikzpicture}[baseline=10]
				\draw[thick] (2.5,0.7) --++ (0,-0.7);
				\draw[thick] (2.3,0) --++ (0.7,0);
				\draw[ultra thick] (2.5,0.6) --++ (0,0.4);
				\draw[thick] (3.0, 0) --++ (0.4, 0);
				
				\pic  at (3.0,0.3){ARl=x/\;g/gx};
				\node[tensor] at (2.5,0.6) {};
				\node[tensor] at (2.5,0.0) {};
			\end{tikzpicture}
			\label{eq:defects_Z2}
		\end{equation}
		and for any $L\geq 0$,
		\begin{equation}
			\begin{tikzpicture}[baseline]
				\def\dy{0.6}
				\tikzset{decoration={snake,amplitude=.4mm,segment length=2mm, post length=0mm,pre length=0mm}}
				\draw[thick] (-1.3,0.) -- (1.3,0.);
				\foreach \x in {-0.9,-0.45,0, 0.45,0.9}{
					\node[tensor] at (\x, 0) {}; }
				\draw[thick] (-0.9,0.) --++ (0.,\dy);
				\draw[thick] (0.9,0.) --++ (0.,\dy);
				\draw[thick] (0,0.) --++ (0.,\dy+0.3);
				\draw[thick] (-0.45,0.) --++ (0.,\dy+0.3);
				\draw[thick] (0.45,0.) --++ (0.,\dy+0.3);
				\draw[ultra thick] (-0.9,\dy) --++ (0.0,0.3);
				\draw[thick, decorate] (+0.9,\dy) --++ (0.,0.3);
				
				\draw[thick, red] (-0.9,\dy) -- (0.9,\dy);
				\node[square, fill=violet] at (0.45, \dy) {};
				\node[square, fill=violet] at (0, \dy) {};
				\node[square, fill=violet] at (-0.45, \dy) {};
				\node[tensor, fill=white] at (0.9, \dy) {};
				\node[tensor] at (-0.9, \dy) {};
				\node[] at (-1.1, \dy) {$g$};
				\node[] at (-1.5, 0) {$x$};
				\node at (0.,\dy+0.5) {$\overbrace{\hspace{1.2cm}}^{L}$};
			\end{tikzpicture}
			=
			\xi_g(gx)\begin{tikzpicture}[baseline]
				\def\dy{0.6}
				\draw[thick] (-1.3,0.) -- (1.3,0.);
				\foreach \x in {-0.45,0,0.45}{
					\node[tensor] at (\x, 0) {}; }
				\draw[thick] (-0.9,0.) --++ (0.,\dy);
				\draw[thick] (0.9,0.) --++ (0.,\dy);
				\draw[thick] (-0.45,0.) --++ (0.,\dy);
				\draw[thick] (0,0.) --++ (0.,\dy);
				\draw[thick] (0.45,0.) --++ (0.,\dy);
				\node[square] at (-0.9,0.) {};
				\node[square] at (0.9,0.) {};
				\node[] at (-0.9, -0.3) {$\inv{g}[gx]$};
				\node[] at (0.9, -0.3) {$g[x]$};
				\node[] at (0.0, -0.3) {$gx$};
			\end{tikzpicture},
			\label{eq:defects_Z2b}
		\end{equation}
		where $\xi_g:\mathsf{X}\to\{\pm 1\}$ is an arbitrary function satisfying $\xi_g(x)\xi_{\inv g}(gx)=-1$. 
		\label{prop:def_tens_anomalous}
	\end{proposition}
	\begin{proof}
		The proof proceeds very much in the same way as that of Proposition \ref{prop:1}, but this time keeping track of MPS labels $x\in \mathsf{X}$, and potential $\Z_2$ anomalies. 
		Firstly,
		\begin{equation}
			\begin{tikzpicture}[baseline=+3mm]
				\tikzset{decoration={snake,amplitude=.4mm,segment length=2mm, post length=0mm,pre length=0mm}}
				\def\dy{0.5}
				\draw[thick] (0.3,0.) --++ (1.7,0.);
				\draw[thick] (-0.3,0.) --++ (-1.7,0.);
				\node[] at (0.,0.0) {$\ldots$};
				\foreach \x in {-0.9,-0.5,0.5,0.9}{
					\node[tensor] at (\x, 0) {}; }
				\draw[thick] (-0.9,0.) --++ (0.,0.6);
				\draw[thick] (0.9,0.) --++ (0.,0.6);
				\draw[ultra thick] (0.9,0.6) --++ (0.0,\dy);
				\draw[thick, decorate] (-0.9,0.6) --++ (0.,\dy);
				\draw[thick] (-0.9,0.6+\dy) --++ (0.,0.3);
				\draw[thick] (0.9,0.6+\dy) --++ (0.,0.3);
				\draw[thick] (-0.5, 0.) -- (-0.5, 0.6+\dy+0.3);
				\draw[thick] (0.5, 0.) -- (0.5, 0.6+\dy+0.3);
				\draw[thick, red] (0.3,0.6+\dy) --++ (0.6,0.);
				\draw[thick, red] (-0.9,0.6+\dy) --++ (0.6,0.);
				
				\pic at (-1.4,0.3) {ALl=gx/g/x};
				\pic at (1.4,0.3) {ARl=gx/g/x};
				\node[red] at (0.,0.6+\dy) {$\ldots$};
				\node[square, fill=violet] at (0.5, 0.6+\dy) {};
				\node[square, fill=violet] at (-0.5, 0.6+\dy) {};
				\node[tensor, fill=white] at (-0.9, .6) {};
				\node[tensor, fill=white] at (-0.9, 0.6+\dy) {};
				\node[tensor] at (0.9, .6) {};
				\node[tensor] at (0.9, 0.6+\dy) {};
			\end{tikzpicture}
			=
			\begin{tikzpicture}[baseline=3mm]
				\tikzset{decoration={snake,amplitude=.4mm,segment length=2mm, post length=0mm,pre length=0mm}}
				\def\dy{0.5}
				\draw[thick] (0.3,0.) --++ (1.,0.);
				\draw[thick] (-0.3,0.) --++ (-1.,0.);
				\node[] at (0.,0.0) {$\ldots$};
				\foreach \x in {-0.9,-0.5,0.5,0.9}{
					\node[tensor] at (\x, 0) {}; }
				\draw[thick] (-0.9,0.) --++ (0.,0.6);
				\draw[thick] (0.9,0.) --++ (0.,0.6);
				\draw[thick] (-0.5,0.) --++ (0.,0.6);
				\draw[thick] (0.5,0.) --++ (0.,0.6);
				\node[irrep] at (1.3, 0.0) {$gx$};
				\node[] at (1.5,0.0) {$.$};
			\end{tikzpicture}
		\end{equation}
		We can identify the tensors on top as the inverse of the truncated symmetry of $g^{-1}$, up to a sign $\sigma_g$ that is only negative for an anomalous $\Z_2$ generator (see Appendix  \ref{app:Z2}):
		\begin{equation}
			\begin{tikzpicture}[baseline=+3mm]
				\tikzset{decoration={snake,amplitude=.4mm,segment length=2mm, post length=0mm,pre length=0mm}}
				\def\dy{0.4}
				\draw[ultra thick] (0.9,0.) --++ (0.0,\dy);
				\draw[thick, decorate] (-0.9,0.) --++ (0.,\dy);
				\draw[thick] (-0.9,0.+\dy) --++ (0.,0.3);
				\draw[thick] (0.9,0.+\dy) --++ (0.,0.3);
				\draw[thick] (-0.5, 0.) -- (-0.5, 0.+\dy+0.3);
				\draw[thick] (0.5, 0.) -- (0.5, 0.+\dy+0.3);
				\draw[thick, red] (0.3,0.+\dy) --++ (0.6,0.);
				\draw[thick, red] (-0.9,0.+\dy) --++ (0.6,0.);

				\node[red] at (0.,0.+\dy) {$\ldots$};
				\node[square, fill=violet] at (0.5, 0.+\dy) {};
				\node[square, fill=violet] at (-0.5, 0.+\dy) {};
				\node[tensor, fill=white] at (-0.9, 0.+\dy) {};
				\node[tensor] at (0.9, 0.+\dy) {};
			\end{tikzpicture}=\sigma_g (U^L_{g^{-1}})^\dagger.
		\end{equation}
		Therefore, we have
		\begin{equation}
			\begin{tikzpicture}
				\def\dy{0.6}
				\tikzset{decoration={snake,amplitude=.4mm,segment length=2mm, post length=0mm,pre length=0mm}}
				\draw[thick] (0.3,0.) --++ (1.,0.);
				\draw[thick] (-0.3,0.) --++ (-1.,0.);
				\node[] at (0.,0.0) {$\ldots$};
				\foreach \x in {-0.9,-0.5,0.5,0.9}{
					\node[tensor] at (\x, 0) {}; }
				\draw[thick] (-0.9,0.) --++ (0.,\dy);
				\draw[thick] (0.9,0.) --++ (0.,\dy);
				\draw[thick] (-0.5,0.) --++ (0.,\dy+0.3);
				\draw[thick] (0.5,0.) --++ (0.,\dy+0.3);
				\draw[ultra thick] (-0.9,\dy) --++ (0.0,0.3);
				\draw[thick, decorate] (+0.9,\dy) --++ (0.,0.3);
				
				\draw[thick, red] (0.3,\dy) --++ (0.6,0.);
				\draw[thick, red] (-0.9,\dy) --++ (0.6,0.);
				\node[red] at (0.,\dy) {$\ldots$};
				\node[square,fill=violet] at (0.5, \dy) {};
				\node[square,fill=violet] at (-0.5, \dy) {};
				\node[tensor, fill=white] at (0.9, \dy) {};
				\node[tensor] at (-0.9, \dy) {};
				\node[] at (-1.1, \dy) {$g$};
				\node[irrep] at (-1.2, 0) {$x$};
			\end{tikzpicture}
			=\sigma_g\;
			\begin{tikzpicture}
				\def\dy{0.6}
				\def\dx{0.9}
				\def\dxb{1.4}
				\tikzset{decoration={snake,amplitude=.4mm,segment length=2mm, post length=0mm,pre length=0mm}}
				\draw[thick] (0.3,0.) --++ (1.3,0.);
				\draw[thick] (-0.3,0.) --++ (-1.3,0.);
				\node[] at (0.,0.0) {$\ldots$};
				\foreach \x in {-\dx,-0.5,0.5,\dx}{
					\node[tensor] at (\x, 0) {}; }
				\draw[thick] (-\dx,0.) --++ (0.,\dy);
				\draw[thick] (\dx,0.) --++ (0.,\dy);
				\draw[thick] (-0.5,0.) --++ (0.,\dy+0.1);
				\draw[thick] (0.5,0.) --++ (0.,\dy+0.1);
				\draw[ultra thick] (\dx,\dy) --++ (0.0,0.3);
				\draw[thick, decorate] (-\dx,\dy) --++ (0.,0.3);
				
				\pic at (\dxb,0.3) {ARl=x/g^{-1}/gx};
				\pic at (-\dxb,0.3) {ALl=x/g^{-1}/gx};
				
				\node[tensor, fill=white] at (-\dx, \dy) {};
				\node[tensor] at (\dx, \dy) {};
			\end{tikzpicture}
		\end{equation}
		for any $L\geq 2$. For $L\geq 3$, we also have, using \eqref{eq:action_exterior}, and the fact that the $X_i$ tensors have (left) inverses,
		\begin{equation}
			\begin{tikzpicture}
				\def\dy{0.6}
				\tikzset{decoration={snake,amplitude=.4mm,segment length=2mm, post length=0mm,pre length=0mm}}
				\draw[thick] (0.3,0.) --++ (1.,0.);
				\draw[thick] (-0.3,0.) --++ (-1.,0.);
				\node[] at (0.,0.0) {$\ldots$};
				\foreach \x in {-0.9,-0.5,0.5,0.9}{
					\node[tensor] at (\x, 0) {}; }
				\draw[thick] (-0.9,0.) --++ (0.,\dy);
				\draw[thick] (0.9,0.) --++ (0.,\dy);
				\draw[thick] (-0.5,0.) --++ (0.,\dy+0.3);
				\draw[thick] (0.5,0.) --++ (0.,\dy+0.3);
				\draw[ultra thick] (-0.9,\dy) --++ (0.0,0.3);
				\draw[thick, decorate] (+0.9,\dy) --++ (0.,0.3);
				
				\draw[thick, red] (0.3,\dy) --++ (0.6,0.);
				\draw[thick, red] (-0.9,\dy) --++ (0.6,0.);
				\node[red] at (0.,\dy) {$\ldots$};
				\node[square,fill=violet] at (0.5, \dy) {};
				\node[square,fill=violet] at (-0.5, \dy) {};
				\node[tensor, fill=white] at (0.9, \dy) {};
				\node[tensor] at (-0.9, \dy) {};
				\node[] at (-1.1, \dy) {$g$};
			\end{tikzpicture}
			=
			\begin{tikzpicture}
				\def\dy{0.6}
				\def\dx{1.4}
				\tikzset{decoration={snake,amplitude=.4mm,segment length=2mm, post length=0mm,pre length=0mm}}
				\draw[thick] (0.3,0.) --++ (1.3,0.);
				\draw[thick] (-0.3,0.) --++ (-1.3,0.);
				\node[] at (0.,0.0) {$\ldots$};
				\foreach \x in {-\dx,-0.5,0.5,\dx}{
					\node[tensor] at (\x, 0) {}; }
				\draw[thick] (-\dx,0.) --++ (0.,\dy);
				\draw[thick] (\dx,0.) --++ (0.,\dy);
				\draw[thick] (-0.5,0.) --++ (0.,\dy+0.1);
				\draw[thick] (0.5,0.) --++ (0.,\dy+0.1);
				\draw[ultra thick] (-\dx,\dy) --++ (0.0,0.3);
				\draw[thick, decorate] (\dx,\dy) --++ (0.,0.3);
				
				\pic at (-0.9,0.3) {ARl=gx/g/x};
				\pic at (0.9,0.3) {ALl=gx/g/x};
				
				\node[tensor, fill=white] at (\dx, \dy) {};
				\node[tensor] at (-\dx, \dy) {};
			\end{tikzpicture}.
		\end{equation}
		Now, using Lemma~\ref{lemma:1} we conclude existence of scalars $\alpha_{g,x}$, $g\in G$, with $\alpha_{g,x}=\sigma_g\alpha_{g^{-1},gx}$, such that
		\begin{align}
			\begin{tikzpicture}
				\tikzset{decoration={snake,amplitude=.4mm,segment length=2mm, post length=0mm,pre length=0mm}}
				\draw[thick] (-0.5, 0) --++ (1.5, 0);
				\draw[thick] (0.5, 0) --++(0.0, 0.6);
				\draw[thick, decorate] (0.5, 0.6) --++ (0.0, 0.3);
				\pic at (0.0,0.3) {ALl=gx/g/x};
				\node[tensor] at (0.5, 0.0) {};
				\node[tensor, fill=white] at (0.5, 0.6) {};
			\end{tikzpicture}
			= \alpha_{g,x}
			\begin{tikzpicture}
				\draw[thick] (-0.5, 0) --++ (1.5, 0);
				\draw[thick] (0., 0) --++(0.0, 0.6);
				\draw[ultra thick] (0., 0.6) --++ (0.0, 0.3);
				\pic at (0.5,0.3) {ARl=x/g^{-1}/gx};
				\node[tensor] at (0., 0.0) {};
				\node[tensor] at (0., 0.6) {};
				\node[] at (1.2,0.0) {$.$};
			\end{tikzpicture}
		\end{align}
		We can in fact recognise $\alpha_{g,x}=L^{gx}_{\inv{g},g}$, using
		\begin{align}
			\begin{tikzpicture}[baseline]
				\tikzset{decoration={snake,amplitude=.4mm,segment length=2mm, post length=0mm,pre length=0mm}}
				\def\dx{0.5}
				\def\dxb{0.9}
				\def\ddx{0.4}
				\def\dy{0.5}
				\def\dyb{0.9}
				\def\ddy{0.2}
				\draw[thick] (-\ddx, 0) -- (\dxb+\ddx, 0);
				\draw[thick, red] (0, \dy) -- (\dx, \dy);
				\draw[thick, red] (0, \dyb) -- (\dxb, \dyb);
				\draw[thick] (0, 0) -- (0, \dy);
				\draw[thick, decorate] (0, \dy) -- (0, \dyb);
				\draw[thick] (0, \dyb) -- (0, \dyb+\ddy);
				\pic at (\dx,0) {actR=\dy///};
				\pic at (\dxb,0) {actR=\dyb///};
				\node[tensor] at (0, \dy) {};
				\node[tensor] at (0, 0) {};
				\node[tensor, fill=white] at (0, \dyb) {};
				\node[tiny] at (0,-0.25) {$gx$};
				\node[tiny] at (0.5*\dx+0.5*\dxb,-0.23) {$x$};
				\node[irrep] at (\dx/2+0.1,\dy) {$\inv{g}$};
				\node[irrep] at (\dx/2,\dyb) {$g$};
				\node[tiny] at (-\ddx,\dy) {$Y_2^{\inv{g}}$};
				\node[tiny] at (-\ddx,\dyb) {$X_1^g$};
			\end{tikzpicture}
			&=
			\begin{tikzpicture}[baseline]
				\tikzset{decoration={snake,amplitude=.4mm,segment length=2mm, post length=0mm,pre length=0mm}}
				\def\dx{0.6}
				\def\dxb{1.0}
				\def\ddx{0.3}
				\def\dy{0.5}
				\def\dyb{0.9}
				\def\ddy{0.2}
				\draw[thick] (-\ddx, 0) -- (\dxb+\ddx, 0);
				\draw[thick, red] (0, \dy) -- (\dx, \dy);
				\draw[thick, red] (0, \dyb) -- (\dxb, \dyb);
				\draw[thick] (0, 0) -- (0, \dy);
				\draw[thick, decorate] (0, \dy) -- (0, \dyb);
				\draw[thick] (0, \dyb) -- (0, \dyb+\ddy);
				\pic at (\dx,0) {actR=\dy///};
				\pic at (\dxb,0) {actR=\dyb///};
				\node[tiny, red] at (\dx/2,0.25) {$T^*_g$};
				\node[tiny] at (-\ddx-0.2,\dy) {$(X_1^{g})^\dagger$};
				\node[tensor, fill=white] at (0, \dy) {};
				\node[tensor] at (0, 0) {};
				\node[tensor, fill=white] at (0, \dyb) {};
				\node[tensor, red] at (\dx/2, \dy) {};
			\end{tikzpicture}
			=
			\begin{tikzpicture}[baseline]
				\def\dx{0.5}
				\def\dxb{0.9}
				\def\ddx{0.3}
				\def\dy{0.5}
				\def\dyb{0.9}
				\def\ddy{0.2}
				\draw[thick] (-\ddx-\dx, 0) -- (\dxb+\ddx, 0);
				\draw[thick, red] (-\dx, \dy) -- (\dx, \dy);
				\draw[thick, red] (-\dx, \dyb) -- (\dxb, \dyb);
				\draw[thick] (0, 0) -- (0, \dy);
				\draw[thick] (0, \dy) -- (0, \dyb);
				\draw[thick] (0, \dyb) -- (0, \dyb+\ddy);
				\pic at (\dx,0) {actR=\dy///};
				\pic at (\dxb,0) {actR=\dyb///};
				\pic at (-\dx,\dy) {fusL=\dyb-\dy///};
				\node[mpo] at (0, \dy) {};
				\node[mpo] at (0, \dyb) {};
				\node[tensor] at (0, 0) {};
				\node[irrep] at (\dx/2+0.1,\dy) {$\inv{g}$};
				\node[irrep] at (\dx/2,\dyb) {$g$};
			\end{tikzpicture}\\
			&=\dfrac{1}{L^{gx}_{g,\inv{g}}}
			\begin{tikzpicture}[baseline]
				\def\ddx{0.3}
				\def\dy{0.5}
				\draw[thick] (-\ddx, 0) -- (\ddx, 0);
				\draw[thick] (0, 0) -- (0, \dy);
				\node[tensor] at (0, 0) {};
				\node[tiny] at (0,-0.25) {$gx$};
			\end{tikzpicture}.
		\end{align}
		If $\sigma_g=1$, we can once again absorb this scalar into the gauge of the action tensors, for instance by doing $A^R_{g,x}\to\sqrt{\alpha_{g,x}}A^R_{g,x}, A^L_{g,x}\to A^L_{g,x}/\sqrt{\alpha_{g,x}}$, which corresponds to setting $L^x_{\inv{g}, g}=1$. If $\sigma_g=-1$, however, this does not work and in general, the best we can do is choose an arbitrary function $\xi_g(x)$ as in the statement so that 
		$A^R_{g,x}\to\sqrt{\xi_g(x)\alpha_{g,x}}A^R_{g,x}, A^L_{g,x}\to A^L_{g,x}/\sqrt{\xi_g(x)\alpha_{g,x}}$, which leads to Eqs. \eqref{eq:defects_Z2}-\eqref{eq:defects_Z2b}, and to $L^x_{\inv{g},g}\in\{\pm1\}$.
	\end{proof}
	Thus, our truncated symmetries still generate pairs of defects (domain walls), but sometimes there may be a minus sign involved. The fusion tensors are the same since they only depend on the MPU representation, but their action on the defects now depends on $L$-symbols that may or may not be trivializable.
	\begin{proposition}The operators defined as in \eqref{def:lambdas} satisfy
		\begin{subequations}
			\begin{equation}
				\begin{tikzpicture}[scale=\TikzScaling, baseline=-10]
					
					\node[rectangle, fill=white] at (0,0) {$\lambda^L_{g,h}$};
					
					\draw[thick] (-0.2,0.25) --++ (0,0.2);
					\draw[thick] (0.2,0.25) --++ (0,0.2);
					\draw[thick] (-0.2,-0.25) --++ (0,-0.4);
					\draw[thick] (0.2,-0.25) --++ (0,-0.4);
					
					\draw[thick] (-0.4,-0.65) --++ (0.8,0);
					\node[square] at (-0.2,-0.65) {};
					\node[square] at (0.2,-0.65) {};
					\node at (-0.3,-0.85) {$g[hx]$};
					\node at (0.3,-0.85) {$h[x]$};
				\end{tikzpicture}    
				= \dfrac{1}{L^x_{\inv{h},\inv{g}}}
				\begin{tikzpicture}[scale=\TikzScaling, baseline = -20]
					
					\draw[thick] (-0.4,-0.65) --++ (0.8,0);
					\node[square] at (-0.2,-0.65) {};
					\node[tensor] at (0.2,-0.65) {};
					\node at (-0.2,-0.85) {$gh[x]$};
					
					\draw[thick] (-0.2,-0.65) --++ (0,0.4);
					\draw[thick] (0.2,-0.65) --++ (0,0.4);
					
					\node[] at (0.5,-0.65) {$,$};
					
				\end{tikzpicture}
				\label{eq:fusop1}
			\end{equation}
			\begin{equation}
				\begin{tikzpicture}[scale=\TikzScaling, baseline = -10]
					
					\node[rectangle, fill=white] at (0,0) {$\lambda^R_{g,h}$};
					
					\draw[thick] (-0.2,0.25) --++ (0,0.2);
					\draw[thick] (0.2,0.25) --++ (0,0.2);
					\draw[thick] (-0.2,-0.25) --++ (0,-0.4);
					\draw[thick] (0.2,-0.25) --++ (0,-0.4);
					
					\draw[thick] (-0.4,-0.65) --++ (0.8,0);
					\node[square] at (-0.2,-0.65) {};
					\node[square] at (0.2,-0.65) {};
					\node at (-0.3,-0.85) {$g[hx]$};
					\node at (0.3,-0.85) {$h[x]$};
				\end{tikzpicture}
				=L^x_{g,h}\;
				\begin{tikzpicture}[scale=\TikzScaling, baseline = -20]
					
					\draw[thick] (-0.4,-0.65) --++ (0.8,0);
					\node[square] at (0.2,-0.65) {};
					\node[tensor] at (-0.2,-0.65) {};
					\node at (0.2,-0.85) {$gh[x]$};
					
					\draw[thick] (-0.2,-0.65) --++ (0,0.4);
					\draw[thick] (0.2,-0.65) --++ (0,0.4);
					\node[] at (0.5,-0.65) {$.$};
					
				\end{tikzpicture}
				\label{eq:fusop2}
			\end{equation}
		\end{subequations}
		and 
		\begin{equation}
			\lambda^L_{g,e} = \lambda^R_{e,g} = \mathds{1},\quad \lambda^L_{e,g} = w^L_g, \quad \lambda^R_{g,e} = w^R_g.
			\label{eq:simplecases}
		\end{equation}
		Moreover, there exists a choice of gauge for the fusion tensors and action tensors, compatible with the ones made in Section \ref{sec:assumptions} and Proposition \ref{prop:def_tens_anomalous}, such that 
		\begin{itemize}
			\item[(i)]  $\omega(g,h,k), L^x_{g,h}\in U(1)$ and if the MPU representation is nonanomalous, $\omega(g,h,k)=1$,
			\item[(ii)] $\lambda^L_{g,h},\lambda_{g,h}^R$ are unitary.
		\end{itemize}
		\label{prop:unitarity_and_gauges}
	\end{proposition}
	The proof of this proposition is given in Appendix \ref{app:proofprops}. From now on, unless explicitly mentioned, we will assume that we have chosen a gauge from this proposition.

	\section{Gauging MPU symmetries}
	\label{sec:true_gauging}
	In this section, we show how to define a gauged MPS tensor and a local representation of the symmetry generalizing the gauging of the onsite symmetry by promoting the defects reviewed in Section~\ref{sec:onsite}. We begin with the injective case, where we will see that this can always be done. In the noninjective case, however, our procedure relies on an assumption on the trivializability of the $L$-symbols that we call block independence, for which the absence of an anomaly is a necessary but not sufficient condition.
	
	\subsection{Injective MPS}
	Define a gauged MPS tensor by
	\begin{equation}
		\begin{tikzpicture}[baseline = +2.0mm]
			\def\dx{0.5};
			\def\ddx{0.08};
			\def\dy{0.6};
			\draw[thick] (-\dx,0) -- (\dx,0);
			\draw[thick, dotted] (\ddx,0) --++ (0,\dy);
			\draw[thick] (-\ddx, 0) --++ (0,\dy);
			\node[widetensor] at (0,0) {};
		\end{tikzpicture}
		\equiv
		\sum_{g\in G}{\;\;
			\begin{tikzpicture}[baseline = +2.0mm]
				\def\dx{0.45};
				\def\dy{0.6};        
				\draw[thick] (0,0.0) --++ (0,\dy);
				\draw[thick] (-\dx, 0) --++ (2*\dx,0);
				\node[square] at (0,0) {};
				\node[] at (0.0,-0.3) {$g$};
			\end{tikzpicture}\otimes \ket{g}}
		\label{eq:gauged_tensor}
	\end{equation}
	Note that this tensor has two physical legs, one corresponding to the matter degree of freedom and the other one, which we depict with a dotted line, to the gauge degree of freedom, taking values in $\C G$. We could split it via a singular value decomposition to go back to the picture we had in the onsite case, however, there is no reason we should in that way recover the same bond dimension or the same MPS tensor that we had originally.
	
	Let us now define our local unitary group representation that will give rise to gauge constraints (Gauss laws). Note that this is the same definition as in \cite{Seifnashri}.
	
	\begin{proposition}
		\label{prop:gausslaws}
		The mapping $\mcG:G\to \mc B{\left(\hilb^{\otimes 2}\otimes \C G^{\otimes 2}\right)}$ given by
		\begin{equation}
			\mcG_g\equiv\mcG(g) = \sum_{a,b\in G}{\lambda^{ag^{-1},gb}_{a,b}\otimes \ketbra{ag^{-1},gb}{a,b}},
			\label{eq:mcG}
		\end{equation}
		is a representation of $G$ and leaves the contraction of two copies of the gauged tensor \eqref{eq:gauged_tensor} invariant, i.e.
		\begin{equation}
			\begin{tikzpicture}
				\def\dx{0.5}
				\def\ddx{0.08}
				\def\dxb{1.0}
				\def\dy{0.5}
				\def\dyb{0.9}
				\draw[thick] (-\dxb, 0) --++ (2*\dxb,0);
				\draw[thick] (-\dx-\ddx,0) --++ (0,\dyb);
				\draw[thick, dotted] (-\dx+\ddx,0) --++ (0,\dyb);
				\draw[thick] (\dx-\ddx,0) --++ (0,\dyb);
				\draw[thick, dotted] (\dx+\ddx,0) --++ (0,\dyb);
				\node[widetensor] at (-\dx,0) {};
				\node[widetensor] at (\dx,0) {};
				\draw[fill=white] (-\dx-0.2,\dy-0.2) rectangle ++ (2*\dx+0.4,0.4);
				\node[] at (-\dx-0.4, \dy) {$\mcG_g$};
			\end{tikzpicture}
			=
			\begin{tikzpicture}
				\def\dx{0.5}
				\def\ddx{0.08}
				\def\dxb{1.0}
				\def\dy{0.5}
				\def\dyb{0.9}
				\draw[thick] (-\dxb, 0) --++ (2*\dxb,0);
				\draw[thick] (-\dx-\ddx,0) --++ (0,\dyb);
				\draw[thick, dotted] (-\dx+\ddx,0) --++ (0,\dyb);
				\draw[thick] (\dx-\ddx,0) --++ (0,\dyb);
				\draw[thick, dotted] (\dx+\ddx,0) --++ (0,\dyb);
				\node[widetensor] at (-\dx,0) {};
				\node[widetensor] at (\dx,0) {};
			\end{tikzpicture}.
		\end{equation}
	\end{proposition}
	\begin{proof}
		The first part follows from the direct computation. For the remaining one, we have
		\begin{align}
			&\mcG_g\left(
			\begin{tikzpicture}
				\def\dx{0.5}
				\def\ddx{0.08}
				\def\dxb{1.0}
				\def\dy{0.5}
				\def\dyb{0.8}
				\draw[thick] (-\dxb, 0) --++ (2*\dxb,0);
				\draw[thick] (-\dx-\ddx,0) --++ (0,\dyb);
				\draw[thick, dotted] (-\dx+\ddx,0) --++ (0,\dyb);
				\draw[thick] (\dx-\ddx,0) --++ (0,\dyb);
				\draw[thick, dotted] (\dx+\ddx,0) --++ (0,\dyb);
				\node[widetensor] at (-\dx,0) {};
				\node[widetensor] at (\dx,0) {};
			\end{tikzpicture}
			\right)=
			\mcG_g\left(\sum_{a,b\in G}{
				\begin{tikzpicture}[baseline = 2mm]
					\draw[thick] (-0.6,0) -- (0.6,0);
					\node[square] at (-0.3,0) {};
					\node[square] at (0.3,0) {};
					\node[] at (-0.3, -0.3) {$a$};
					\node[] at (0.3, -0.3) {$b$};
					\draw[thick] (-0.3,0) --++ (0., 0.6);
					\draw[thick] (0.3,0) --++ (0., 0.6);
				\end{tikzpicture}
				\otimes \ket{a,b}}\right)\nonumber\\
			&=\!\!\sum_{a,b\in G}{
				\begin{tikzpicture}[baseline = 2mm]
					\draw[thick] (-0.6,0) -- (0.6,0);
					\node[square] at (-0.3,0) {};
					\node[square] at (0.3,0) {};
					\node[] at (-0.3, -0.3) {$a$};
					\node[] at (0.3, -0.3) {$b$};
					\draw[thick] (-0.3,0) --++ (0., 0.9);
					\draw[thick] (0.3,0) --++ (0., 0.9);
					\draw[] (-0.3,0.3) rectangle ++ (0.6,0.4);
					\node[] at (.0,0.5) {$\lambda$};
				\end{tikzpicture}
				\!\!\otimes \ket{ag^{-1},gb}}\!=\\
				&=
			\sum_{a,b\in G}{
				\begin{tikzpicture}[baseline = 2mm]
					\draw[thick] (-0.6,0) -- (0.6,0);
					\node[square] at (-0.3,0) {};
					\node[square] at (0.3,0) {};
					\node[] at (-0.3, -0.3) {$ag^{-1}$};
					\node[] at (0.3, -0.3) {$gb$};
					\draw[thick] (-0.3,0) --++ (0., 0.6);
					\draw[thick] (0.3,0) --++ (0., 0.6);
				\end{tikzpicture}
				\!\!\otimes \ket{ag^{-1},gb}}=  
			\begin{tikzpicture}
				\def\dx{0.5}
				\def\ddx{0.08}
				\def\dxb{1.0}
				\def\dy{0.5}
				\def\dyb{0.8}
				\draw[thick] (-\dxb, 0) --++ (2*\dxb,0);
				\draw[thick] (-\dx-\ddx,0) --++ (0,\dyb);
				\draw[thick, dotted] (-\dx+\ddx,0) --++ (0,\dyb);
				\draw[thick] (\dx-\ddx,0) --++ (0,\dyb);
				\draw[thick, dotted] (\dx+\ddx,0) --++ (0,\dyb);
				\node[widetensor] at (-\dx,0) {};
				\node[widetensor] at (\dx,0) {};
			\end{tikzpicture}\nonumber.
		\end{align}
	\end{proof}
	
	The original MPS can be recovered by projecting all the gauge degrees of freedom of the gauged MPS onto the neutral element $\ket{e}$. We can define the projectors on the gauge invariant subspace in the usual way
	\begin{equation}
		P_j=\dfrac{1}{|G|}\sum_{g\in G}{\mcG_g^{j,j+1}}
		\label{eq:gauge_projs}
	\end{equation}
	where we have introduced a position index $j$ along the chain. We will soon see in Proposition \ref{prop:comm_proj} that the projectors $P_j, P_j'$ commute with each other in the absence of an anomaly, which is always the case when the invariant MPS is injective. The parent Hamiltonian of the gauged MPS thus commutes with the gauge constraints and therefore constitutes a gauging of the parent Hamiltonian of the original MPS.
	
	\subsection{Noninjective MPS}
	\label{subsec:true_gauging_noninj}
	For a noninjective invariant MPS, the gauged MPS tensor is a generalization of Eq. \eqref{eq:gauged_tensor}. Similarly to a globally invariant MPS (cf. \cite{Schuch11}), it will have a block structure, and the virtual Hilbert space will be given by
	\begin{equation}
		\mc H_v = \bigoplus_{x}{\C^{\otimes \chi_x}},
	\end{equation}
	where $\chi_x$ is the bond dimension of the MPS labeled by $x$. We will additionally denote by
	\begin{equation}
		I_x: \C^{\otimes \chi_x}\rightarrow \mc H_v
	\end{equation}
	the natural isometric embedding of the $\C^{\otimes \chi_x}$ summand into $\mc H_v$, which we will use to access specific blocks in the MPS. With that notation, our gauged MPS is given by
	\begin{equation}
		\begin{tikzpicture}[baseline=1mm]
			\draw[thick] (-0.1,0) --++ (0,0.6);
			\draw[thick, dotted] (0.1,0) --++ (0,0.6);
			\draw[thick] (-0.6, 0) --++ (1.2,0);
			\draw[fill=black] (-0.2,-0.11) rectangle ++ (0.4,0.22);
			\node[] at (-0.8,0.0) {$I^\dagger_y$};
			\node[] at (0.8,0.0) {$I_x$};
		\end{tikzpicture}
		\equiv
		\sum_{g\in G}{\;\;
			\begin{tikzpicture}[baseline=1mm]
				\draw[thick] (0,0.0) --++ (0,0.6);
				\draw[thick] (-0.4, 0) --++ (0.8,0);
				\node[square] at (0,0) {};
				\node[] at (0.0,-0.3) {$g[x]$};
			\end{tikzpicture}\otimes \ket{g}\, .}
		\label{eq:gauged_tensor_dw}
	\end{equation}
	This time, the action of the $\mc G_g$ operators on the different MPS blocks gives to scalar factors:
	\begin{align}
		&\begin{tikzpicture}
			\def\dx{0.4}
			\def\dxx{1.}
			\def\dy{0.8}
			\def\ddy{0.35}
			\def\dd{0.1}
			\draw[thick] (\dx-\dd,0) --++ (0,\dy);
			\draw[thick, dotted] (\dx+\dd,0) --++ (0,\dy);
			\draw[thick] (-\dx-\dd,0) --++ (0,\dy);
			\draw[thick, dotted] (-\dx+\dd,0) --++ (0,\dy);
			\draw[thick] (-\dxx, 0) --++ (2*\dxx,0);
			\node[widetensor] at (\dx, 0) {};
			\node[widetensor] at (-\dx, 0) {};
			\draw (-\dx-2*\dd, \ddy) rectangle (\dx+2*\dd, \ddy+0.3) [fill=white, thick];
			\node[] at (-\dxx-0.2,0.0) {$I^\dagger_y$};
			\node[] at (\dxx+0.2,0.0) {$I_x$};
			\node[] at (-\dx-\dd-0.4,\ddy+0.15) {$\mc G_g$};
		\end{tikzpicture}=\nonumber\\
		&=\mc G_g\left(\sum_{\substack{h_1,h_2\\h_1h_2x = y}}{
			\begin{tikzpicture}
				\def\dx{0.6}
				\def\dxx{1.2}
				\def\dy{0.6}
				\draw[thick] (\dx,0) --++ (0,\dy);
				\draw[thick] (-\dx,0) --++ (0,\dy);
				\draw[thick] (-\dxx, 0) --++ (2*\dxx,0);
				\node[square] at (\dx, 0) {};
				\node[square] at (-\dx, 0) {};
				\node[] at (\dx,-0.25) {$h_2[x]$};
				\node[] at (-\dx,-0.25) {$h_1[h_2x]$};
			\end{tikzpicture}
			\otimes\ket{h_1,h_2}
		}\right)\nonumber\\
		&=\!\!\!\!\sum_{\substack{h_1,h_2\\h_1h_2x = y}}\!\!\!\!{\dfrac{L^x_{h_1,h_2}}{L^x_{h_1g^{-1},gh_2}}
			\begin{tikzpicture}
				\def\dx{0.6}
				\def\dxx{1.2}
				\def\dy{0.6}
				\draw[thick] (\dx,0) --++ (0,\dy);
				\draw[thick] (-\dx,0) --++ (0,\dy);
				\draw[thick] (-\dxx, 0) --++ (2*\dxx,0);
				\node[square] at (\dx, 0) {};
				\node[square] at (-\dx, 0) {};
				\node[] at (\dx,-0.3) {$gh_2[x]$};
				\node[] at (-\dx,-0.3) {$h_1\inv{g}[gh_2x]$};
			\end{tikzpicture}
			\otimes\ket{h_1g^{-1},gh_2}
		}\nonumber\\
		&=\sum_{\substack{h_1,h_2\\h_1h_2x = y}}{\dfrac{L^x_{h_1g,g^{-1}h_2}}{L^x_{h_1,h_2}}
			\begin{tikzpicture}
				\def\dx{0.6}
				\def\dxx{1.2}
				\def\dy{0.6}
				\draw[thick] (\dx,0) --++ (0,\dy);
				\draw[thick] (-\dx,0) --++ (0,\dy);
				\draw[thick] (-\dxx, 0) --++ (2*\dxx,0);
				\node[square] at (\dx, 0) {};
				\node[square] at (-\dx, 0) {};
				\node[] at (\dx,-0.25) {$h_2[x]$};
				\node[] at (-\dx,-0.25) {$h_1[h_2x]$};
			\end{tikzpicture}
			\otimes\ket{h_1,h_2}.
		}
	\end{align}
	In the absence of $x$-dependence, the factor 
	\begin{equation}
		\ell_{h_1,h_2;g}^x\equiv\dfrac{L^x_{h_1g,g^{-1}h_2}}{L^x_{h_1,h_2}}
		\label{eq:scalarfactor}
	\end{equation} could be absorbed into a redefinition of $\mc G_g$ and we would get a locally invariant MPS. We can thus define the condition of \textit{block independence} (BI) as the existence of a gauge where \eqref{eq:scalarfactor} is independent of the block label $x$, i.e. $\ell_{h_1,h_2;g}^x=\ell_{h_1,h_2;g}^y, \forall x,y$. We prove in Appendix~\ref{app:block_indep} that this condition is equivalent to the existence of a gauge where all $L$-symbols are equal to 1, and in Appendix~\ref{app:proofprops} that such a gauge can be found while respecting all of the previous gauge choices we have made (those encoded in Proposition \ref{prop:unitarity_and_gauges}). This already implies that anomalous MPUs cannot satisfy the BI condition (as in the corresponding gauge the anomaly would be trivial due to \eqref{eq:omega_and_Ls}). However, in the absence of an anomaly, and even restricting to onsite symmetries, both BI and non-BI cases exist (for onsite symmetries, BI can be phrased as the extendability of the $2$-cocycle of the virtual representation of the unbroken symmetry group; see the examples below and Appendix~\ref{app:block_indep} for details). Thus, for a noninjective invariant MPS satisfying BI, we can promote the symmetry defects to quantum degrees of freedom, resulting in a gauged MPS satisfying local gauge constraints. Moreover, the resulting projectors on the gauge-invariant subspace will commute on different sites, due to the following
	\begin{proposition}
		Let $P_j$ be given by \eqref{eq:gauge_projs}. Then $[P_j, P_{j'}]=0$, for all $j, j'$, if and only if the symmetry is nonanomalous.
		\label{prop:comm_proj}
	\end{proposition}
	\begin{proof}
		The proof is analogous to the one in \cite[Section~3.2]{Seifnashri}. By locality, it is enough to consider the commutators $[P_j, P_{j-1}]$. Then it is shown that if there exists a scalar function $F: G^3 \rightarrow \mathbb{C}$ relating the two fusion paths for three consecutive defects,
		\begin{equation}
			\begin{tikzpicture}
				\def\dx{1.0}
				\def\dya{.5}
				\def\dyb{1.25}
				\def\dyc{2.0}
				\def\dyd{2.5}
				\def\ddx{0.1}
				\def\ddy{0.3}
				\draw[thick] (0,0) --++ (0,\dyd);
				\draw[thick] (-\dx,0) --++ (0,\dyd);
				\draw[thick] (\dx,0) --++ (0,\dyd);
				\draw[fill=white] (0-\ddx,\dya-\ddy) rectangle (\dx+\ddx, \dya + \ddy);
				\draw[fill=white] (-\dx-\ddx,\dyb-\ddy) rectangle (\ddx, \dyb + \ddy);
				\draw[fill=white] (-\ddx,\dyc-\ddy) rectangle (\dx+\ddx, \dyc + \ddy);
				\node[] at (\dx/2,\dya) {$\lambda^R_{g_2,g_3}$};
				\node[] at (\dx/2,\dyc) {$\lambda^R_{g_1,g_2g_3}$};
				\node[] at (-\dx/2,\dyb) {$\lambda^R_{g_1,e}$};
			\end{tikzpicture}
			=
			F(g_1,g_2,g_3)\;
			\begin{tikzpicture}
				\def\dx{1.0}
				\def\dyb{2.5/3}
				\def\dyc{5/3}
				\def\dyd{2.5}
				\def\ddx{0.1}
				\def\ddy{0.3}
				\draw[thick] (0,0) --++ (0,\dyd);
				\draw[thick] (-\dx,0) --++ (0,\dyd);
				\draw[thick] (\dx,0) --++ (0,\dyd);
				\draw[fill=white] (-\dx-\ddx,\dyb-\ddy) rectangle (\ddx, \dyb + \ddy);
				\draw[fill=white] (-\ddx,\dyc-\ddy) rectangle (\dx+\ddx, \dyc + \ddy);
				\node[] at (\dx/2,\dyc) {$\lambda^R_{g_1,g_2}$};
				\node[] at (-\dx/2,\dyb) {$\lambda^R_{g_1g_2,g_3}$};
			\end{tikzpicture}
			,
			\label{eq:assoc}
		\end{equation}
		then the commutativity condition $[P_j,P_{j-1}]=0$ reads
		\begin{equation}
			\frac{ F(h_1,h_2h_3g_3^{-1},g_3)}{F(g_1,g_2,g_3)}=\frac{F(h_1,h_2,h_3)}{F(g_1,g_{1}^{-1}h_1h_2,h_3)},
		\end{equation}
		for all $h_1,h_2,h_3,g_1,g_2,g_3\in G$ (this relies on the fact that the fusion operators $\lambda^R$ can be chosen to be unitary). Picking $h_1=e$, $h_2=g_1$ and $h_3=g_2g_3$, we get
		\begin{equation}F(g_1,g_2,g_3)=\frac{F(g_1,e,g_2g_3)F(e,g_1g_2,g_3)}{F(e,g_1,g_2g_3)}.
		\end{equation}
		In Appendix~\ref{app:associativity} we prove that our fusion operators indeed satisfy \eqref{eq:assoc} with $F(g, h, k)=\omega(g,h,k)$. With our gauge choices, we also have $\omega(e,g,h)=\omega(g,e,h)=1$, for all $g,h\in G$. Therefore, $[P_j,P_{j-1}]=0$ if and only if $\omega$ is trivial, i.e., if the anomaly vanishes.
	\end{proof}
	
	\subsection{Examples}
	\subsubsection{Onsite symmetries}
	\label{subsec:onsite}
	As a first example, we reproduce the formalism for onsite symmetries from Section~\ref{sec:onsite}. Onsite global symmetry groups constitute MPUs with bond dimension $1$, thus all fusion tensors are scalars. We need to choose $X_i\text{--}Y_i$ decomposition for our MPUs, which we do as follows:
	\begin{equation}
		X^g_1 = Y^g_2 = u_g,\qquad Y^g_1 = X^g_2 = \mathds{1}.
		\label{eq:XYonsite}
	\end{equation}
	We also initially pick the fusion tensors to be trivial and the action tensors to be given by the virtual representations:
	\begin{equation}
		F^{<}_{g,h} = F^>_{g,h} = 1,\qquad A^\Gamma_g = S_g, \qquad A^\ammaG_g = S^\dagger_g.
		\label{eq:onsitechoice}
	\end{equation}
	This way our truncated symmetries read
	\begin{equation}
		U^L_g = \overbrace{u_g\otimes\ldots\otimes u_g}^{L-1}\otimes\mathds{1},
		\label{eq:ULonsite}
	\end{equation}
	and the defect tensors are given by 
	\begin{equation}
		\begin{tikzpicture}[baseline]
			\def\dx{0.3}
			\def\dy{0.5}
			\draw[thick] (0,0) --++ (0,\dy);
			\draw[thick] (-\dx, 0) --++ (2*\dx,0);
			\node[square] at (0,0) {};
			\node[] at (-0.25,0.25) {$g$};
		\end{tikzpicture}
		=\begin{tikzpicture}[baseline]
			\def\dx{0.2}
			\def\dxx{0.5}
			\def\dy{0.5}
			\draw[thick] (\dx,0) --++ (0,\dy);
			\draw[thick] (-\dxx, 0) --++ (2*\dxx,0);
			\node[square, fill=blue] at (-\dx,0) {};
			\node[tensor] at (\dx,0) {};
			\node[blue] at (-\dx,0.3) {$S_g$};
			\node[] at (3.5*\dx,0) {$.$};
		\end{tikzpicture}
	\end{equation}
	The defect fusion operators in this gauge read $\lambda^R_{g,h}=u_g\otimes\mathds{1}$, $\lambda_{g,h}^L= u_{h}^\dagger\otimes\mathds{1}$. They are unitary and satisfy
	\begin{equation}
		\begin{tikzpicture}[baseline=6]
			\def\dx{0.3}
			\def\dxx{0.7}
			\def\dy{0.8}
			\def\ddy{0.3}
			\draw[thick] (-\dx,0) --++ (0,\dy);
			\draw[thick] (\dx,0) --++ (0,\dy);
			\draw[thick] (-\dxx, 0) --++ (2*\dxx,0);
			\node[square] at (-\dx,0) {};
			\node[square] at (\dx,0) {};
			\draw[fill=white] (-\dx-0.1,\ddy) rectangle (\dx +0.1, \ddy + 0.3);
			\node[] at (\dx,-0.25) {$h$};
			\node[] at (-\dx,-0.25) {$g$};
			\node[] at (-\dx-0.45,\ddy+0.15) {$\lambda^R_{g,h}$};
		\end{tikzpicture}
		=\Omega(g,h)\;
		\begin{tikzpicture}[baseline=6]
			\def\dx{0.3}
			\def\dxx{0.7}
			\def\dy{0.7}
			\draw[thick] (-\dx,0) --++ (0,\dy);
			\draw[thick] (\dx,0) --++ (0,\dy);
			\draw[thick] (-\dxx, 0) --++ (2*\dxx,0);
			\node[tensor] at (-\dx,0) {};
			\node[square] at (\dx,0) {};
			\node[] at (\dx,-0.25) {$gh$};
			\node[] at (3*\dx,0) {$,$};
		\end{tikzpicture}
	\end{equation}
	\begin{equation}
		\begin{tikzpicture}[baseline=6]
			\def\dx{0.3}
			\def\dxx{0.7}
			\def\dy{0.8}
			\def\ddy{0.3}
			\draw[thick] (-\dx,0) --++ (0,\dy);
			\draw[thick] (\dx,0) --++ (0,\dy);
			\draw[thick] (-\dxx, 0) --++ (2*\dxx,0);
			\node[square] at (-\dx,0) {};
			\node[square] at (\dx,0) {};
			\draw[fill=white] (-\dx-0.1,\ddy) rectangle (\dx +0.1, \ddy + 0.3);
			\node[] at (\dx,-0.25) {$h$};
			\node[] at (-\dx,-0.25) {$g$};
			\node[] at (-\dx-0.45,\ddy+0.15) {$\lambda^L_{g,h}$};
		\end{tikzpicture}
		=\Omega(g,h)\;
		\begin{tikzpicture}[baseline=6]
			\def\dx{0.3}
			\def\dxx{0.7}
			\def\dy{0.7}
			\draw[thick] (-\dx,0) --++ (0,\dy);
			\draw[thick] (\dx,0) --++ (0,\dy);
			\draw[thick] (-\dxx, 0) --++ (2*\dxx,0);
			\node[square] at (-\dx,0) {};
			\node[tensor] at (\dx,0) {};
			\node[] at (-\dx,-0.25) {$gh$};
			\node[] at (3*\dx,0) {$,$};
		\end{tikzpicture}
	\end{equation}
	since the $L$-symbols in this gauge are given by the $2$-cocycle $\Omega(g,h)$. We can regauge the fusion tensors so that the $L$-symbols become one, which results in fusion operators
	\begin{equation}
		\lambda^R_{g,h} = \dfrac{u_g\otimes\mathds{1}}{\Omega(g,h)},\qquad \lambda^L_{g,h} = \dfrac{u^\dagger_h\otimes\mathds{1}}{\Omega(g,h)},
	\end{equation}
	and Gauss laws
	\begin{align}
		\mc G_g&=\sum_{a,b\in G}{\dfrac{\Omega(ag^{-1},gb)}{\Omega(a,b)}(u_g\otimes\mathds{1})\otimes\ketbra{ag^{-1},gb}{a,b}}\nonumber\\
		&=\sum_{a,b\in G}{\dfrac{u_g\otimes\mathds{1}}{\Omega(a,g^{-1})\Omega(g,b)}\otimes\ketbra{ag^{-1},gb}{a,b}}\nonumber\\
		&= u_g\otimes\mathds{1}\otimes R^{\Omega}_g\otimes  L^{\Omega}_g,
	\end{align}
	which act only on two gauge legs and one matter leg, and leave the gauged MPS generated by 
	\begin{equation}
		\begin{tikzpicture}[baseline=1mm]
			\def\dx{0.1}
			\def\dxx{0.5}
			\def\dy{0.5}
			\draw[thick] (-\dx,0) --++ (0,\dy);
			\draw[thick, dotted] (\dx,0) --++ (0,\dy);
			\draw[thick] (-\dxx, 0) --++ (2*\dxx,0);
			\node[widetensor] at (0,0) {};
		\end{tikzpicture}
		\equiv
		\sum_{g\in G}{\;\;
			\begin{tikzpicture}[baseline=1mm]
				\def\dx{0.2}
				\def\dxx{0.5}
				\def\dy{0.5}
				\draw[thick] (\dx,0) --++ (0,\dy);
				\draw[thick] (-\dxx, 0) --++ (2*\dxx,0);
				\node[square, fill=blue] at (-\dx,0) {};
				\node[tensor] at (\dx,0) {};
				\node[blue] at (-\dx,0.3) {$S_g$};
			\end{tikzpicture}\otimes \ket{g}}
	\end{equation}
	invariant. Since the left action commutes with the right one, $[P_j,P_k]=0$ for all positions $j,k$.
	
	\subsubsection{Non-onsite $\Z_2$ symmetry}
	\label{sec:CZ}
	Consider now a genuine MPU symmetry, albeit the simplest example with the product of local $\mathsf{CZ}$ gates
	\begin{equation}
		U = \prod_{j}{\mathsf{CZ_{j,j+1}}}.
	\end{equation}
	Using the following representation for the controlled-$\mathsf{Z}$ gate,
	\begin{equation}
		\begin{tikzpicture}
			\draw[thick] (-0.3,-0.4)--++(0.0, 0.8);  \draw[thick] (0.3,-0.4)--++(0.0, 0.8);  
			\node[rectangle, fill=antiquewhite] at (0,0) {$\mathsf{CZ}$};
			
			\node[] at (0.75, 0.0) {$=$};
			
			\draw[thick] (1.0,-0.25) --++ (0.0,0.5);
			\draw[thick] (1.5,-0.25) --++ (0.0,0.5);
			\draw[thick,red] (1.0,0.0) --++ (0.5,0.0);
			\node[tensor, fill=blue] at (1.25,0.0) {};
			\node[] at (1.75,0.0) {$.$};   
		\end{tikzpicture}
		\label{eq:CZ}
	\end{equation}
	where 
	\begin{equation}
		\begin{tikzpicture}
			\draw[red,thick] (0.0,0.0) --++ (0.5,0.0);
			\node[tensor, fill= blue] at (0.25,0.0) {};
			\node[] at (1.85,0.0) {$=\sqrt{2}H=\begin{pmatrix}
					1&1\\
					1&-1
				\end{pmatrix},$};
			\def\Sc{1.8};
			\draw[thick] (0.0+\Sc*\Shift, -0.5) --++(0.0,1.0);
			\draw[red,thick] (0.0+\Sc*\Shift, 0.0) --++(0.5,0.0);
			\node[] at (0.0+\Sc*\Shift, -0.65) {$i$};
			\node[] at (0.0+\Sc*\Shift, 0.65) {$k$};
			\node[] at (0.6+\Sc*\Shift, 0.0) {$j$};
			
			\node[] at (0.9+\Sc*\Shift, 0.0) {$=$};
			
			\def\co{1.475};
			
			\draw[thick] (0.0+\co*\Sc*\Shift, -0.5) --++(0.0,1.0);
			\draw[red,thick] (0.0+\co*\Sc*\Shift, 0.0) --++(-0.5,0.0);
			\node[] at (0.0+\co*\Sc*\Shift, -0.65) {$i$};
			\node[] at (0.0+\co*\Sc*\Shift, 0.65) {$k$};
			\node[] at (-0.6+\co*\Sc*\Shift, 0.0) {$j$};
			
			\node[] at (0.8+\co*\Sc*\Shift,0.0) {$=\ \delta_{i,j,k},$};
		\end{tikzpicture}
		\label{eq:deltas}
	\end{equation}
	the many-body unitary $U$ can be written as an MPU of bond dimension $2$,
	\begin{equation}
		U=
		\begin{tikzpicture}
			\def\dx{0.6}
			\def\ddx{0.4}
			\def\dy{0.15}
			\def\dyb{0.4}
			\draw[thick] (0,-\dyb) --++ (0,2*\dyb);
			\draw[thick] (\dx,-\dyb) --++ (0,2*\dyb);
			\draw[thick] (2*\dx,-\dyb) --++ (0,2*\dyb);
			\draw[thick] (3*\dx,-\dyb) --++ (0,2*\dyb);
			\draw[thick, red] (-\ddx, -\dy) --++ (\ddx,0);
			\draw[thick, red] (0, \dy) --++ (\dx,0);
			\draw[thick, red] (\dx, -\dy) --++ (\dx,0);
			\draw[thick, red] (2*\dx, \dy) --++ (\dx,0);
			\draw[thick, red] (3*\dx, -\dy) --++ (\ddx,0);
			\node[tensor, blue] at (0.5*\dx, \dy) {};
			\node[tensor, blue] at (1.5*\dx, -\dy) {};
			\node[tensor, blue] at (2.5*\dx, \dy) {};
			\node[tensor, blue] at (3*\dx+0.5*\ddx, -\dy) {};
			\node[] at (0-\ddx-0.3,0) {$\ldots$};
			\node[] at (3*\dx+\ddx+0.3,0) {$\ldots$};
		\end{tikzpicture}
	\end{equation}
	and since it squares to the identity, it yields a representation of $\Z_2=\langle g\mid g^2=e\rangle$ which can be checked to be nonanomalous. Now we make a convenient choice of an MPS state that is invariant under this symmetry:
	\begin{equation}
		A^0 = \left(\begin{array}{ccc}
			1 & 0 & 0\\
			0 & 0 & 0\\
			0 & 1 & 1
		\end{array}\right), \qquad A^1 = \left(\begin{array}{ccc}
			0 & 1 & -1\\
			-1 & 0 & 0\\
			0 & 0 & 0
		\end{array}\right). 
	\end{equation}
	 This state appeared in \cite{simps}, where it was pointed out that, apart from the CZ symmetry, it also has an onsite $\Z_2$ symmetry $\prod_{j}{Z_j}$ and belongs to the nontrivial SPT phase of the combined $\Z_2\times\Z_2$ theory. Additionally, it has a split-index MPS structure (see below). 
		
	First of all, we note that this is not an injective, but rather a normal MPS: it becomes injective after blocking sites. It turns out, however, that we can carry out the formalism without blocking in this case, since all the necessary objects exist before blocking. In particular, the defect tensor
	\begin{equation}
		A_g^0 \!=\! \left(\begin{array}{ccc}
			0 & -1 & -1\\
			0 & 0 & 0\\
			-1 & 0 & 0
		\end{array}\right), \quad A_g^1 \!=\! \left(\begin{array}{ccc}
			1 & 0 & 0\\
			0 & -1 & 1\\
			0 & 0 & 0
		\end{array}\right), 
	\end{equation}
	satisfies 
	\begin{equation}
		(-1)^{ij}A^iA^j = A_g^iA_g^j,\qquad (-1)^{ij}A_g^iA^j = A^iA_g^j,
		\label{eq:Z2_def}
	\end{equation}
	or in diagrams
	\begin{equation}
		\begin{tikzpicture}
			\def\dx{0.3}
			\def\dxb{0.6}
			\def\dy{0.5}
			\def\dyb{0.7}
			\draw[thick] (-\dx,0) --++ (0,\dyb);
			\draw[thick] (\dx,0) --++ (0,\dyb);
			\draw[thick] (-\dxb, 0) --++ (2*\dxb,0);
			\draw[thick, red] (-\dx, \dy) --++ (2*\dx,0);
			\node[tensor] at (-\dx,0) {};
			\node[tensor] at (\dx, 0) {};
			\node[tensor, blue] at (0, \dy) {};
			\node[] at (-\dx,-0.3) {$A$};
			\node[] at (\dx,-0.3) {$A$};
		\end{tikzpicture}
		=  
		\begin{tikzpicture}
			\def\dx{0.3}
			\def\dxb{0.6}
			\def\dy{0.5}
			\def\dyb{0.7}
			\draw[thick] (-\dx,0) --++ (0,\dyb);
			\draw[thick] (\dx,0) --++ (0,\dyb);
			\draw[thick] (-\dxb, 0) --++ (2*\dxb,0);
			\node[square] at (-\dx,0) {};
			\node[square] at (\dx, 0) {};
			\node[] at (-\dx,-0.35) {$A_g$};
			\node[] at (\dx,-0.35) {$A_g$};
		\end{tikzpicture}
		,\qquad
		\begin{tikzpicture}
			\def\dx{0.3}
			\def\dxb{0.6}
			\def\dy{0.5}
			\def\dyb{0.7}
			\draw[thick] (-\dx,0) --++ (0,\dyb);
			\draw[thick] (\dx,0) --++ (0,\dyb);
			\draw[thick] (-\dxb, 0) --++ (2*\dxb,0);
			\draw[thick, red] (-\dx, \dy) --++ (2*\dx,0);
			\node[square] at (-\dx,0) {};
			\node[tensor] at (\dx, 0) {};
			\node[tensor, blue] at (0, \dy) {};
			\node[] at (-\dx,-0.35) {$A_g$};
			\node[] at (\dx,-0.3) {$A$};
		\end{tikzpicture}
		=  
		\begin{tikzpicture}
			\def\dx{0.3}
			\def\dxb{0.6}
			\def\dy{0.5}
			\def\dyb{0.7}
			\draw[thick] (-\dx,0) --++ (0,\dyb);
			\draw[thick] (\dx,0) --++ (0,\dyb);
			\draw[thick] (-\dxb, 0) --++ (2*\dxb,0);
			\node[tensor] at (-\dx,0) {};
			\node[square] at (\dx, 0) {};
			\node[] at (-\dx,-0.3) {$A$};
			\node[] at (\dx,-0.35) {$A_g$};
		\end{tikzpicture},
	\end{equation}
	that is, $\lambda_{g_1,g_2}=\mathsf{CZ}$ as long as $g_1\neq e$ or $g_2\neq e$. To gauge the symmetry, we define the gauged tensor
	\begin{equation}
		C^{0e} = A^0,~~C^{1e} = A^1,~~C^{0g} = A_g^0,~~C^{1g} = A_g^1,
	\end{equation}
	and the resulting MPS can be checked to be invariant under the gauge constraint $\mcG_g = \mathsf{CZ} \otimes X\otimes X$, where the first factor acts on the matter indices and the others on the gauge indices; this is an easy consequence of Eq.~\eqref{eq:Z2_def}.
	
	As mentioned above, the defect and gauged tensors can be identified from the formalism in \cite{simps}, which shows a hidden structure in this example. Indeed, the MPS generated by $A^i$ can also be expressed as a split-index MPS (SIMPS) of bond dimension 2,
	\begin{equation}
		\begin{tikzpicture}
			\def\dx{0.8}
			\def\ddx{0.2}
			\def\ddxb{0.1}
			\def\dy{0.4}
			\def\dyb{0.7}
			\draw[thick] (-\ddx,0) --++ (0,\dy);
			\draw[thick] (\ddx,0) --++ (0,\dy);
			\draw[thick] (\dx-\ddx,0) --++ (0,\dy);
			\draw[thick] (\dx+\ddx,0) --++ (0,\dy);
			\draw[thick] (-\dx-\ddx,0) --++ (0,\dy);
			\draw[thick] (\ddx-\dx,0) --++ (0,\dy);
			\draw[thick] (0.5*\dx, \dy) --++ (0,\dyb-\dy);
			\draw[thick] (-0.5*\dx, \dy) --++ (0,\dyb-\dy);
			\draw[thick] (1.5*\dx, \dy) --++ (0,\dyb-\dy);
			\draw[thick] (-1.5*\dx, \dy) --++ (0,\dyb-\dy);
			\draw[thick] (-1.5*\dx-\ddxb, 0) -- (1.5*\dx+\ddxb,0);
			\draw[thick] (-\dx+\ddx, \dy) -- (-\ddx,\dy);
			\draw[thick] (\ddx, \dy) -- (\dx-\ddx,\dy);
			\draw[thick] (-1.5*\dx-\ddxb, \dy) -- (-\dx-\ddx,\dy);
			\draw[thick] (\dx+\ddx, \dy) -- (1.5*\dx+\ddxb,\dy);
			\node[simpstensor] at (-\dx,0) {};
			\node[simpstensor] at (0,0) {};
			\node[simpstensor] at (\dx,0) {};
			\node[] at (-1.5*\dx-0.4,0) {$\ldots$};
			\node[] at (1.5*\dx+0.4,0) {$\ldots \ ,$};
		\end{tikzpicture}
	\end{equation}
	where the three-legged intersections of wires host copy tensors $\delta_{i,j,k}$ as in Eq. \eqref{eq:deltas} and the SIMPS tensor satisfies
	\begin{equation}
		\begin{tikzpicture}[baseline = 2mm]
			\def\dx{1.0}
			\def\ddx{0.2}
			\def\dxb{0.7}
			\def\dy{0.5}
			\def\dyb{0.9}
			\draw[thick] (-\ddx,0) --++ (0,\dy);
			\draw[thick] (\ddx,0) --++ (0,\dy);
			\draw[thick] (\dx-\ddx,0.12) --++ (0,\dy-0.12);
			\draw[thick] (\ddx-\dx,0.12) --++ (0,\dy-0.12);
			\draw[thick] (0.5*\dx, \dy) --++ (0,\dyb-\dy);
			\draw[thick] (-0.5*\dx, \dy) --++ (0,\dyb-\dy);
			\draw[thick] (-\dxb, 0) -- (\dxb,0);
			\draw[thick] (-\dx+\ddx, \dy) -- (-\ddx,\dy);
			\draw[thick] (\ddx, \dy) -- (\dx-\ddx,\dy);
			\draw[thick, red] (-0.5*\dx, 0.5*\dyb+0.5*\dy) --++ (\dx,0);
			\node[simpstensor] at (0,0) {};
			\node[tensor, blue] at (0, 0.5*\dyb+0.5*\dy) {};
		\end{tikzpicture}
		\ = \
		\begin{tikzpicture}[baseline = 2mm]
			\def\dx{1.0}
			\def\ddx{0.2}
			\def\dxb{0.7}
			\def\dy{0.5}
			\def\dyb{0.9}
			\draw[thick] (-\ddx,0) --++ (0,\dy);
			\draw[thick] (\ddx,0) --++ (0,\dy);
			\draw[thick] (\dx-\ddx,0.12) --++ (0,\dy-0.12);
			\draw[thick] (\ddx-\dx,0.12) --++ (0,\dy-0.12);
			\draw[thick] (0.5*\dx, \dy) --++ (0,\dyb-\dy);
			\draw[thick] (-0.5*\dx, \dy) --++ (0,\dyb-\dy);
			\draw[thick] (-\dxb, 0) -- (\dxb,0);
			\draw[thick] (-\dx+\ddx, \dy) -- (-\ddx,\dy);
			\draw[thick] (\ddx, \dy) -- (\dx-\ddx,\dy);
			\draw[thick, red] (-\ddx, .1+0.5*\dy) --++ (2*\ddx,0);
			\node[simpstensor] at (0,0) {};
			\node[tensor, blue] at (0, .1+0.5*\dy) {};
		\end{tikzpicture}
		\ = \
		\begin{tikzpicture}[baseline = 2mm]
			\def\dx{1.0}
			\def\ddx{0.2}
			\def\dxb{0.8}
			\def\dy{0.5}
			\def\dyb{0.9}
			\draw[thick] (-\ddx,0) --++ (0,\dy);
			\draw[thick] (\ddx,0) --++ (0,\dy);
			\draw[thick] (\dx-\ddx,0.12) --++ (0,\dy-0.12);
			\draw[thick] (\ddx-\dx,0.12) --++ (0,\dy-0.12);
			\draw[thick] (0.5*\dx, \dy) --++ (0,\dyb-\dy);
			\draw[thick] (-0.5*\dx, \dy) --++ (0,\dyb-\dy);
			\draw[thick] (-\dxb, 0) -- (\dxb,0);
			\draw[thick] (-\dx+\ddx, \dy) -- (-\ddx,\dy);
			\draw[thick] (\ddx, \dy) -- (\dx-\ddx,\dy);
			\node[simpstensor] at (0,0) {};
			\node[square, fill=white] at (-0.5*\dx, 0) {};
			\node[square, fill=white] at (0.5*\dx, 0) {};
			\node[] at (-0.5*\dx, -0.3) {$X$};
			\node[] at (0.5*\dx, -0.3) {$X$};
			\node[] at (\dx, 0) {$.$};
		\end{tikzpicture}
	\end{equation}
	The SIMPS tensor is then straightforward to decompose into two tensors $L,R$ such that
	\begin{equation}
		\begin{tikzpicture}[baseline = 2mm]
			\def\dx{.6}
			\def\ddx{0.2}
			\def\dy{0.5}
			\draw[thick] (-\ddx,0) --++ (0,\dy);
			\draw[thick] (\ddx,0) --++ (0,\dy);
			\draw[thick] (-\dx, 0) -- (\dx,0);
			\node[simpstensor] at (0,0) {};
		\end{tikzpicture}
		=
		\begin{tikzpicture}[baseline = 2mm]
			\def\dx{.4}
			\def\dxb{0.8}
			\def\dy{0.5}
			\draw[thick] (-\dx,0) --++ (0,\dy);
			\draw[thick] (\dx,0) --++ (0,\dy);
			\draw[thick] (-\dxb, 0) -- (\dxb,0);
			\draw[ultra thick] (-\dx, 0) -- (\dx,0);
			\draw[thick, fill=Goldenrod] (-\dx-0.2,-0.125) rectangle (-\dx +0.2, 0.125);
			\draw[thick, fill=Goldenrod] (\dx-0.2,-0.125) rectangle (\dx +0.2, 0.125);
			\node[] at (-\dx, -0.3) {$L$};
			\node[] at (\dx, -0.3) {$R$};
			\node[] at (2.5*\dx,0) {$,$};
		\end{tikzpicture}
	\end{equation}
	and we have
	\begin{equation}
		\begin{tikzpicture}[baseline]
			\def\dx{0.4}
			\def\dy{0.5}
			\draw[thick] (0,0) --++ (0,\dy);
			\draw[thick] (-\dx, 0) -- (\dx,0);
			\node[tensor] at (0,0) {};
			\node[] at (0, -0.3) {$A$};
		\end{tikzpicture}
		=
		\begin{tikzpicture}[baseline]
			\def\dx{.45}
			\def\dxb{0.9}
			\def\dy{0.3}
			\def\dyb{0.5}
			\draw[thick] (-\dx,0) --++ (0,\dy);
			\draw[thick] (\dx,0) --++ (0,\dy);
			\draw[thick] (0,\dy) -- (0,\dyb);
			\draw[thick] (-\dxb, 0) -- (\dxb,0);
			\draw[thick] (-\dx, \dy) --++ (2*\dx,0);
			\draw[ultra thick] (-\dxb, 0) -- (-\dx,0);
			\draw[ultra thick] (\dx, 0) -- (\dxb,0);
			\draw[thick, fill=Goldenrod] (-\dx-0.2,-0.125) rectangle (-\dx +0.2, 0.125);
			\draw[thick, fill=Goldenrod] (\dx-0.2,-0.125) rectangle (\dx +0.2, 0.125);
			\node[] at (-\dx, -0.3) {$R$};
			\node[] at (\dx, -0.3) {$L$};
		\end{tikzpicture}
		,\quad
		\begin{tikzpicture}[baseline]
			\def\dx{0.4}
			\def\dy{0.5}
			\draw[thick] (0,0) --++ (0,\dy);
			\draw[thick] (-\dx, 0) -- (\dx,0);
			\node[square] at (0,0) {};
			\node[] at (0, -0.3) {$A_g$};
		\end{tikzpicture}
		=
		\begin{tikzpicture}[baseline]
			\def\dx{.45}
			\def\dxb{0.9}
			\def\dy{0.3}
			\def\dyb{0.5}
			\draw[thick] (-\dx,0) --++ (0,\dy);
			\draw[thick] (\dx,0) --++ (0,\dy);
			\draw[thick] (0,\dy) -- (0,\dyb);
			\draw[thick] (-\dxb, 0) -- (\dxb,0);
			\draw[thick] (-\dx, \dy) --++ (2*\dx,0);
			\draw[ultra thick] (-\dxb, 0) -- (-\dx,0);
			\draw[ultra thick] (\dx, 0) -- (\dxb,0);
			\draw[thick, fill=Goldenrod] (-\dx-0.2,-0.125) rectangle (-\dx +0.2, 0.125);
			\draw[thick, fill=Goldenrod] (\dx-0.2,-0.125) rectangle (\dx +0.2, 0.125);
			\node[] at (-\dx, -0.3) {$R$};
			\node[] at (\dx, -0.3) {$L$};
			\node[] at (0, -0.3) {$X$};
			\node[square, fill=white] at (0, 0) {};
		\end{tikzpicture}
		,
	\end{equation}
	and
	\begin{equation}
		\begin{tikzpicture}[baseline]
			\def\dx{0.4}
			\def\dd{0.1}
			\def\dy{0.5}
			\draw[thick] (-\dd, 0) --++ (0,\dy);
			\draw[thick, dotted] (\dd, 0) --++ (0,\dy);
			\draw[thick] (-\dx, 0) -- (\dx,0);
			\node[widetensor] at (0,0) {};
			\node[] at (0, -0.3) {$C$};
		\end{tikzpicture}
		=
		\begin{tikzpicture}[baseline]
			\def\dx{.45}
			\def\dxb{0.9}
			\def\dy{0.3}
			\def\dyb{0.5}
			\draw[thick] (-\dx,0) --++ (0,\dy);
			\draw[thick] (\dx,0) --++ (0,\dy);
			\draw[thick] (-0.3,\dy) --++ (0,\dyb-\dy);
			\draw[thick, dotted] (0, 0) --++ (0,\dyb);
			\draw[thick] (-\dxb, 0) -- (\dxb,0);
			\draw[thick] (-\dx, \dy) --++ (2*\dx,0);
			\draw[ultra thick] (-\dxb, 0) -- (-\dx,0);
			\draw[ultra thick] (\dx, 0) -- (\dxb,0);
			\draw[thick, fill=Goldenrod] (-\dx-0.2,-0.125) rectangle (-\dx +0.2, 0.125);
			\draw[thick, fill=Goldenrod] (\dx-0.2,-0.125) rectangle (\dx +0.2, 0.125);
			\node[] at (-\dx, -0.3) {$R$};
			\node[] at (\dx, -0.3) {$L$};
			\node[square, fill=white] at (0, 0) {};
		\end{tikzpicture}
		,
	\end{equation}
	where the internal tensor with the gauge physical leg is defined as in \eqref{eq:gaugedoftensor} for the virtual representation $S_e=\mathds{1}$, $S_g = X$.
	
	\section{State-level gauging}
	\label{sec:state_level}
	
	If we find ourselves in a block-dependent case, it is not obvious that we can modify the group representation $\mcG$ so that our candidate MPS is invariant. For this to be the case, the modified Gauss law should be able to account for different scalar factors on distinct blocks. The most straightforward way to do this is to manually remove the $L$-symbol factors by modifying the fusion operators $\lambda^R$, effectively ``enforcing'' block independence. This can be achieved if the injective blocks of our MPS are \textit{locally orthogonal}, which can be achieved by blocking until we approach a fixed point of the renormalization group.
	\begin{lemma}
		Assume that all injective MPS tensors are locally orthogonal, i.e.
		\begin{equation}
			\begin{tikzpicture}
				\def\dx{0.5}
				\def\dy{0.6}
				\draw[thick] (0,0) --++ (0,\dy);
				\draw[thick] (-\dx, 0) --++ (2*\dx,0);
				\draw[thick] (-\dx, \dy) --++ (2*\dx,0);
				\node[tensor] at (0.0, 0.0) {};
				\node[tensor] at (0.0, \dy) {};
				\node[] at (0.0,-0.25) {$x$};
				\node[] at (0.0,\dy+0.25) {$y$};
				\node[] at (0.15,\dy+0.1) {$*$};
			\end{tikzpicture}
			=0,\qquad x\neq y, 
		\end{equation}
		and let
		\begin{equation}
			v(g, h, x) = 
			\begin{tikzpicture}
				\def\dx{0.6}
				\def\dxx{1.2}
				\def\dy{0.6}
				\draw[thick] (\dx,0) --++ (0,\dy);
				\draw[thick] (-\dx,0) --++ (0,\dy);
				\draw[thick] (-\dxx, 0) --++ (2*\dxx,0);
				\node[square] at (\dx, 0) {};
				\node[square] at (-\dx, 0) {};
				\node[] at (\dx,-0.25) {$h[x]$};
				\node[] at (-\dx,-0.25) {$g[hx]$};
			\end{tikzpicture}
		\end{equation}
		Then for all $g,h\in G$ there exists a unitary $\tilde\lambda_{g,h}\in \text{\emph U}\left(\mc H^{\otimes 2}\right)$ such that 
		\begin{equation}
			\tilde\lambda_{g,h}v(g, h, x) = v(e, gh, x)
			\label{eq:withoutL}
		\end{equation}
		where $\tilde\lambda_{g,h}$ acts on the physical legs. That is, $\tilde\lambda_{g,h}$ satisfies Eq.~\eqref{eq:fusop2} without the $L$-symbol factor.
		\label{lemma:modif}
	\end{lemma}
	\begin{proof}
		It is sufficient to prove that $v(g, h, x)$ and $v(g, h, y)$ are orthogonal whenever $x\neq y$ so that we can absorb potentially different scalars $L^x_{g,h},L^y_{g,h}$ into a redefinition of the unitary $\lambda_{g,h}$. We compute the inner product
		\begin{equation}
			\begin{tikzpicture}
				\def\dx{0.6}
				\def\dxx{1.0}
				\def\dy{0.8}
				\draw[thick] (\dx,0) --++ (0,\dy);
				\draw[thick] (-\dx,0) --++ (0,\dy);
				\draw[thick] (-\dxx, 0) --++ (2*\dxx,0);
				\draw[thick] (-\dxx, \dy) --++ (2*\dxx,0);
				\node[square] at (\dx, 0) {};
				\node[square] at (-\dx, 0) {};
				\node[square] at (\dx, \dy) {};
				\node[square] at (-\dx, \dy) {};
				\node[] at (\dx,-0.25) {$h[x]$};
				\node[] at (-\dx,-0.25) {$g[hx]$};
				\node[] at (\dx,\dy+0.35) {$h[y]$};
				\node[] at (-\dx,\dy+0.35) {$g[hy]$};
				\node[] at (\dx+0.25,\dy+0.15) {$*$};
				\node[] at (-\dx+0.25,\dy+0.15) {$*$};
			\end{tikzpicture}
			=
			\begin{tikzpicture}
				\def\dx{0.6}
				\def\dxx{1.0}
				\def\dy{1.8}
				\def\dd{0.25}
				\draw[thick] (\dx,0) --++ (0,\dy);
				\draw[thick] (-\dx,0) --++ (0,\dy);
				\draw[thick] (-\dxx, 0) --++ (2*\dxx,0);
				\draw[thick] (-\dxx, \dy) --++ (2*\dxx,0);
				\node[square] at (\dx, 0) {};
				\node[square] at (-\dx, 0) {};
				\node[square] at (\dx, \dy) {};
				\node[square] at (-\dx, \dy) {};
				\draw (-\dx-0.1, \dy/3-\dd) rectangle (\dx+0.1, \dy/3+\dd) [fill=white, thick];
				\draw (-\dx-0.1, 2*\dy/3-\dd) rectangle (\dx+0.1, 2*\dy/3+\dd) [fill=white, thick];
				\node[] at (\dx,-0.25) {$h[x]$};
				\node[] at (-\dx,-0.25) {$g[hx]$};
				\node[] at (\dx,\dy+0.35) {$h[y]$};
				\node[] at (-\dx,\dy+0.35) {$g[hy]$};
				\node[] at (\dx+0.25,\dy+0.15) {$*$};
				\node[] at (-\dx+0.25,\dy+0.15) {$*$};
				\node[] at (0,\dy/1.5){$(\lambda^L_{g,h})^\dagger$};
				\node[] at (0,\dy/3) {$\lambda^L_{g,h}$};
			\end{tikzpicture}
			\propto 
			\begin{tikzpicture}
				\def\dx{0.4}
				\def\dxx{0.8}
				\def\dy{0.8}
				\draw[thick] (\dx,0) --++ (0,\dy);
				\draw[thick] (-\dx,0) --++ (0,\dy);
				\draw[thick] (-\dxx, 0) --++ (2*\dxx,0);
				\draw[thick] (-\dxx, \dy) --++ (2*\dxx,0);
				\node[square] at (\dx, 0) {};
				\node[square] at (-\dx, 0) {};
				\node[square] at (\dx, \dy) {};
				\node[square] at (-\dx, \dy) {};
				\node[] at (\dx,-0.25) {$x$};
				\node[] at (-\dx,-0.25) {$gh[x]$};
				\node[] at (\dx,\dy+0.35) {$y$};
				\node[] at (-\dx,\dy+0.35) {$gh[y]$};
				\node[] at (\dx+0.25,\dy+0.15) {$*$};
				\node[] at (-\dx+0.25,\dy+0.15) {$*$};
			\end{tikzpicture}
			=0
		\end{equation}
		and conclude that it vanishes from the local orthogonality assumption. Note that, unless the physical legs of $\{v(g,h,x)\}_{x\in \mathsf{X}}$ span $\mc H^{\otimes 2}$, $\tilde \lambda_{g,h}$ is not fully fixed by Eq.~\eqref{eq:withoutL}, which can have infinitely many solutions.
	\end{proof}
	In terms of these modified fusion operators, we then define modified Gauss laws,
	\begin{equation}
		\tilde{\mc G}_g = \sum_{a,b\in G}{\bigl(\tilde\lambda^R_{ag^{-1},gb}\bigr)^\dagger\tilde\lambda^R_{a,b}\otimes\ketbra{ag^{-1},gb}{a,b}}.
	\end{equation}
	It can be seen that these still form a representation of $G$, which now leaves the gauged MPS \eqref{eq:gauged_tensor_dw} invariant \footnote{Alternatively, we could have directly arrived at modified Gauss generators by precomposing $\mc G$ with a unitary $Q_{g}$ that satisfies 
		\begin{equation}
			Q'_{g}v(h_1, h_2, x)=\dfrac{L^x_{h_1,h_2}}{L^x_{h_1g, g^{-1}h_2}}v(h_1, h_2, x)
		\end{equation}
		which can be seen to exist in the locally orthogonal case by a very similar argument to Lemma \ref{lemma:modif}.}. However, the proof in Appendix \ref{app:associativity} does no longer apply if we modify $\lambda\to\tilde\lambda$, since the new fusion operators may not associate up to a scalar. Hence, Proposition \ref{prop:comm_proj} does not necessarily hold for projectors derived from the modified Gauss laws $\tilde{\mc G}_g$, which may then not commute. Of course, there will be a nonempty subspace where they do, inhabited at least by the gauged MPS \eqref{eq:gauged_tensor_dw}. In this subspace, we have gauged the symmetry as much as we have obtained a locally invariant MPS by defect promotion, keeping the same bond dimension, such that projection of the gauge degrees of freedom onto $\ket{e}$ recovers the original invariant state and the Gauss law projectors commute. Since we have no guarantee that this subspace contains any other states than our gauged MPS, we call this procedure \textit{state-level} gauging.
	
	\subsection{Examples}
	\label{sec:examples_statelevel}
	\subsubsection{$\Z_4\times\Z_2$}
	\label{example:z4z2}
	Let us extend the discussion about onsite symmetries from Subsection~\ref{subsec:onsite} into the noninjective case. We make the same choice \eqref{eq:XYonsite} for the $X,Y$ tensors so that \eqref{eq:ULonsite} remains true. As in the injective case, the onsite symmetry action with $u_g$ can be moved to the virtual level. Individual blocks transform as
	\begin{equation}
		\begin{tikzpicture}
			\def\dx{0.5};
			\def\dy{0.5};
			\draw[thick] (-\dx,0) --++ (2*\dx,0) {};
			\draw[thick] (0,0)--++(0,1.5*\dy) {};  
			\node[tensor] at (0,0) {};
			\node[mpo] at (0,\dy) {};
			\node[] at (0,-0.5*\dy) {$x$};
			\node[] at (0.65*\dx, 1.15*\dy) {$u_g$};
			\node[] at (1.5*\dx,0) {$=$};
			\draw[thick] (2*\dx,0)--++(3*\dx,0) {};
			\draw[thick] (3.5*\dx,0)--++(0,1.5*\dy) {};
			\node[tensor] at (3.5*\dx,0) {};
			\node[square, fill=blue] at (2.5*\dx,0) {};
			\node[square, fill=blue] at (4.5*\dx,0) {};
			\node[] at (3.5*\dx,-0.5*\dy) {$gx$};
			\node[] at (2.5*\dx,-0.75*\dy) {$S_{g,x}^\dagger$};
			\node[] at (4.75*\dx,-0.75*\dy) {$S_{g,x}$};
			\node[] at (5.5*\dx,0) {$,$};
		\end{tikzpicture}
	\end{equation}
	with unitaries $S_{g,x}$ defined up to phases. These now satisfy
	\begin{equation}
		S_{g,hx}S_{h,x}=L^x_{g,h} S_{gh,x},
	\end{equation}
	for scalars $L^x_{g,h}$, which we have already identified as the $L$-symbols, provided we make the choice analogous to \eqref{eq:onsitechoice} for the action tensors,
	\begin{equation}
		A^\Gamma_{g,x}=S_{g,x}, \qquad A^{\ammaG}_{g,x}=S^\dagger_{g,x},
	\end{equation}
	which results in analogous defect tensors,
	\begin{equation}
		\begin{tikzpicture}
			\def\dx{0.5};
			\def\dy{0.5};
			\draw[thick] (-\dx,0)--++(2*\dx,0) {};
			\draw[thick] (0,0)--++(0,\dy);
			\node[tensor] at (0,0) {};
			\node[square, fill=blue] at (-0.65*\dx,0) {};
			\node[] at (0,-0.5*\dy) {$x$};
			\node[] at (-0.75*\dx,-0.65*\dy) {$S_{g,x}$};
			\node[] at (1.25*\dx,0) {$.$};
		\end{tikzpicture}
	\end{equation}
	The fusion operators are again $\lambda^R_{g,h}=u_g\otimes\I$.
	
	We can already find a minimal example of block dependence within this class. It is known that the SPT phases under this group symmetry are classified by a pair $(H, \psi)$, where $H\leq G$ is the unbroken symmetry group and $\psi\in H^2(H, U(1))$ is a 2-cocycle on $H$ \cite{Schuch11, GarreLootensMolnar}. We prove in Appendix \ref{app:block_indep} that block independence is equivalent to the extendability of $\psi$, that is, the existence of a 2-cocycle on $G$, $\Psi\in H^2(G, U(1))$ such that $\Psi$ restricts to $\psi$ on $H$.
	
	We therefore take $G=\Z_4\times \Z_2$, with its subgroup $H=\Z_2\times \Z_2$ and $\psi\in H^2(H,U(1))\cong\Z_2$ being the non-trivial $2$-cocycle. It can be seen that $\psi$ cannot be extended to a $2$-cocycle on the whole group $G$ \cite{BlanikGarreSchuch}. Thus, a set of $L$-symbols labeled by $(\Z_2\times \Z_2, \psi)$ cannot be block-independent. We present one such example in the following.
	
	The (unnormalized) GHZ state $\ket{00\ldots0}+\ket{11\ldots1}$ can be written as a bond dimension two MPS using the copy tensor from \eqref{eq:deltas}
	\begin{equation}
		\begin{tikzpicture}
			\def\dx{0.5};
			\def\dy{0.5};
			\draw[thick] (-0.5*\dx,0)--++(3.0*\dx,0) {};
			\draw[thick] (0,0)--++(0,\dy);
			\draw[thick] (\dx,0)--++(0,\dy);
			\draw[thick] (2*\dx,0)--++(0,\dy);
			\node[] at (-0.65*\dx, 0.2) {$\ldots$};
			\node[] at (2.65*\dx, 0.2) {$\ldots$};
		\end{tikzpicture}.
	\end{equation}
	The cluster state, which is the ground state of the stabilizer Hamiltonian
	\begin{equation}
		H = -\sum_j{Z_jX_{j+1}Z_{j+2}},
	\end{equation}
	can also be written as a bond dimension two MPS, in terms of the copy tensor and the Hadamard gate $H$ from \eqref{eq:deltas}:
	\begin{equation}
		\begin{tikzpicture}
			\def\dx{0.5};
			\def\dy{0.5};
			\draw[thick] (-0.5*\dx,0)--++(3.5*\dx,0) {};
			\draw[thick] (0,0)--++(0,\dy);
			\draw[thick] (\dx,0)--++(0,\dy);
			\draw[thick] (2*\dx,0)--++(0,\dy);
			\node[tensor, fill=blue] at (\dx/2,0) {};
			\node[tensor, fill=blue] at (3*\dx/2,0) {};
			\node[tensor, fill=blue] at (5*\dx/2,0) {};
			\node[] at (-0.65*\dx, 0.2) {$\ldots$};
			\node[] at (3.15*\dx, 0.2) {$\ldots$};
		\end{tikzpicture}.
	\end{equation}
	After blocking the cluster state once (merging two sites into one), the resulting state, symmetric with respect to the $\Z_2\times\Z_2$ representation generated by the $X$ operators on each of the qubits on a site, belongs to the non-trivial SPT phase labeled by $\psi$.
	
	Let us take this blocked cluster state and tensor the GHZ state with it,
	\begin{equation}
		\begin{tikzpicture}[baseline]
			\def\dx{0.5};
			\def\dy{0.5};
			\draw[thick] (-0.5*\dx,0)--++(\dx,0) {};
			\draw[thick] (0,0)--++(0,\dy);
		\end{tikzpicture}
		\otimes
		\begin{tikzpicture}[baseline]
			\def\dx{0.5};
			\def\dy{0.5};
			\draw[thick] (-0.5*\dx,0)--++(2.5*\dx,0) {};
			\draw[thick] (0,0)--++(0,\dy);
			\draw[thick] (\dx,0)--++(0,\dy);
			\node[tensor, fill=blue] at (\dx/2,0) {};
			\node[tensor, fill=blue] at (3*\dx/2,0) {};
		\end{tikzpicture}
		.
		\label{eq:z4z2MPS}
	\end{equation}
	The resulting state has bond dimension four and physical dimension eight, where for each site we denote the qubit of the GHZ state as $1$, and those of the cluster state by $2$ and $3$. The MPS tensor is the direct sum of two injective blocks labeled by the first qubit
	\begin{equation}
		\begin{tikzpicture}[baseline]
			\def\dy{0.4};
			\node[] at (0, 0) {$\ket{0}\bra{0}$};
			\node[] at (0, \dy) {$\ket{0}$};
		\end{tikzpicture}\otimes
		\begin{tikzpicture}[baseline]
			\def\dx{0.5};
			\def\dy{0.5};
			\draw[thick] (-0.5*\dx,0)--++(2.5*\dx,0) {};
			\draw[thick] (0,0)--++(0,\dy);
			\draw[thick] (\dx,0)--++(0,\dy);
			\node[tensor, fill=blue] at (\dx/2,0) {};
			\node[tensor, fill=blue] at (3*\dx/2,0) {};
		\end{tikzpicture}
		,\qquad
		\begin{tikzpicture}[baseline]
			\def\dy{0.4};
			\node[] at (0, 0) {$\ket{1}\bra{1}$};
			\node[] at (0, \dy) {$\ket{1}$};
		\end{tikzpicture}\otimes
		\begin{tikzpicture}[baseline]
			\def\dx{0.5};
			\def\dy{0.5};
			\draw[thick] (-0.5*\dx,0)--++(2.5*\dx,0) {};
			\draw[thick] (0,0)--++(0,\dy);
			\draw[thick] (\dx,0)--++(0,\dy);
			\node[tensor, fill=blue] at (\dx/2,0) {};
			\node[tensor, fill=blue] at (3*\dx/2,0) {};
		\end{tikzpicture}
		.
	\end{equation}
	Note that these are locally orthogonal because the first qubit supports are orthogonal. Although it has more symmetries, we look at it as an invariant MPS under the symmetry group $G=\Z_4\times \Z_2$ generated by $U_{(1,0)}=(X_1\otimes X_2)\cdot \mathrm{CNOT}_{1\rightarrow 2}$ and $U_{(0,1)}=X_3$. The action of an element $(a,b)\in \Z_4\times\Z_2$ ($a=0,\ldots,3$, $b=0,1$) can be expressed on the virtual level by the operators $A^\Gamma_{(a, b),x}$ given by
	\begin{equation*}
		\begin{tabular}{c|c c c c}
			$(a,b)$ & $(0,0)$ & $(2,0)$ & $(0,1)$ & $(2,1)$ \\\hline
			$x=0$ & $\I$ & $X$ & $Z$ & $Y$  \\
			$x=1$ & $\I$ & $X$ & $Z$ & $Y$  \\
		\end{tabular}
	\end{equation*}
	
	\begin{equation*}
		\begin{tabular}{c|c c c c}
			$(a,b)$ & $(1,0)$ & $(3,0)$ & $(1,1)$ & $(3,1)$ \\\hline
			$x=0$ & $X$ & $\I$ & $Y$ & $Z$ \\
			$x=1$ & $\I$ & $X$ & $Z$ & $Y$ \\
		\end{tabular}
	\end{equation*}
	where we have used the scalar freedom following the conventions of this paper, to get well-defined defect tensors (see Proposition~\ref{prop:def_tens_anomalous}: because the symmetry is onsite, it is nonanomalous and we do not need to worry about the $\sigma_g$ signs). When restricting to the unbroken symmetry subgroup (isomorphic to $\Z_2\times\Z_2$ and generated by $U_{(2,0)} = U^2_{(1,0)}=X_2$ and $U_{(0,1)}=X_3$) we recover the blocked cluster state, which represents the symmetry projectively on the virtual space with nontrivial cocycle $\psi$. Thus, this example is labeled by $(H, \psi)$, and from the argument above it violates block independence. It can be checked that, indeed, the local symmetries $\mc G_g$ defined from $\lambda^R_{g,h}=u_g\otimes\I$ do not leave \eqref{eq:z4z2MPS} invariant. Since the two injective blocks are locally orthogonal, the state-level gauging can be performed. We pick an easy modification of the fusion operators, $\tilde\lambda_{g,h}=u_g\otimes Q_{g,h}$ where $Q_{g,h}$ is just added a phase based on the value of the first qubit of the second site,
	\begin{equation}
		Q_{g,h}=\left(\begin{array}{cc}
			\dfrac{1}{L^{0}_{g,h}} & 0 \\
			0 & \dfrac{1}{L^{1}_{g,h}}
		\end{array}\right)\otimes \I\otimes \I
	\end{equation}
	Because the $L$-symbols take values in $\{\pm1, \pm i\}$, the $Q$ operators are $Z$ gates or phase gates, up to global phases. Then it is straightforward to build the local symmetry operators $\tilde{\mc G_g}$ and check that they leave \eqref{eq:z4z2MPS} invariant. We also constructed the associated Gauss law projectors
	\begin{equation}
		\tilde P_j=\dfrac{1}{|G|}\sum_{g\in G}{\tilde{\mcG}_g^{j,j+1}},
		\label{eq:gauge_projs_modif}
	\end{equation}
	and found that neither the $\tilde\lambda$ associate to a scalar (Eq.~\eqref{eq:assoc}) nor the $\tilde P_j$ commute for neighboring sites. 
	
	\subsubsection{The CZX model}
	
	We consider here the anomalous $\Z_2=\langle g\ |\ g^2=e\rangle$ representation specified by the following MPU tensor (blocked once to ensure injectivity) 
	\begin{equation}
		\label{eq:MPU_CZX}
		\begin{tikzpicture}
			\draw[red,thick] (-0.3,0.0)--++(0.6,0.0);
			\draw[thick] (0.0,-0.3)--++(0.0,0.6);
			\node[mpo] at (0.0,0.0) {};
			
			\node[] at (0.8,0.0) {$\equiv$};
			
			\draw[red,thick] (1.2,0.0) --++ (0.3, 0.0);
			\draw[red,thick] (2.5,0.0) --++ (0.5, 0.0);
			
			\draw[thick] (1.5,0.0) --++ (0.0, -0.6);
			\draw[thick] (1.5,0.0) --++ (0.0, 0.8);
			
			\draw[thick] (2.5,0.0) --++ (0.0, -0.6);
			\draw[thick] (2.5,0.0) --++ (0.0, 0.8);
			
			\draw[red, thick] (1.5, 0.25) --++ (1.0, 0.0);
			\node[tensor, fill=blue] at (2.0, 0.25) {}; 
			
			\node[square, fill= white] at (1.5, -0.25) {};
			
			\node[square, fill= white] at (2.5, -0.25) {};
			
			\node[square, fill= red] at (1.5, 0.5) {};
			
			\node[square, fill= red] at (2.5, 0.5) {};
			
			\node[] at (1.25, -0.45) {$X$};
			\node[] at (2.75, -0.45) {$X$};
			
			\node[] at (1.25, 0.55) {$Z$};
			\node[] at (2.75, 0.55) {$Z$};
			
			\node[tensor, fill=blue] at (2.75, 0.0) {}; 
			
			\node[] at (3.2,0.0) {$,$};       
		\end{tikzpicture}
	\end{equation}
	where the CZ gate is expressed as in \eqref{eq:CZ}. Its adjoint reads
	\begin{equation}
		\begin{tikzpicture}
			\draw[thick,red] (-0.7,0.0)--++(1.4,0.0);
			\draw[thick] (0.0,-0.3)--++(0.0,0.6);
			\node[mpo] at (0.0,0.0) {};   
			\node[square, fill=yellow] at (-0.4,0.0) {};
			\node[square, fill=yellow] at (0.4,0.0) {};
			\node[] at (-0.4, -0.25) {$Y$};
			\node[] at (0.4,-0.25) {$Y$};
			\node[] at (0.9,0.0) {$,$};
		\end{tikzpicture}
	\end{equation}
	so that $T_g = Y$ and $\sigma_g=-1$ are nontrivial. We pick decompositions satisfying \eqref{eq:modif_XY},
	\begin{subequations}
		\begin{equation}
			\begin{tikzpicture}
				\draw[thick, red] (-0.5,0.0) --++ (0.5,0.0);
				\draw[thick] (0.0,0.0) --++ (0.0,0.5);
				\draw[ultra thick] (0.0,-0.5) --++ (0.0,0.5);
				\node[tensor] at (0.0,0.0) {};
				
				\node[] at (0.5,0.0) {$=$};
				
				\draw[thick, red] (1.0,0.0) --++ (0.5,0.0);
				\draw[thick] (1.5,-0.5) --++ (0.0, 1.2);
				\draw[thick] (2.0,-0.5) --++ (0.0, 1.2);
				\node[square, fill=red] at (1.5,0.5) {};
				\node[square, fill=red] at (2.0,0.5) {};
				\node[] at (1.25,0.55) {$Z$};
				\node[] at (2.25,0.55) {$Z$};
				
				\node[] at (2.5,0.0) {$,$};
				
				\def\Sc{1.4};
				
				\draw[thick, red] (0.5+\Sc*\Shift,0.0) --++ (-0.5,0.0);
				\draw[ultra thick] (0.0+\Sc*\Shift,0.0) --++ (0.0,0.5);
				\draw[thick] (0.0+\Sc*\Shift,-0.5) --++ (0.0,0.5);
				\node[tensor] at (0.0+\Sc*\Shift,0.0) {};
				
				\node[] at (1.0+\Sc*\Shift,0.0) {$=$};
				
				\draw[thick] (1.5+\Sc*\Shift,-0.5) --++ (0.0, 1.2);
				\draw[thick] (2.0+\Sc*\Shift,-0.5) --++ (0.0, 1.2);
				\node[square, fill=white] at (1.5+\Sc*\Shift,-0.25) {};
				\node[square, fill=white] at (2.0+\Sc*\Shift,-0.25) {};
				\node[] at (1.25+\Sc*\Shift,-0.3) {$X$};
				\node[] at (2.25+\Sc*\Shift,-0.3) {$X$};
				
				\draw[thick,red] (1.5+\Sc*\Shift,0.5) --++ (0.5, 0.0);
				
				\draw[thick,red] (2.0+\Sc*\Shift,0.25) --++ (0.5, 0.0);
				
				\node[tensor, fill=blue] at (1.75+\Sc*\Shift,0.5) {};
				\node[tensor, fill=blue] at (2.25+\Sc*\Shift,0.25) {};
				\node[] at (2.75+\Sc*\Shift,0.0) {$,$};
			\end{tikzpicture}
		\end{equation}
		\begin{equation}
			\begin{tikzpicture}
				
				\tikzset{decoration={snake,amplitude=.4mm,segment length=2mm, post length=0mm,pre length=0mm}}
				
				\draw[thick, red] (-0.5,0.0) --++ (0.5,0.0);
				\draw[thick, decorate] (0.0,0.0) --++ (0.0,0.5);
				\draw[thick] (0.0,-0.5) --++ (0.0,0.5);
				\node[tensor, fill=white] at (0.0,0.0) {};
				
				\node[] at (0.5,0.0) {$=$};
				
				\draw[thick, red] (0.75,0.0) --++ (0.75,0.0);
				\draw[thick] (1.5,-0.7) --++ (0.0, 1.2);
				\draw[thick] (2.0,-0.7) --++ (0.0, 1.2);
				\node[square, fill=red] at (1.5,-0.5) {};
				\node[square, fill=red] at (2.0,-0.5) {};
				\node[] at (1.25,-0.55) {$Z$};
				\node[] at (2.25,-0.55) {$Z$};
				
				\node[square, fill=yellow] at (1.1,0.0) {};
				
				\node[] at (1.1,0.25) {$Y$};
				
				\node[] at (2.25,0.0) {$,$};
				
				\def\Sc{1.4};
				
				\draw[thick, red] (0.5+\Sc*\Shift,0.0) --++ (-0.5,0.0);
				\draw[thick] (0.0+\Sc*\Shift,0.0) --++ (0.0,0.5);
				\draw[thick,decorate] (0.0+\Sc*\Shift,-0.5) --++ (0.0,0.5);
				\node[tensor, fill=white] at (0.0+\Sc*\Shift,0.0) {};
				
				\node[] at (0.75+\Sc*\Shift,0.0) {$=$};
				
				\draw[thick] (1.5+\Sc*\Shift,-0.7) --++ (0.0, 1.2);
				\draw[thick] (2.0+\Sc*\Shift,-0.7) --++ (0.0, 1.2);
				\node[square, fill=white] at (1.5+\Sc*\Shift,0.25) {};
				\node[square, fill=white] at (2.0+\Sc*\Shift,0.25) {};
				\node[] at (1.25+\Sc*\Shift,0.3) {$X$};
				\node[] at (2.25+\Sc*\Shift,0.3) {$X$};
				
				\draw[thick,red] (1.5+\Sc*\Shift,-0.25) --++ (0.5, 0.0);
				
				\draw[thick,red] (2.0+\Sc*\Shift,-0.45) --++ (0.8, 0.0);
				
				\node[tensor, fill=blue] at (1.75+\Sc*\Shift,-0.25) {};
				\node[tensor, fill=blue] at (2.25+\Sc*\Shift,-0.45) {};
				\node[] at (3.0+\Sc*\Shift,0.0) {$.$};
				
				\node[square, fill=yellow] at (2.6+\Sc*\Shift,-0.45) {};
				\node[] at (2.6+\Sc*\Shift, -0.2) {$Y$};
			\end{tikzpicture}
		\end{equation}
	\end{subequations}
	This MPU group leaves the GHZ MPS invariant. After blocking once, the two injective blocks are given by
	\begin{equation}
		\begin{tikzpicture}
			\draw[thick] (0.0,0.0) --++ (0.0,0.4) {};
			\node[tensor, fill=white] at (0.0,0.0) {};
			\node[] at (0.8,0.0) {$= \ \ket{00},$};
			\draw[thick] (1.8,0.0) --++ (0.0,0.4) {};
			\node[tensor] at (1.8,0.0) {};
			\node[] at (2.6,0.0) {$= \ \ket{11},$};
		\end{tikzpicture}
	\end{equation}
	and the only nontrivial fusion tensors are given by \eqref{eq:Fgg},
	\begin{equation}
		\begin{tikzpicture}
			\node[] at (0.0,0.0) {$F_{g,g}^<=$}; 
			\draw[thick,red] (0.5,-0.3) --++ (0.0,0.6);
			\draw[thick,red] (0.5,-0.3) --++ (0.4,0.0);
			\draw[thick,red] (0.5,0.3) --++ (0.4,0.0);
			\node[square, fill=yellow] at (0.7, 0.3) {};
			\node[] at (0.7,0) {$Y$};
			\node[] at (1.2,0.0) {$,$};
			\node[] at (2.25,0.0) {$F_{g,g}^>=\frac{1}{2}$};
			\draw[thick,red] (3.25,-0.3) --++ (0.0,0.6);
			\draw[thick,red] (3.25,-0.3) --++ (-0.4,0.0);
			\draw[thick,red] (3.25,0.3) --++ (-0.4,0.0);
			\node[square, fill=yellow] at (3.05,0.3) {};
			\node[] at (3.05,0) {$Y$};
			\node[] at (3.5,0.0) {$,$};
		\end{tikzpicture}
	\end{equation}
	and we choose action tensors,
	\begin{align}
		&A^\Gamma_{g,0}=-i\bra{1}, &A^\Gamma_{g,1}=\bra{0},\\
		&A^\ammaG_{g,0}=-\frac{i}{\sqrt 2}\ket{-}, &A^\ammaG_{g,1}=\frac{\ket{+}}{\sqrt 2}.
	\end{align}
	Therefore, from the very definition, we see that 
	\begin{equation}
		\omega(h_1,h_2,h_3)=(-1)^{\delta_{h_1,g}\delta_{h_2,g}\delta_{h_3,g}},
	\end{equation}
	and the only nontrivial $L$-symbol is $L^{0}_{g,g}=-1$ (because of the anomaly, we already know that block independence does not hold, hence no gauge transformation will make that last $L$-symbol equal to 1). The truncated symmetry operators
	\begin{equation}
		\begin{tikzpicture}
			\node[] at (0.0,0.0) {$U_g^L \equiv $};
			\draw[thick,red] (0.6,0.0) --++ (0.8,0.0);
			\draw[thick] (0.6,0.0) --++ (0.0,-0.4);
			\draw[thick] (1.0,0.0) --++ (0.0,-0.4);
			\draw[thick] (1.4,0.0) --++ (0.0,-0.4);
			\draw[thick, ultra thick] (0.6,0.0) --++ (0.0,0.4);
			\draw[thick] (1.0,0.0) --++ (0.0,0.4);
			\draw[thick, decorate] (1.4,0.0) --++ (0.0,0.4);
			\node[tensor] at (0.6,0.0) {}; 
			\node[tensor, fill=white] at (1.4,0.0) {};
			\node[mpo] at (1.0,0.0) {};
			\node[] at (1.9,0.0) {$\equiv\ i$};
			\def\Sc{1.1};
			\draw[thick] (0.0+\Sc*\Shift,-0.7) --++ (0.0,1.2);
			\draw[thick] (0.4+\Sc*\Shift,-0.7) --++ (0.0,1.2);
			\draw[thick] (0.8+\Sc*\Shift,-0.7) --++ (0.0,1.4);
			\draw[thick] (1.2+\Sc*\Shift,-0.7) --++ (0.0,1.4);
			\draw[thick] (1.6+\Sc*\Shift,-0.7) --++ (0.0,1.4);
			\draw[thick] (2.0+\Sc*\Shift,-0.5) --++ (0.0,1.2);
			
			\draw[thick, red] (0.0+\Sc*\Shift,0.2) --++ (0.4,0.0);
			\draw[thick, red] (0.8+\Sc*\Shift,0.2) --++ (0.4,0.0);
			\draw[thick, red] (0.4+\Sc*\Shift,-0.2) --++ (0.4,0.0);
			\draw[thick, red] (1.2+\Sc*\Shift,-0.2) --++ (0.4,0.0);
			\node[square, fill=white] at (0.0+\Sc*\Shift,-0.5) {};
			\node[square, fill=white] at (0.4+\Sc*\Shift,-0.5) {};
			\node[square, fill=white] at (0.8+\Sc*\Shift,-0.5) {};
			\node[square, fill=white] at (1.2+\Sc*\Shift,-0.5) {};
			\node[square, fill=white] at (1.6+\Sc*\Shift,-0.5) {};
			\node[square, fill=red] at (0.8+\Sc*\Shift,0.5) {};
			\node[square, fill=red] at (1.2+\Sc*\Shift,0.5) {};
			\node[square, fill=white] at (1.6+\Sc*\Shift,0.5) {};
			\node[square, fill=red] at (2.0+\Sc*\Shift,0.5) {};
			
			\node[tensor, fill =blue] at (0.2+\Sc*\Shift, 0.2) {};
			\node[tensor, fill =blue] at (1.0+\Sc*\Shift, 0.2) {};
			\node[tensor, fill =blue] at (0.6+\Sc*\Shift, -0.2) {};
			\node[tensor, fill =blue] at (1.4+\Sc*\Shift, -0.2) {};
			
			\node[] at (-0.15+\Sc*\Shift,-0.75) {$X$};
			\node[] at (0.25+\Sc*\Shift,-0.75) {$X$};
			\node[] at (0.65+\Sc*\Shift,-0.75) {$X$};
			\node[] at (1.05+\Sc*\Shift,-0.75) {$X$};
			\node[] at (1.45+\Sc*\Shift,-0.75) {$X$};
			
			\node[] at (1.75+\Sc*\Shift,0.75) {$X$};
			\node[] at (2.15+\Sc*\Shift,0.75) {$Z$};
			\node[] at (1.35+\Sc*\Shift,0.75) {$Z$};
			\node[] at (0.95+\Sc*\Shift,0.75) {$Z$};
			
			\def\co{2.2};
			
			\node[] at (4.75,0.0) {$=$};
			
			\draw[thick] (0.0+\co*\Sc*\Shift,-0.7) --++ (0.0,1.2);
			\draw[thick] (0.4+\co*\Sc*\Shift,-0.7) --++ (0.0,1.2);
			\draw[thick] (0.8+\co*\Sc*\Shift,-0.7) --++ (0.0,1.4);
			\draw[thick] (1.2+\co*\Sc*\Shift,-0.7) --++ (0.0,1.4);
			\draw[thick] (1.6+\co*\Sc*\Shift,-0.7) --++ (0.0,1.4);
			\draw[thick] (2.0+\co*\Sc*\Shift,-0.7) --++ (0.0,1.4);
			
			\draw[thick, red] (0.0+\co*\Sc*\Shift,0.2) --++ (0.4,0.0);
			\draw[thick, red] (0.8+\co*\Sc*\Shift,0.2) --++ (0.4,0.0);
			\draw[thick, red] (0.4+\co*\Sc*\Shift,-0.2) --++ (0.4,0.0);
			\draw[thick, red] (1.2+\co*\Sc*\Shift,-0.2) --++ (0.4,0.0);
			\node[square, fill=white] at (0.0+\co*\Sc*\Shift,-0.5) {};
			\node[square, fill=white] at (0.4+\co*\Sc*\Shift,-0.5) {};
			\node[square, fill=white] at (0.8+\co*\Sc*\Shift,-0.5) {};
			\node[square, fill=white] at (1.2+\co*\Sc*\Shift,-0.5) {};
			\node[square, fill=white] at (1.6+\co*\Sc*\Shift,-0.5) {};
			\node[square, fill=yellow] at (2.0+\co*\Sc*\Shift,-0.5) {};
			\node[square, fill=red] at (0.8+\co*\Sc*\Shift,0.5) {};
			\node[square, fill=red] at (1.2+\co*\Sc*\Shift,0.5) {};
			\node[square, fill=white] at (1.6+\co*\Sc*\Shift,0.5) {};
			\node[square, fill=white] at (2.0+\co*\Sc*\Shift,0.5) {};
			
			\node[tensor, fill =blue] at (0.2+\co*\Sc*\Shift, 0.2) {};
			\node[tensor, fill =blue] at (1.0+\co*\Sc*\Shift, 0.2) {};
			\node[tensor, fill =blue] at (0.6+\co*\Sc*\Shift, -0.2) {};
			\node[tensor, fill =blue] at (1.4+\co*\Sc*\Shift, -0.2) {};
			
			\node[] at (-0.15+\co*\Sc*\Shift,-0.75) {$X$};
			\node[] at (0.25+\co*\Sc*\Shift,-0.75) {$X$};
			\node[] at (0.65+\co*\Sc*\Shift,-0.75) {$X$};
			\node[] at (1.05+\co*\Sc*\Shift,-0.75) {$X$};
			\node[] at (1.45+\co*\Sc*\Shift,-0.75) {$X$};
			\node[] at (1.85+\co*\Sc*\Shift,-0.75) {$Y$};
			
			\node[] at (1.75+\co*\Sc*\Shift,0.75) {$X$};
			\node[] at (2.15+\co*\Sc*\Shift,0.75) {$X$};
			\node[] at (1.35+\co*\Sc*\Shift,0.75) {$Z$};
			\node[] at (0.95+\co*\Sc*\Shift,0.75) {$Z$};

		\end{tikzpicture}
	\end{equation}
	satisfy $U_g^L U_g^{\infty}=- U_g^{\infty-L}$. Next, we concentrate on the corresponding defects:
	\begin{subequations}
		\begin{equation}
			\begin{tikzpicture}
				\tikzset{decoration={snake,amplitude=.4mm,segment length=2mm, post length=0mm,pre length=0mm}}
				\node[] at (0.0,-0.1) {$g[1]\equiv 0|1 \ = $};
				
				\draw[thick] (1.5,-0.5) --++ (0.0,0.5);
				\draw[decorate, thick] (1.5,0.0) --++ (0.0,0.5);
				
				\draw[thick,red] (1.5,0.0) --++ (-0.5,0.0);
				\draw[thick,red] (1.0,0.0) --++ (0.0,-0.55);
				
				\node[tensor] at (1.5,-0.5) {};
				\node[tensor, fill=white] at (1.5,0.0) {};
				
				\node[] at (2.1,-0.1) {$=\ -i $};
				\node[tensor] at (2.6, -0.4) {};
				\draw[thick] (2.6,-0.4) --++ (0.0, 0.6);
				
				\node[] at (3.1,-0.1) {$=\ - $};
				\draw[ultra thick] (3.6,0.0) --++ (0.0,0.5);
				\draw[thick,red] (3.6,0.0) --++ (0.5,0.0);
				\draw[thick] (3.6,0.0) --++ (0.0,-0.5);
				\draw[thick,red] (4.1,0,0) --++ (0.0,-0.55);
				\node[tensor, fill=white] at (3.6, -0.5) {};
				\node[tensor] at (3.6, 0.0) {};
				
				\node[] at (4.3, 0.0) {$,$};
			\end{tikzpicture}
		\end{equation}
		\begin{equation}
			\begin{tikzpicture}
				\tikzset{decoration={snake,amplitude=.4mm,segment length=2mm, post length=0mm,pre length=0mm}}
				\node[] at (0.0,-0.1) {$g[0]\equiv 1|0 \ = $};
				
				\draw[thick] (1.5,-0.5) --++ (0.0,0.5);
				\draw[decorate, thick] (1.5,0.0) --++ (0.0,0.5);
				
				\draw[thick,red] (1.5,0.0) --++ (-0.5,0.0);
				\draw[thick,red] (1.0,0.0) --++ (0.0,-0.55);
				
				\node[tensor, fill=white] at (1.5,0.0) {};
				\node[tensor, fill=white] at (1.5,-0.5) {};
				
				\node[] at (2.1,-0.1) {$=$};
				
				\draw[thick] (2.6,-0.4) --++ (0.0, 0.6);
				\node[tensor, fill=white] at (2.6, -0.4) {};
				
				\node[] at (3.1,-0.1) {$=$};
				\draw[ultra thick] (3.6,0.0) --++ (0.0,0.5);
				\draw[thick,red] (3.6,0.0) --++ (0.5,0.0);
				\draw[thick] (3.6,0.0) --++ (0.0,-0.5);
				\draw[thick,red] (4.1,0,0) --++ (0.0,-0.55);
				\node[tensor] at (3.6, -0.5) {};
				\node[tensor] at (3.6, 0.0) {};
				
				\node[] at (4.3, 0.0) {$,$};
			\end{tikzpicture}
		\end{equation}
	\end{subequations}
	on which the movement operators, 
	\begin{equation}
		\begin{tikzpicture}
			\node[] at (0.0,0.0) {$w_R=$};
			\draw[thick] (0.5,-0.6) --++ (0.0,1.6);
			\draw[thick] (1.0,-0.6) --++ (0.0,1.6);
			\draw[thick] (1.5,-0.8) --++ (0.0,1.4);
			\draw[thick] (2.0,-0.8) --++ (0.0,1.4);
			\node[square, fill=white] at (0.5,0.8) {};
			\node[square, fill=white] at (1.0,0.8) {};
			\node[square, fill=red] at (1.5,-0.6) {};
			\node[square, fill=red] at (2.0,-0.6) {};
			\draw[thick,red] (0.5,0.3) --++ (0.5,0.0);
			\draw[thick,red] (1.0,-0.2) --++ (0.5,0.0);
			\node[tensor, fill=blue] at (0.75, 0.3) {};
			\node[tensor, fill=blue] at (1.25, -0.2) {};
			\node[] at (0.65,1.05) {$X$};
			\node[] at (1.15,1.05) {$X$};
			\node[] at (1.35,-0.85) {$Z$};
			\node[] at (1.85,-0.85) {$Z$};
			\node[] at (2.2,0.0) {$,$};
			
			\def\Sc{1.8};
			
			\node[] at (0.0+\Sc*\Shift,0.0) {$w_L=w_R^\dagger$};
		\end{tikzpicture}
	\end{equation}
	satisfy
	\begin{subequations}
		\begin{equation}
			\begin{tikzpicture}
				\draw[thick] (-0.2,-0.65) --++ (0,0.9);
				\draw[thick] (0.2,-0.65) --++ (0,0.9);
				
				\node[rectangle, fill=citrine] at (0,-0.2) {$w_R$};
				\node[tensor, fill=white] at (-0.2,-0.65) {};
				\node[tensor, fill=white] at (0.2,-0.65) {};
				\node at (-0.2,-0.9) {$1|0$};
				\node at (0.2,-0.9) {$0$};
				
				\node[] at (0.6, -0.65) {$=$};
				
				\def\Sc{0.6};
				
				\draw[thick] (-0.2+\Sc*\Shift,-0.65) --++ (0,0.4);
				\draw[thick] (0.2+\Sc*\Shift,-0.65) --++ (0,0.4);
				\node[tensor] at (-0.2+\Sc*\Shift,-0.65) {};
				\node[tensor, fill=white] at (0.2+\Sc*\Shift,-0.65) {};
				\node at (-0.2+\Sc*\Shift,-0.9) {$1$};
				\node at (0.2+\Sc*\Shift,-0.9) {$1|0$};
				
				\node[] at (0.4+\Sc*\Shift,-0.65) {$,$};
				
				\def\co{2.5};
				\draw[thick] (-0.2+\co*+\Sc*\Shift,-0.65) --++ (0,0.9);
				\draw[thick] (0.2+\co*+\Sc*\Shift,-0.65) --++ (0,0.9);

				\node[] at (-0.6+\co*\Sc*\Shift, -0.2) {$-i$};
				\node[rectangle, fill=citrine] at (0+\co*+\Sc*\Shift,-0.2) {$w_R$};
				\node[tensor] at (-0.2+\co*+\Sc*\Shift,-0.65) {};
				\node[tensor] at (0.2+\co*+\Sc*\Shift,-0.65) {};
				\node at (-0.2+\co*+\Sc*\Shift,-0.9) {$0|1$};
				\node at (0.2+\co*+\Sc*\Shift,-0.9) {$1$};
				
				\node[] at (0.6+\co*+\Sc*\Shift, -0.65) {$=$};
				
				\def\coo{1.4};
				
				\draw[thick] (-0.2+\co*\coo*\Sc*\Shift,-0.65) --++ (0,0.4);
				\draw[thick] (0.4+\co*\coo*\Sc*\Shift,-0.65) --++ (0,0.4);
				\node[] at (0.2+\co*\coo*\Sc*\Shift,-0.45) {$-i$};
				\node[tensor, fill=white] at (-0.2+\co*\coo*\Sc*\Shift,-0.65) {};
				\node[tensor] at (0.4+\co*\coo*\Sc*\Shift,-0.65) {};
				\node at (-0.2+\co*\coo*\Sc*\Shift,-0.9) {$0$};
				\node at (0.4+\co*\coo*\Sc*\Shift,-0.9) {$0|1$};
				\node[] at (0.6+\co*\coo*\Sc*\Shift,-0.65) {$,$};
			\end{tikzpicture}
		\end{equation}
		\begin{equation}
			\begin{tikzpicture}
				\draw[thick] (-0.2,-0.65) --++ (0,0.9);
				\draw[thick] (0.2,-0.65) --++ (0,0.9);
				
				\node[rectangle, fill=carrotorange] at (0,-0.2) {$w_L$};
				\node[tensor] at (-0.2,-0.65) {};
				\node[tensor, fill=white] at (0.2,-0.65) {};
				\node at (-0.2,-0.9) {$1$};
				\node at (0.2,-0.9) {$1|0$};
				
				\node[] at (0.6, -0.65) {$=$};
				
				\def\Sc{0.6};
				
				\draw[thick] (-0.2+\Sc*\Shift,-0.65) --++ (0,0.4);
				\draw[thick] (0.2+\Sc*\Shift,-0.65) --++ (0,0.4);
				\node[tensor, fill=white] at (-0.2+\Sc*\Shift,-0.65) {};
				\node[tensor, fill=white] at (0.2+\Sc*\Shift,-0.65) {};
				\node at (-0.2+\Sc*\Shift,-0.9) {$1|0$};
				\node at (0.2+\Sc*\Shift,-0.9) {$0$};
				
				\node[] at (0.4+\Sc*\Shift,-0.65) {$,$};
				
				\def\co{2.5};
				\draw[thick] (-0.2+\co*+\Sc*\Shift,-0.65) --++ (0,0.9);
				\draw[thick] (0.2+\co*+\Sc*\Shift,-0.65) --++ (0,0.9);
				
				\node[rectangle, fill=carrotorange] at (0+\co*+\Sc*\Shift,-0.2) {$w_L$};
				\node[] at (0.03+\co*\Sc*\Shift, -0.575) {$\!\text{-} i$};
				\node[tensor, fill=white] at (-0.2+\co*+\Sc*\Shift,-0.65) {};
				\node[tensor] at (0.2+\co*+\Sc*\Shift,-0.65) {};
				\node at (-0.2+\co*+\Sc*\Shift,-0.9) {$0$};
				\node at (0.2+\co*+\Sc*\Shift,-0.9) {$0|1$};
				
				\node[] at (0.6+\co*+\Sc*\Shift, -0.65) {$=$};
				
				\def\coo{1.4};
				
				\draw[thick] (-0.2+\co*\coo*\Sc*\Shift,-0.65) --++ (0,0.4);
				\draw[thick] (0.2+\co*\coo*\Sc*\Shift,-0.65) --++ (0,0.4);
				\node[] at (-0.4+\co*\coo*\Sc*\Shift,-0.45) {$-i$};
				\node[tensor] at (-0.2+\co*\coo*\Sc*\Shift,-0.65) {};
				\node[tensor] at (0.2+\co*\coo*\Sc*\Shift,-0.65) {};
				\node at (-0.2+\co*\coo*\Sc*\Shift,-0.9) {$0|1$};
				\node at (0.2+\co*\coo*\Sc*\Shift,-0.9) {$1$};
				\node[] at (0.6+\co*\coo*\Sc*\Shift,-0.65) {$.$};
			\end{tikzpicture}
		\end{equation}
	\end{subequations}
	The fusion operators read $\lambda^R_{e,e}=\lambda^R_{e,g}=\I$, $\lambda^R_{g,e}=w_R$ and
	\begin{equation}
		\begin{tikzpicture}
			\node[] at (-0.3,0.0) {$\lambda^R_{g,g}=-i$};
			\draw[thick] (0.5,-0.6) --++ (0.0,1.6);
			\draw[thick] (1.0,-0.6) --++ (0.0,1.6);
			\draw[thick] (1.5,-0.8) --++ (0.0,1.8);
			\draw[thick] (2.0,-0.8) --++ (0.0,1.8);
			\node[square, fill=white] at (0.5,0.8) {};
			\node[square, fill=white] at (1.0,0.8) {};
			\node[square, fill=white] at (1.5,-0.6) {};
			\node[square, fill=white] at (1.5,0.8) {};
			\node[square, fill=red] at (2.0,0.8) {};
			\draw[thick,red] (0.5,0.3) --++ (0.5,0.0);
			\draw[thick,red] (1.0,-0.2) --++ (0.5,0.0);
			\node[tensor, fill=blue] at (0.75, 0.3) {};
			\node[tensor, fill=blue] at (1.25, -0.2) {};
			\node[] at (0.65,1.05) {$X$};
			\node[] at (1.15,1.05) {$X$};
			\node[] at (1.65,1.05) {$X$};
			\node[] at (2.15,1.05) {$Z$};
			\node[] at (1.35,-0.85) {$X$};
			\node[] at (2.2,0.0) {$.$};
		\end{tikzpicture}
	\end{equation}
	The last one satisfies
	\begin{equation}
		\begin{tikzpicture}
			\draw[thick] (-0.2,-0.65) --++ (0,1.0);
			\draw[thick] (0.2,-0.65) --++ (0,1.0);
			
			\node[rectangle, fill=white] at (0,-0.1) {$\lambda^R_{g,g}$};
			\node[tensor] at (-0.2,-0.65) {};
			\node[tensor, fill=white] at (0.2,-0.65) {};
			\node at (-0.25,-0.9) {$0|1$};
			\node at (0.25,-0.9) {$1|0$};
			\node at (-0.6,-0.65) {$-i$};
			
			\node[] at (0.6, -0.65) {$=$};
			
			\def\Sc{0.6};
			
			\draw[thick] (-0.2+\Sc*\Shift,-0.65) --++ (0,0.4);
			\draw[thick] (0.2+\Sc*\Shift,-0.65) --++ (0,0.4);
			\node[tensor, fill=white] at (-0.2+\Sc*\Shift,-0.65) {};
			\node[tensor, fill=white] at (0.2+\Sc*\Shift,-0.65) {};
			\node at (-0.2+\Sc*\Shift,-0.9) {$0$};
			\node at (0.2+\Sc*\Shift,-0.9) {$0$};
			\node[] at (-0.4+\Sc*\Shift,-0.45) {$-$};
			
			\node[] at (0.4+\Sc*\Shift,-0.65) {$,$};
			
			\def\co{2.5};
			\draw[thick] (-0.2+\co*+\Sc*\Shift,-0.65) --++ (0,1.0);
			\draw[thick] (0.2+\co*+\Sc*\Shift,-0.65) --++ (0,1.0);

			\node[] at (0.02+\co*\Sc*\Shift, -0.525) {$\text{-}i$};
			\node[rectangle, fill=white] at (0+\co*+\Sc*\Shift,-0.1) {$\lambda^R_{g,g}$};
			\node[tensor, fill=white] at (-0.2+\co*+\Sc*\Shift,-0.65) {};
			\node[tensor] at (0.2+\co*+\Sc*\Shift,-0.65) {};
			\node at (-0.25+\co*+\Sc*\Shift,-0.9) {$1|0$};
			\node at (0.25+\co*+\Sc*\Shift,-0.9) {$0|1$};
			
			\node[] at (0.6+\co*+\Sc*\Shift, -0.65) {$=$};
			
			\def\coo{1.4};
			
			\draw[thick] (-0.2+\co*\coo*\Sc*\Shift,-0.65) --++ (0,0.4);
			\draw[thick] (0.2+\co*\coo*\Sc*\Shift,-0.65) --++ (0,0.4);
			
			\node[tensor] at (-0.2+\co*\coo*\Sc*\Shift,-0.65) {};
			\node[tensor] at (0.2+\co*\coo*\Sc*\Shift,-0.65) {};
			\node at (-0.2+\co*\coo*\Sc*\Shift,-0.9) {$1$};
			\node at (0.2+\co*\coo*\Sc*\Shift,-0.9) {$1$};
			\node[] at (0.4+\co*\coo*\Sc*\Shift,-0.65) {$,$};
		\end{tikzpicture}
	\end{equation}
	what shows that $L_{g,g}^1=-L^0_{g,g}=1$. By a straightforward computation, we see that
	\begin{equation}
		\begin{tikzpicture}
			\draw[thick] (-0.2,-0.65) --++ (0,1.7);
			\draw[thick] (0.2,-0.65) --++ (0,1.7);
			\draw[thick] (0.6,-0.65) --++ (0,1.7);
			
			\node[rectangle, fill=white] at (0,-0.1) {$\lambda^R_{g,g}$};
			\node[rectangle, fill=white] at (0.4,0.6) {$\lambda^R_{e,g}$};
			\node[tensor, fill=white] at (-0.2,-0.65) {};
			\node[tensor] at (0.2,-0.65) {};
			\node[tensor, fill=white] at (0.6,-0.65) {};
			\node at (-0.25,-0.9) {$1|0$};
			\node at (0.2,-0.9) {$0|1$};
			\node at (0.02,-0.525) {$\text{-}i$};
			\node at (0.65,-0.9) {$1|0$};
			
			\node at (1.0,-0.2) {$=$};
			
			\def\Sc{0.9};
			
			\node at (-0.5+\Sc*\Shift,-0.2) {$-$};
			\draw[thick] (-0.2+\Sc*\Shift,-0.65) --++ (0,0.8);
			\draw[thick] (0.2+\Sc*\Shift,-0.65) --++ (0,0.8);
			\draw[thick] (0.6+\Sc*\Shift, -0.65) --++ (0,0.8);
			\node[tensor] at (-0.2+\Sc*\Shift, -0.65) {};
			\node[tensor] at (0.2+\Sc*\Shift,-0.65) {};
			\node[tensor, fill=white] at (0.6+\Sc*\Shift, -0.65) {};
			\node[] at (-0.2+\Sc*\Shift,-0.9) {$1$};
			\node[] at (0.2+\Sc*\Shift,-0.9) {$1$};
			\node[] at (0.6+\Sc*\Shift,-0.9) {$1|0$};
			
			\node at (1.0+\Sc*\Shift,-0.2) {$=$};
			
			\def\Scc{1.9};
			\node at (-0.7+\Scc*\Shift,-0.2) {$-$};
			\draw[thick] (-0.2+\Scc*\Shift,-1.0) --++ (0,2.3);
			\draw[thick] (0.2+\Scc*\Shift,-1.0) --++ (0,2.3);
			\draw[thick] (0.6+\Scc*\Shift,-1.0) --++ (0,2.3);
			
			\node[tensor, fill=white] at (-0.2+\Scc*\Shift,-1.0) {};
			\node[tensor] at (0.2+\Scc*\Shift,-1.0) {};
			\node[tensor, fill=white] at (0.6+\Scc*\Shift,-1.0) {};
			\node[] at (-0.25+\Scc*\Shift,-1.25) {$1|0$};
			\node[] at (0.2+\Scc*\Shift,-1.25) {$0|1$};
			\node[] at (0.65+\Scc*\Shift,-1.25) {$1|0$};
			\node[] at (0.01+\Scc*\Shift,-0.9) {$\text{-}i$};
			\node[rectangle, fill=white] at (0.4+\Scc*\Shift,-0.5) {$\lambda^R_{g,g}$};
			\node[rectangle, fill=white] at (0.0+\Scc*\Shift, 0.2) {$\lambda^R_{g,e}$};
			\node[rectangle, fill=white] at (0.4+\Scc*\Shift,0.9) {$\lambda^R_{g,e}$};
			
			\node[] at (1.0+\Scc*\Shift,-0.2) {$,$};
		\end{tikzpicture} 
	\end{equation}
	hence $\omega(g,g,g)=\frac{L^0_{g,g} L^0_{g,e}}{L^1_{g,g} L^{0}_{e,g}}=-1$, demonstrating the presence of the anomaly. The local symmetry representation is given by $\mathcal{G}_e=\I$ and
	\begin{equation}
		\begin{split}
			\mathcal{G}_g=&\bigl(\lambda^R_{g,g}\bigr)^\dagger\otimes|g,g\rangle\langle e,e | +w_R^\dagger\otimes| g,e\rangle \langle e,g| \\
			&+w_R\otimes|e,g\rangle \langle g,e| + \lambda^R_{g,g}\otimes|e,e\rangle\langle g,g |.
			\label{eq:mcGCZX}
		\end{split}
	\end{equation}
	We can give the gauged tensor explicitly, using \eqref{eq:gauged_tensor_dw},
	\begin{equation}
		\label{CZX_gaugedMPS}
		\begin{tikzpicture}[baseline]
			\def\dx{0.1}
			\def\dxx{0.5}
			\def\dy{0.5}
			\draw[thick] (-\dx,0) --++ (0,\dy);
			\draw[thick, dotted] (\dx,0) --++ (0,\dy);
			\draw[thick] (-\dxx, 0) --++ (2*\dxx,0);
			\node[widetensor] at (0,0) {};
		\end{tikzpicture}
		=
		\begin{tikzpicture}[baseline]
			\def\dx{0.4}
			\def\dd{0.075}
			\def\dxb{0.6}
			\def\dy{0.5}
			\draw[thick] (-\dxb,0) --++ (2*\dxb,0);
			\draw[thick] (0-\dd,0) -- (0-\dd,\dy);
			\draw[thick] (0+\dd,0) -- (0+\dd,\dy);
			\draw[thick, dotted] (-\dx,0) -- (-\dx,\dy/2)--(\dx, \dy/2)--(\dx,\dy);
			\node[square, fill=green] at (-\dx,0) {};
		\end{tikzpicture}
		,
	\end{equation}
	where
	\begin{equation}
		\begin{tikzpicture}[baseline]
			\def\dx{0.4}
			\def\dy{0.4}
			\draw[thick] (-\dx,0) --++ (2*\dx,0);
			\draw[thick, dotted] (0,0) --++ (0,\dy);
			\node[square, fill=green] at (0,0) {};
			\node[] at (0, \dy+0.2) {0};
		\end{tikzpicture}
		=\I,\qquad
		\begin{tikzpicture}[baseline]
			\def\dx{0.4}
			\def\dy{0.4}
			\draw[thick] (-\dx,0) --++ (2*\dx,0);
			\draw[thick, dotted] (0,0) --++ (0,\dy);
			\node[square, fill=green] at (0,0) {};
			\node[] at (0, \dy+0.2) {1};
		\end{tikzpicture}
		=\left(\begin{array}{cc}
			0 & 1 \\
			-i & 0
		\end{array}\right).
	\end{equation}
	This MPS is not invariant under \eqref{eq:mcGCZX} due to the nontrivial $L$-symbol. A simple modification of the fusion operators that compensates for this is
	\begin{equation}
		\lambda^R_{g,g}\to\tilde\lambda^R_{g,g}\equiv \lambda^R_{g,g}Z_1
	\end{equation}
	where $Z_1$ is the Pauli $Z$ operator on the first qubit, and $\lambda^R_{e,e}, \lambda^R_{e,g}, \lambda^R_{g,e}$ do not change. The modified Gauss law $\tilde {\mc G}_g$ arising from the redefined fusion operators leaves \eqref{CZX_gaugedMPS} invariant, at the cost once more of the resulting projectors not commuting.
	
	\section{Relation to projective gauging}
	\label{sec:projHam}
	The gauging method introduced above is based on the promotion of symmetry defects to gauge degrees of freedom and yields an MPS of the same bond dimension as the original one. This departs from the strategy presented in the seminal paper \cite{PhysRevX.5.011024}, where a state is gauged by tensoring it with a product state of the gauge degrees of freedom and projecting it onto the gauge-invariant subspace. Since the Gauss law projectors commute, this can be easily accomplished by applying each of them successively. This strategy presents the advantage that it only requires the knowledge of the local gauge symmetry representation, specifically of the Gauss law projectors, to be applied to any state, regardless of whether it has an MPS form. On the other hand, the original paper only proposed these local representations for the case of onsite symmetry, and \cite{GarreKull} proposed a way to generate such a representation for MPO group algebras, in terms of MPO tensors and their associated fusion tensors. However, that construction only works whenever the MPO algebras satisfy stronger constraints than general MPU groups (including, as we mentioned above, that Eq.~\eqref{eq:fusion_1} holds for $m=0$). For example, it can be checked that the $\mathsf{CZ}$ MPU representation of $\Z_2$ in Section \ref{sec:CZ} does not produce a group representation under an analogous construction. Our proposal yields a local symmetry representation for any finite MPU group representation, although acting on two physical and two gauge degrees of freedom (as opposed to one physical and two gauge degrees of freedom), given by \eqref{eq:mcG}. This construction is independent of the globally symmetric state we are trying to gauge, and its associated Gauss law projectors commute in the nonanomalous case; see Proposition~\ref{prop:comm_proj}. Thus, we could see our construction as providing an input for a generalization of the projective strategy of \cite{PhysRevX.5.011024} to nonanomalous MPU symmetries, regardless of whether the state is block independent or even an MPS at all.
	
	Therefore, when block independence holds, we have two choices to gauge an MPU symmetric MPS, which give potentially different outcomes. Choosing to promote the symmetry defects to degrees of freedom, that is, using the formalism introduced in this paper, the gauge configurations that can appear in the resulting superposition need to satisfy $g_1g_2\ldots g_n\cdot x = x$ for some $x\in\mathsf{X}$. However, if we tensor the initial MPS with the gauge configuration made of neutral elements $\ket{ee\ldots e}$ and then apply the projectors built from \eqref{eq:gauge_projs}, we will be enforcing $g_1g_2\ldots g_n = e$, which is a stronger constraint (unless $g\cdot x = x$ for some $x$ implies $g=e$, that is, we are in the maximally symmetry broken case). This additional constraint on the global charge sector, which is nonlocal in nature, explains the potential increase in bond dimension. 
	
	To clarify this issue, we conclude this section with a minimal example of this dichotomy, involving an onsite symmetry. Consider the fixed-point MPS with bond dimension $D=2$ and physical dimension $d=4$, given by the tensor\footnote{Sometimes known as a \textit{corner double line} or CDL tensor.}
	\begin{equation}
		\begin{tikzpicture}
			\draw[thick] (0.0,0.0) --++ (1.3,0.0);
			\draw[thick] (0.65,0.0) --++ (0.0,0.8);
			\node[tensor] at (0.65,0.0) {};
			\node[] at (1.7,0.25) {$=$};
			\draw[thick] (2.0,0.0) --++ (0.8,0.0);
			\draw[thick] (2.8,0.0) --++ (0.0,0.8);
			\draw[thick] (3.1,0.0) --++ (0.8,0.0);
			\draw[thick] (3.1,0.0) --++ (0.0,0.8);
			\node[] at (4.0,0.0) {$,$};
		\end{tikzpicture}
	\end{equation}
	whose iterated contraction gives rise to a sequence of maximally entangled pairs. This state displays an onsite $\Z_2$ symmetry given by the operator $XX$:
	\begin{equation}
			\begin{tikzpicture}
				\draw[thick] (0.0,0.0) --++ (0.5,0.0);
				\draw[thick] (0.5,0.0) --++ (0.0,0.8);
				\draw[thick] (0.8,0.0) --++ (0.5,0.0);
				\draw[thick] (0.8,0.0) --++ (0.0,0.8);
				\node[tensor,red] at (0.5,0.5) {};
				\node[tensor,red] at (0.8,0.5) {};
				\node[] at (0.2,0.5) {{\color{red}$X$}};
				\node[] at (1.1,0.5) {{\color{red}$X$}};
				\node[] at (1.7,0.25) {$=$};
				\draw[thick] (2.0,0.0) --++ (0.8,0.0);
				\draw[thick] (2.8,0.0) --++ (0.0,0.8);
				\draw[thick] (3.1,0.0) --++ (0.8,0.0);
				\draw[thick] (3.1,0.0) --++ (0.0,0.8);
				\node[tensor,red] at (2.45,0.0) {};
				\node[tensor,red] at (3.45,0.0) {};
				\node[] at (2.45,0.35) {{\color{red}$X$}};
				\node[] at (3.45,0.35) {{\color{red}$X$}};
				\node[] at (4.0,0.0) {$.$};
			\end{tikzpicture}
	\end{equation}
	If we do \textit{promotion} gauging, the promoted defects give rise to sum tensors, and the state remains a product of unentangled copies. It is apparent that this state has a nonzero amplitude for all configurations of gauge variables $g_{j+\frac{1}{2}}\in\{0,1\}$:
		\begin{equation}
			\ldots
			\begin{tikzpicture}
				\def\dx{0.5}
				\def\ddx{0.15}
				\def\dy{0.7}
				\def\dr{0.12}
				\draw[thick] (\ddx,0)--++(2*\dx,0);
				\draw[thick] (3*\ddx+2*\dx,0)--++(\dx,0);
				\draw[thick] (-\ddx,0)--++(-2*\dx,0);
				\draw[thick] (-3*\ddx-2*\dx,0)--++(-\dx,0);
				\foreach\x in {-3*\ddx-2*\dx,-\ddx-2*\dx,-\ddx-\dx,-\ddx,\ddx,\ddx+\dx,\ddx+2*\dx,3*\ddx+2*\dx}
				{\draw[thick] (\x,0)--++(0,\dy);}
				\draw (\ddx+\dx,0) circle (\dr);
				\draw (-\ddx-\dx,0) circle (\dr);
				\draw[thick] (\ddx+\dx,-\dr)--++(0,\dr);
				\draw[thick] (-\ddx-\dx,-\dr)--++(0,\dr);
				\node[] at (-1*\dx-1*\ddx,\dy+0.2) {$g_{j-\ha}$};
				\node[] at (1*\dx+1*\ddx,\dy+0.2) {$g_{j+\ha}$};
			\end{tikzpicture}
			\ldots
		\end{equation}
	On the other had, we can do \textit{projection} gauging, by tensoring the state with $\Z_2$ gauge variables in the state $\ket{00\ldots0}$ and applying the Gauss law projectors
		\begin{equation}
			P_j=\dfrac{\mathds{1}+X_{j-\frac{1}{2}}(XX)_jX_{j+\frac{1}{2}}}{2},
		\end{equation}
	resulting in the state
		\begin{equation}
			\ldots
			\begin{tikzpicture}
				\def\dx{0.5}
				\def\ddx{0.15}
				\def\dy{0.7}
				\def\dh{0.4}
				\def\dyb{0.3}
				\draw[thick] (\ddx,0)--++(2*\dx,0);
				\draw[thick] (3*\ddx+2*\dx,0)--++(1.5*\dx,0);
				\draw[thick] (-\ddx,0)--++(-2*\dx,0);
				\draw[thick] (-3*\ddx-2*\dx,0)--++(-1.5*\dx,0);
				\foreach\x in {-3*\ddx-2*\dx,-\ddx-2*\dx,-\ddx,\ddx,\ddx+2*\dx,3*\ddx+2*\dx}
				{\draw[thick] (\x,0)--++(0,\dy+2*\dyb+2*\dh);}
				\foreach\x in {-3*\ddx-3*\dx,-\ddx-\dx,\ddx+\dx,3*\ddx+3*\dx}
				{\draw[thick] (\x,\dy/2)--++(0,\dy/2+2*\dyb+2*\dh);}
				\node[tensor] at (\ddx+\dx,\dy/2) {};
				\node[tensor] at (-\ddx-\dx,\dy/2) {};
				\node[tensor] at (3*\ddx+3*\dx,\dy/2) {};
				\node[tensor] at (-3*\ddx-3*\dx,\dy/2) {};
				\draw[fill=white] (-\dx-2*\ddx,\dy) rectangle (\dx+2*\ddx,\dy+\dh);
				\draw[fill=white] (-4*\ddx-3*\dx,\dy+\dh+\dyb) rectangle (-\dx,\dy+2*\dh+\dyb);
				\draw[fill=white] (\dx,\dy+\dh+\dyb) rectangle (4*\ddx+3*\dx,\dy+2*\dh+\dyb);
				\node[] at (0.0,\dy+\dh/2-0.025) {$P_j$};
				\node[] at (-2*\dx-2*\ddx,\dy+\dh+\dyb+\dh/2-0.025) {$P_{j-1}$};
				\node[] at (2*\dx+2*\ddx,\dy+\dh+\dyb+\dh/2-0.025) {$P_{j+1}$};
				\node[] at (3*\dx+5*\ddx,\dy/2-0.025) {$\ket{0}$};
				\node[] at (1*\dx+3*\ddx,\dy/2-0.025) {$\ket{0}$};
				\node[] at (-1*\dx+1*\ddx,\dy/2-0.025) {$\ket{0}$};
				\node[] at (-3*\dx-1*\ddx,\dy/2-0.025) {$\ket{0}$};
				\node[] at (-3*\dx-5*\ddx,\dy/2-0.025) {$\phantom{\ket{0}}$};
				
			\end{tikzpicture}
			\ldots
		\end{equation}
		Since both the trivial gauge configuration $\ket{00\ldots0}$ and the Gauss-law projectors $P_j$ are left invariant by $\prod_i Z_{i+\frac{1}{2}}$, the parity operator associated with the emergent $\widehat{\mathbb{Z}}_2$ symmetry, the new gauged state has support only on gauge configurations with even charge, i.e. +1-eigenstates of $\prod_i Z_{i+\frac{1}{2}}$, as opposed to what happened after promotion gauging.

	\section{Discussion and outlook}
	\label{sec:outlook}
	
	In this work, we have studied MPS symmetric under MPU representations of finite groups. We have found that, given certain assumptions, we can define defect tensors (symmetry twists) that can be moved and fused by local unitaries. Using these defects, we have introduced a formalism analogous to \cite{Seifnashri}, which allows us to gauge the symmetry, recovering a locally symmetric MPS of the same bond dimension whenever the block independence condition holds. If the latter does not happen, we can still use our method to obtain a local gauge group representation, giving rise to Gauss law projectors that will commute in the absence of an anomaly. From a physical perspective, this provides a direct tensor network mechanism to implement gauging beyond the standard onsite setting, while keeping full control over the resulting local constraints. Additionally, our proposal based on defects, when applicable, maintains the bond dimension of the MPS, in comparison with the approach based on projectors, which generally increases it. This could lead to the new method having a smaller computational cost, at the expense of it keeping multiple total charge sectors, as explained in Section \ref{sec:projHam}.
	
	It is interesting to consider the relationship between the gauging at the Hamiltonian and the state levels. Suppose that the ground state of the Hamiltonian for our model is unique and has an MPS structure. Then, one can gauge the ground state and find a parent Hamiltonian for the resulting MPS. The question that arises naturally is: How is the new Hamiltonian related to the one obtained by gauging the original Hamiltonian? Conversely, one can start by gauging at the Hamiltonian level and then compare the corresponding ground-state subspaces. We leave the discussion of this problem for future work. 
	
	When block independence does not hold, we can still find a local symmetry representation that leaves the gauged MPS invariant by modifying the defect fusion operators. We have called this process state-level gauging. However, there is no guarantee that the resulting projectors will commute on the whole Hilbert space. They will do so in their common $+1$-eigenspace, the gauge-invariant subspace, which we know is nonempty, since it contains at least the gauged MPS. We leave for future work the study of the structure of this subspace to see if it allows for nontrivial physics, and in particular, whether its dimension is bounded or grows with system size. For the examples in Section \ref{sec:examples_statelevel}, preliminary numerics for small system sizes leave the door open for the latter to be true. This method reflects the physical cost, or trade-off, between generality —since state-level gauging applies also to anomalous symmetries— and the complexity of the resulting gauge structure, resulting from the non-commuting nature of the constraints.
	
	Gauging procedures can also be seen as means of mapping between different phases of matter (see \cite{Lootens_2023, rubio2024, BlanikGarreSchuch} for recent tensor network-based examples). Our formalism should thus play a similar role for phases characterized by MPU symmetries. Yet another intriguing potential extension is to consider MPSs invariant under generic MPO symmetries, not necessarily those possessing the MPU structure. The defect formalism from \cite{Seifnashri} has been generalized to non-invertible symmetries in \cite{SSS}. Since we have made intensive use of the structural theorem for MPUs, however, it is at first glance nontrivial how to proceed in the case of MPO algebra symmetries, where we expect the connection to weak Hopf algebras \cite{weakHopf} and fusion categories to play an essential role. 
	
	Finally, we emphasize that the present framework is not restricted to exotic or anomalous non-onsite symmetries, but naturally encompasses simpler and widely studied cases such as onsite symmetries and non-onsite non-anomalous realizations. In this sense, it provides a unified tensor network perspective on gauging across a broad range of physically relevant settings. We hope that this unified viewpoint could be useful for practical applications, for instance in constructing and analyzing ground states and dynamical properties of lattice gauge theories and other strongly correlated quantum systems.
	
    \subsection*{Author contributions}
	All authors contributed to the development of the ideas presented in this work, discussed the results, contributed to the preparation of the manuscript, and approved the final version.

	\subsection*{Acknowledgements}
	We thank David Blanik, Jos\'{e} Garre-Rubio, Laurens Lootens, András Molnár, Gerardo Ortiz, Norbert Schuch, Nathan Seiberg, Sahand Seifnashri, Alex Turzillo and Frank Verstraete for illuminating discussions. Part of this work was conducted at the Institute for Advanced Study, Princeton. A.B. is supported by the Alexander von Humboldt Foundation. The work was partially funded by the Deutsche Forschungsgemeinschaft (DFG, German Research Foundation) under Germany’s Excellence Strategy - EXC-2111 – 390814868 and the Österreichischer Wissenschaftsfonds (FWF, Austrian Science Fund) under Grant DOI 10.55776/F71.
	
	\appendix
	
	\section{$\Z_2$ anomalies}
	\label{app:Z2}
	In this appendix, we will show the meaning of having $\sigma_g=-1$ (and thus $T_g\neq\I$) for a group element $g$, which from our definitions (see Section \ref{sec:assumptions}) necessarily must be of order two, $g=g^{-1}$. 
	
	First, we prove Eq. \eqref{eq:wggg}.
	Note that from the definition of the anomaly cocycle,
	\begin{equation}
		\begin{tikzpicture}[baseline]
			\def\dxl{0.4}
			\def\dxr{0.5}
			\def\dy{0.6}
			\def\dd{0.2}
			
			\draw[thick] (0,-\dy-\dd) -- (0, \dy+\dd);
			
			\draw[thick,red] (-\dxl, 0) -- (\dxr, 0);
			\draw[thick,red] (-\dxl, \dy) -- (\dxr, \dy);
			\draw[thick,red] (-\dxl, -\dy) -- (\dxr, -\dy);
			
			\draw[fusion] (-\dxl, 0) --++ ( 0, \dy);
			\draw[fusion] (\dxr, -\dy) --++ ( 0, \dy);
			
			\node[mpo] at (0, 0) {};
			\node[mpo] at (0, -\dy) {};
			\node[mpo] at (0, \dy) {};
			\node[irrep, anchor=south west] at (0.1, \dy) {$g$};
			\node[irrep, anchor=south west] at (0.1, 0) {$\inv g$};
			\node[irrep, anchor=south west] at (0.1, -\dy) {$g$};
		\end{tikzpicture}
		=\omega(g, g^{-1}, g)\;
		\begin{tikzpicture}[baseline]
			\draw[thick,red] (-0.4, 0) -- (0.4, 0);
			\draw[thick] (0, -0.4) -- (0, 0.4);
			\node[mpo]  at (0,0) {};
		\end{tikzpicture},
	\end{equation}
	while, using our choice of fusion operators, cf. Eq. \eqref{eq:Fgg}, 
	\begin{equation}
		\begin{tikzpicture}[baseline]
			\def\dxl{0.4}
			\def\dxr{0.5}
			\def\dy{0.6}
			\def\dd{0.2}
			
			\draw[thick] (0,-\dy-\dd) -- (0, \dy+\dd);
			
			\draw[thick,red] (-\dxl, 0) -- (\dxr, 0);
			\draw[thick,red] (-\dxl, \dy) -- (\dxr, \dy);
			\draw[thick,red] (-\dxl, -\dy) -- (\dxr, -\dy);
			
			\draw[fusion] (-\dxl, 0) --++ ( 0, \dy);
			\draw[fusion] (\dxr, -\dy) --++ ( 0, \dy);
			
			\node[mpo] at (0, 0) {};
			\node[mpo] at (0, -\dy) {};
			\node[mpo] at (0, \dy) {};
			\node[irrep, anchor=south west] at (0.1, \dy) {$g$};
			\node[irrep, anchor=south west] at (0.1, 0) {$\inv g$};
			\node[irrep, anchor=south west] at (0.1, -\dy) {$g$};
		\end{tikzpicture}
		=
		\begin{tikzpicture}[baseline]
			\def\dxl{0.8}
			\def\dxr{1.2}
			\def\dy{0.6}
			\def\dd{0.2}
			
			\draw[thick] (0,-\dy-\dd) -- (0, \dy+\dd);
			
			\draw[thick,red] (-\dxl, 0) -- (\dxr, 0);
			\draw[thick,red] (-\dxl, \dy) -- (\dxr, \dy);
			\draw[thick,red] (-\dxl, -\dy) -- (\dxr, -\dy);
			
			\draw[thick,red] (-\dxl, 0) --++ ( 0, \dy);
			\draw[thick,red] (\dxr, -\dy) --++ ( 0, \dy);
			
			\node[mpo] at (0, 0) {};
			\node[mpo] at (0, -\dy) {};
			\node[mpo] at (0, \dy) {};
			\node[tensor, red] at (-\dxl/2, 0) {};
			\node[tensor, red] at (-\dxl/2, \dy) {};
			\node[tensor, red] at (\dxr*0.45, 0) {};
			\node[tensor, red] at (\dxr*0.8, 0) {};
			\node[square, fill=white] at (\dxr, -\dy/2) {};
			\node[] at (0.18,0.15) {$^\dagger$};
			\node[red] at (-\dxl/2, \dy+0.3) {$T_g^\dagger$};
			\node[red] at (-\dxl/2, 0.25) {$T_g$};
			\node[red] at (\dxr*0.45, 0.3) {$T^\dagger_g$};
			\node[red] at (\dxr*0.8, 0.25) {$T_g$};
		\end{tikzpicture}
		=\sigma_g\;
		\begin{tikzpicture}[baseline]
			\def\dx{0.4}
			\def\dy{0.6}
			\def\dd{0.2}
			
			\draw[thick] (0,-\dy-\dd) -- (0, \dy+\dd);
			
			\draw[thick,red] (-\dx, 0) -- (\dx, 0);
			\draw[thick,red] (-\dx, \dy) -- (\dx, \dy);
			\draw[thick,red] (-\dx, -\dy) -- (\dx, -\dy);
			
			\draw[thick,red] (-\dx, 0) --++ ( 0, \dy);
			\draw[thick,red] (\dx, -\dy) --++ ( 0, \dy);
			\node[] at (0.18,0.15) {$^\dagger$};
			
			\node[mpo] at (0, 0) {};
			\node[mpo] at (0, -\dy) {};
			\node[mpo] at (0, \dy) {};
			\node[square, fill=white] at (\dx, -\dy/2) {};
		\end{tikzpicture}
		=\sigma_g\;\begin{tikzpicture}[baseline]
			\draw[thick,red] (-0.4, 0) -- (0.4, 0);
			\draw[thick] (0, -0.4) -- (0, 0.4);
			\node[mpo]  at (0,0) {};
		\end{tikzpicture},
		\label{eq:sigmadiag}
	\end{equation}
	where the last equality appears in the proof of Proposition~\ref{prop:MPUsplus}. Thus we have $\omega(g, g^{-1}, g)=\sigma_g$ as intended. We also check
	\begin{equation}
		\omega(g, g^{-1}, g) = \dfrac{L^x_{\inv g, g}L^x_{g, e}}{L^{gx}_{g, \inv g}L^x_{e, g}} = \dfrac{L^x_{\inv g, g}}{L^{gx}_{g, \inv g}},
	\end{equation}
	which follows from Eqs. \eqref{eq:omega_and_Ls} and \eqref{eq:triv_Ls}.
	
	Now consider an MPU representation of $\Z_2=\langle g\mid g^2=e\rangle$ (which could be part of a representation of a larger group). The MPU associated with $g$ squares to the identity, in other words, it is time-reversal invariant, $U_g=U_g^\dagger$. Recall that this implies the existence of a gauge transformation between the MPU tensor and its adjoint (which can be taken to be unitary since the tensors are in canonical form)
	\begin{equation}
		\mc U_g = T_g\mc U_g^\dagger T_g^\dagger,
	\end{equation}
	where it can be seen that, for consistency, $T_gT_g^\ast \equiv \sigma_g \mathds{1}$, $\sigma_g=\pm 1$. This sign labels two SPT phases associated with time-reversal invariance.
	\begin{appxprop}
		An MPU representation of $\Z_2$ is anomalous (i.e. $\omega$ is a nontrivial 3-cocycle) if and only if the MPU belongs to the nontrivial SPT class for time-reversal invariance \footnote{This anomaly index coincides with what is called the \textit{Frobenius-Schur indicator} in category theory. The authors thank Sahand Seifnashri for informing us of this and thus motivating this proposition.}.
	\end{appxprop}
	\begin{proof}
		Consider the 3-cocycle $\omega$ that defines the anomaly. As in Section~\ref{sec:assumptions}, we can pick all fusion tensors involving the identity to be trivial, which automatically leads to $\omega(h_1,h_2,h_3)=1$ whenever $h_i=e$ for some $i$. On the other hand, from the argument immediately above, we have $\omega(g,g,g)=\sigma_g$. The two possible signs thus correspond to the two cohomology classes in $H^3(\Z_2, \C^\times)\cong \Z_2$, which are therefore labeled by the index of the SPT phase.
	\end{proof}
	In the anomalous $\Z_2$ case, whenever $\mc U_g\neq\mc U^\dagger_g$, it is not obvious that we can find a decomposition of the form \eqref{eq:3-leg} such that the compatibility conditions \eqref{eq:compat} hold. Rather, we generalize these conditions to read \eqref{eq:modif_XY}. We next prove that there exists a decomposition that satisfies the generalized compatibility conditions \eqref{eq:modif_XY}, as well as all the pleasant conditions that were shown for the $X, Y$ tensors in Section~\ref{sec:MPUsstr}.
	\begin{appxprop}
		Consider an MPU that squares to identity, with a time-reversal SPT index $\sigma$, generated by a simple injective tensor $\mc U$ whose adjoint is also simple. Then there exist decompositions $\mc U=X_1Y_1=X_2Y_2$, satisfying \eqref{eq:modif_XY}, such that Eqs.~\eqref{eq:3-leg}-\eqref{eq:thm1c}, \eqref{eq:wLR}-\eqref{eq:MPUnice2} and \eqref{eq:MPUnice3}-\eqref{eq:MPUnice4} (the ``pleasant properties'') all hold. Furthermore, for the truncated symmetries introduced in Section~\ref{sec:defs}, Eq.~\eqref{eq:complementary} holds up to a factor of $\sigma$,
		\begin{equation}
			U^L U^{\infty} = \sigma U^{\infty - L}.
			\label{eq:UUU}
		\end{equation}
	\end{appxprop}
	\begin{proof}
		From the simplicity of the tensor and its adjoint, and from the results in Section \ref{sec:MPUsstr}, we know that there exists a decomposition
		\begin{equation}
			\begin{tikzpicture}
				\draw[thick,red] (-0.4, 0) -- (0.4, 0);
				\draw[thick] (0, -0.4) -- (0, 0.4);
				\node[mpo]  at (0,0) {};
			\end{tikzpicture}
			= 
			\begin{tikzpicture}
				\def\dy{0.2}
				\def\dz{0.3}
				\draw[thick,red] (-0.35, -\dy) -- (0, -\dy);
				\draw[thick,red] (0, \dy) -- (0.35, \dy);
				\draw[thick,decorate] (0, -\dz) -- (0, \dz);
				\draw[thick] (0,-\dz)--(0,-1.5*\dz);
				\draw[thick] (0,\dz)--(0,1.5*\dz);
				
				\node[tensor, fill = white]  at (0,\dy) {};
				\node[tensor, fill = white]  at (0,-\dy) {};
				
				\node[] at (-0.3, \dy) {$\tilde X_1$};
				\node[] at (0.3, -\dy) {$\tilde Y_1$};
			\end{tikzpicture}
			=
			\begin{tikzpicture}
				\def\dy{0.2}
				\def\dz{0.3}
				\draw[thick,red] (0.35, -\dy) -- (0, -\dy);
				\draw[thick,red] (0, \dy) -- (-0.35, \dy);
				\draw[ultra thick] (0, -\dz) -- (0, \dz);
				\draw[thick] (0,\dz)--(0,1.5*\dz);
				\draw[thick] (0,-\dz)--(0,-1.5*\dz);
				
				\node[tensor]  at (0,\dy) {};
				\node[tensor]  at (0,-\dy) {};
				
				\node[] at (0.3, \dy) {$X_2$};
				\node[] at (-0.3, -\dy) {$Y_2$};
			\end{tikzpicture}\ ,
		\end{equation}
		such that all the {\it pleasant properties} hold (we have added a tilde to $\tilde X_1, \tilde Y_1$ to differentiate them from the tensors we are about to define). By assumption, the tensor can also be decomposed as 
		\begin{equation}
			\begin{tikzpicture}
				\draw[thick,red] (-0.4, 0) -- (0.4, 0);
				\draw[thick] (0, -0.4) -- (0, 0.4);
				\node[mpo]  at (0,0) {};
			\end{tikzpicture}
			= 
			\begin{tikzpicture}
				\draw[thick,red] (-0.6, 0) -- (0.6, 0);
				\draw[thick] (0, -0.4) -- (0, 0.4);
				
				\node[mpo]  at (0,0) {};
				\node[tensor, red] at (0.4,0) {};
				\node[tensor, red] at (-0.4,0) {};
				\node[] at (0.2, 0.1) {$^{\dagger}$};
				\node[red] at (-0.45, -0.25) {$^{\phantom\dagger}T$};
				\node[red] at (0.45, -0.25) {$T^{\dagger}$};
			\end{tikzpicture}        
			=
			\begin{tikzpicture}
				\def\dy{0.2}
				\def\dz{0.3}
				\draw[thick,red] (-0.6, -\dy) -- (0, -\dy);
				\draw[thick,red] (0, \dy) -- (0.6, \dy);
				\draw[ultra thick] (0, -\dz) -- (0, \dz);
				\draw[thick] (0,\dz)--(0,1.5*\dz);
				\draw[thick] (0,-\dz)--(0,-1.5*\dz);
				
				\node[tensor]  at (0,\dy) {};
				\node[tensor]  at (0,-\dy) {};
				\node[tensor, red] at (0.4,\dy) {};
				\node[tensor, red] at (-0.4,-\dy) {};
				
				\node[] at (-0.35, \dy+0.15) {$Y^\dagger_2$};
				\node[] at (0.35, -\dy-0.15) {$X^\dagger_2$};
				
				\node[red] at (-0.45, -\dy-0.25) {$^{\phantom\dagger}T$};
				\node[red] at (0.45, \dy+0.25) {$T^{\dagger}$};
			\end{tikzpicture}\ ,
		\end{equation}
		thus we define our candidate decomposition as $X_1 \equiv (\mathds{1}\otimes T^*) Y_2^\dagger$, $Y_1 \equiv X_2^\dagger(\mathds{1}\otimes T^T)$. Because of the left (resp. right) invertibility of $X_1$, $\tilde X_1$ (resp. $Y_1$, $\tilde Y_1$) we know there exists an invertible matrix $K$ such that 
		\begin{equation}
			\tilde Y_1 = K  Y_1,\qquad  \tilde X_1 = X_1 K^{\dagger}.
		\end{equation}
		Indeed, we only need to define
		\begin{equation}
			K= \tilde X_1^{\dagger} X^{\phantom{1}}_1 = ^{\phantom{1}} \tilde Y_1  Y_1^{\dagger}.
		\end{equation}
		Now we can use
		\begin{equation}
			Y_1 = K(\mathds{1}\otimes T^T)X_2^\dagger,\qquad Y_2 = K^\dagger(\mathds{1}\otimes T^*)\tilde X_1^\dagger,
		\end{equation}
		to check
		\begin{equation}
			u = 
			\begin{tikzpicture}
				\def\dx{0.3}
				\def\dy{0.4}
				\draw[thick] (-\dx,-\dy) --++ (0, \dy);
				\draw[thick] (\dx,-\dy) --++ (0, \dy);
				\draw[ultra thick] (-\dx,0)--++(0,\dy);
				\draw[thick, decorate] (\dx,0)--++(0,\dy);
				\draw[thick,red] (-\dx,0) --++ (2*\dx, 0);
				\node[tensor] at (-\dx, 0) {};
				\node[tensor, fill = white] at (\dx, 0) {};
				\node[] at (\dx +0.3, 0) {$\tilde Y_1$};
				\node[] at (-\dx -0.3, 0) {$Y_2$};
			\end{tikzpicture}=
			\begin{tikzpicture}
				\def\dx{0.5}
				\def\dy{0.4}
				\def\dd{0.3}
				\draw[thick] (-\dx,-\dy) --++ (0, \dy);
				\draw[thick] (\dx,-\dy) --++ (0, \dy);
				\draw[ultra thick] (-\dx,\dy) --++(0,\dd);
				\draw[thick,decorate] (-\dx,0)--++(0,\dy);
				\draw[ultra thick] (\dx,0)--++(0,\dy);
				\draw[thick, decorate] (\dx,\dy)--++(0,\dd);
				\draw[thick,red] (-\dx,0) --++ (2*\dx, 0);
				\node[tensor] at (\dx, 0) {};
				\node[tensor, fill = white] at (-\dx, 0) {};
				\node[tensor, red] at (-\dx/3, 0) {};
				\node[tensor, red] at (\dx/3, 0) {};
				\node[tensor, fill = blue] at (-\dx, \dy) {};
				\node[tensor, fill = blue] at (\dx, \dy) {};
				\node[red] at (\dx/3,0.3) {$^{\phantom{*}}T$};
				\node[red] at (-\dx/3,0.3) {$T^*$};
				\node[blue] at (\dx+0.2, \dy+0.2) {$K$};
				\node[blue] at (-\dx-0.25, \dy+0.2) {$K^\dagger$};
				\node[] at (\dx +0.3, 0) {$X_2^\dagger$};
				\node[] at (-\dx -0.3, 0) {$\tilde X_1^\dagger$};
			\end{tikzpicture}
			=\sigma (K^\dagger\otimes K)v^\dagger,
		\end{equation}
		which proves that $K^\dagger\otimes K$ is unitary and hence so is $K$, thanks to the unitarity of the $u$, $v$ gates defined from $\{\tilde X_1, \tilde Y_1, X_2, Y_2\}$. Thus, the gates built from $X_1$, $Y_1$ will be unitary, since they differ by $K$ from the ones defined from $\tilde X_1$, $\tilde Y_1$, and all the other {\it pleasant properties} follow similarly. Finally, Eq. \eqref{eq:UUU} can be checked by writing the associated tensor network diagram, which we leave to the reader as an exercise.
	\end{proof}
	
	\section{Proof of Proposition \ref{prop:unitarity_and_gauges}}
	\label{app:proofprops}
	
	In this appendix, we will give a proof of Proposition~\ref{prop:unitarity_and_gauges}, which if restricted to the injective case gives Proposition~\ref{prop:3}. First, we will show that for any gauge of the fusion and action tensors, Eqs.~\eqref{eq:fusop1}-\eqref{eq:fusop2} hold with the corresponding $L$-symbols for said gauge, and the $\lambda$ operators are unitary up to a scalar. Then, in the second part, we will show that a gauge can be found where this scalar is one, and all other properties listed in the proposition, depending on the particular situation, together with the general assumptions of Section \ref{sec:assumptions}, hold.
	
	\begin{lemma}
		With our assumptions for the MPU group representation permuting injective MPS, the following identities hold 
		\begin{subequations}
			\begin{equation}
				\begin{tikzpicture}
					\def\dx{0.5}
					\def\ddx{0.3}
					\def\dddx{0.2}
					\def\dya{0.4}
					\def\dyb{0.8}
					\def\dyc{1.0}
					\draw[thick] (-\ddx,0) --  (3*\dx+\ddx,0);
					\draw[thick, red] (0,\dyb) --  (3*\dx+\dddx,\dyb);
					\draw[thick, red] (2*\dx,\dya) --  (3*\dx+\dddx,\dya);
					\draw[thick] (\dx,0) --++ (0,\dyc);
					\draw[thick] (3*\dx,0) --++  (0,\dyc);
					\node[tensor] at (\dx, 0) {};
					\node[tensor] at (3*\dx,0) {};
					\node[mpo] at (\dx, \dyb) {};
					\node[mpo] at (3*\dx, \dyb) {};
					\node[mpo] at (3*\dx, \dya) {};
					\node[] at (\dx,-0.25) {$hx$};
					\node[] at (3*\dx,-0.25) {$x$};
					\pic[] at (0, 0) {actL=\dyb//g/};
					\pic[] at (2*\dx, 0) {actL=\dya//h/};
				\end{tikzpicture}
				=
				L^x_{g,h}\;
				\begin{tikzpicture}
					\def\dx{0.5}
					\def\ddx{0.3}
					\def\dddx{0.2}
					\def\dya{0.4}
					\def\dyb{0.8}
					\def\dyc{1.0}
					\draw[thick] (-\ddx,0) --  (3*\dx+\ddx,0);
					\draw[thick, red] (0,0.5*\dya+0.5*\dyb) --  (\dx,0.5*\dya+0.5*\dyb);
					\draw[thick, red] (\dx,\dyb) --  (3*\dx+\dddx,\dyb);
					\draw[thick, red] (\dx,\dya) --  (3*\dx+\dddx,\dya);
					\draw[thick] (2*\dx,0) --++ (0,\dyc);
					\draw[thick] (3*\dx,0) --++  (0,\dyc);
					\node[tensor] at (2*\dx, 0) {};
					\node[tensor] at (3*\dx,0) {};
					\node[mpo] at (2*\dx, \dyb) {};
					\node[mpo] at (3*\dx, \dyb) {};
					\node[mpo] at (3*\dx, \dya) {};
					\node[mpo] at (2*\dx, \dya) {};
					\node[] at (2*\dx,-0.25) {$x$};
					\node[] at (3*\dx,-0.25) {$x$};
					\pic[] at (0, 0) {actL=0.5*\dya+0.5*\dyb//gh/};
					\pic[] at (\dx, \dya) {fusL=\dyb-\dya//g/h};
				\end{tikzpicture},
			\end{equation}
			\begin{equation}
				\begin{tikzpicture}
					\def\dx{-0.5}
					\def\ddx{-0.3}
					\def\dddx{-0.2}
					\def\dya{0.4}
					\def\dyb{0.8}
					\def\dyc{1.0}
					\draw[thick] (-\ddx,0) --  (3*\dx+\ddx,0);
					\draw[thick, red] (0,\dyb) --  (3*\dx+\dddx,\dyb);
					\draw[thick, red] (2*\dx,\dya) --  (3*\dx+\dddx,\dya);
					\draw[thick] (\dx,0) --++ (0,\dyc);
					\draw[thick] (3*\dx,0) --++  (0,\dyc);
					\node[tensor] at (\dx, 0) {};
					\node[tensor] at (3*\dx,0) {};
					\node[mpo] at (\dx, \dyb) {};
					\node[mpo] at (3*\dx, \dyb) {};
					\node[mpo] at (3*\dx, \dya) {};
					\node[] at (\dx,-0.25) {$hx$};
					\node[] at (3*\dx,-0.25) {$x$};
					\pic[] at (0, 0) {actR=\dyb/g//};
					\pic[] at (2*\dx, 0) {actR=\dya/h//};
				\end{tikzpicture}
				=
				\dfrac{1}{L^x_{g,h}}\;
				\begin{tikzpicture}
					\def\dx{-0.5}
					\def\ddx{-0.3}
					\def\dddx{-0.2}
					\def\dya{0.4}
					\def\dyb{0.8}
					\def\dyc{1.0}
					\draw[thick] (-\ddx,0) --  (3*\dx+\ddx,0);
					\draw[thick, red] (0,0.5*\dya+0.5*\dyb) --  (\dx,0.5*\dya+0.5*\dyb);
					\draw[thick, red] (\dx,\dyb) --  (3*\dx+\dddx,\dyb);
					\draw[thick, red] (\dx,\dya) --  (3*\dx+\dddx,\dya);
					\draw[thick] (2*\dx,0) --++ (0,\dyc);
					\draw[thick] (3*\dx,0) --++  (0,\dyc);
					\node[tensor] at (2*\dx, 0) {};
					\node[tensor] at (3*\dx,0) {};
					\node[mpo] at (2*\dx, \dyb) {};
					\node[mpo] at (3*\dx, \dyb) {};
					\node[mpo] at (3*\dx, \dya) {};
					\node[mpo] at (2*\dx, \dya) {};
					\node[] at (2*\dx,-0.25) {$x$};
					\node[] at (3*\dx,-0.25) {$x$};
					\pic[] at (0, 0) {actR=0.5*\dya+0.5*\dyb/gh//};
					\pic[] at (\dx, \dya) {fusR=\dyb-\dya/g/h/};
				\end{tikzpicture}.
			\end{equation}
		\end{subequations}
		\label{lemma:action_L}
	\end{lemma}
	\begin{proof}
		These are merely slightly stronger versions of the equations that define the $L$-symbols (see \eqref{eq:defL}). To prove the first one, we first attach an additional MPS tensor to the left, move the leftmost action tensor through two MPS tensors to the left, the second action tensor through the corresponding three MPS tensors, and finally apply the definition of the $L$-symbols. Finally, we remove the tensor we added by applying its inverse which exists because of injectivity. The second equation follows analogously.
	\end{proof}
	Eqs.~\eqref{eq:fusop1}-\eqref{eq:fusop2} can then be seen to hold as a corollary of this lemma, using the definition of the fusion operators and Eq.~\eqref{eq:thm1c}, while Eq.~\eqref{eq:simplecases} is an easy consequence of the assumptions made in Section \eqref{sec:assumptions}.
	Next, we observe that, also in any gauge,
	\begin{equation}
		\left(\lambda^{L}_{h^{-1},g^{-1}}\otimes\lambda^R_{g,h}\right)U^{L-2}_{g}U^{L}_{h}=U^{L}_{gh},
	\end{equation}
	where the relative location of the operators is given as:
	\begin{equation}
		\begin{tikzpicture}
			\tikzset{decoration={snake,amplitude=.4mm,segment length=2mm, post length=0mm,pre length=0mm}}
			
			\draw[ultra thick] (-0.3,0.3)--++(0.0,0.4);
			\draw[thick] (0.3,0.3)--++(0.0,0.4);
			\draw[ultra thick] (-0.3,-0.3)--++(0.0,-1.0);
			\draw[ultra thick] (0.3,-0.3)--++(0.0,-0.4);
			
			\node[rectangle, fill=white] at (0,0) {$\lambda^L_{g,h}$};
			
			\draw[thick,red] (-0.3,-1.2)--++(0.9,0.0);
			\draw[thick,red] (0.3,-0.8)--++(0.3,0.0);
			\draw[thick] (0.3,-0.8)--++(0.0,-0.7);
			\draw[thick] (-0.3,-1.2)--++(0.0,-0.3);
			
			\node[tensor] at (-0.3,-1.2) {};
			\node[tensor] at (0.3,-0.8) {};
			\node[tensor, fill=white] at (0.3, -1.2) {};
			
			\node[] at (0.9,-0.8) {$\ldots$};
			\node[] at (0.9,-1.2) {$\ldots$};
			
			\draw[thick,red] (1.2,-1.2)--++(0.9,0.0);
			\draw[thick,red] (1.2,-0.8)--++(0.3,0.0);
			
			\draw[thick] (1.5,-0.8)--++(0.0,-0.7);
			\draw[thick] (2.1,-1.2)--++(0.0,-0.3);
			\draw[decorate,thick] (1.5,-0.8)--++(0.0,0.5);
			\draw[decorate,thick] (2.1,-1.2)--++(0.0,0.9);
			\draw[thick] (1.5,0.3)--++(0.0,0.4);
			\draw[decorate,thick] (2.1,0.3)--++(0.0,0.4);
			\node[rectangle, fill=white] at (1.8,0.0) {$\lambda^R_{g,h}$};
			
			\node[tensor,fill=white] at (1.5,-1.2) {};
			\node[tensor,fill=white] at (2.1,-1.2) {};
			\node[tensor, fill=white] at (1.5,-0.8) {};
			
			\node[irrep] at (0.1,-0.9) {$g$};
			\node[irrep] at (-0.5,-1.3) {$h$};
			
			\node[] at (2.6,-1.2) {$=$};
			
			\draw[red, thick] (3.0,-1.2)--++(0.9,0.0);
			
			\draw[ultra thick] (3.0,-1.2)--++(0.0,0.3);
			\draw[decorate, thick] (5.4,-1.2)--++(0.0,0.3);
			\draw[thick] (3.0,-1.2)--++(0.0,-0.3);
			\draw[thick] (5.4,-1.2)--++(0.0,-0.3);
			\draw[thick] (3.6,-1.5)--++(0.0,0.6);
			\draw[thick] (4.8,-1.5)--++(0.0,0.6);
			
			\node[tensor] at (3.0,-1.2) {};
			\node[tensor, fill=white] at (3.6,-1.2) {};
			\node[] at (4.2,-1.2) {$\ldots$};
			\draw[red, thick] (4.5,-1.2)--++(0.9,0.0);
			\node[tensor, fill=white] at (4.8,-1.2) {};
			\node[tensor, fill=white] at (5.4,-1.2) {};
			\node[irrep] at (5.1,-1.2) {$gh$};
			
			\node[] at (5.6,-1.2) {$.$};
		\end{tikzpicture}
	\end{equation}
	This implies the unitarity of $\lambda^{L}_{h^{-1},g^{-1}}\otimes\lambda^R_{g,h}$ for all $g,h$, and thus that of $\lambda^R_{g,h}, \lambda^{L}_{g,h}$ up to scalars. To compute this scalar, we take the trace:
	\begin{align}\dfrac{1}{d^2}
		\begin{tikzpicture}
			\def\dx{0.3}
			\def\dya{0.4}
			\def\dyb{1.0}
			\def\dd{0.2}
			\draw[thick, decorate] (\dx,-\dyb) --++ (0,2*\dyb);
			\draw[thick, decorate] (-\dx,-\dya) --++ (0,2*\dya);
			\draw[thick] (-\dx,-\dyb) -- (-\dx,-\dya);
			\draw[thick] (-\dx,\dya) -- (-\dx,\dyb);
			\draw[fill=white] (-\dx-0.1,\dya-\dd) rectangle (\dx+0.1,\dya+\dd);
			\draw[fill=white] (-\dx-0.1,-\dya-\dd) rectangle (\dx+0.1,-\dya+\dd);
			\pic[] at (-\dx,\dyb-0.1) {ddash};
			\pic[] at (-\dx,-\dyb+0.1) {ddash};
			\pic[] at (\dx,\dyb-0.1) {ddash};
			\pic[] at (\dx,-\dyb+0.1) {ddash};
			\node[] at (-\dx-0.5, \dya) {$\lambda^R_{g,h}$};
			\node[] at (-\dx-0.5, -\dya) {${\lambda^R_{g,h}}^\dagger$};
		\end{tikzpicture}
		=\dfrac{1}{d^2}
		\begin{tikzpicture}
			\def\dxa{0.8};
			\def\dxb{1.8};
			\def\dxc{2.3};
			\def\dya{0.6};
			\def\dyb{1.2};
			\def\dyc{1.8};
			\def\dyd{2.2};
			\draw[thick] (0,\dyb) -- (0,\dyd);
			\draw[thick] (\dxa,\dya) -- (\dxa,\dyc);
			\draw[thick, decorate] (0,-\dyb) -- (0,\dyb);
			\draw[thick, decorate] (\dxa,-\dya) -- (\dxa,\dya);
			\draw[thick, decorate] (\dxa,\dyc) -- (\dxa,\dyd);
			\draw[thick, decorate] (\dxa,-\dyc) -- (\dxa,-\dyd);
			\draw[thick] (0,-\dyb) -- (0,-\dyd);
			\draw[thick] (\dxa,-\dya) -- (\dxa,-\dyc);
			\draw[thick, red] (\dxa,\dya) -- (\dxb,\dya);
			\draw[thick, red] (0,\dyb) -- (\dxb,\dyb);
			\draw[thick, red] (\dxa,\dyc) -- (\dxc,\dyc);
			\draw[thick, red] (\dxb,\dya/2 +\dyb/2) -- (\dxc,\dya/2 +\dyb/2 );
			\draw[thick, red] (\dxa,-\dya) -- (\dxb,-\dya);
			\draw[thick, red] (0,-\dyb) -- (\dxb,-\dyb);
			\draw[thick, red] (\dxa,-\dyc) -- (\dxc,-\dyc);
			\draw[thick, red] (\dxb,-\dya/2 -\dyb/2) -- (\dxc,-\dya/2-\dyb/2);
			\draw[thick, red] (\dxc,0.5*\dya+0.5*\dyb) -- (\dxc,\dyc);
			\draw[thick, red] (\dxc,-0.5*\dya-0.5*\dyb) -- (\dxc,-\dyc);
			\pic[] at (\dxb, \dya) {fusR=\dyb-\dya/g/h/};
			%\pic[] at (\dxc, \dya/2+\dyb/2) {fusR=\dyc-\dya*0.5-\dyb*0.5/\inv{(gh)}/gh/};
			\pic[] at (\dxb, -\dya) {fusRdag=-\dyb+\dya/\inv{g}/\inv{h}/};
			%\pic[] at (\dxc, -\dyc) {fusRdag=\dyc-\dya*0.5-\dyb*0.5/\inv{(gh)}/gh/};
			%
			\node[mpo] at (\dxa,\dyb) {};
			\node[mpo] at (\dxa,-\dyb) {};
			\node[square, fill=white] at (\dxc,.25*\dya+.25*\dyb+.5*\dyc) {};
			\node[square, fill=white] at (\dxc,-.25*\dya-.25*\dyb-.5*\dyc) {};
			\node[] at (\dxc+0.4,.25*\dya+.25*\dyb+.5*\dyc) {$\rho_{gh}$};
			\node[] at (\dxc+0.4,-.25*\dya-.25*\dyb-.5*\dyc) {$\rho_{gh}$};
			\node[tensor, fill=white] at (0,\dyb) {};
			\node[tensor, fill=white] at (\dxa,\dya) {};
			\node[tensor, fill=white] at (\dxa,\dyc) {};
			\node[tensor, fill=white] at (0,-\dyb) {};
			\node[tensor, fill=white] at (\dxa,-\dya) {};
			\node[tensor, fill=white] at (\dxa,-\dyc) {};
			\pic[] at (0,\dyd-0.1) {ddash};
			\pic[] at (0,-\dyd+0.1) {ddash};
			\pic[] at (\dxa,\dyd-0.1) {ddash};
			\pic[] at (\dxa,-\dyd+0.1) {ddash};
			\node[] at (-0.4, \dyb) {$X_1^{g}$};
			\node[] at (-0.4, -\dyb) {${X_1^{g}}^\dagger$};
			\node[] at (\dxa-0.4, \dya) {$X_1^{h}$};
			\node[] at (\dxa-0.4, \dyc) {${X_1^{gh}}^\dagger$};
			\node[] at (\dxa-0.4, -\dya) {${X_1^{h}}^\dagger$};
			\node[] at (\dxa-0.4, -\dyc) {$X_1^{gh}$};
			\node[] at (\dxa+0.18, -\dyb+0.15) {$\dagger$};
		\end{tikzpicture}\nonumber\\
		=\dfrac{1}{d^2}\dfrac{\sigma_{g}\sigma_{h}}{\zeta_{\inv{h},\inv{g}}}\begin{tikzpicture}
			\def\dxa{0.5};
			\def\dxb{1.0};
			\def\dxc{1.7};
			\def\dxd{2.3};
			\def\dxe{3.0};
			\def\dxf{3.9};
			\def\dya{0.3};
			\def\dyb{0.6};
			\def\dyc{0.9};
			\def\dyd{1.5};
			\def\dye{2.0};
			\draw[white] (-0.3,0) --++ (0,\dya); %left spacing
			\draw[thick, red] (0,-\dyd) -- (0,\dyd);
			\draw[thick, red] (\dxa,-\dyc) -- (\dxa,\dyc);
			\draw[thick, red] (\dxb,-\dya) -- (\dxb,\dya);
			\draw[thick] (\dxc,-\dye) -- (\dxc,\dye);
			\draw[thick] (\dxd,-\dye) -- (\dxd,\dye);
			\draw[thick, red] (\dxf,\dyb) -- (\dxf,\dyd);
			\draw[thick, red] (\dxf,-\dyd) -- (\dxf,-\dyb);
			\draw[thick, red] (\dxb,\dya) -- (\dxe,\dya);
			\draw[thick, red] (\dxe,\dyb) -- (\dxf,\dyb);
			\draw[thick, red] (\dxa,\dyc) -- (\dxe,\dyc);
			\draw[thick, red] (0,\dyd) -- (\dxf,\dyd);
			\draw[thick, red] (\dxb,-\dya) -- (\dxe,-\dya);
			\draw[thick, red] (\dxe,-\dyb) -- (\dxf,-\dyb);
			\draw[thick, red] (\dxa,-\dyc) -- (\dxe,-\dyc);
			\draw[thick, red] (0,-\dyd) -- (\dxf,-\dyd);
			\draw[white] (0,-\dye-0.2) -- (0.1,-\dye-0.2);
			\draw[white] (0,\dye+0.2) -- (0.1,\dye+0.2);
			\pic[] at (\dxe, \dya) {fusR=\dyc-\dya/g/h/};
			\pic[] at (\dxe, -\dyc) {fusR=\dyc-\dya/\inv{h}/\inv{g}/};
			\pic[] at (\dxf, \dyb) {fusR=\dyd-\dyb/\inv{(gh)}/gh/};
			\pic[] at (\dxf, -\dyd) {fusR=\dyd-\dyb/\inv{(gh)}/gh/};
			\pic[] at (0, -\dyd) {fusL=2*\dyd//\inv{(gh)}/gh};
			\pic[] at (\dxa, -\dyc) {fusL=2*\dyc//g/\inv{g}};
			\pic[] at (\dxb, -\dya) {fusL=2*\dya//h/\inv{h}};
			\foreach \y in {-\dyd,-\dyc,-\dya,\dya,\dyc,\dyd}{
				\node[mpo] at (\dxc,\y) {};
				\node[mpo] at (\dxd,\y) {};}
			\pic[] at (\dxc,\dye-0.1) {ddash};
			\pic[] at (\dxc,-\dye+0.1) {ddash};
			\pic[] at (\dxd,\dye-0.1) {ddash};
			\pic[] at (\dxd,-\dye+0.1) {ddash};
		\end{tikzpicture}
		\nonumber\\
		=\dfrac{1}{d^2}\dfrac{\sigma_{g}\sigma_{h}}{\zeta_{\inv{h},\inv{g}}} \dfrac{\omega_{gh,\inv h,\inv g}\,\omega_{gh,\inv{(gh)},gh}}{\omega_{g,h,h^{-1}}} d^2\nonumber\\
		=\dfrac{\sigma_{g}\sigma_{h}}{\sigma_{gh}}\dfrac{1}{\zeta_{\inv{h},\inv{g}}} \dfrac{\omega_{gh,\inv h,\inv g}}{\omega_{g,h,h^{-1}}}>0,
		\label{eq:trlambda}
	\end{align}
	where $d$ is the physical dimension of the MPU and MPS, and we have used
	\begin{equation}
		\label{eq:red1}
		\begin{tikzpicture}
			\tikzset{decoration={snake,amplitude=.4mm,segment length=2mm, post length=0mm,pre length=0mm}}
			
			\def\sx{0.75};
			\def\dx{0.6};
			\def\dy{0.6};
			
			\draw[thick] (0.0,0.5*\dy) --++ (0.0,0.5*\dy);
			\draw[thick, decorate] (0.0,-0.5*\dy)--++(0.0,\dy);
			\draw[thick] (0.0,-0.5*\dy)--++(0.0,-0.5*\dy);
			\draw[thick,red] (0.0,0.5*\dy) --++ (0.5*\dx,0.0);
			\draw[thick,red] (0.0,-0.5*\dy) --++ (0.5*\dx,0.0);
			\node[tensor, fill=white] at (0.0,0.5*\dy) {};
			\node[tensor, fill=white] at (0.0,-0.5*\dy) {};
			\node[] at (0.0-0.6*\dx, -0.75*\dy) {$X_1^{g\dagger}$};
			\node[] at (0.0-0.5*\dx, 0.75*\dy) {$X_1^g$};
			
			\node[] at (0.0+\sx,0.0) {$=$};
			
			\draw[thick,decorate] (0.0+2*\sx,0.5*\dy) --++ (0.0,\dy);
			\draw[thick] (0.0+2*\sx,-0.5*\dy)--++(0.0,\dy);
			\draw[thick,decorate] (0.0+2*\sx,-0.5*\dy)--++(0.0,-\dy);
			\draw[thick] (0.0+2*\sx, 1.5*\dy) --++(0.0,0.5*\dy);
			\draw[thick] (0.0+2*\sx, -1.5*\dy) --++(0.0,-0.5*\dy);
			\draw[thick,red] (0.0+2*\sx,0.5*\dy) --++(-0.75*\dx,0.0);
			\draw[thick,red] (0.0+2*\sx,-0.5*\dy) --++ (-0.75*\dx,0.0);
			\draw[thick,red] (0.0+2*\sx-0.75*\dx,0.5*\dy)--++(0.0,-\dy);
			\draw[thick,red] (0.0+2*\sx,1.5*\dy)--++(0.5*\dx,0.0);
			\draw[thick,red] (0.0+2*\sx,-1.5*\dy)--++(0.5*\dx,0.0);
			
			\node[tensor, fill=white] at (0.0+2*\sx,0.5*\dy) {};
			\node[tensor,fill=white] at (0.0+2*\sx,-0.5*\dy) {};
			\node[tensor, fill=white] at (0.0+2*\sx, 1.5*\dy) {};
			\node[tensor, fill=white] at (0.0+2*\sx, -1.5*\dy) {};
			\node[] at (0.0+2*\sx-0.6*\dx, -1.75*\dy) {$X_1^{g\dagger}$};
			\node[] at (0.0+2*\sx-0.5*\dx, 1.75*\dy) {$X_1^g$};
			\node[] at (0.0+2*\sx+0.5*\dx, -0.75*\dy) {$Y_1^{g\dagger}$};
			\node[] at (0.0+2*\sx+0.5*\dx, 0.75*\dy) {$Y_1^g$};
			
			\node[] at (0.0+3*\sx,0.0) {$=$};
			
			\draw[thick] (0.0+4*\sx,-\dy)--++(0.0,2*\dy);
			\draw[thick,red] (0.0+4*\sx-0.75*\dx,0.5*\dy)--++(1.25*\dx,0.0);
			\draw[thick,red] (0.0+4*\sx-0.75*\dx,-0.5*\dy)--++(1.25*\dx,0.0);
			\draw[thick,red] (0.0+4*\sx-0.75*\dx,-0.5*\dy)--++(0.0,\dy);
			\node[mpo] at (0.0+4*\sx,0.5*\dy) {};
			\node[mpo] at (0.0+4*\sx,-0.5*\dy) {};
			\node[] at (0.0+4*\sx+0.25*\dx,-0.25*\dy) {$\dagger$};
			\node[] at (0.0+4*\sx+0.5*\dx,0.75*\dy) {$g$};
			\node[] at (0.0+4*\sx+0.5*\dx,-0.25*\dy) {$g$};
			
			\node[] at (0.0+5*\sx,0.0) {$=$};
			
			\draw[thick] (0.0+6*\sx,-\dy)--++(0.0,2*\dy);
			\draw[thick,red] (0.0+6*\sx-0.75*\dx,0.5*\dy)--++(1.75*\dx,0.0);
			\draw[thick,red] (0.0+6*\sx-0.75*\dx,-0.5*\dy)--++(1.75*\dx,0.0);
			\draw[thick,red] (0.0+6*\sx-0.75*\dx,-0.5*\dy)--++(0.0,\dy);
			\node[mpo] at (0.0+6*\sx,0.5*\dy) {};
			\node[mpo] at (0.0+6*\sx,-0.5*\dy) {};
			\node[] at (0.0+6*\sx+0.5*\dx,-0.2*\dy) {$\inv{g}$};
			\node[] at (0.0+6*\sx+0.35*\dx,0.75*\dy) {$g$};
			\node[tensor, red] at (0.0+6*\sx-0.5*\dx,-0.5*\dy) {};
			\node[tensor, red] at (0.0+6*\sx+0.75*\dx,-0.5*\dy) {};
			\node[red] at (0.0+6*\sx-0.5*\dx,-0.9*\dy) {$T_g^\dagger$};
			\node[red] at (0.0+6*\sx+0.75*\dx,-0.9*\dy) {$T_g$}; 
			
			\node[] at (0.0+7.25*\sx,0.0) {$=\, \sigma_g$};
			
			\draw[thick] (0.0+8.25*\sx,-\dy)--++(0.0,2*\dy);
			\draw[thick,red] (0.0+8.25*\sx-0.5*\dx,0.5*\dy)--++(1.5*\dx,0.0);
			\draw[thick,red] (0.0+8.25*\sx-0.5*\dx,-0.5*\dy)--++(1.5*\dx,0.0);
			\draw[fusion] (0.0+8.25*\sx-0.5*\dx, -0.5*\dy) -- (0.0+8.25*\sx-0.5*\dx, 0.5*\dy);
			\node[mpo] at (0.0+8.25*\sx,0.5*\dy) {};
			\node[mpo] at (0.0+8.25*\sx,-0.5*\dy) {};
			\node[] at (0.0+8.25*\sx+0.5*\dx,-0.2*\dy) {$\inv{g}$};
			\node[] at (0.0+8.25*\sx+0.35*\dx,0.75*\dy) {$g$};
			\node[tensor, red] at (0.0+8.25*\sx+0.75*\dx,-0.5*\dy) {};
			\node[red] at (0.0+8.25*\sx+0.75*\dx,-0.9*\dy) {$T_g$}; 
			
		\end{tikzpicture}
	\end{equation}
	as well as Eqs. \eqref{eq:defzeta} and \eqref{eq:wggg}. The resulting quantity is positive since it is the trace of a positive operator.
	
	We now head towards the second part of the proof, showing the existence of a gauge choice that makes the $\lambda$ operators unitary (i.e. the scalar in \eqref{eq:trlambda} equal to 1) while preserving all other conventions. We will first set some notation to ease the computations below, which will involve elements with group indices, labeled $g, h, k, \ldots\in G$ and potentially a block label $x\in\mathsf{X}$, such as
	\begin{equation}
		\chi_g, \gamma_g^x, \beta_{g,h}, L^x_{g,h},\ldots .
	\end{equation}
	Consider a general such function $f_{g_1,\ldots,g_n}^x$. We define
	\begin{equation}
		\hat f^x_{g_1,\ldots,g_n} := f^{g_1\cdot\ldots\cdot g_n\cdot x}_{\inv{g_n}\ldots\inv{g_1}}
	\end{equation}
	which is an involution, $\hat{\hat f}=f$. We also define the coboundary operator for an object with one group index, $f^x_g$,
	\begin{equation}
		[df]^x_{g,h}=\dfrac{f^{hx}_gf^x_h}{f^x_{gh}},
	\end{equation}
	and two group indices, $f^x_{g,h}$,
	\begin{equation}
		[df]^x_{g,h,k}=\dfrac{f^{x}_{g, hk}f^{x}_{h,k}}{f^{kx}_{g,h}f^{x}_{gh,x}},
	\end{equation}
	which in fact satisfies $d(df)=1$. When applied to objects with no block index, these operations are defined similarly, ignoring all block indices (as if $|\mathsf{X}|=1$ with a trivial action). With this notation, we get the highly compact versions of many expressions, for instance, Eqs.~\eqref{eq:inv_omega}, \eqref{eq:Lsymbgauge} and \eqref{eq:omega_and_Ls},
	\begin{equation}
		\hat\omega^*\omega = d\zeta,\qquad
		L \rightarrow L \dfrac{\beta}{d\gamma},\qquad
		\omega=dL.
	\end{equation}
	
	We will also need the following result on the triviality of positive (generalized) cocycles:
	
	\begin{lemma}
		Let $\Omega^x_{g,h}$ be such that $d\Omega =1$ for $G$ being a finite group. Then $\Omega^{|G|}$ is a coboundary, that is, there exist $\chi^x_g$ such that 
		\begin{equation}
			\Omega^{|G|}=d\chi.
		\end{equation}
		\label{lemma:G_cocycle}
	\end{lemma}
	\begin{proof}
		Consider the group algebra $\C G$ as a Hilbert space and define the following operators $X^x_g, g\in G$,
		\begin{equation}
			X^x_g\ket{h}=\Omega^{\inv{h}x}_{g,h}\ket{gh}.
		\end{equation}
		Then it can be checked that $d\Omega=1$ implies $X^{hx}_gX^{x}_h = \Omega^x_{g,h}X^{x}_{gh}$. Thus, by taking the determinants: 
		\begin{equation}
			\left(\Omega^x_{g,h}\right)^{|G|} = \dfrac{\det{X^{hx}_{g}}\det{X^{x}_{g}}}{\det{X^{x}_{gh}}};
		\end{equation}
		and $\Omega^{|G|}=d\chi$ with $\chi^x_g:=\det{X^x_g}$. 
	\end{proof}
	
	\begin{corollary}
		Let $\Omega^x_{g,h}>0$ be such that $d\Omega =1$ for $G$ being a finite group. Then there exists $\chi_g>0$ such that $\Omega=d\chi$. In particular, any positive $2$-cocycle of a finite group is cohomologically trivial. Also, if $\Omega\hat\Omega=1$, $\chi_g$ can be chosen such that $\chi\hat\chi=1$.
		\label{cor:positive_trivial}
	\end{corollary}
	\begin{proof}
		From the previous lemma it follows that $\Omega^{|G|}$ is a coboundary, and since $x\mapsto x^{\frac{1}{|G|}}$ is a bijection on $\R_{>0}$, we have $\Omega=d\chi$ with $\chi$ positive. If $\Omega\hat\Omega=1$,
		\begin{equation}
			\Omega=\dfrac{\sqrt{\Omega}}{\sqrt{\hat\Omega}} = d\left(\sqrt{\dfrac{\chi}{\hat\chi}}\right).
		\end{equation}
		where the square roots do not pose any problem, since we are working with positive numbers. Thus, we can redefine $\chi$ to $\sqrt{\chi/\hat\chi}$, which satisfies $\chi\hat\chi=1$.
	\end{proof}
	As a remark, note that for $|G|<\infty$ these results imply
	\begin{equation}
		H^2(G, \C^*)\cong H^2(G, U(1)) \cong H^2(G, \Z_{|G|}),
	\end{equation}
	with the cyclic group in the last term represented multiplicatively by the $|G|$-th roots of unity. 
	
	Yet another useful lemma is
	\begin{lemma}
		Let $\omega$ be a normalized 3-cocycle, and let
		\begin{equation}
			\Xi_{g,h}:=\dfrac{\omega_{\inv h, \inv g, g}}{\omega_{\inv h \inv g, g, h}}=\dfrac{\omega_{\inv h, \inv g, gh}}{\omega_{\inv g, g, h}},\qquad \sigma_g:=\omega_{g,\inv g, g}.
			\label{eq:introXi}
		\end{equation}
		Then $\hat\omega=\omega\,d\Xi$, with $\Xi=\hat \Xi\,d\sigma$.
		\label{lemma:omegaXi}
	\end{lemma}
	\begin{proof}
		$\hat\omega=\omega\,d\Xi$ follows from combining the 3-cocycle condition of $\omega$, Eq.~\eqref{eq:3-cocycle}, for the tuples $(k^{-1}, h^{-1}, g^{-1}, g)$, $(k^{-1}, h^{-1} g^{-1}, g, h)$ and $(k^{-1}h^{-1}g^{-1}, g, h, k)$, while ${\Xi=\hat \Xi\,d\sigma}$ follows from combining the 3-cocycle conditions for tuples $(h^{-1}g^{-1}, g,h,h^{-1}g^{-1})$, $(h^{-1}, h,h^{-1}, g^{-1})$, $(h^{-1},g^{-1}, g,g^{-1})$ and $(g, g^{-1}, g, g^{-1})$. In both cases, we use the normalization of the cocycle, that is, $\omega_{g, h, k}=1$ whenever either of $g,h,k$ equals the neutral element, which was one of the conditions imposed in Section \ref{sec:assumptions}.
	\end{proof}
	Note that having introduced $\Xi$, Eq. \eqref{eq:trlambda} reads
	\begin{equation}
		d\sigma \dfrac{1}{\hat\zeta}\dfrac{1}{\hat\Xi}>0\implies \zeta\,\Xi\,d\sigma=\zeta\,\hat\Xi>0
		\label{eq:positive_zeta}
	\end{equation}
	where we have made use of $\sigma_g = \sigma_{g^{-1}}=1/\sigma_g$.
	
	Finally, we recall that in the main text we introduced the condition of block independence (BI), for a family of $L$-symbols, as the existence of a gauge where the combinations $\ell_{h_1,h_2;g}^x$ given by \eqref{eq:scalarfactor} are independent of $x$. In Appendix~\ref{app:block_indep} we show that this condition is equivalent to the stronger constraint of the existence of a gauge where $L=1$, i.e., all $L$-symbols are 1. This can be seen as a constraint on the same footing as the system being nonanomalous (NA), which corresponds to the existence of a gauge where $\omega = 1$. BI is strictly stronger than NA, as the example in Section~\ref{example:z4z2} shows.
	
	For the final part of the proof, we will show that we can find gauge transformations 
	\begin{equation}
		F^>_{g,h}\to \beta_{g,h}F^>_{g,h},\qquad A^\ammaG_{g,x}\to \gamma^x_g A^\ammaG_{g,x}, 
	\end{equation}
	such that, in the new gauge,
	\begin{enumerate}
		\item $\omega = \zeta = L = 1$, if BI is satisfied,
		\item $\omega = \zeta = 1,  |L| = 1$, if BI does not hold, but NA does,
		\item $|\omega| = |\zeta| = |L| = 1$, if NA does not hold.
	\end{enumerate}
	Thus, in all cases, the scalar in \eqref{eq:trlambda} equals $1$. These gauge transformations should also be constrained to satisfy
	\begin{align}
		\beta_{g,e} = \beta_{e,g} = \beta_{g,g^{-1}} &= 1,\qquad \forall g,\label{eq:lig1}\\
		\gamma^x_{e} = \gamma^x_{g}\gamma^{gx}_{g^{-1}} &= 1,\qquad \forall g,x,\label{eq:lig2}
	\end{align}
	so that our assumptions from Section \ref{sec:assumptions} and Proposition \ref{prop:def_tens_anomalous} are preserved. Note that the last condition, in our compact notation, reads $\gamma\hat\gamma=1$.
	
	\textit{Case 1 (BI)}: From block independence follows the existence of gauge transformations $\beta, \gamma$ such that $L=1$ (and hence $\omega=1$). This does not immediately imply that they satisfy the constraints \eqref{eq:lig1}-\eqref{eq:lig2}. However, we know that they preserve the $L$-symbols $L_{e,g}^x,L_{g,e}^x,L_{g,\inv{g}}^x$, which were already set to 1 (recall Section~\ref{sec:assumptions}). Thus we deduce
	\begin{equation}
		\dfrac{\beta_{g,e}}{\gamma^x_e} = \dfrac{\beta_{e,g}}{\gamma^{gx}_e} = \dfrac{\beta_{g,\inv{g}}\gamma^x_e}{\gamma_{\inv g}^{gx}\gamma^x_g}=1,\qquad\forall g,x,
	\end{equation}
	that is,
	\begin{equation}
		\beta_{g,e} = \beta_{e,g} = \gamma^{gx}_e = \beta_{e,e},\quad \beta_{g,\inv{g}}=\beta_{\inv{g},g}=\dfrac{\gamma_{\inv g}^{gx}\gamma^x_g}{\beta_{e,e}}.
	\end{equation}
	We then define $\theta_g:=(\beta_{e,e}\beta_{g,g^{-1}})^{\frac{1}{2}}$, with an adequate choice of square root so that $\theta_e=\beta_{e,e}$ and $\theta_g\theta_{\inv g}=\beta_{g,\inv{g}}\beta_{e,e}$. It is then easy to check that we can redefine $\beta,\gamma$ as $\beta/d\theta, \gamma/\theta$ giving rise to the same $L$-symbols and satisfying conditions \eqref{eq:lig1}-\eqref{eq:lig2}. Applying this gauge transformation leaves us in a gauge where, by \eqref{eq:positive_zeta}, $\zeta>0$, while \eqref{eq:zeta} and \eqref{eq:inv_omega} read
	\begin{equation}
		\zeta\hat\zeta=1,\qquad d\zeta =1.
	\end{equation}
	Thus, by Corollary~\ref{cor:positive_trivial}, there exists $\chi>0, \chi\hat\chi=1$ such that $\zeta=d\chi$. Then it can be checked that $\beta=(d\chi)^{\frac{1}{2}}, \gamma=\chi^{\frac{1}{2}}$ is a gauge transformation that preserves all the previous properties and makes $\zeta=1$.
	
	\textit{Case 2 (NA, not BI)}:
	Due to the absence of an anomaly, there exists a gauge transformation $\beta$ of the fusion tensors that trivializes $\omega$. This means that it preserves $\omega_{g, e, e}$, $\omega_{e, e, g}$, and $\omega_{g, \inv{g}, g}$, which are all equal to $1$ for our choices for the fusion tensors. This implies that $\beta_{g,e}=\beta_{g,e}=\beta_{e,e}$ and $\beta_{g,\inv{g}}=\beta_{\inv g,g}$. Defining $\theta$ exactly as above, we see that $\beta/d\theta$ is a gauge transformation that preserves our choices and still leads to $\omega=1$. In this new gauge, we have
	\begin{align}
		d\zeta&=1, & \zeta&>0, & \zeta\hat\zeta &=1,\\
		d|L|&=1, & |L|&>0, & |L||\hat L| &=1,\label{eq:absLcocycle}
	\end{align}
	which imply
	\begin{align}
		\zeta&=d\chi, &\chi&>0,& \chi\hat\chi&=1,\\
		|L|&=d\eta, &\eta&>0,& \eta\hat\eta&=1.\label{eq:absLtrivial}
	\end{align}
	The gauge transformation $\beta=(d\chi)^{\frac{1}{2}}, \gamma=\chi^{\frac{1}{2}}\eta$ then satisfies all constraints and leads to the desired outcome.
	
	\textit{Case 3 (anomalous)}:
	From Proposition \ref{prop:zetas} and Lemma \ref{lemma:omegaXi} it follows that 
	\begin{equation}
		|\omega|=d\left|\dfrac{\zeta}{\Xi}\right|^{\frac{1}{2}},
	\end{equation}
	thus it can be checked that the gauge transformation $\beta=\left|\dfrac{\Xi}{\zeta}\right|^{\frac{1}{2}}$ leaves $|\omega|=|\zeta|=1$. Thus in this new gauge, Eqs.~\eqref{eq:absLcocycle} hold and imply Eqs.~\eqref{eq:absLtrivial}. The gauge transformation $\gamma=\eta$ finishes the job without altering what we previously achieved.

	\section{Associativity of the fusion tensors}
	\label{app:associativity}
	In this Appendix, we demonstrate that the obstruction to the associativity of the fusion operators $\lambda^R_{g,h}$ defined in \eqref{eq:lambda_R}, is given by a scalar function $F$, i.e. they satisfy Eq.~\eqref{eq:assoc}. We will prove this, thanks to the unitarity of the fusion tensors, by computing
	\begin{equation}
		\label{eq:Fghk}
		(\I\otimes \lambda^R_{gh,k})(\lambda^R_{g,h}\otimes \I)\bigl[(\I\otimes \lambda^R_{g,hk})(w_R^g\otimes \I)(\I\otimes \lambda^R_{h,k})\bigr]^\dagger,
	\end{equation}
	and showing that it is a scalar multiple of the identity. We will use the following representation of the fusion operators
	\begin{equation}
		\lambda^R_{g,h}=\sigma_{gh}~
		\begin{tikzpicture}
			\def\dxa{1.2}
			\def\dxb{1.8}
			\def\dxc{2.4}
			\def\dya{0.4}
			\def\dyb{1.0}
			\def\dyc{1.6}
			\def\dyd{2.0}
			\tikzset{decoration={snake,amplitude=.4mm,segment length=2mm, post length=0mm,pre length=0mm}}
			\draw[thick, red] (\dxa,\dya)--++(\dxb-\dxa,0.0);
			\draw[thick, red] (0,\dyb)--++(\dxb,0.);
			\draw[thick, red] (\dxa,\dyc)--++(\dxc-\dxa,0.);
			\draw[thick, red] (\dxb,0.5*\dya+0.5*\dyb)--++(\dxc-\dxb,0.);
			\draw[thick] (\dxa,\dya)--++(0,\dyc-\dya);
			\draw[thick] (0,\dyb)--++(0,\dyd-\dyb);
			\draw[decorate, thick] (\dxa,0) --++ (0, \dya);
			\draw[decorate, thick] (\dxa,\dyc) --++ (0., \dyd-\dyc);
			\draw[decorate, thick] (0,0) --++ (0., \dyb);
			\pic () at (\dxb,\dya) {fusR=\dyb-\dya/g/h/};
			\pic () at (\dxc,0.5*\dya+0.5*\dyb) {fusR=\dyc-0.5*\dya-0.5*\dyb/\inv{(gh)}/gh/};
			\node[tensor, fill=white] at (0.,\dyb) {};
			\node[tensor, fill=white] at (\dxa,\dya) {};
			\node[tensor] at (\dxa,\dyc) {};
			\node[mpo] at (\dxa,\dyb) {};
			\node[] at (-0.32, \dyb) {$X_1^g$};
			\node[] at (\dxa-0.32, \dya) {$X_1^h$};
			\node[] at (\dxa-0.5, \dyc+0.1) {$Y^{\inv{(gh)}}_2$};
		\end{tikzpicture}
		\label{eq:lambda_R_new}
	\end{equation}
	where we have used Eqs.~\eqref{eq:modif_XY} and \eqref{eq:Fgg}. Note that in this appendix 
	we will unify the notation of the squiggle legs of $X_1, Y_1$ and the thick legs of $X_2, Y_2$ and show them all as squiggles since we will be contracting them together using
	\begin{equation}
		\begin{tikzpicture}
			\tikzset{decoration={snake,amplitude=.4mm,segment length=2mm, post length=0mm,pre length=0mm}}
			
			\def\sx{1.0};
			\def\dx{0.8};
			\def\dy{0.6};
			
			\draw[thick] (0.0,0.5*\dy) --++ (0.0,0.5*\dy);
			\draw[thick, decorate] (0.0,-0.5*\dy)--++(0.0,\dy);
			\draw[thick] (0.0,-0.5*\dy)--++(0.0,-0.5*\dy);
			\draw[thick,red] (0.0,0.5*\dy) --++ (0.5*\dx,0.0);
			\draw[thick,red] (0.0,-0.5*\dy) --++ (0.5*\dx,0.0);
			\node[tensor, fill=white] at (0.0,0.5*\dy) {};
			\node[tensor] at (0.0,-0.5*\dy) {};
			\node[] at (0.0-0.6*\dx, -0.75*\dy) {$Y_2^{\inv{g}}$};
			\node[] at (0.0-0.5*\dx, 0.75*\dy) {$X_1^g$};
			
			\node[] at (-0.2+\sx,0.0) {$=$};
			
			\draw[thick] (0.0+2*\sx,-\dy)--++(0.0,2*\dy);
			\draw[thick,red] (0.0+2*\sx-0.75*\dx,0.5*\dy)--++(1.25*\dx,0.0);
			\draw[thick,red] (0.0+2*\sx-0.75*\dx,-0.5*\dy)--++(1.25*\dx,0.0);
			\pic () at (0.0+2*\sx-0.75*\dx,-0.5*\dy) {fusL=\dy//g/\inv{g}};
			\node[mpo] at (0.0+2*\sx,0.5*\dy) {};
			\node[mpo] at (0.0+2*\sx,-0.5*\dy) {};
			
			\node[] at (0.0+2.75*\sx,0.0) {$,$};
		\end{tikzpicture}
		\label{eq:XYvsMPO}
	\end{equation}
	which follows from Eqs.~\eqref{eq:modif_XY} and \eqref{eq:red1}. The fusion tree $(\I\otimes \lambda^R_{gh,k})(\lambda^R_{g,h}\otimes \I)$ can then be expressed as
	\begin{equation}
		\begin{tikzpicture}
			\tikzset{decoration={snake,amplitude=.4mm,segment length=2mm, post length=0mm,pre length=0mm}}
			
			\def\dx{0.6};
			\def\dy{0.6};
			
			\draw[thick,red] (0.0,0.0)--++(0.75*\dx,0.0);
			\draw[thick,red] (0.0,\dy)--++(0.75*\dx,0.0);
			\draw[thick] (0.0,0.0)--++(0.0,\dy);
			\draw[thick] (0.0,\dy)--++(0.0,\dy);
			\draw[thick,decorate] (0.0,2*\dy)--++(0.0,\dy);
			\draw[thick,red] (0.0,2*\dy)--++(1.5*\dx,0.0);
			\draw[thick,red] (1.5*\dx,2*\dy)--++(0.0,-1.5*\dy);
			\draw[thick,red] (1.5*\dx,0.5*\dy)--++(-0.75*\dx,0.0);
			\draw[thick,red] (0.0,\dy)--++(-2*\dx,0.0);
			\draw[thick] (-2*\dx,\dy)--++(0.0,2*\dy);
			\draw[thick,decorate] (-2*\dx,\dy)--++(0.0,-\dy);
			\draw[thick] (-2*\dx,0.0)--++(0.0,-2*\dy);
			\draw[thick,decorate] (-2*\dx,-2*\dy)--++(0.0,-\dy);
			\draw[thick,red] (-3*\dx,-\dy)--++(2*\dx,0.0);
			\draw[thick,red] (-2*\dx,-2*\dy)--++(\dx,0.0);
			\draw[thick,red] (-2*\dx,0.0)--++(1.5*\dx,0.0);
			\draw[thick,red] (-0.5*\dx,0.0)--++(0.0,-1.5*\dy);
			\draw[thick,red] (-\dx,-1.5*\dy)--++(0.5*\dx,0.0);
			\draw[thick,decorate] (0.0,0.0)--++(0.0,-3*\dy);
			\draw[thick] (-3*\dx,-\dy)--++(0.0,4*\dy);
			\draw[thick,decorate] (-3*\dx,-\dy)--++(0.0,-2*\dy);

			\draw[fusion] (0.75*\dx,0.0)--++(0.0,\dy);
			\draw[fusion] (-\dx,-2*\dy)--++(0.0,\dy);
			\pic () at (1.5*\dx,0.5*\dy) {fusR=1.5*\dy///};
			\pic () at (-0.5*\dx,-1.5*\dy) {fusR=1.5*\dy///};
			
			\node[tensor] at (0.0,2*\dy) {};
			\node[mpo] at (0.0,\dy) {};
			\node[tensor, fill=white] at (0.0,0.0) {};
			\node[tensor,fill=white] at (-2*\dx,\dy) {};
			\node[tensor] at (-2*\dx,0.0) {};
			\node[mpo] at (-2*\dx,-\dy) {};
			\node[tensor, fill=white] at (-2*\dx,-2*\dy) {};
			\node[tensor, fill=white] at (-3*\dx, -\dy) {};
			
			\node[] at (-4.5*\dx,0.0) {$\sigma_{gh}\sigma_{ghk}$};
			\node[] at (-1.5*\dx,1.35*\dy) {$X_1^{gh}$};
			\node[] at (-1.1*\dx,0.4*\dy) {$Y_2^{\inv{(gh)}}$};
			
		\end{tikzpicture}
	\end{equation}
	and, using \eqref{eq:XYvsMPO}, simplified into
	\begin{equation}
		\begin{tikzpicture}
			\tikzset{decoration={snake,amplitude=.4mm,segment length=2mm, post length=0mm,pre length=0mm}}
			
			\def\sx{0.75};
			\def\dx{0.6};
			\def\dy{0.6};
			
			\draw[thick,red] (0.0,0.0)--++(0.75*\dx,0.0);
			\draw[thick,red] (0.0,\dy)--++(0.75*\dx,0.0);
			\draw[thick] (0.0,0.0)--++(0.0,\dy);
			\draw[thick] (0.0,\dy)--++(0.0,\dy);
			\draw[thick,decorate] (0.0,2*\dy)--++(0.0,\dy);
			\draw[thick,red] (0.0,2*\dy)--++(1.5*\dx,0.0);
			\draw[thick,red] (1.5*\dx,2*\dy)--++(0.0,-1.5*\dy);
			\draw[thick,red] (1.5*\dx,0.5*\dy)--++(-0.75*\dx,0.0);
			\draw[thick,decorate] (0.0,0.0)--++(0.0,-3*\dy);
			\draw[thick,red] (0.0,\dy)--++(-2*\dx,0.0);
			\draw[thick] (-2*\dx,\dy)--++(0.0,2*\dy);
			\draw[thick] (-2*\dx,\dy)--++(0.0,-\dy);
			\draw[thick] (-2*\dx,0.0)--++(0.0,-2*\dy);
			\draw[thick,red] (-2*\dx,0.0) --++ (-0.5*\dx,0.0);
			\draw[thick,red] (-2*\dx,\dy)--++(-0.5*\dx,0.0);
			\draw[thick,red] (-2.5*\dx,0.0)--++(0.0,\dy);
			\draw[thick,decorate] (-2*\dx,-2*\dy)--++(0.0,-\dy);
			\draw[thick,red] (-3*\dx,-\dy)--++(2*\dx,0.0);
			\draw[thick,red] (-2*\dx,-2*\dy)--++(\dx,0.0);
			\draw[thick,red] (-2*\dx,0.0)--++(1.5*\dx,0.0);
			\draw[thick,red] (-0.5*\dx,0.0)--++(0.0,-1.5*\dy);
			\draw[thick,red] (-\dx,-1.5*\dy)--++(0.5*\dx,0.0);
			\draw[thick] (-3*\dx,-\dy)--++(0.0,4*\dy);
			\draw[thick,decorate] (-3*\dx,-\dy)--++(0.0,-2*\dy);
			
			\draw[fusion] (0.75*\dx,0.0)--++(0.0,\dy);
			\draw[fusion] (-\dx,-2*\dy)--++(0.0,\dy);
			\pic () at (1.5*\dx,0.5*\dy) {fusR=1.5*\dy///};
			\pic () at (-0.5*\dx,-1.5*\dy) {fusR=1.5*\dy///};
			\pic () at (-2.5*\dx,0) {fusL=\dy///};
			
			\node[tensor] at (0.0,2*\dy) {};
			\node[mpo] at (0.0,\dy) {};
			\node[tensor, fill=white] at (0.0,0.0) {};
			\node[mpo] at (-2*\dx,0.0) {};
			\node[mpo] at (-2*\dx,\dy) {};
			\node[mpo] at (-2*\dx,-\dy) {};
			\node[tensor,fill=white] at (-2*\dx,-2*\dy) {};
			\node[tensor,fill=white] at (-3*\dx,-\dy) {};
			
			\node[] at (-4.5*\dx,0.0) {$\sigma_{gh}\sigma_{ghk}$};
			
			\node[] at (2*\dx,0.0) {$.$};
		\end{tikzpicture}
	\end{equation}
	Similarly, we write $(\I\otimes \lambda^R_{g,hk})(w_R^g\otimes \I)(\I\otimes \lambda^R_{h,k})$ in the form
	
	\begin{equation}
		\begin{tikzpicture}
			\tikzset{decoration={snake,amplitude=.4mm,segment length=2mm, post length=0mm,pre length=0mm}}
			
			\def\sx{0.75};
			\def\dx{0.8};
			\def\dy{0.6};
			
			\draw[thick,decorate] (0.0,-0.5*\dy)--++(0.0,\dy);
			\draw[thick] (0.0,0.5*\dy)--++(0.0,2*\dy);
			\draw[thick,decorate] (0.0,2.5*\dy)--++(0.0,\dy);
			\draw[thick] (0.0,-0.5*\dy)--++(0.0,-2*\dy);
			\draw[thick, decorate] (0.0,-2.5*\dy)--++(0.0,-\dy);
			
			\draw[thick] (-\dx,3.5*\dy)--++(0.0,-2*\dy);
			\draw[thick,decorate] (-\dx,1.5*\dy)--++(0.0,-1.5*\dy);
			\draw[thick] (-\dx,0.0)--++(0.0,-1.5*\dy);
			\draw[thick,decorate] (-\dx,-1.5*\dy)--++(0.0,-2*\dy);
			
			\draw[thick] (-2*\dx,0.0)--++(0.0,3.5*\dy);
			\draw[thick,decorate] (-2*\dx,0.0)--++(0.0,-3.5*\dy);
			
			\draw[thick,red] (-2*\dx,0.0)--++(\dx,0.0);
			\draw[thick,red] (-\dx,1.5*\dy)--++(1.75*\dx,0.0);
			\draw[thick,red] (-\dx,-1.5*\dy)--++(1.75*\dx,0.0);
			\draw[thick,red] (0.0,-2.5*\dy)--++(0.75*\dx,0.0);
			\draw[thick,red] (0.0,-0.5*\dy)--++(1.5*\dx,0.0);
			\draw[thick,red] (1.5*\dx,-0.5*\dy)--++(0.0,-1.5*\dy);
			\draw[thick,red] (1.5*\dx,-2*\dy)--++(-0.75*\dx,0.0);
			\draw[thick,red] (0.0,0.5*\dy)--++(0.75*\dx,0.0);
			\draw[thick,red] (0.75*\dx,\dy)--++(0.75*\dx,0.0);
			\draw[thick,red] (1.5*\dx,\dy)--++(0.0,1.5*\dy);
			\draw[thick,red] (0.0,2.5*\dy)--++(1.5*\dx,0.0);
			
			\draw[fusion] (0.75*\dx,-2.5*\dy)--++(0.0,\dy);
			\draw[fusion] (0.75*\dx,1.5*\dy)--++(0.0,-\dy);
			\pic () at (1.5*\dx,-2*\dy) {fusR=1.5*\dy///};
			\pic () at (1.5*\dx,\dy) {fusR=1.5*\dy///};
			
			\node[tensor, fill=white] at (-2*\dx,0.0) {};
			\node[tensor, fill=white] at (-\dx,0.0) {};
			\node[tensor, fill=white] at (-\dx,1.5*\dy) {};
			\node[tensor, fill=white] at (-\dx,-1.5*\dy) {};
			\node[tensor,fill=white] at (0.0,-2.5*\dy) {};
			\node[mpo] at (0.0,-1.5*\dy) {};
			\node[tensor] at (0.0,-0.5*\dy) {};
			\node[tensor, fill=white] at (0.0,0.5*\dy) {};
			\node[mpo] at (0.0,1.5*\dy) {};
			\node[tensor] at (0.0,2.5*\dy) {};
			
			\node[] at (-3.5*\dx,0.0) {$\sigma_{ghk}\sigma_{hk}$};
			\node[] at (-0.5*\dx,0.75*\dy) {$X_1^{hk}$};
			\node[] at (0.85*\dx,-0.15*\dy) {$Y_2^{\inv{(hk)}}$};
			
		\end{tikzpicture}
	\end{equation}
	and again use \eqref{eq:XYvsMPO} to express it as
	\begin{equation}
		\begin{tikzpicture}
			\tikzset{decoration={snake,amplitude=.4mm,segment length=2mm, post length=0mm,pre length=0mm}}
			
			\def\sx{0.75};
			\def\dx{0.8};
			\def\dy{0.6};
			
			\draw[thick] (0.0,-0.5*\dy)--++(0.0,\dy);
			\draw[thick] (0.0,0.5*\dy)--++(0.0,2*\dy);
			\draw[thick,decorate] (0.0,2.5*\dy)--++(0.0,\dy);
			\draw[thick] (0.0,-0.5*\dy)--++(0.0,-2*\dy);
			\draw[thick, decorate] (0.0,-2.5*\dy)--++(0.0,-\dy);
			
			\draw[thick] (-\dx,3.5*\dy)--++(0.0,-2*\dy);
			\draw[thick] (-\dx,1.5*\dy)--++(0.0,-1.5*\dy);
			\draw[thick] (-\dx,0.0)--++(0.0,-1.5*\dy);
			\draw[thick,decorate] (-\dx,-1.5*\dy)--++(0.0,-2*\dy);
			
			\draw[thick] (-2*\dx,1.5*\dy)--++(0.0,2*\dy);
			\draw[thick,decorate] (-2*\dx,1.5*\dy)--++(0.0,-5*\dy);
			
			\draw[thick,red] (-2*\dx,1.5*\dy)--++(\dx,0.0);
			\draw[thick,red] (-\dx,1.5*\dy)--++(1.75*\dx,0.0);
			\draw[thick,red] (-\dx,-1.5*\dy)--++(1.75*\dx,0.0);
			\draw[thick,red] (0.0,-2.5*\dy)--++(0.75*\dx,0.0);
			\draw[thick,red] (0.0,-0.5*\dy)--++(1.5*\dx,0.0);
			% \draw[thick,red] (1.5*\dx,-0.5*\dy)--++(0.0,-1.5*\dy);
			\draw[thick,red] (1.5*\dx,-2*\dy)--++(-0.75*\dx,0.0);
			\draw[thick,red] (0.0,0.5*\dy)--++(0.75*\dx,0.0);
			\draw[thick,red] (0.75*\dx,\dy)--++(0.75*\dx,0.0);
			% \draw[thick,red] (1.5*\dx,\dy)--++(0.0,1.5*\dy);
			\draw[thick,red] (0.0,2.5*\dy)--++(1.5*\dx,0.0);
			\draw[thick,red] (0.0,0.5*\dy)--++(-0.5*\dx,0.0);
			\draw[thick,red] (0.0,-0.5*\dy)--++(-0.5*\dx,0.0);
			% \draw[thick,red] (-0.5*\dx,0.5*\dy)--++(0.0,-\dy);

			\draw[fusion] (0.75*\dx,-2.5*\dy)--++(0.0,\dy);
			\draw[fusion] (0.75*\dx,1.5*\dy)--++(0.0,-\dy);
			\pic () at (1.5*\dx,-2*\dy) {fusR=1.5*\dy///};
			\pic () at (1.5*\dx,\dy) {fusR=1.5*\dy///};
			\pic () at (-0.5*\dx,-0.5*\dy) {fusL=\dy///};
			
			\node[tensor, fill=white] at (-2*\dx,1.5*\dy) {};
			\node[mpo] at (-\dx,1.5*\dy) {};
			\node[tensor, fill=white] at (-\dx,-1.5*\dy) {};
			\node[tensor,fill=white] at (0.0,-2.5*\dy) {};
			\node[mpo] at (0.0,-1.5*\dy) {};
			\node[mpo] at (0.0,-0.5*\dy) {};
			\node[mpo] at (0.0,0.5*\dy) {};
			\node[mpo] at (0.0,1.5*\dy) {};
			\node[tensor] at (0.0,2.5*\dy) {};
			
			\node[] at (-3.5*\dx,0.0) {$\sigma_{ghk}\sigma_{hk}$};
			\node[] at (2.0*\dx,0.0) {$.$};

		\end{tikzpicture}
	\end{equation}
	We now compute the expression \eqref{eq:Fghk}. After every step, we will denote the successive scalar factors by $\mathcal{F}^i_{g,h,k}, i=1,2\ldots$. Note that we need not keep track of them at this point since they can be easily computed a posteriori by acting with the two unitaries in \eqref{eq:assoc} on a suitable MPS and using \eqref{eq:fusop2}. Nevertheless, we will, just as a check. We start by writing \eqref{eq:Fghk} as
	\begin{equation}
		\begin{tikzpicture}
			\tikzset{decoration={snake,amplitude=.4mm,segment length=2mm, post length=0mm,pre length=0mm}}
			
			\def\sx{0.75};
			\def\dx{0.8};
			\def\dy{0.6};
			
			\draw[thick,red] (0.0,0.0)--++(0.75*\dx,0.0);
			\draw[thick,red] (0.0,\dy)--++(0.75*\dx,0.0);
			\draw[thick] (0.0,0.0)--++(0.0,\dy);
			\draw[thick] (0.0,\dy)--++(0.0,\dy);
			\draw[thick,decorate] (0.0,2*\dy)--++(0.0,\dy);
			\draw[thick,red] (0.0,2*\dy)--++(1.5*\dx,0.0);
			% \draw[thick,red] (1.5*\dx,2*\dy)--++(0.0,-1.5*\dy);
			\draw[thick,red] (1.5*\dx,0.5*\dy)--++(-0.75*\dx,0.0);
			\draw[thick,decorate] (0.0,0.0)--++(0.0,-3*\dy);
			\draw[thick,red] (0.0,\dy)--++(-2*\dx,0.0);
			\draw[thick] (-2*\dx,\dy)--++(0.0,2*\dy);
			\draw[thick] (-2*\dx,\dy)--++(0.0,-\dy);
			\draw[thick] (-2*\dx,0.0)--++(0.0,-2*\dy);
			\draw[thick,red] (-2*\dx,0.0) --++ (-0.5*\dx,0.0);
			\draw[thick,red] (-2*\dx,\dy)--++(-0.5*\dx,0.0);
			% \draw[thick,red] (-2.5*\dx,0.0)--++(0.0,\dy);
			\draw[thick,decorate] (-2*\dx,-2*\dy)--++(0.0,-2*\dy);
			\draw[thick,red] (-3*\dx,-\dy)--++(2*\dx,0.0);
			\draw[thick,red] (-2*\dx,-2*\dy)--++(\dx,0.0);
			\draw[thick,red] (-2*\dx,0.0)--++(1.5*\dx,0.0);
			% \draw[thick,red] (-0.5*\dx,0.0)--++(0.0,-1.5*\dy);
			\draw[thick,red] (-\dx,-1.5*\dy)--++(0.5*\dx,0.0);
			\draw[thick] (-3*\dx,-\dy)--++(0.0,4*\dy);
			\draw[thick,decorate] (-3*\dx,-\dy)--++(0.0,-6*\dy);
			
			\draw[thick] (0.0,-3*\dy)--++(0.0,-5*\dy);
			\draw[thick, decorate] (0.0,-8*\dy)--++(0.0,-\dy);
			\draw[thick] (-3*\dx,-7*\dy)--++(0.0,-2*\dy);
			\draw[thick] (-2*\dx,-4*\dy)--++(0.0,-5*\dy);
			
			\draw[thick,red] (0.0,-3*\dy)--++(0.75*\dx,0.0);
			\draw[thick,red] (0.0,-4*\dy)--++(0.75*\dx,0.0);
			\draw[thick,red] (0.75*\dx,-3.5*\dy)--++(0.75*\dx,0.0);
			\draw[thick,red] (0.0,-4*\dy)--++(-2*\dx,0.0);
			% \draw[thick,brightube] (1.5*\dx,-3.5*\dy)--++(0.0,-1.5*\dy);
			\draw[thick,red] (-0.5*\dx,-5*\dy)--++(2*\dx,0.0);
			% \draw[thick,red] (-0.5*\dx,-5*\dy)--++(0.0,-\dy);
			\draw[thick,red] (-0.5*\dx,-6*\dy)--++(1.25*\dx,0.0);
			\draw[thick,red] (-3*\dx,-7*\dy)--++(3.75*\dx,0.0);
			\draw[thick,red] (0.75*\dx,-6.5*\dy)--++(0.75*\dx,0.0);
			% \draw[thick,brightube] (1.5*\dx,-6.5*\dy)--++(0.0,-1.5*\dy);
			\draw[thick,red] (1.5*\dx,-8*\dy)--++(-1.5*\dx,0.0);
			
			\draw[fusion] (0.75*\dx,0.0)--++(0.0,\dy);
			\draw[fusion] (-\dx,-2*\dy)--++(0.0,\dy);
			\draw[fusion_dagger,brightube] (0.75*\dx,-3*\dy)--++(0.0,-\dy);
			\draw[fusion_dagger,brightube] (0.75*\dx,-6*\dy)--++(0.0,-\dy);
			\pic () at (1.5*\dx,0.5*\dy) {fusR=1.5*\dy///};
			\pic () at (1.5*\dx,-5*\dy) {fusRdag=1.5*\dy///};
			\pic () at (1.5*\dx,-8*\dy) {fusRdag=1.5*\dy///};
			\pic () at (-0.5*\dx,-6*\dy) {fusLdag=\dy///};
			\pic () at (-0.5*\dx,-1.5*\dy) {fusR=1.5*\dy///};
			\pic () at (-2.5*\dx,0) {fusL=\dy///};
			
			\node[tensor] at (0.0,2*\dy) {};
			\node[mpo] at (0.0,\dy) {};
			\node[tensor, fill=white] at (0.0,0.0) {};
			\node[mpo] at (-2*\dx,0.0) {};
			\node[mpo] at (-2*\dx,\dy) {};
			\node[mpo] at (-2*\dx,-\dy) {};
			\node[tensor,fill=white] at (-2*\dx,-2*\dy) {};
			\node[tensor,fill=white] at (-3*\dx,-\dy) {};
			\node[tensor, fill=white] at (0.0,-3*\dy) {};
			\node[mpo] at (0.0,-4*\dy) {};
			\node[mpo] at (0.0,-5*\dy) {};
			\node[mpo] at (0.0,-6*\dy) {};
			\node[mpo] at (0.0,-7*\dy) {};
			\node[tensor] at (0.0,-8*\dy) {};
			\node[tensor,fill=white] at (-3*\dx,-7*\dy) {};
			\node[tensor,fill=white] at (-2*\dx,-4*\dy) {};
			\node[mpo] at (-2*\dx,-7*\dy) {};

			\node[] at (-4.5*\dx,-3.5*\dy) {$\mathcal{F}^1_{g,h,k}$};
			
			\node[] at (-3.65*\dx,-7*\dy) {$X_1^{g\dagger}$};
			\node[] at (-3.5*\dx,-\dy) {$X_1^{g}$};
			\node[] at (-1.5*\dx,-4.35*\dy) {$X_1^{h\dagger}$};
			\node[] at (-1.5*\dx,-2.35*\dy) {$X_1^h$};
			\node[] at (-0.6*\dx,-3*\dy) {$X_1^{k\dagger}$};
			\node[] at (-1.75*\dx,-6.75*\dy) {$\dagger$};
			\node[] at (0.25*\dx,-6.75*\dy) {$\dagger$};
			\node[] at (0.25*\dx,-5.75*\dy) {$\dagger$};
			\node[] at (0.25*\dx,-4.75*\dy) {$\dagger$};
			\node[] at (0.25*\dx,-3.75*\dy) {$\dagger$};
			\node[] at (-0.5*\dx,-3.75*\dy) {$h$};
			\node[] at (0.95*\dx,-4.75*\dy) {$\inv{(hk)}$};
			\node[] at (0.7*\dx,-5.6*\dy) {$hk$};
			\node[] at (0.4*\dx,-7.3*\dy) {$g$};
			\node[] at (1.3*\dx,-8.5*\dy) {$\bigl(Y_2^{\inv{(ghk)}}\bigr)^\dagger$};
			\node[] at (0.5*\dx,-0.4*\dy) {$X_1^k$};
			\node[] at (\dx,2.4*\dy) {$Y_2^{\inv{(ghk)}}$};

		\end{tikzpicture}
	\end{equation}
	with $\mathcal{F}^1_{g,h,k}=\sigma_{gh}\sigma_{ghk}\sigma_{hk}\sigma_{ghk}=\sigma_{gh}\sigma_{hk}$. Using \eqref{eq:red1}, the above expression can be rephrased as
	
	\begin{equation}
		\begin{tikzpicture}
			\tikzset{decoration={snake,amplitude=.4mm,segment length=2mm, post length=0mm,pre length=0mm}}
			
			\def\sx{0.75};
			\def\dx{0.8};
			\def\dy{0.6};
			
			\draw[thick,red] (0.0,0.0)--++(0.75*\dx,0.0);
			\draw[thick,red] (0.0,\dy)--++(0.75*\dx,0.0);
			\draw[thick] (0.0,0.0)--++(0.0,\dy);
			\draw[thick] (0.0,\dy)--++(0.0,\dy);
			\draw[thick,decorate] (0.0,2*\dy)--++(0.0,\dy);
			\draw[thick,red] (0.0,2*\dy)--++(1.5*\dx,0.0);
			\draw[thick,red] (1.5*\dx,2*\dy)--++(0.0,-1.5*\dy);
			\draw[thick,red] (1.5*\dx,0.5*\dy)--++(-0.75*\dx,0.0);
			\draw[thick] (0.0,0.0)--++(0.0,-3*\dy);
			\draw[thick,red] (0.0,\dy)--++(-2*\dx,0.0);
			\draw[thick] (-2*\dx,\dy)--++(0.0,2*\dy);
			\draw[thick] (-2*\dx,\dy)--++(0.0,-\dy);
			\draw[thick] (-2*\dx,0.0)--++(0.0,-2*\dy);
			\draw[thick,red] (-2*\dx,0.0) --++ (-0.5*\dx,0.0);
			\draw[thick,red] (-2*\dx,\dy)--++(-0.5*\dx,0.0);
			\draw[thick,red] (-2.5*\dx,0.0)--++(0.0,\dy);
			\draw[thick] (-2*\dx,-2*\dy)--++(0.0,-2*\dy);
			\draw[thick,red] (-3*\dx,-\dy)--++(1.5*\dx,0.0);
			\draw[thick,red] (-2*\dx,-2*\dy)--++(.5*\dx,0.0);
			\draw[thick,red] (-2*\dx,0.0)--++(1.0*\dx,0.0);
			\draw[thick,red] (-0.5*\dx,0.0)--++(0.0,-1.5*\dy);
			\draw[thick,red] (-1.5*\dx,-1.5*\dy)--++(0.5*\dx,0.0);
			\draw[thick] (-3*\dx,-\dy)--++(0.0,4*\dy);
			\draw[thick] (-3*\dx,-\dy)--++(0.0,-6*\dy);
			
			\draw[thick] (0.0,-3*\dy)--++(0.0,-5*\dy);
			\draw[thick, decorate] (0.0,-8*\dy)--++(0.0,-\dy);
			\draw[thick] (-3*\dx,-7*\dy)--++(0.0,-2*\dy);
			\draw[thick] (-2*\dx,-4*\dy)--++(0.0,-5*\dy);
			
			\draw[thick,red] (0.0,-3*\dy)--++(0.75*\dx,0.0);
			\draw[thick,red] (0.0,-4*\dy)--++(0.75*\dx,0.0);
			\draw[thick,red] (0.75*\dx,-3.5*\dy)--++(0.75*\dx,0.0);
			\draw[thick,red] (0.0,-4*\dy)--++(-2*\dx,0.0);
			\draw[thick,brightube] (1.5*\dx,-3.5*\dy)--++(0.0,-1.5*\dy);
			\draw[thick,red] (-0.5*\dx,-5*\dy)--++(2*\dx,0.0);
			\draw[thick,red] (-0.5*\dx,-5*\dy)--++(0.0,-\dy);
			\draw[thick,red] (-0.5*\dx,-6*\dy)--++(1.25*\dx,0.0);
			\draw[thick,red] (-3*\dx,-7*\dy)--++(3.75*\dx,0.0);
			\draw[thick,red] (0.75*\dx,-6.5*\dy)--++(0.75*\dx,0.0);
			\draw[thick,brightube] (1.5*\dx,-6.5*\dy)--++(0.0,-1.5*\dy);
			\draw[thick,red] (1.5*\dx,-8*\dy)--++(-1.5*\dx,0.0);
			\draw[thick,red] (-3*\dx,-\dy)--++(-0.5*\dx,0.0);
			\draw[thick,red] (-3*\dx,-7*\dy)--++(-0.5*\dx,0.0);
			%    \draw[thick,red] (-3.5*\dx,-\dy)--++(0.0,-6*\dy);
			\draw[thick,red] (-2*\dx,-2*\dy)--++(-0.5*\dx,0.0);
			\draw[thick,red] (-2*\dx,-4*\dy)--++(-0.5*\dx,0.0);
			\draw[thick,red] (-2.5*\dx,-2*\dy)--++(0.0,-2*\dy);
			\draw[thick,red] (0.0,0.0)--++(-0.5*\dx,0.0);
			\draw[thick,red] (0.0,-3*\dy)--++(-0.5*\dx,0.0);
			
			\draw[fusion] (0.75*\dx,0.0)--++(0.0,\dy);
			\draw[fusion] (-1.5*\dx,-2*\dy)--++(0.0,\dy);
			\draw[fusion_dagger,brightube] (0.75*\dx,-3*\dy)--++(0.0,-\dy);
			\draw[fusion_dagger,brightube] (0.75*\dx,-6*\dy)--++(0.0,-\dy);
			\pic () at (1.5*\dx,0.5*\dy) {fusR=1.5*\dy///};
			\pic () at (1.5*\dx,-5*\dy) {fusRdag=1.5*\dy///};
			\pic () at (1.5*\dx,-8*\dy) {fusRdag=1.5*\dy///};
			\pic () at (-0.5*\dx,-6*\dy) {fusLdag=\dy///};
			\pic () at (-\dx,-1.5*\dy) {fusR=1.5*\dy///};
			\pic () at (-2.5*\dx,0) {fusL=\dy///};
			\pic () at (-0.5*\dx,-3*\dy) {fusL=3*\dy///};
			\pic () at (-2.5*\dx,-4*\dy) {fusL=2*\dy///};
			\pic () at (-3.5*\dx,-7*\dy) {fusL=6*\dy///};

			\node[tensor] at (0.0,2*\dy) {};
			\node[mpo] at (0.0,\dy) {};
			\node[mpo] at (0.0,0.0) {};
			\node[mpo] at (-2*\dx,0.0) {};
			\node[mpo] at (-2*\dx,\dy) {};
			\node[mpo] at (-2*\dx,-\dy) {};
			\node[mpo] at (-2*\dx,-2*\dy) {};
			\node[mpo] at (-3*\dx,-\dy) {};
			\node[mpo] at (0.0,-3*\dy) {};
			\node[mpo] at (0.0,-4*\dy) {};
			\node[mpo] at (0.0,-5*\dy) {};
			\node[mpo] at (0.0,-6*\dy) {};
			\node[mpo] at (0.0,-7*\dy) {};
			\node[tensor] at (0.0,-8*\dy) {};
			\node[mpo] at (-3*\dx,-7*\dy) {};
			\node[mpo] at (-2*\dx,-4*\dy) {};
			\node[mpo] at (-2*\dx,-7*\dy) {};
			\node[tensor,red] at (0.4*\dx,-3*\dy) {};
			\node[tensor,red] at (-1.5*\dx,-4*\dy) {};
			\node[tensor,red] at (-2.5*\dx,-7*\dy) {};

			\node[] at (-4.5*\dx,-3.5*\dy) {$\mathcal{F}^2_{g,h,k}$};
			
			\node[] at (0.25*\dx,-6.75*\dy) {$\dagger$};
			\node[] at (0.25*\dx,-5.75*\dy) {$\dagger$};
			\node[] at (0.25*\dx,-4.75*\dy) {$\dagger$};
			\node[] at (0.25*\dx,-3.75*\dy) {$\dagger$};
			\node[] at (-0.5*\dx,-3.75*\dy) {$h$};
			\node[] at (0.95*\dx,-4.75*\dy) {$\inv{(hk)}$};
			\node[] at (0.7*\dx,-5.6*\dy) {$hk$};
			\node[] at (0.4*\dx,-7.3*\dy) {$g$};
			\node[] at (1.3*\dx,-8.5*\dy) {$\bigl(Y_2^{\inv{(ghk)}}\bigr)^\dagger$};
			\node[] at (\dx,2.4*\dy) {$Y_2^{\inv{(ghk)}}$};
			\node[] at (-3.5*\dx,-7.3*\dy) {$\inv{g}$};
			\node[] at (-3.5*\dx,-0.75*\dy) {$g$};
			\node[] at (-2.5*\dx,-4.25*\dy) {$\inv{h}$};
			\node[red] at (0.45*\dx,-2.6*\dy) {$T_{k}$};
			\node[red] at (-1.35*\dx,-3.6*\dy) {$T_{h}$};
			\node[red] at (-2.5*\dx,-6.6*\dy) {$T_g$};
			\node[] at (-2.65*\dx, 1.25*\dy) {$gh$};
			
		\end{tikzpicture}
	\end{equation}
	with $\mathcal{F}^2_{g,h,k}=\sigma_g\sigma_h\sigma_k\sigma_{gh}\sigma_{hk}$. Next, we use the relations between the fusion tensors and their adjoints, cf. Proposition~\ref{prop:zetas}, as well as between the MPO tensors and their adjoints. As a result, the expression in \eqref{eq:Fghk} can be presented in the following form
	\begin{equation}
		\begin{tikzpicture}
			\tikzset{decoration={snake,amplitude=.4mm,segment length=2mm, post length=0mm,pre length=0mm}}
			
			\def\sx{0.75};
			\def\dx{0.8};
			\def\dy{0.6};
			
			\draw[thick,red] (0.0,0.0)--++(0.75*\dx,0.0);
			\draw[thick,red] (0.0,\dy)--++(0.75*\dx,0.0);
			\draw[thick] (0.0,0.0)--++(0.0,\dy);
			\draw[thick] (0.0,\dy)--++(0.0,\dy);
			\draw[thick,decorate] (0.0,2*\dy)--++(0.0,\dy);
			\draw[thick,red] (0.0,2*\dy)--++(1.5*\dx,0.0);
			\draw[thick,red] (1.5*\dx,2*\dy)--++(0.0,-1.5*\dy);
			\draw[thick,red] (1.5*\dx,0.5*\dy)--++(-0.75*\dx,0.0);
			\draw[thick] (0.0,0.0)--++(0.0,-3*\dy);
			\draw[thick,red] (0.0,\dy)--++(-2*\dx,0.0);
			\draw[thick] (-2*\dx,\dy)--++(0.0,2*\dy);
			\draw[thick] (-2*\dx,\dy)--++(0.0,-\dy);
			\draw[thick] (-2*\dx,0.0)--++(0.0,-2*\dy);
			\draw[thick,red] (-2*\dx,0.0) --++ (-0.5*\dx,0.0);
			\draw[thick,red] (-2*\dx,\dy)--++(-0.5*\dx,0.0);
			\draw[thick,red] (-2.5*\dx,0.0)--++(0.0,\dy);
			\draw[thick] (-2*\dx,-2*\dy)--++(0.0,-2*\dy);
			\draw[thick,red] (-3*\dx,-\dy)--++(1.5*\dx,0.0);
			\draw[thick,red] (-2*\dx,-2*\dy)--++(0.5*\dx,0.0);
			\draw[thick,red] (-2*\dx,0.0)--++(1.*\dx,0.0);
			\draw[thick,red] (-0.5*\dx,0.0)--++(0.0,-1.5*\dy);
			\draw[thick,red] (-1.5*\dx,-1.5*\dy)--++(0.5*\dx,0.0);
			\draw[thick] (-3*\dx,-\dy)--++(0.0,4*\dy);
			\draw[thick] (-3*\dx,-\dy)--++(0.0,-6*\dy);
			
			\draw[thick] (0.0,-3*\dy)--++(0.0,-5*\dy);
			\draw[thick, decorate] (0.0,-8*\dy)--++(0.0,-\dy);
			\draw[thick] (-3*\dx,-7*\dy)--++(0.0,-2*\dy);
			\draw[thick] (-2*\dx,-4*\dy)--++(0.0,-5*\dy);
			
			\draw[thick,red] (0.0,-3*\dy)--++(0.75*\dx,0.0);
			\draw[thick,red] (0.0,-4*\dy)--++(0.75*\dx,0.0);
			\draw[thick,red] (0.75*\dx,-3.5*\dy)--++(0.75*\dx,0.0);
			\draw[thick,red] (0.0,-4*\dy)--++(-2*\dx,0.0);
			\draw[thick,red] (1.5*\dx,-3.5*\dy)--++(0.0,-1.5*\dy);
			\draw[thick,red] (-0.5*\dx,-5*\dy)--++(2*\dx,0.0);
			\draw[thick,red] (-0.5*\dx,-5*\dy)--++(0.0,-\dy);
			\draw[thick,red] (-0.5*\dx,-6*\dy)--++(1.25*\dx,0.0);
			\draw[thick,red] (-3*\dx,-7*\dy)--++(3.75*\dx,0.0);
			\draw[thick,red] (0.75*\dx,-6.5*\dy)--++(0.75*\dx,0.0);
			\draw[thick,red] (1.5*\dx,-6.5*\dy)--++(0.0,-1.5*\dy);
			\draw[thick,red] (1.5*\dx,-8*\dy)--++(-1.5*\dx,0.0);
			\draw[thick,red] (-3*\dx,-\dy)--++(-0.5*\dx,0.0);
			\draw[thick,red] (-3*\dx,-7*\dy)--++(-0.5*\dx,0.0);
			\draw[thick,red] (-3.5*\dx,-\dy)--++(0.0,-6*\dy);
			\draw[thick,red] (-2*\dx,-2*\dy)--++(-0.5*\dx,0.0);
			\draw[thick,red] (-2*\dx,-4*\dy)--++(-0.5*\dx,0.0);
			\draw[thick,red] (-2.5*\dx,-2*\dy)--++(0.0,-2*\dy);
			\draw[thick,red] (0.0,0.0)--++(-0.5*\dx,0.0);
			\draw[thick,red] (0.0,-3*\dy)--++(-0.5*\dx,0.0);
			\draw[thick,red] (-0.5*\dx,0.0)--++(0.0,-3*\dy);
			
			\draw[fusion] (0.75*\dx,0.0)--++(0.0,\dy);
			\draw[fusion] (-1.5*\dx,-2*\dy)--++(0.0,\dy);
			\draw[fusion] (0.75*\dx,-3*\dy)--++(0.0,-\dy);
			\draw[fusion] (0.75*\dx,-6*\dy)--++(0.0,-\dy);
			\pic () at (1.5*\dx,0.5*\dy) {fusR=1.5*\dy///};
			\pic () at (1.5*\dx,-5*\dy) {fusR=1.5*\dy///};
			\pic () at (1.5*\dx,-8*\dy) {fusR=1.5*\dy///};
			\pic () at (-0.5*\dx,-6*\dy) {fusL=\dy///};
			\pic () at (-\dx,-1.5*\dy) {fusR=1.5*\dy///};
			\pic () at (-2.5*\dx,0) {fusL=\dy///};
			\pic () at (-0.5*\dx,-3*\dy) {fusL=3*\dy///};
			\pic () at (-2.5*\dx,-4*\dy) {fusL=2*\dy///};
			\pic () at (-3.5*\dx,-7*\dy) {fusL=6*\dy///};
			
			\node[tensor] at (0.0,2*\dy) {};
			\node[mpo] at (0.0,\dy) {};
			\node[mpo] at (0.0,0.0) {};
			\node[mpo] at (-2*\dx,0.0) {};
			\node[mpo] at (-2*\dx,\dy) {};
			\node[mpo] at (-2*\dx,-\dy) {};
			\node[mpo] at (-2*\dx,-2*\dy) {};
			\node[mpo] at (-3*\dx,-\dy) {};
			\node[mpo] at (0.0,-3*\dy) {};
			\node[mpo] at (0.0,-4*\dy) {};
			\node[mpo] at (0.0,-5*\dy) {};
			\node[mpo] at (0.0,-6*\dy) {};
			\node[mpo] at (0.0,-7*\dy) {};
			\node[tensor, fill=white] at (0.0,-8*\dy) {};
			\node[mpo] at (-3*\dx,-7*\dy) {};
			\node[mpo] at (-2*\dx,-4*\dy) {};
			\node[mpo] at (-2*\dx,-7*\dy) {};
			
			\node[] at (-4.5*\dx,-3.5*\dy) {$\mathcal{F}^3_{g,h,k}$};
			
			\node[] at (0.95*\dx,-4.75*\dy) {$hk$};
			\node[] at (0.7*\dx,-5.6*\dy) {$\inv{(hk)}$};
			\node[] at (0.8*\dx,-8.35*\dy) {$X_1^{ghk}$};
			\node[] at (\dx,2.4*\dy) {$Y_2^{\inv{(ghk)}}$};
			\node[] at (-3.5*\dx,-7.3*\dy) {$\inv{g}$};
			\node[] at (-3.5*\dx,-0.75*\dy) {$g$};
			\node[] at (-2.5*\dx,-4.25*\dy) {$\inv{h}$};
			\node[] at (-2.65*\dx, 1.25*\dy) {$gh$};        
		\end{tikzpicture}
	\end{equation}
	with
	\begin{align}
		\mathcal{F}^3_{g,h,k}\!\!&=\!\!\frac{\sigma_g\sigma_h\sigma_k\sigma_{gh}\sigma_{hk}}{\zeta_{\inv{k},\inv{h}}\zeta_{hk,\inv{(hk)}}\zeta_{\inv{(hk)},hk}\zeta_{\inv{(hk)},\inv{g}}\zeta_{\inv{(ghk)},ghk}}\nonumber\\ &=\frac{\sigma_g\sigma_h\sigma_k\sigma_{gh}\sigma_{hk}\sigma_{ghk}}{\zeta_{\inv{k},\inv{h}}\zeta_{\inv{(hk)},\inv{g}}},\phantom{\dfrac{\Big|}{1}}
	\end{align}
	where we have used Eq.~\eqref{eq:zetas_in_gauge}. Using the properties of the fusion operators, we can write the latter diagram as
	\begin{equation}
		\begin{tikzpicture}
			\tikzset{decoration={snake,amplitude=.4mm,segment length=2mm, post length=0mm,pre length=0mm}}
			
			\def\sx{0.75};
			\def\dx{0.8};
			\def\dy{0.6};
			
			\draw[thick,decorate] (0.0,0.0)--++(0.0,\dy);
			\draw[thick] (0.0,\dy)--++(0.0,10*\dy);
			\draw[thick,decorate] (0.0,11*\dy)--++(0.0,\dy);
			\draw[thick,red] (0.0,\dy)--++(1.5*\dx,0.0);
			\draw[thick,red] (1.5*\dx,2.5*\dy)--++(-0.75*\dx,0.0);
			\draw[thick,red] (0.75*\dx,2*\dy)--++(-6.25*\dx,0.0);
			\draw[thick,red] (0.75*\dx,3*\dy)--++(-3.75*\dx,0.0);
			\draw[thick] (-2*\dx,0.0)--++(0.0,12*\dy);
			\draw[thick] (-4*\dx,0.0)--++(0.0,12*\dy);
			\draw[thick,red] (-3*\dx,4*\dy)--++(2*\dx,0.0);
			\draw[thick,red] (-\dx,5*\dy)--++(-2*\dx,0.0);
			\draw[thick,red] (-3*\dx,4.5*\dy)--++(-2*\dx,0.0);
			\draw[thick,red] (-3*\dx,5.5*\dy)--++(-1.5*\dx,0.0);
			\draw[thick,red] (-4.5*\dx,6*\dy)--++(5*\dx,0.0);
			\draw[thick,red] (-5*\dx,7*\dy)--++(2*\dx,0.0);
			\draw[thick,red] (-5.5*\dx,8*\dy)--++(2.5*\dx,0.0);
			\draw[thick,red] (-3*\dx,7.5*\dy)--++(2*\dx,0.0);
			\draw[thick,red] (-\dx,9*\dy)--++(-2*\dx,0.0);
			\draw[thick,red] (-3*\dx,10*\dy)--++(3.5*\dx,0.0);
			\draw[thick,red] (0.5*\dx,8*\dy)--++(0.5*\dx,0.0);
			\draw[thick,red] (0.0,11*\dy)--++(\dx,0.0);
			
			\draw[fusion] (1.5*\dx,\dy)--++(0.0,1.5*\dy);
			\draw[fusion] (-5.5*\dx,2*\dy)--++(0.0,6*\dy);
			\draw[fusion] (-3*\dx,3*\dy)--++(0.0,\dy);
			\draw[fusion] (-\dx,4*\dy)--++(0.0,\dy);
			\draw[fusion] (-4.5*\dx,5.5*\dy)--++(0.0,0.5*\dy);
			\draw[fusion] (-5*\dx,4.5*\dy)--++(0.0,2.5*\dy);
			\draw[fusion] (-\dx,7.5*\dy)--++(0.0,1.5*\dy);
			\draw[fusion] (-3*\dx,9*\dy)--++(0.0,\dy);
			\draw[fusion] (\dx,8*\dy)--++(0.0,3*\dy);
			\draw[fusion] (0.75*\dx,2*\dy)--++(0.0,\dy);
			\draw[fusion] (-3*\dx,4.5*\dy)--++(0.0,\dy);
			\draw[fusion] (-3*\dx, 7*\dy)--++(0.0,\dy);
			\draw[fusion] (0.5*\dx,6*\dy)--++(0.0,4*\dy);
			
			\node[tensor, fill=white] at (0.0,\dy) {};
			\node[mpo] at (0.0,2*\dy) {};
			\node[mpo] at (0.0,3*\dy) {};
			\node[mpo] at (0.0,6*\dy) {};
			\node[mpo] at (0.0,10*\dy) {};
			\node[tensor] at (0.0,11*\dy) {};
			\node[mpo] at (-2*\dx,2*\dy) {};
			\node[mpo] at (-2*\dx,3*\dy) {};
			\node[mpo] at (-2*\dx,4*\dy) {};
			\node[mpo] at (-2*\dx,5*\dy) {};
			\node[mpo] at (-2*\dx,6*\dy) {};
			\node[mpo] at (-2*\dx,7.5*\dy) {};
			\node[mpo] at (-2*\dx,9*\dy) {};
			\node[mpo] at (-2*\dx,10*\dy) {};
			\node[mpo] at (-4*\dx,2*\dy) {};
			\node[mpo] at (-4*\dx,4.5*\dy) {};
			\node[mpo] at (-4*\dx,5.5*\dy) {};
			\node[mpo] at (-4*\dx,6*\dy) {};
			\node[mpo] at (-4*\dx,7*\dy) {};
			\node[mpo] at (-4*\dx,8*\dy) {};
			
			\node[] at (-6.5*\dx, 6*\dy) {$\mathcal{F}^3_{g,h,k}$};
			\node[] at (-0.8*\dx,0.8*\dy) {$X_1^{ghk}$};
			\node[] at (-0.9*\dx,11.2*\dy) {$Y_2^{\inv{(ghk)}}$};
			\node[] at (-5*\dx,2.35*\dy) {$\inv{g}$};
			\node[] at (-5*\dx,8.35*\dy) {$g$};
			\node[] at (-3.5*\dx,4.25*\dy) {$\inv{h}$};
			\node[] at (-3.5*\dx,5.25*\dy) {$\inv{k}$};
			\node[] at (-3.5*\dx,6.25*\dy) {$k$};
			\node[] at (-3.5*\dx,7.25*\dy) {$h$};
			\node[] at (-2.75*\dx,10.25*\dy) {$gh$};
			\node[] at (-2.75*\dx,8.65*\dy) {$\inv{(gh)}$};
			\node[] at (-1.3*\dx,5.3*\dy) {$\inv{(hk)}$};
			\node[] at (-1.4*\dx,4.25*\dy) {$hk$};
			\node[] at (-1.25*\dx,3.25*\dy) {$\inv{(hk)}$};
			\node[] at (1.5*\dx,6*\dy) {$.$};
			
		\end{tikzpicture}
		\label{eq:big_diagramm_4}
	\end{equation}
	Using Eq.~\eqref{eq:sigmadiag}, the diagram in \eqref{eq:big_diagramm_4} can be reduced into
	\begin{equation}
		\begin{tikzpicture}
			\tikzset{decoration={snake,amplitude=.4mm,segment length=2mm, post length=0mm,pre length=0mm}}
			
			\def\sx{0.75};
			\def\dx{0.6};
			\def\dy{0.6};
			
			\draw[thick,decorate] (0.0,0.0)--++(0.0,-\dy);
			\draw[thick] (0.0,0.0)--++(0.0,9*\dy);
			\draw[thick,decorate] (0.0,9*\dy)--++(0.0,\dy);
			
			\draw[thick] (-\dx,-\dy)--++(0.0,11*\dy);
			
			\draw[thick] (-3*\dx,-\dy)--++(0.0,11*\dy);
			
			\draw[thick,red] (0.0,0.0)--++(2*\dx,0.0);
			\draw[thick,red] (-6*\dx,\dy)--++(7*\dx,0.0);
			\draw[thick,red] (\dx,2.5*\dy)--++(\dx,0.0);
			\draw[thick,red] (-5*\dx,3*\dy)--++(3*\dx,0.0);
			\draw[thick,red] (-4*\dx,5*\dy)--++(2*\dx,0.0);
			\draw[thick,red] (-2*\dx,4*\dy)--++(3*\dx,0.0);
			\draw[thick,red] (-4*\dx,6*\dy)--++(5*\dx,0.0);
			\draw[thick,red] (-2*\dx,8*\dy)--++(3*\dx,0.0);
			\draw[thick,red] (-5*\dx,7*\dy)--++(3*\dx,0.0);
			\draw[thick,red] (-6*\dx,9*\dy)--++(4*\dx,0.0);
			\draw[thick,red] (0.0,9*\dy)--++(2*\dx,0.0);
			\draw[thick,red] (\dx,7*\dy)--++(\dx,0.0);
			
			\draw[fusion] (\dx,\dy)--++(0.0,3*\dy);
			\draw[fusion] (2*\dx,0.0)--++(0.0,2.5*\dy);
			\draw[fusion] (-2*\dx,3*\dy)--++(0.0,2*\dy);
			\draw[fusion] (-4*\dx,5*\dy)--++(0.0,\dy);
			\draw[fusion] (-5*\dx,3*\dy)--++(0.0,4*\dy);
			\draw[fusion] (-6*\dx,\dy)--++(0.0,8*\dy);
			\draw[fusion] (-2*\dx,7*\dy)--++(0.0,2*\dy);
			\draw[fusion] (\dx,6*\dy)--++(0.0,2*\dy);
			\draw[fusion] (2*\dx,7*\dy)--++(0.0,2*\dy);
			
			\node[tensor, fill=white] at (0.0,0.0) {};
			\node[mpo] at (0.0,\dy) {};
			\node[mpo] at (0.0,4*\dy) {};
			\node[mpo] at (0.0,6*\dy) {};
			\node[mpo] at (0.0,8*\dy) {};
			\node[tensor] at (0.0,9*\dy) {};
			
			\node[mpo] at (-\dx,\dy) {};
			\node[mpo] at (-\dx,4*\dy) {};
			\node[mpo] at (-\dx,6*\dy) {};
			\node[mpo] at (-\dx,8*\dy) {};
			
			\node[mpo] at (-3*\dx,\dy) {};
			\node[mpo] at (-3*\dx, 3*\dy) {};
			\node[mpo] at (-3*\dx, 5*\dy) {};
			\node[mpo] at (-3*\dx, 6*\dy) {};
			\node[mpo] at (-3*\dx, 7*\dy) {};
			\node[mpo] at (-3*\dx, 9*\dy) {};
			
			\node[] at (-5.5*\dx, 1.4*\dy) {$\inv{g}$};
			\node[] at (0.5*\dx, 1.4*\dy) {$\inv{g}$};
			\node[] at (0.65*\dx,-0.35*\dy) {$X_1^{ghk}$};
			\node[] at (\dx, 9.4*\dy) {$Y_2^{\inv{(ghk)}}$};
			\node[] at (0.45*\dx, 6.3*\dy) {$k$};
			\node[] at (-2.4*\dx, 5.3*\dy) {$\inv{k}$};
			\node[] at (-2.4*\dx, 3.3*\dy) {$\inv{k}$};
			\node[] at (0.7*\dx, 4.4*\dy) {$\inv{(hk)}$};
			\node[] at (-1.5*\dx, 8.35*\dy) {$gh$};
			
			\node[] at (-7.5*\dx,5.5*\dy) {$\mathcal{F}^4_{g,h,k}$};
			
		\end{tikzpicture}
	\end{equation}
	where
	\begin{equation}
		\mathcal{F}^4_{g,h,k}=\frac{\sigma_g \sigma_h \sigma_k\sigma_{ghk}}{\zeta_{k^{-1},h^{-1}}\zeta_{(hk)^{-1},g^{-1}}}.
	\end{equation}
	It can be equivalently written as
	\begin{equation}
		\begin{tikzpicture}
			\tikzset{decoration={snake,amplitude=.4mm,segment length=2mm, post length=0mm,pre length=0mm}}
			
			\def\sx{0.75};
			\def\dx{0.6};
			\def\dy{0.6};
			
			\draw[thick, decorate] (0.0,-\dy)--++(0.0,\dy);
			\draw[thick] (0.0,0.0)--++(0.0,10*\dy);
			\draw[thick,decorate] (0.0,10*\dy)--++(0.0,\dy);
			\draw[thick] (-\dx,-\dy)--++(0.0,12*\dy);
			\draw[thick] (-2*\dx,-\dy)--++(0.0,12*\dy);
			\draw[thick,red] (0.0,0.0)--++(3*\dx,0.0);
			\draw[thick,red] (2*\dx,2.5*\dy)--++(\dx,0.0);
			\draw[thick,red] (-5*\dx,\dy)--++(7*\dx,0.0);
			\draw[thick,red] (-4*\dx,3*\dy)--++(5*\dx,0.0);
			\draw[thick,red] (\dx,4*\dy)--++(\dx,0.0);
			\draw[thick, red] (-3*\dx,5*\dy)--++(4*\dx,0.0);
			\draw[thick, red] (-3*\dx,6*\dy)--++(5*\dx,0.0);
			\draw[thick, red] (\dx,8*\dy)--++(\dx,0.0);
			\draw[thick,red] (2*\dx,7*\dy)--++(\dx,0.0);
			\draw[thick,red] (0.0,10*\dy)--++(3*\dx,0.0);
			\draw[thick,red] (-5*\dx,9*\dy)--++(6*\dx,0.0);
			\draw[thick,red] (-4*\dx,7*\dy)--++(5*\dx,0.0);
			
			\node[tensor, fill=white] at (0.0,0.0) {};  
			\node[tensor] at (0.0,10*\dy) {};
			
			\node[mpo] at (0.0,\dy) {};
			\node[mpo] at (0.0,3*\dy) {};
			\node[mpo] at (0.0,5*\dy) {};
			\node[mpo] at (0.0,6*\dy) {};
			\node[mpo] at (0.0,7*\dy) {};
			\node[mpo] at (0.0,9*\dy) {};
			
			\node[mpo] at (-\dx,\dy) {};
			\node[mpo] at (-\dx, 3*\dy) {};
			\node[mpo] at (-\dx,5*\dy) {};
			\node[mpo] at (-\dx,6*\dy) {};
			\node[mpo] at (-\dx,7*\dy) {};
			\node[mpo] at (-\dx,9*\dy) {};
			
			\node[mpo] at (-2*\dx,\dy) {};
			\node[mpo] at (-2*\dx, 3*\dy) {};
			\node[mpo] at (-2*\dx,5*\dy) {};
			\node[mpo] at (-2*\dx,6*\dy) {};
			\node[mpo] at (-2*\dx,7*\dy) {};
			\node[mpo] at (-2*\dx,9*\dy) {}; 
			
			\draw[fusion] (3*\dx,0.0)--++(0.0,2.5*\dy);
			\draw[fusion] (2*\dx,\dy)--++(0.0,3*\dy);
			\draw[fusion] (\dx,3*\dy)--++(0.0,2*\dy);
			
			\draw[fusion] (\dx,7*\dy)--++(0.0,2*\dy);
			\draw[fusion] (2*\dx,6*\dy)--++(0.0,2*\dy);
			\draw[fusion] (3*\dx,7*\dy)--++(0.0,3*\dy);
			
			\draw[fusion] (-3*\dx,5*\dy)--++(0.0,\dy);
			\draw[fusion] (-4*\dx,3*\dy)--++(0.0,4*\dy);
			\draw[fusion] (-5*\dx,\dy)--++(0.0,8*\dy);
			
			\node[] at (4*\dx, 5.5*\dy) {$.$};
			\node[] at (-6.5*\dx,5.5*\dy) {$\mathcal{F}^4_{g,h,k}$};
			\node[] at (0.65*\dx,-0.4*\dy) {$X_1^{ghk}$};
			\node[] at (\dx,10.4*\dy) {$Y_2^{\inv{(ghk)}}$};
			\node[] at (-1.5*\dx,9.3*\dy) {$g$};
			\node[] at (-1.5*\dx,7.3*\dy) {$h$};
			\node[] at (-1.5*\dx,6.3*\dy) {$k$};
			\node[] at (-1.4*\dx,5.3*\dy) {$\inv{k}$};
			\node[] at (-1.4*\dx,3.3*\dy) {$\inv{h}$};
			\node[] at (-1.4*\dx,1.3*\dy) {$\inv{g}$};
			\node[] at (1.5*\dx,8.3*\dy) {$gh$};
			\node[] at (1.75*\dx,4.4*\dy) {$\inv{(hk)}$};
		\end{tikzpicture}
	\end{equation}
	From the definition of the cocycle $\omega$, we then get that the above expression takes the following form
	\begin{equation}
		\begin{tikzpicture}
			\tikzset{decoration={snake,amplitude=.4mm,segment length=2mm, post length=0mm,pre length=0mm}}
			
			\def\sx{0.75};
			\def\dx{0.6};
			\def\dy{0.6};
			
			\draw[thick,decorate] (0.0,0.0)--++(0.0,-\dy);
			\draw[thick] (0.0,0.0)--++(0.0,3*\dy);
			\draw[thick,decorate] (0.0,3*\dy)--++(0.0,\dy);
			\draw[thick] (-2*\dx,-\dy)--++(0.0,5*\dy);
			\draw[thick] (-3*\dx,-\dy)--++(0.0,5*\dy);
			\draw[thick,red] (0.0,0.0)--++(\dx,0.0);
			\draw[thick,red] (-\dx,\dy)--++(2*\dx,0.0);
			\draw[thick,red] (-\dx,2*\dy)--++(2*\dx,0.0);
			\draw[thick,red] (0.0,3*\dy)--++(\dx,0.0);
			
			\draw[fusion] (\dx,2*\dy)--++(0.0,\dy);
			\draw[fusion] (\dx,0.0)--++(0.0,\dy);
			\draw[fusion] (-\dx,\dy)--++(0.0,\dy);
			
			\node[tensor, fill=white] at (0.0,0.0) {};
			\node[mpo] at (0.0,\dy) {};
			\node[mpo] at (0.0,2*\dy) {};
			\node[tensor] at (0.0,3*\dy) {};
			
			\node[] at (-0.65*\dx,-0.5*\dy) {$X_1^{ghk}$};
			
			\node[] at (-\dx,3.3*\dy) {$Y_2^{\inv{(ghk)}}$};
			
			\node[] at (-4.5*\dx,1.5*\dy) {$\mathcal{F}^5_{g,h,k}$};
			
			\node[] at (1.5*\dx,1.5*\dy) {$,$};
		\end{tikzpicture}
		\label{eq:last_diag}
	\end{equation}
	with
	\begin{align}
		\mathcal{F}^5_{g,h,k}&=\frac{\mathcal{F}^4_{g,h,k}}{\omega_{g,h,k}}\frac{\omega_{k,\inv{k},\inv{h}}\;\omega_{ghk,\inv{(hk)},\inv{g}}}{\omega_{h,k,\inv{(hk)}}\;\omega_{g,hk,\inv{(hk)}}}\nonumber\\
		&=\frac{\mathcal{F}^4_{g,h,k}}{\omega_{g,h,k}}\dfrac{1}{\hat\Xi_{h,k}\;\hat\Xi_{g, hk}}\nonumber\\
		&=\frac{1}{\omega_{g,h,k}}\dfrac{\sigma_h\sigma_k\sigma_{hk}}{\hat\zeta_{h,k}\,\hat\Xi_{h,k}\;}\dfrac{\sigma_g\sigma_{hk}\sigma_{ghk}}{\hat\zeta_{g,hk}\,\hat\Xi_{g, hk}}\nonumber\\
		&=\frac{1}{\omega_{g,h,k}}
	\end{align}
	where $\Xi$ is defined as in \eqref{eq:introXi}, and in the last step we have used that in our gauge the $\lambda$ operators are unitary, so that, cf.~Eq.~\eqref{eq:trlambda},
	\begin{equation}
		\dfrac{d\sigma}{\hat\zeta\,\hat\Xi}=1.
	\end{equation}
	Finally, using successively \eqref{eq:XYvsMPO} and \eqref{eq:red1}, we find that the last tensor leg in \eqref{eq:last_diag} is just the identity operator. Hence the associator of the two fusion trees is indeed given by a function $F$, whose value is given by the scalar factor we computed,
	\begin{equation}
		F(g,h,k) = \dfrac{1}{\mathcal{F}^5_{g,h,k}}=\omega_{g,h,k}.
	\end{equation}
	
	\section{Characterizations of block independence}
	\label{app:block_indep}
	In this Appendix, we prove two characterizations of the block independence (BI) condition, defined in Section~\ref{subsec:true_gauging_noninj} as the existence of a gauge for which
	\begin{equation}
		\ell_{a,b;g}^x = \ell_{a,b;g}^y\qquad \forall g, a, b\in G, \forall x,y\in \mathsf{X},
	\end{equation}
	where 
	\begin{equation}
		\ell_{a,b;g}^x \equiv\dfrac{L^x_{ag,\inv{g}b}}{L^x_{a,b}}.
	\end{equation}
	The next proposition shows that this is equivalent to a much stronger condition:
	\begin{proposition}
		\label{prop:technical01}
		The following conditions are equivalent: \begin{enumerate}
			\item There exists a choice of gauge such that $L^x_{g,h}=1$, for all $g,h\in G$ and $x\in\mathsf{X}$.
			
			\item There exists a choice of gauge such that $L_{a,b}^x = L_{a,b}^y$, for all $a,b\in G$ and $x,y\in \mathsf{X}$. 
			
			\item There exists a choice of gauge such that $\ell_{a,b;g}^x=1$, for all $a,b,g\in G$ and $x\in \mathsf{X}$.
			
			\item The pair (MPS, MPU) satisfies the BI condition.
		\end{enumerate}
	\end{proposition}
	\begin{proof}
		It is enough to demonstrate that the last condition implies the first one. The existence of a gauge transformation of the form \eqref{eq:Lsymbgauge} that maps every $L$-symbol to 1 requires them to be of the form
		\begin{equation}
			L_{g,h}^x = \beta_{g,h}\dfrac{\gamma^x_{gh}}{\gamma^{hx}_{g}\gamma^x_g},
			\label{eq:Lgauge}
		\end{equation}
		for some $\beta_{g,h}, \gamma^{x}_g$. In the following, we prove that this is the case.
		Assuming $\ell$ is independent of $x$, we can factorize the $(g,h)$ and $x$ dependences 
		\begin{equation}
			L_{g,h}^x = \ell_{g,h;h}L^x_{gh,e}= \dfrac{\ell_{g,h;h}}{\omega(gh,e,e)}L^x_{e,e}
			\label{eq:Lfact}
		\end{equation}
		where we have used \eqref{eq:omega_and_Ls} in the particular case $h=k=e$,
		\begin{equation}
			L^x_{g,e} = \dfrac{L^x_{e,e}}{\omega(g,e,e)}.
			\label{eq:aux1}
		\end{equation}
		Next we will show that $x\mapsto L^x_{e,e}$ is a $G$-homomorphism, that is, there exists a one-dimensional representation $\rho\in\hat G$ such that 
		\begin{equation}
			L^{gx}_{e,e}=\rho(g)L^x_{e,e}.
			\label{eq:defrho}
		\end{equation}
		For that we need another particular case of \eqref{eq:omega_and_Ls}, namely
		\begin{equation}
			L^x_{e, g} = \omega(e, e, g)L^{gx}_{e,e}.
			\label{eq:aux2}
		\end{equation}
		Indeed, combining \eqref{eq:aux1} and \eqref{eq:aux2} yields \eqref{eq:defrho} with
		\begin{equation}
			\rho(g)=\dfrac{\ell_{e,g;g}}{\omega(g, e, e)\omega(e,e,g)}.
		\end{equation}
		Such a $G$-homomorphism can always be written as a gauge transformation of the form \eqref{eq:Lgauge}, we pick, for instance,
		\begin{equation}
			\beta_{g,h}\equiv\dfrac{1}{\rho(h)},\qquad \gamma^x_g\equiv \dfrac{1}{L^{gx}_{e,e}}, 
		\end{equation}
		which yields
		\begin{equation}
			L^x_{e,e} = \beta_{g,h}\dfrac{\gamma^x_{gh}}{\gamma^{hx}_{g}\gamma^x_g}
		\end{equation}
		for any $g,h\in G$. By redefining $\beta_{g,h}$ to include the $(g,h)$-dependent prefactor in \eqref{eq:Lfact}, we have found a gauge transformation that trivializes the $L$-symbols, and the result follows.
	\end{proof}
	Next, we will use the classification of SPT phases under MPO group symmetries \cite[Section~4.1]{GarreLootensMolnar} to give a group cohomological characterization of the BI condition. Consider a block-injective MPS invariant under an MPU representation of a finite group $G$ and, as above, denote $\mathsf{X}$ the set of injective blocks, which are permuted by the transitive action of $G$. Now, choose $x_0\in\mathsf{X}$ and define $H\leq G$ by $H\equiv\text{Stab}(x_0)=\{g\in G|gx_0=x_0\}$. Due to transitivity, we have then an identification $\mathsf{k}:\mathsf{X}\leftrightarrow\frac{G}{H}$ as sets, and we can choose a set of representatives, $\{k_x\}_{x\in \mathsf{X}}\subseteq G$, such that $\mathsf{k}(x)=k_xH$ and $k_xx_0=x$. For convenience, we choose $k_{x_0}=e$. We will need the following
	\begin{lemma}
		\label{lemma:h1h2}
		In the above setting, if the MPU representation is nonanomalous, there exists a gauge of the fusion and action tensors such that $\omega=1$ and 
		\begin{equation}
			L^x_{g_1,g_2}=\psi(h_1,h_2),
			\label{eq:Lcocycle}
		\end{equation}
		where $\psi\in H^2(H, U(1))$ and $h_1, h_2\in H$ are given by
		\begin{equation}
			h_1=k^{-1}_{g_1g_2x}\,g_1\,k_{g_1x}, \quad h_2=k^{-1}_{g_2x}\,g_2\,k_x.
			\label{eq:h1h2}
		\end{equation}
	\end{lemma}
	The pair $(H, \psi)$ then classifies the SPT phase to which the MPS belongs \cite{GarreLootensMolnar}. We now prove the following result that characterizes block independence in terms of $\psi$:
	\begin{proposition}
		\label{prop:BI_psi}
		The set of $L$-symbols arising from an MPS invariant under a nonanomalous MPU group representation is BI if and only if the cocycle $\psi$ extends to $G$, that is, if there exists $\Psi\in H^2(G, U(1))$ such that its restriction to the subgroup $H$ is $\psi$.
	\end{proposition}
	\begin{proof}
		Suppose that there exists $\Psi\in H^2(G,U(1))$ such that $\Psi\left.\right|_{H\times H}=\psi$. For $g_1,g_2\in G$ and $x\in X$, take $h_1,h_2$ as in \eqref{eq:h1h2}. Then, applying the cocycle condition for $\Psi$ a few times, we can prove that
		\begin{equation}
			\psi(h_1, h_2)=\Psi(h_1, h_2)=\Psi(g_1, g_2)\dfrac{\gamma^x_{g_1g_2}}{\gamma^{g_2x}_{g_1}\gamma^x_{g_2}},
		\end{equation}
		where 
		\begin{equation}
			\gamma^x_{g}\equiv\dfrac{\Psi(k^{-1}_{gx}, gk_x)\Psi(g, k_x)}{\Psi(e,e)\Psi(k_x,k_x^{-1})}.
		\end{equation} Thus, the $L$-symbols \eqref{eq:Lcocycle} are of the form \eqref{eq:Lgauge} and can be made trivial by regauging. Conversely, if BI holds, there exist $\beta_{g_1, g_2}$ and $\gamma_g^x$ such that 
		\begin{equation}
			\psi(h_1, h_2) = \beta_{g_1,g_2}\dfrac{\gamma^x_{g_1g_2}}{\gamma^{g_2x}_{g_1}\gamma^x_{g_2}}.
		\end{equation}
		$\beta_{g_1,g_2}$ has to be a 2-cocycle, since the gauge transformation leaves $\omega=1$ invariant. Then, defining
		\begin{equation}
			\Psi(g_1,g_2) \equiv \beta_{g_1,g_2}\dfrac{\gamma^{x_0}_{g_1g_2}}{\gamma^{x_0}_{g_1}\gamma^{x_0}_{g_2}},
		\end{equation}
		we find that $\Psi$ is a 2-cocycle that restricts to $\psi$ on $H$ (if we set $x=x_0$, $g_1,g_2\in H\implies h_1=g_1, h_2=g_2$ since we picked $k_{x_0}=e$) and thus $\psi$ extends to $G$.
	\end{proof}
	
	\bibliographystyle{quantum}
	\bibliography{refs.bib}
	
\end{document}